%% file: main.tex
\newif\ifcomments   
\newif\ifanon       
\newif\ifcrypto     
\newif\ifllncs      
\newif\ifpublish    

\cryptotrue

\ifllncs
  \documentclass[runningheads,a4paper]{llncs}

\else
  \documentclass[11pt,letter]{article}
  \usepackage[in]{fullpage}

\fi

\usepackage{amsmath,amssymb}  
\usepackage{amsthm}
\usepackage{todonotes}
\usepackage{setspace} 
 \usepackage[colorlinks]{hyperref}
\hypersetup{linkcolor=blue,filecolor=blue,citecolor=blue,urlcolor=blue}
\usepackage{cleveref}

\usepackage{wrapfig}
\usepackage{footnote} 
\newcommand{\cblue}[1]{{\color{blue} #1}}

\usepackage{comment}

\usepackage{textcomp}
\usepackage{paralist}

\usepackage{ulem}
\ifllncs
\else
\newtheorem{theorem}{Theorem}
\newtheorem{proposition}[theorem]{Proposition}
\newtheorem{definition}[theorem]{Definition}
\newtheorem{lemma}[theorem]{Lemma}
\newtheorem{claim}[theorem]{Claim}
\newtheorem{remark}[theorem]{Remark}
\newtheorem{corollary}[theorem]{Corollary}
\newtheorem{conjecture}[theorem]{Conjecture}
\fi

\crefname{definition}{Definition}{Definitions}
\crefname{lemma}{Lemma}{Lemmas}
\crefname{proposition}{Proposition}{Propositions}
\crefname{corollary}{Corollary}{Corollaries}
\crefname{conjecture}{Conjecture}{Conjectures}

\input{headers}

\newcommand{\threshold}{{threshold}}
\newcommand{\distrind}{\distr_{\mathsf{ind}}}
\newcommand{\distriden}{\distr_{\mathsf{identical}}}
\newcommand{\distrcor}{\distr_{\mathsf{iden\text{-}bit\text{,}ind\text{-}msg}}}

\newcommand{\distrprod}{{\mathcal{U}}}
\newcommand{\distrid}{{\sf Id}_{{\cal U}}}

\newcommand{\distrcipherind}{\distr_{\mathsf{ind}\text{-}\mathsf{msg}}}
\newcommand{\distrcipherid}{\distr_{\mathsf{identical}\text{-}\mathsf{cipher}}}

\newcommand{\ordistr}{\mathsf{Oracle}\text{-}\distr}
\newcommand{\ordistrind}{\mathsf{Oracle}\text{-}\distr_{\mathsf{ind}}}

\newcommand{\ordistrcor}{\mathsf{Oracle}\text{-}\distr_{\mathsf{correlated}}}

\newcommand{\fclass}{\mathcal{F}}
\newcommand{\gclass}{\mathcal{G}}

\newcommand{\secparam}{\lambda}

\newcommand{\keygen}{\mathsf{KeyGen}}
\newcommand{\qkeygen}{\mathsf{QKeyGen}}

\newcommand{\hybrid}{\mathsf{Hybrid}}

\newcommand{\cA}{\mathcal{A}}

\newcommand{\calU}{\mathcal{U}}

\newcommand{\negl}{\mathsf{negl}}

\newcommand{\poly}{\mathrm{poly}}

\newcommand{\distr}{\mathcal{D}}

\newcommand{\bra}[1]{\langle #1|}
\newcommand{\ket}[1]{|#1\rangle}
\newcommand{\braket}[2]{\langle #1|#2\rangle}
\newcommand{\ketbra}[2]{|#1\rangle\langle #2|}

\newcommand{\adversary}{\mathcal{A}}
\newcommand{\prob}{\mathsf{Pr}}

\newcommand{\Zq}{\mathbb{Z}_q}
\newcommand{\ZQ}{\mathbb{Z}_Q}
\newcommand{\embed}{\mathsf{Embed}}

\newcommand{\obf}{\mathsf{Obf}}
\newcommand{\eval}{\mathsf{Eval}}

\newcommand{\qdec}{\mathrm{Dec}}

\newcommand{\postproc}{\cllz\text{ }{post}\text{-}{processing}}
\newcommand{\PPROC}{\cllz\text{ }\mathsf{Post}\text{-}\mathsf{Process}}
\newcommand{\encproc}{\mathsf{Enc}\mathsf{Post}\mathsf{Process}}
\newcommand{\decproc}{\mathsf{Dec}\mathsf{Post}\mathsf{Process}}

\newcommand{\prg}{\mathsf{PRG}}

\newcommand{\gen}{\mathsf{Gen}}
\newcommand{\gentrig}{\mathsf{Gen}\text{-}\mathsf{Trigger}}
\newcommand{\trig}{\mathsf{trigger}}

\newcommand{\distrc}{\mathcal{D}_{\cktclass}}
\newcommand{\cktclassf}{\cktclass^\fclass}
\newcommand{\keyspacef}{{\keyspace^\fclass}}

\newcommand{\cktclassw}{\mathcal{W}}
\newcommand{\cpdigitalsig}{\mathsf{CP}\text{-}\mathsf{DS}}
\newcommand{\digitalsig}{\mathsf{DS}}
\newcommand{\distrunr}{\mathcal{U}_{\fclass^r}}
\newcommand{\distrunrminus}{\mathcal{U}_{\fclass^{r-1}}}
\newcommand{\distrsf}{\distr^r\text{-}S_f}
\newcommand{\distrsg}{{\distr^{r-1}}\text{-}S_g}
\newcommand{\inductdistr}{\mathsf{Induced}\text{-}\distr^r}
\newcommand{\inductdistrdash}{\mathsf{Induced}\text{-}{\distr^{r-1}}}

\newcommand{\aux}{\mathsf{aux}}

\newcommand{\inpclass}{{\cal X}}
\newcommand{\imgclass}{{\cal I}}

\newcommand{\cllz}{\mathsf{CLLZ}}

\newcommand{\tokengen}{\mathsf{Token}\text{-}{Gen}}

\newcommand{\prf}{\mathsf{PRF}}
\newcommand{\puncture}{\mathsf{Puncture}}
\newcommand{\io}{\mathsf{iO}}
\newcommand{\compile}{\mathsf{Compile}}

\newcommand{\cpppiracy}{{preponed\ security}}

\newcommand{\sde}{\mathrm{single\ }\allowbreak\mathrm{decryptor\ }\allowbreak\mathrm{encryption}}
\newcommand{\SDE}{\mathsf{SDE}}
\newcommand{\cpa}{\mathsf{CPA}}

\usetikzlibrary{calc}

\newcommand{\reduct}{{\cal R}}

\newcommand{\genpuncture}{\mathsf{GenPuncture}}

\newcommand{\POVM}{\mathsf{POVM}}

\newcommand{\ppoly}{\sf{P/poly}}

\newcommand{\presamp}{\ensuremath{\mathsf{preimage}\text{-}\mathsf{samplable}}}
\newcommand{\timp}{\mathsf{TI}}
\newcommand{\atimp}{\mathsf{ATI}}
\newcommand{\satimp}{small\text{-}{range}\text{-}\mathsf{ATI}}

\newcommand{\reduc}{\mathcal{R}}

\newcommand{\upo}{\text{unclonable}\allowbreak\text{ puncturable}\allowbreak\text{  obfuscation}}
\newcommand{\UPO}{\mathsf{UPO}}
\newcommand{\upoexpt}{\mathsf{UPO.Expt}}
\newcommand{\glexpt}{\mathsf{GL.Expt}}

\newif\ifCommentsON
\CommentsONtrue 
\newcommand{\EE}{\mathbb{E}}

\ifcrypto
\newcommand{\pnote}[1]{}
\newcommand{\anote}[1]{}
\else 
\newcommand{\pnote}[1]{{\color{red} P: #1}}
\newcommand{\anote}[1]{{\color{violet} A: #1}}
\fi

\newcommand{\cpexpt}{\mathsf{CP.Expt}}
\newcommand{\indrsdeexpt}{{\sf Ind\text{-}random.SDE.Expt}}
\newcommand{\srchsdeexpt}{{\sf Search.SDE.Expt}}
\newcommand{\selcpasdeexpt}{{\sf SelCPA.SDE.Expt}}
\newcommand{\cpasdeexpt}{{\sf CPA.SDE.Expt}}
\newcommand{\ueexpt}{{\sf UE.Expt}}
\newcommand{\density}[1]{{\cal D}(#1)}
\newcommand{\distrclass}{\mathfrak{D}_X}
\newcommand{\abc}{({\alice,\bob,\charlie})}
\newcommand{\adv}{\mathsf{Adv}}
\newcommand{\UE}{\mathsf{UE}}

\newcommand{\calX}{\mathcal{X}}

\newcommand{\NN}{\mathbb{N}}
\newcommand{\copyprotect}{\mathsf{CopyProtect}}
\newcommand{\expt}{\mathsf{Expt}}
\newcommand{\ch}{\mathsf{Ch}}
\newcommand{\alice}{\mathcal{A}}
\newcommand{\bob}{\mathcal{B}}
\newcommand{\charlie}{\mathcal{C}}
\newcommand{\bfB}{\mathbf{B}}
\newcommand{\bfC}{\mathbf{C}}
\newcommand{\keyspace}{\ensuremath{\mathcal{K}}}
\newcommand{\bitspace}{\ensuremath{\{0,1\}}}

\newcommand{\cp}{\mathsf{CP}}
\newcommand{\iO}{\mathsf{iO}}

\newcommand{\cS}{\mathcal{S}}
\newcommand{\dmx}{{\cal D}}
\newcommand{\TD}{\mathsf{TD}}
\newcommand{\norm}[1]{\| #1 \|}
\newcommand{\Tr}{\mathsf{Tr}}
\newcommand{\N}{\mathbb{N}}

 \usepackage[style=alphabetic,minalphanames=3,maxalphanames=4,maxnames=99,backref=true]{biblatex}

  \DeclareFieldFormat{eprint:iacr}{Cryptology ePrint Archive: \href{https://ia.cr/#1}{\texttt{#1}}}
  \DeclareFieldFormat{eprint:iacrarchive}{Cryptology ePrint Archive: \href{https://eprint.iacr.org/archive/#1}{\texttt{#1}}}
\title{A Modular Approach to Unclonable Cryptography\ifllncs\thanks{{\bf Full version attached under "Optional Supplementary Materials".}}\fi}

\ifllncs
    \author{}
    \institute{} 
\else
    \author{Prabhanjan Ananth\thanks{\texttt{prabhanjan@cs.ucsb.edu}} \\{\small UCSB} \and Amit Behera\thanks{\texttt{behera@post.bgu.ac.il}} \\{\small Ben-Gurion University}} 
\fi

\date{}

\begin{document}

\maketitle




\begin{abstract}
\noindent We explore a new pathway to designing unclonable cryptographic primitives. We propose a new notion called unclonable puncturable obfuscation (UPO) and study its implications for unclonable cryptography. Using UPO, we present modular (and in some cases, arguably, simple) constructions of many primitives in unclonable cryptography, including, public-key quantum money, quantum copy-protection for many classes of functionalities, unclonable encryption, and single-decryption encryption. \par Notably, we obtain the following new results assuming the existence of UPO:  
\begin{itemize}
    \item We show that any cryptographic functionality can be copy-protected as long as this functionality satisfies a notion of security, which we term puncturable security. Prior feasibility results focused on copy-protecting specific cryptographic functionalities. 
    \item We show that copy-protection exists for any class of evasive functions as long as the associated distribution satisfies a preimage-sampleability condition. Prior works demonstrated copy-protection for point functions, which follows as a special case of our result.  
\end{itemize}
We put forward a candidate construction of UPO and prove two notions of security, each based on the existence of (post-quantum) sub-exponentially secure indistinguishability obfuscation and one-way functions, the quantum hardness of learning with errors, and a new conjecture called simultaneous inner product conjecture. 
\end{abstract}

\ifllncs

\else
\newpage 

\begin{spacing}{0.95}
\tableofcontents
\end{spacing} 

\newpage 
\fi

\input{intro}

\ifllncs

\else
\input{prelims}
\fi

\input{upodefinition}
\input{conjectures}

\section*{Part I: Constructions} 
\addcontentsline{toc}{section}{\protect\numberline{}Part I: Constructions of Unclonable Puncturable Obfuscation}
\input{upoconstruction}

\ifllncs

\else

\input{qsio}

\fi

\section*{Part II: Applications}
\addcontentsline{toc}{section}{\protect\numberline{}Part II: Applications}
\input{upoapplications}

\ifllncs

\else 
    \subsection*{Acknowledgements}
    \vspace{-0.5em}
    \BeforeBeginEnvironment{wrapfigure}{\setlength{\intextsep}{0pt}}
            \begin{wrapfigure}{r}{100px}
                \includegraphics[width=100px]{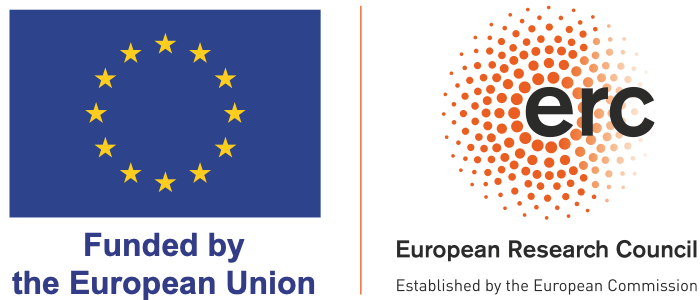}
            \end{wrapfigure}
            A.B. has received funding from the European Union (ERC-2022-COG, ACQUA, 101087742). Views and opinions expressed are, however, those of the author(s) only and do not necessarily reflect those of the European Union or the European Research Council Executive Agency. Neither the European Union nor the granting authority can be held responsible for them.
            \par P.A. is supported in part by the National Science Foundation under Grant No. 2329938 and Grant No. 2341004.
    \par We thank Supartha Podder for discussions during the early stages of the project.
\fi

\renewcommand{\emph}[1]{\textit{#1}}
\printbibliography
\newpage 
\ifllncs
{\Huge Supplementary Material} 
\input{prelims}
\else 
\appendix 
\fi 
\ifllncs
\section{Concurrent and Subsequent Works} 

\paragraph{Concurrent Work.} Concurrent to our work is a recent work by Coladangelo and Gunn~\cite{CG23} who also showed the feasibility of copy-protecting puncturable functionalities and point functions albeit using a completely different approach. At a high level, the themes of the two papers are quite different. Our goal is to identify a central primitive in unclonable cryptography whereas their work focuses on exploring applications of quantum state indistinguishability obfuscation, a notion of indistinguishability obfuscation for quantum computations, to unclonable cryptography. 
\par We discuss the other differences below. 
\begin{itemize}
    \item Unlike our work, which only focuses on \textit{search} puncturing security, their work considers both \textit{search} and \textit{decision} puncturing security. 
    \item The two notions of obfuscation considered in both works seem to be incomparable. While the problem of obfuscating quantum computations has been notoriously challenging, their work considers the (weaker) problem of obfuscating a subclass of quantum computations that are implementations of classical functionalities.
    \item They demonstrate the feasibility of quantum state indistinguishability obfuscation in the quantum oracle model. We demonstrate the feasibility of UPO based on well-studied cryptographic assumptions and a new conjecture. 
\end{itemize}

\paragraph{Subsequent Work.} Subsequent to~\cite{CG23}, we were able to show that the notion of quantum state iO introduced by Colandangelo and Gunn implies UPO, assuming a strong form of unclonable encryption, referred to as leakage-resilient unclonable encryption. We discuss this in~\Cref{sec:cons:qsio}. 
\par Subsequent to~\cite{CG23}, Bartusek, Brakerski and Vaikuntanathan~\cite{BVV24} obtained a construction of quantum state iO in the classical oracle model. 
\fi

\input{ucdefs}

\input{relatedwork}
\input{additionalprelims}

\ifllncs
    \input{composition-theorem}
\else

\fi

\ifllncs 
    \input{qsio}
    \section{Proof of security and correctness for the direct construction}

\input{proofofsde_construct}\input{proofofcp_construct}\input{proofupocons}\input{proofcorrectness}
\else 

\fi

\ifllncs
\section{Applications}

\input{punccrypto}
\input{sde}
\input{evasive}
\else

\fi

\end{document}

%% file: headers.tex
\newcommand{\bfu}{\mathbf{u}}

\newcommand{\cktclass}{\mathfrak{C}}

\newcommand{\enc}{\mathsf{Enc}}
\newcommand{\dec}{\mathsf{Dec}}
\newcommand{\sign}{\mathsf{Sign}}
\newcommand{\verify}{\mathsf{Verify}}
\newcommand{\sk}{\mathsf{sk}}
\newcommand{\pk}{\mathsf{pk}}
\newcommand{\vk}{\mathsf{vk}}
\newcommand{\ct}{\mathsf{ct}}
\newcommand{\sig}{\mathsf{sig}}

\newcommand{\id}{\mathrm{id}}
\newcommand{\tr}{\text{Tr}}

\newcommand{\cH}{\mathcal{H}}
\newcommand{\cC}{\mathcal{C}}

%% file: intro.tex
\anote{
To add in the contributions section of the introduction,
\begin{enumerate}
\item Highlight the connection of the primitives to the fundamental conjectures
\item Mention the new results that are based on the conjecture and are not trivial.
\item Finally, we should compare our CP construction for PRFs compared to CLLZ. The fact that we generally copy protects any PRF.
\item The side result about making a CPA secure sde from selectively secure sde.
\end{enumerate}
}
\anote{ To do in general
\begin{enumerate}
    \item Add the concrete result statement for the applications without the UPO such as which conjecture implies the unclonable encryption, public key quantum money, copy protection of point function, etc.
    \item Try to get a simpler construction of unclonable encryption from UPO rather than going via SDE.
    \item Try to construct UPO from $qsio$ and unclonable encryption.
    \item Try to prove the conjectures.
\end{enumerate}
\newpage 
}

\section{Introduction}
\noindent Unclonable cryptography leverages the no-cloning principle of quantum mechanics~\cite{WZ82,Dieks82} to build many novel cryptographic notions that are otherwise impossible to achieve classically. This has been an active area of interest since the 1980s~\cite{Wiesner83}. In the past few years, researchers have investigated a dizzying variety of unclonable primitives such as quantum money~\cite{AC12,Zha17,Shm22,LMZ23} and its variants~\cite{RS19,BS20,RZ21}, quantum one-time programs~\cite{BGS13}, copy-protection~\cite{Aar09,CLLZ21}, tokenized signatures~\cite{BDS16,CLLZ21}, unclonable encryption~\cite{Got02,BL20} and its variants~\cite{KN23}, secure software leasing~\cite{AL20}, single-decryptor encryption~\cite{GZ20,CLLZ21}, and many more~\cite{BKL23,GMR23,JK23}.
\par Establishing the feasibility of unclonable primitives has been quite challenging. The adversarial structure considered in the unclonability setting (i.e., spatially separated and entangled) is quite different from what we typically encounter in the traditional cryptographic setting. This makes it difficult to leverage traditional classical techniques, commonly used in cryptographic proofs, to argue the security of unclonable primitives. As a result, there are two major gaping holes in the area. 
\begin{itemize}
    \item \underline{\textsc{Unsolved Foundational Questions}}: Despite the explosion of results in the past few years, many fundamental questions in this area remain to be solved. One particular research direction relevant to our work is the design of quantum copy-protection schemes. Quantum copy-protection, first invented by~\cite{Aar05}, is arguably one of the most fundamental primitives of unclonable cryptography besides quantum money. 
    
    \item \underline{\textsc{Lack of Abstractions}}: Due to the lack of good abstractions, proofs in the area of unclonable cryptography tend to be complex and use sophisticated tools, making the literature less accessible to the broader research community. This makes not only verification of proofs difficult but also makes it harder to use the techniques to obtain new feasibility results. 
\end{itemize}


\paragraph{Overarching goal of our work.} We advocate for a modular approach to designing unclonable cryptography. Our goal is to identify an important unclonable cryptographic primitive that would serve as a useful abstraction leading to the design of other unclonable primitives. Ideally, we would like to abstract away all the complex details in the instantiation of this primitive, and it should be relatively easy, even to classical cryptographers, to use this primitive to study unclonability in the context of other cryptographic primitives. We believe that the identification and instantiation of such a  primitive will speed up the progress in the design of unclonable primitives. 
\par Indeed, similar explorations in other contexts, such as classical cryptography, have been fruitful. For instance, the discovery of indistinguishability obfuscation~\cite{BGIRSVY01,GGHRSW13} (iO) revolutionized cryptography and led to the resolution of many open problems (for instance:~\cite{SW14,GGHR14,BZ17,BPR15}). Hence, there is merit to exploring the possibility of such a primitive in unclonable cryptography, as well.
\par Thus, we ask the following question: 
\begin{quote} 
\begin{center}
\textit{Is there an ``iO-like" primitive for unclonable cryptography?}
\end{center}
\end{quote}
We seek the pursuit of identifying unclonable primitives that would have a similar impact on unclonable cryptography as iO did on classical cryptography.

\subsection{Our Contributions in a Nutshell} 
In our search for an \textit{``iO-like"} primitive for unclonable cryptography, we propose a new notion called \textit{unclonable puncturable obfuscation} ($\UPO$) and explore its impact on unclonable cryptography. \\

\noindent \underline{\textsc{New Feasibility Results.}} Specifically, using UPO and other well-studied cryptographic tools, we demonstrate the following new results. 
\begin{itemize}
    \item We show any class of functionalities can be copy-protected as long as they are puncturable (more details in~\Cref{sec:ourresults}). 
    \item We show that a large class of evasive functionalities can be copy-protected. 
\end{itemize}
\noindent \textit{The above two results not only subsume all the copy-protectable functionalities studied in prior works but also capture new functionalities}.
\par Even for functionalities that have been studied before our work, we get qualitatively new results. For instance, our result shows that {\bf any} puncturable digital signature can be copy-protected whereas the work of~\cite{LLQ+22} shows a weaker result that the digital signature of~\cite{SW14} can be copy-protected. We get similar conclusions for copy-protection for pseudorandom functions. \\

\noindent \underline{\textsc{Implication to Unclonable Cryptography.}} Apart from quantum copy-protection, UPO implies many of the foundational unclonable primitives such as public-key quantum money, unclonable encryption, and single-decryptor encryption. \textit{The resulting constructions from UPO are conceptually different compared to the prior works}. Since building unclonable primitives is a daunting task even when relying on exotic computational assumptions, it becomes crucial to venture into alternative approaches. Moreover, this endeavor could potentially yield fresh perspectives on unclonable cryptography.\\

\noindent \underline{\textsc{Simpler Constructions.}} We believe that some of our constructions are simpler than the prior works, albeit the underlying assumptions are incomparable\footnote{We assume UPO whereas the previous works assume post-quantum iO and other well-studied assumptions.}. The construction of copy-protection for puncturable functionalities yields simpler constructions of copy-protection for pseudorandom functions, studied in~\cite{CLLZ21}, and copy-protection for signatures, studied in~\cite{LLQ+22}. \\

\noindent One potential criticism of our work is that our construction of UPO is based on a new conjecture. Specifically, we show that UPO can be based on the existence of post-quantum secure iO, learning with errors and a new conjecture. 
\par However, it is essential to keep in mind the following facts: 
\begin{itemize}
    \item \textsc{Assumptions}: If our conjectures are true, then this would mean that we can construct UPO from indistinguishability obfuscation and other standard assumptions. On the other hand, we currently do not know whether the other direction is true, i.e., whether UPO implies post-quantum indistinguishability obfuscation. As a result, it is plausible that UPO could be a weaker assumption than post-quantum iO! One consequence of this is the construction of public-key quantum money from generic assumptions weaker than post-quantum iO.

    If our conjectures are false, by itself, this does not refute the existence of UPO. \textit{We would like to emphasize that there is no reason to believe these conjectures are necessary for the existence of UPO.} Instead, it merely suggests that we need a different approach to investigate the feasibility of UPO. 
    
    \item \textsc{Pushing the Feasibility Landscape}: Time and time again, in cryptography, we have been forced to invent new assumptions. In numerous instances, these assumptions have unveiled a previously uncharted realm of cryptographic primitives, expanding our understanding beyond what we once deemed feasible. While not all of the computational assumptions have survived the test of time, in some cases\footnote{Several candidates of post-quantum indistinguishability obfuscation had to be broken before a secure candidate was proposed~\cite{JLS21}.}, the insights gained from their cryptanalysis have helped us to come up with more secure instantiations in the future. In a similar vein, being aggressive with exploring new assumptions could push the boundaries of unclonable cryptography.
    
    \anote{We need to complete the sentence here right?}
\end{itemize}

\noindent We also present another construction of UPO from quantum state iO and unclonable encryption. We discuss this more at the end of~\Cref{sec:ourresults}.

\subsection{Our Contributions}
\label{sec:ourresults}
\paragraph{{\bf Definition}.} We discuss our results in more detail. Roughly speaking, unclonable puncturable obfuscation ($\UPO$) defined for a class of circuits $\cktclass$ in ${\sf P}/{\sf Poly}$, consists of two QPT algorithms $(\obf,\eval)$ defined as follows:
\begin{itemize}
    \item \textsc{Obfuscation algorithm}: $\obf$ takes as input a classical circuit $C \in \cktclass$ and outputs a quantum state $\rho_C$. 
    \item \textsc{Evaluation algorithm}: $\eval$ takes as input a quantum state $\rho_C$, an input $x$, and outputs a value $y$.  
\end{itemize}
In terms of correctness, we require $y=C(x)$. To define security, as is typically the case for unclonable primitives, we consider non-local adversaries of the form $(\alice,\bob,\charlie)$. The security experiment, parameterized by a distribution $\distr_{{\cal X}}$, is defined as follows: 
\begin{itemize}
    \item $\alice$ (Alice) receives as input a quantum state $\rho^*$ that is generated as follows. $\alice$ sends a circuit $C$ to the challenger, who then samples a bit $b$ uniformly at random and samples $\left( x^{\bob},x^{\charlie} \right)$ from $ \distr_{{\cal X}}$. If $b=0$, it sets $\rho^*$ to be the output of $\obf$ on input $C$, or if $b=1$, it sets $\rho^*$ to be the output of $\obf$ on $G$, where $G$ is a punctured circuit that has the same functionality as $C$ on all the points except $x^{\bob}$ and $x^{\charlie}$. It is important to note that $\alice$ only receives $\rho^*$ and in particular, $x^{\bob}$ and $x^{\charlie}$ are hidden from $\alice$.  
    \item $\alice$ then creates a bipartite state and shares one part with $\bob$ (Bob) and the other part with $\charlie$ (Charlie). 
    \item $\bob$ and $\charlie$ cannot communicate with each other. In the challenge phase, $\bob$ receives $x^{\bob}$ and $\charlie$ receives $x^{\charlie}$. Then, they each output bits $b_{\bob}$ and $b_{\charlie}$. 
\end{itemize}
$(\alice,\bob,\charlie)$ win if $b_{\bob}=b_{\charlie}=b$. The scheme is secure if they can only win with probability at most $0.5$ (ignoring negligible additive factors).\\

\noindent \textsc{Keyed Circuits.} Towards formalizing the notion of puncturing circuits in a way that will be useful for applications, we consider keyed circuit classes in the above definition. Every circuit in a keyed circuit class is of the form $C_k(\cdot)$ for some key $k$. Any circuit class can be implemented as a keyed circuit class using universal circuits and thus, by considering keyed circuits, we are not compromising on the generality of the above definition. \\

\noindent \textsc{Challenge Distributions.} We could consider different settings of $\distr_{{\cal X}}$. In this work, we focus on two settings. In the first setting (referred to as \textit{independent} challenge distribution), sampling $(x^{\bob},x^{\charlie})$ from $\distr_{{\cal X}}$ is the same as sampling $x^{\bob}$ and $x^{\charlie}$ uniformly at random (from the input space of $C$). In the second setting (referred to as \textit{identical} challenge distribution), sampling $(x^{\bob},x^{\charlie})$ from $\distr_{{\cal X}}$ is the same as sampling $x$ uniformly at random and setting $x=x^{\bob}=x^{\charlie}$. \\

\noindent \textsc{Generalized UPO.} In the above security experiment, we did not quite specify the behavior of the punctured circuit on the points $x^{\bob}$ and $x^{\charlie}$. There are two ways to formalize and this results in two different definitions; we consider both of them in~\Cref{def:upo-definition}. In the first (basic) version, the output of the punctured circuit $G$ on the punctured points is set to be $\bot$. This version would be the regular UPO definition. In the second (generalized) version, we allow $\alice$ to control the output of the punctured circuit on inputs $x^{\bob}$ and $x^{\charlie}$. For instance, $\alice$ can choose and send the circuits $\mu_{\bob}$ and $\mu_{\charlie}$ to the challenger. On input $x^{\bob}$ (resp., $x^{\charlie}$), the challenger programs the punctured circuit $G$ to output $\mu_{\bob}(x^{\bob})$ (resp., $\mu_{\charlie}(x^{\charlie})$). We refer to this version as \textit{generalized UPO}. 

\paragraph{{\bf Applications}.} We demonstrate several applications of $\UPO$ to unclonable cryptography. 
\par We summarise the applications\footnote{We refer the reader unfamiliar with copy-protection, single-decryptor encryption, or unclonable encryption to the introduction section of~\cite{AKL23} for an informal explanation of these primitives.} in~\Cref{fig:intro:implications}. For a broader context of these results, we refer the reader to~\Cref{sec:relatedwork} (Related Work). \\ 

\begin{figure}[!htb]
\hspace*{-3em}
\scalebox{1}{
\begin{tikzpicture}[%
    scale=1,auto,
    block/.style={
      rectangle,
      draw=gray,
      thick,
      fill=blue!5,
      align=center,
      rounded corners,
      minimum height=2em
    },
    block1/.style={
      rectangle,
      draw=gray,
      thick,
      fill=red!5,
      align=center,
      rounded corners,
      minimum height=2em
    },
    line/.style={
      draw,thick,
      -latex',
      shorten >=2pt
    },
    cloud/.style={
      draw=red,
      thick,
      ellipse,
      fill=red!20,
      minimum height=1em
    }
  ]
\coordinate (o) at (0,0) node[above] {{\bf Unclonable Puncturable Obfuscation}};
\draw[->] (0,0) -- node[left] {{\color{blue} \Cref{sec:cp:punc:cryptoscheme}\ \ \ }} (-5,-2) node[block,below] {Copy-Protection for ${\cal S}_{{\sf punc}}$};

\draw[dashed,->] (-5,-2.76) -- (-2,-4) node[block1,below] {Copy-Protection for Signatures};

\draw[->,dashed] (-1.5,-4.76) -- node[left] {} (-1,-6) node[block1,below] {Quantum Money};

\draw[->] (-5,-2.76) -- node[left] {{\color{blue} \Cref{sec:qcp:punc}}} (-6,-5) node[block,below] {Copy-Protection for ${\cal F}_{{\sf punc}}$};

\draw[->,dashed] (-6,-5.76) -- (-6,-7) node[block1,below] {Copy-Protection for PRFs};

\draw[->] (0,0) -- node[left] {{\color{blue} \Cref{sec:qcp:evasive}}} (1.5,-2.5) node[block,below] {Copy-Protection for ${\cal F}_{{\sf evasive}}$};

\draw[dashed,->] (1.5,-3.26) -- (3.5,-7) node[block1,below] {Copy-Protection for Point Functions};

\draw[->] (0,0) -- node[right] {{\color{blue}\ \ \ \ \Cref{sec:pksde:app}}} (7,-2) node[block1,below] {Single-Decryptor Encryption};


\draw[->] (0,0) -- node[xshift=3cm, yshift=-2cm] {{\color{blue} \Cref{sec:simple-construct-uenc-upo}}} (8,-5) node[block,below] {Unclonable Encryption};
\end{tikzpicture}

}
\caption{Applications of Unclonable Puncturable Obfuscation. ${\cal S}_{{\sf punc}}$ denotes cryptographic schemes satisfying puncturable property. ${\cal F}_{{\sf punc}}$ denotes cryptographic functionalities satisfying functionalities satisfying puncturable property. ${\cal F}_{{\sf evasive}}$ denotes functionalities that are evasive with respect to a distribution $\distr$ satisfying preimage-sampleability property. The dashed lines denote corollaries of our main results. The {\color{blue} blue-filled} boxes represent primitives whose feasibility was unknown prior to our work. The {\color{red} red-filled} boxes represent primitives for which we get qualitatively different results or from incomparable assumptions when compared to previous works. \pnote{should we remove copy-protection for ${\cal F}_{{\sf punc}}$?}}
\label{fig:intro:implications}
\end{figure} 

\noindent {\textsc{Copy-Protection for Puncturable Cryptographic Schemes (\Cref{sec:qcp:punc} and~\Cref{sec:cp:punc:cryptoscheme}).} We consider cryptographic schemes satisfying a property called puncturable security. Informally speaking, puncturable security says the following: given a secret key $\sk$, generated using the scheme, it is possible to puncture the key at a couple of points $x^{\bob}$ and $x^{\charlie}$ such that it is computationally infeasible to use the punctured secret key on $x^{\bob}$ and $x^{\charlie}$. We formally define this in~\Cref{sec:cp:punc:cryptoscheme}. 
\par We show the following: 

\begin{theorem}
\label{thm:intro:punccrypt}
Assuming UPO for ${\sf P/poly}$, there exists copy-protection for any puncturable cryptographic scheme. 
\end{theorem}

\noindent Prior works~\cite{CLLZ21,LLQ+22} aimed at copy-protecting specific cryptographic functionalities whereas we, for the first time, characterize a broad class of cryptographic functionalities that can be copy-protected. \\

\noindent As a corollary, we obtain the following results assuming UPO. 
\begin{itemize}

    \item We show that {\bf any} class of puncturable pseudorandom functions that can be punctured at two points~\cite{BW13,BGI14} can be copy-protected. The feasibility result of copy-protecting pseudorandom functions was first established in~\cite{CLLZ21}. A point to note here is that in~\cite{CLLZ21}, given a class of puncturable pseudorandom functions, they transform this into a different class of pseudorandom functions\footnote{Specifically, they add a transformation to generically make the pseudorandom function extractable.} that is still puncturable and then copy-protect the resulting class. On the other hand, we show that \textit{any} class of puncturable pseudorandom functions, which allows for the puncturing of two points, can be copy-protected. Hence, our result is qualitatively different than~\cite{CLLZ21}. 
    
    \item We show that {\bf any} digital signature scheme, where the signing key can be punctured at two points, can be copy-protected. Roughly speaking, a digital signature scheme is puncturable if the signing key can be punctured on two messages $m^{\bob}$ and $m^{\charlie}$ such that given the punctured signing key, it is computationally infeasible to produce a signature on one of the punctured messages. Our result rederives and generalizes a recent result by~\cite{LLQ+22} who showed how to copy-protect the digital signature scheme of~\cite{SW14}.
    
\end{itemize}

\noindent In the technical sections, we first present a simpler result where we copy-protect puncturable functionalities (\Cref{sec:qcp:punc}) and we then extend this result to achieve copy-protection for puncturable cryptographic schemes (\Cref{sec:cp:punc:cryptoscheme}). \\

\noindent \textsc{Copy-Protection for Evasive Functions (\Cref{sec:qcp:evasive}).} We consider a class of evasive functions associated with a distribution $\distr$ satisfying a property referred to as preimage-sampleability which is informally defined as follows: there exists a distribution $\distr'$ such that sampling an evasive function from $\distr$ along with an accepting point (i.e., the output of the function on this point is 1) is computationally indistinguishable from sampling a function from $\distr'$ and then modifying this function by injecting a uniformly random point as the accepting point. We show the following. 

\begin{theorem}
Assuming generalized UPO for ${\sf P/poly}$, there exists copy-protection for any class of functions that is evasive with respect to a distribution $\distr$ satisfying preimage-sampleability property.  
\end{theorem}

\noindent Unlike~\Cref{thm:intro:punccrypt}, we assume generalized UPO in the above theorem. 
\par As a special case, we obtain copy-protection for point functions. A recent work~\cite{CHV23} presented construction of copy-protection for point functions from post-quantum iO and other standard assumptions. Qualitatively, our results are different in the following ways: 
\begin{itemize}
    \item The challenge distribution considered in the security definition of~\cite{CHV23} is arguably not a natural one: with probability $\frac{1}{3}$, $\bob$ and $\charlie$ get as input the actual point, with probability $\frac{1}{3}$, $\bob$ gets the actual point while $\charlie$ gets a random value and finally, with probability $\frac{1}{3}$, $\bob$ gets a random value while $\charlie$ gets the actual point. On the other hand, we consider identical challenge distribution; that is, $\bob$ and $\charlie$ both receive the actual point with probability $\frac{1}{2}$ or they both receive a value picked uniformly at random. 
    \item While the result of~\cite{CHV23} is restricted to point functions, we show how to copy-protect functions where the number of accepting points is a fixed polynomial. 
\end{itemize}

\noindent We clarify that none of the above results on copy-protection contradicts the impossibility result by~\cite{AL20} who present a conditional result ruling out the possibility of copy-protecting contrived functionalities. \\

\noindent \textsc{Unclonable Encryption (\Cref{sec:simple-construct-uenc-upo,sec:pksde:app}).} Finally, we show, for the first time, an approach to construct unclonable encryption in the plain model. We give a direct and simple construction of unclonable encryption for bits, see \Cref{sec:simple-construct-uenc-upo}.
\begin{theorem}
Assuming generalized UPO for ${\sf P/poly}$, there exists a one-time unclonable bit-encryption scheme in the plain model. 
\end{theorem}

We also obtain a construction of unclonable encryption for arbitrary fixed length messages by first constructing public-key single-decryptor encryption (SDE) with an identical challenge distribution.

\begin{theorem}
Assuming generalized UPO for ${\sf P/poly}$, post-quantum indistinguishability obfuscation (iO), and post-quantum one-way functions, there exists a public-key single-decryptor encryption scheme with security against identical challenge distribution, see \Cref{sec:pksde:app}.
\end{theorem}

\noindent ~\cite{GZ20} showed that SDE with such a challenge distribution implies unclonable encryption. Prior work by~\cite{CLLZ21} demonstrated the construction of public-key single-decryptor encryption with security against independent challenge distribution, which is not known to imply unclonable encryption. 
We, thus, obtain the following corollary.

\begin{corollary}
Assuming generalized UPO, post-quantum iO, and post-quantum one-way functions, there exists a one-time unclonable encryption scheme in the plain model. 
\end{corollary}

\par Note that using the compiler of~\cite{AK21}, we can generically transform a one-time unclonable encryption into a public-key unclonable encryption in the plain model under the same assumptions as above. 

\noindent We note that this is the first construction of unclonable encryption in the plain model. All the previous works~\cite{BL20,AKLLZ22,AKL23} construct unclonable encryption in the quantum random oracle model. The disadvantage of our construction is that they leverage computational assumptions whereas the previous works~\cite{BL20,AKLLZ22,AKL23} are information-theoretically secure. 

\par Apart from unclonable encryption, single-decryptor encryption also implies public-key quantum money, thereby giving the following corollary.
\begin{corollary}
Assuming generalized UPO, post-quantum iO, and post-quantum one-way functions, there exists a public-key quantum money scheme. 
\end{corollary}

\noindent The construction of quantum money from UPO offers a conceptually different approach to constructing public-key quantum money in comparison with other quantum money schemes such as~\cite{Zha17,LMZ23,Zha23}. 
\par As an aside, we also present a lifting theorem that lifts a selectively secure single-decryptor encryption into an adaptively secure construction, assuming the existence of post-quantum iO. Such a lifting theorem was not known prior to our work.

\paragraph{Construction.} Finally we demonstrate a construction of generalized $\UPO$ for all classes of efficiently computable keyed circuits. We show that the same construction is secure with respect to both identical and independent challenge distributions. Specifically, we show the following: 

\begin{theorem}[Informal]
Suppose $\cktclass$ consists of polynomial-sized keyed circuits. Assuming the following:
\begin{itemize}
    \item Post-quantum sub-exponentially secure indistinguishability obfuscation for ${\sf P/poly}$,
    \item Post-quantum sub-exponentially secure one-way functions,
    \item Compute-and-compare obfuscation secure against QPT adversaries and, 
    \item Simultaneous inner product conjecture.
\end{itemize} 
there exists generalized $\UPO$ with respect to identical $\distr_{{\cal X}}$ for $\cktclass$. 
\end{theorem}

\noindent \textsc{On the Simultaneous Inner Product Conjecture}: Technically we need two different versions of the simultaneous inner product conjecture (\Cref{conj:goldreich-levin-identical} and~\Cref{conj:goldreich-levin-correlated}) to prove the security of our construction with respect to identical and independent challenge distributions. At a high level, the simultaneous inner product conjecture states that two (possibly entangled) QPT adversaries (i.e., non-local adversaries) should be unsuccessful in distinguishing $({\bf r},\langle {\bf r}, {\bf x} \rangle + m)$ versus $({\bf r},\langle {\bf r}, {\bf x} \rangle)$, where ${\bf r} \xleftarrow{\$} \Zq^{n},{\bf x} \xleftarrow{\$} \Zq^n,m \xleftarrow{\$} \Zq$ for every prime $q \geq 1$. Moreover, the adversaries receive as input a bipartite state $\rho$ that could depend on ${\bf x}$ with the guarantee that it should be computationally infeasible to recover ${\bf x}$. As mentioned above, we consider two different versions of the conjecture. In the first version (\textit{identical}), both the adversaries get the same sample $({\bf r},\langle {\bf r},{\bf x} \rangle)$ or they both get $({\bf r},\langle {\bf r},{\bf x} \rangle + m)$. In the second version (\textit{independent}), the main difference is that ${\bf r}$ and ${\bf x}$ are sampled independently for both adversaries. Weaker versions of this conjecture have been investigated and proven to be unconditionally true~\cite{AKL23,KT22}. \\

\noindent \pnote{add the result about the other construction.}  


\noindent \textsc{Composition}: Another contribution of ours is a composition theorem (see \Cref{sec:composition-theorem}), where we show how to securely compose unclonable puncturable obfuscation with a functionality-preserving compiler. In more detail, we show the following. Suppose $\UPO$ is a secure unclonable puncturable obfuscation scheme and let ${\sf Compiler}$ be a functionality-preserving circuit compiler. We define another scheme $\UPO'$ such that the obfuscation algorithm of $\UPO'$, on input a circuit $C$, first runs the circuit compiler on $C$ to obtain $\widetilde{C}$ and then it runs the obfuscation of $\UPO$ on $\widetilde{C}$ and outputs the result. The evaluation process can be similarly defined. We show that the resulting scheme $\UPO'$ is secure as long as $\UPO$ is secure. Our composition result allows us to compose $\UPO$ with other primitives such as different forms of program obfuscation without compromising on security. We use our composition theorem in some of the applications discussed earlier. 

\ifllncs

\else 
\paragraph{Concurrent and Independent Work.} Concurrent to our work is a recent work by Coladangelo and Gunn~\cite{CG23} who also showed the feasibility of copy-protecting puncturable functionalities and point functions albeit using a completely different approach. At a high level, the themes of the two papers are quite different. Our goal is to identify a central primitive in unclonable cryptography whereas their work focuses on exploring applications of quantum state indistinguishability obfuscation, a notion of indistinguishability obfuscation for quantum computations, to unclonable cryptography. 
\par We discuss the other differences below. 
\begin{itemize}
    \item Unlike our work, which only focuses on \textit{search} puncturing security, their work considers both \textit{search} and \textit{decision} puncturing security. 
    \item The two notions of obfuscation considered in both works seem to be incomparable. While the problem of obfuscating quantum computations has been notoriously challenging, their work considers the (weaker) problem of obfuscating a subclass of quantum computations that are implementations of classical functionalities.
    \item They demonstrate the feasibility of quantum state indistinguishability obfuscation in the quantum oracle model. We demonstrate the feasibility of UPO based on well-studied cryptographic assumptions and a new conjecture. 
\end{itemize}

\paragraph{Subsequent Work.} Subsequent to~\cite{CG23}, we were able to show that the notion of quantum state iO introduced by Colandangelo and Gunn implies UPO, assuming unclonable encryption for bits and injective one-way functions. We discuss this in~\Cref{sec:cons:qsio}. 
\par Subsequent to~\cite{CG23}, Bartusek, Brakerski and Vaikuntanathan~\cite{BVV24} obtained a construction of quantum state iO in the classical oracle model. 
\fi



\subsection{Technical Overview} 
\noindent We give an overview of the techniques behind our construction of UPO and the applications of UPO. We start with applications. 

\subsubsection{Applications} 
\paragraph{Copy-Protecting Puncturable Cryptographic Schemes.} We begin by exploring methods to copy-protect puncturable pseudorandom functions. Subsequently, we generalize this approach to achieve copy-protection for a broader class of puncturable cryptographic schemes. \\

\noindent \textsc{Case Study: Puncturable Pseudorandom Functions.} Let $\fclass=\{f_k(\cdot):\{0,1\}^n \rightarrow \{0,1\}^m\ :\ k \in \keyspace_{\secparam}\}$ be a puncturable pseudorandom function (PRF) with $\secparam$ being the security parameter and $\keyspace_{\secparam}$ being the key space. To copy-protect $f_k(\cdot)$, we simply obfuscate $f_k(\cdot)$ using an unclonable puncturable obfuscation scheme $\UPO$. To evaluate the copy-protected circuit on an input $x$, run the evaluation procedure of $\UPO$. 
\par To argue security, let us look at two experiments: 
\begin{itemize}
\item The first experiment corresponds to the regular copy-protection security experiment. That is, $\alice$ receives as input a copy-protected state $\rho_{f_k}$, which is copy-protection of $f_k$ where $k$ is sampled uniformly at random from the key space. It then creates a bipartite state which is split between $\bob$ and $\charlie$, who are two non-communicating adversaries who can share some entanglement. Then, $\bob$ and $\charlie$ independently receive as input $x$, which is picked uniformly at random. $(\bob,\charlie)$ win if they simultaneously guess $f_k(x)$.

\item The second experiment is similar to the first experiment except $\alice$ receives as input copy-protection of $f_k$ punctured at the point $x$, where $x$ is the same input given to both $\bob$ and $\charlie$. 
\end{itemize}
Thanks to the puncturing security of $\fclass$, the probability that $(\bob,\charlie)$ succeeds in the second experiment is negligible in $\secparam$. We would like to argue that $(\bob,\charlie)$ succeed in the first experiment also with probability negligible in $\secparam$. Suppose not, we show that the security of $\UPO$ is violated. \\

\noindent \textit{Reduction to $\UPO$}: The reduction $\reduct_{\alice}$ samples a uniformly random $f_k$ and forwards it to the challenger of the $\UPO$ game. The challenger of the $\UPO$ game then generates either an obfuscation of $f_k$ or the punctured circuit $f_k$ punctured at $x$ which is then sent to $\reduct_{\alice}$, who then forwards this to $\alice$ who prepares the bipartite state. The reduction $\reduct_{\bob}$ (resp., $\reduct_{\charlie}$) then receives as input $x$ which it duly forwards to $\bob$ (resp., $\charlie$). Then, $\bob$ and $\charlie$ each output $y_{\bob}$ and $y_{\charlie}$. Then, $\reduct_{\bob}$ outputs the {\bf bit 0} if $f_k(x)=y_{\bob}$, otherwise it outputs 1. Similarly, $\reduct_{\charlie}$ outputs  {\bf bit 0} if $f_k(x)=y_{\charlie}$, otherwise it outputs 1. The reason behind boldifying ``bit 0" part will be discussed below. 
\par Let us see how $(\reduct_{\alice},\reduct_{\bob},\reduct_{\charlie})$ fares in the $\UPO$ game.
\begin{itemize}
    \item \textit{Case 1. Challenge bit is $b=0$.} In this case, $\reduct_{\alice}$ receives as input obfuscation of $f_k$ with respect to $\UPO$. Denote $p_0$ to be the probability that $(\reduct_{\bob},\reduct_{\charlie})$ output $(0,0)$.
    \item \textit{Case 2. Challenge bit is $b=1$.} Here, $\reduct_{\alice}$ receives as input obfuscation of the circuit $f_k$ punctured at $x$. Similarly, denote $p_1$ to be the probability that $(\reduct_{\bob},\reduct_{\charlie})$ output $(1,1)$. 
\end{itemize}
\noindent From the security of $\UPO$, we have the following: $\frac{p_0 + p_1}{2} \leq \frac{1}{2} + \mu(\secparam),$ for some negligible function $\mu(\cdot)$. From the puncturing security of $\fclass$, the probability that $(\reduct_{\bob},\reduct_{\charlie})$ outputs $(1,1)$ is at least $1-\nu(\secparam)$, for some negligible function $\nu$. In other words, $p_1 \geq 1 - \nu(\secparam)$. From this, we can conclude, $p_0$ is negligible which proves the security of the copy-protection scheme.  \\

\noindent Perhaps surprisingly (at least to the authors), we do not know how to make the above reduction work if $\reduct_{\bob}$ (resp., $\reduct_{\charlie}$) instead output bit 1 in the case when $f_k(x)=y_{\bob}$ (resp., $f_k(x)=y_{\charlie}$). This is because we only get an upper bound for $p_1$ which cannot be directly used to determine an upper bound for $p_0$. \\

\noindent \textsc{Generalizing to Puncturable Cryptographic Schemes.} We present two generalizations of the above approach. We first generalize the above approach to handle puncturable circuit classes in~\Cref{sec:pksde:app}. A circuit class $\cktclass$, equipped with an efficient puncturing algorithm $\puncture$, is said to be puncturable\footnote{We need a slightly more general version than this. Formally, in Definition~\ref{def:puncturable-security-function-classes}, we puncture the circuit at two points (and not one), and then we require the adversary to predict the output of the circuit on one of the points.} if given a circuit $C \in \cktclass$, we can puncture $C$ on a point $x$ to obtain a punctured circuit $G$ such that given punctured circuit $G$, it is computationally infeasible to predict $C(x)$. As we can see, puncturable pseudorandom functions are a special case of puncturable circuit classes. The template to copy-protect an arbitrary puncturable circuit class, say $\cktclass$, is essentially the same as the above template to copy-protect puncturable pseudorandom functions. To copy-protect $C$, obfuscate $C$ using the scheme $\UPO$. The evaluation process and the proof of security proceed along the same lines as above. 

We then generalize this further to handle puncturable\footnote{We again consider a more general version where the circuit is punctured at two points. } cryptographic schemes. We consider an abstraction of a cryptographic scheme consisting of efficient algorithms $(\gen,\eval,\puncture,\verify)$ with the following correctness guarantee: the verification algorithm on input $(\pk,x,y)$ outputs 1, where $\gen(1^{\secparam})$ produces the secret key-public key pair $(\sk,\pk)$ and the value $y$ is the output of $\eval$ on input $(\sk,x)$. The algorithm $\puncture$ on input $(\sk,x)$ outputs a punctured circuit that has the same functionality as $\eval(\sk,\cdot)$ on all the points except $x$. The security property roughly states that predicting the output $\eval(\sk,x)$ given the punctured circuit should be computationally infeasible. The above template of copy-protecting PRFs can similarly be adopted for copy-protecting puncturable cryptographic schemes. 

\paragraph{Copy-Protecting Evasive Functions.} Using UPO to construct copy-protection for evasive functions turns out to be more challenging. To understand the difficulty, let us compare both the notions below: 
\begin{itemize}
    \item In a UPO scheme, $\alice$ gets as input an obfuscation of a circuit $C$ (if the challenge bit is $b=0$) or a circuit $C$ (if $b=1$) punctured at two points $x^{\bob}$ and $x^{\charlie}$. In the challenge phase, $\bob$ gets $x^{\bob}$ and $\charlie$ gets $x^{\charlie}$. 
    \item In the copy-protection for evasive function scheme, $\alice$ gets as input copy-protection of $C$, where $C$ is a circuit implements an evasive function. In the challenge phase, $\bob$ gets $x^{\bob}$ and $\charlie$ gets $x^{\charlie}$, where $(x^{\bob},x^{\charlie})=(x,x)$ is sampled as follows: $x$ is sampled uniformly at random (if challenge bit is $b=0$), otherwise $x$ is sampled uniformly at random from the set of points on which $C$ outputs 1 (if challenge bit is $b=1$).  
\end{itemize}
\noindent In other words, the distribution from which $\alice$ gets its input from depends on the bit $b$ in UPO but the challenges given to $\bob$ and $\charlie$ are always sampled from the same distribution. The setting in the case of copy-protection is the opposite: the distribution from which $\alice$ gets its input is always fixed while the challenge distribution depends on the bit $b$.\\

\noindent \textsc{Preimage Sampling Property}: To handle this discrepancy, we consider a class of evasive functions called preimage sampleable evasive functions. The first condition we require is that there is a distribution $\distr$ from which we can efficiently sample a circuit $C$ (representing an evasive function) together with an input $x$ such that $C(x)=1$. The second condition states that there exists another distribution $\distr'$ from which we can sample $(C',x')$, where $x'$ is sampled uniformly at random and then a punctured circuit $C'$ is sampled conditioned on $C'(x')=1$, satisfying the following property: the distributions $\distr$ and $\distr'$ are computationally indistinguishable. The second condition is devised precisely to ensure that we can reduce the security of copy-protection to UPO. \\

\noindent \textsc{Construction and Proof Idea}:  But first let us discuss the construction of copy-protection: to copy-protect a circuit $C$, compute two layers of obfuscation of $C$. First, obfuscate $C$ using a post-quantum iO scheme and then obfuscate the resulting circuit using UPO. To argue security, we view the obfuscated state given to $\alice$ as follows: first sample $C$ from $\distr$ and then do the following: (a) give $\rho_C$ to $\alice$ if $b=0$ and, (b) $\rho_C$ to $\alice$ if $b=1$, where $\rho_C$ is the copy-protected state and $b$ is the challenge bit that is used in the challenge phase. So far, we have done nothing. Now, we will modify (b). We will leverage the above conditions to modify (b) as follows: we will instead sample from $\distr'$. Since $\distr$ and $\distr'$ are computationally indistinguishable, the adversary will not notice the change. Now, let us examine the modified experiment: if $b=0$, the adversary receives $\rho_C$ (defined above), where $(C,x)$ is sampled from $\distr$ and if $b=1$, the adversary receives $\rho_{C'}$, where $(C',x')$ is sampled from $\distr'$. We can show that this precisely corresponds to the UPO experiment and thus, we can successfully carry out the reduction. \anote{For next version: We also need to change a) and argue that the change in a) involves invoking the UPO security. I do not think we need to add this but I just wanted to point it out. We can say that after this we change a) to be based on $\distr'$ instead of $\distr$.}  

\paragraph{Single-Decryptor Encryption.} A natural attempt to construct single-decryptor encryption would be to leverage UPO for puncturable cryptographic schemes. After all, it would seem that finding a public-key encryption scheme where the decryption key can be punctured at the challenge ciphertexts would give us our desired result. Unfortunately, this does not quite work: the reason lies in the way we defined the challenge distribution of UPO. We required that the marginals of the challenge distribution for a UPO scheme have to be uniform. Any public-key encryption scheme where the decryption keys can be punctured would not necessarily satisfy this requirement and hence, we need to find schemes that do\footnote{Of course, we could try the aforementioned issue in a different way: we could instead relax the requirements on the challenge distribution of UPO. Unfortunately, we currently do not know how to design an UPO for challenge distributions that do not have uniform marginals.}. 
\par We start with the public-key encryption scheme due to Sahai and Waters~\cite{SW14}. The advantage of this scheme is that the ciphertexts are pseudorandom. First, we show that this public-key encryption scheme can be made puncturable. Once we show this, using UPO for puncturable cryptographic schemes (and standard iO tricks), we construct single-decryptor encryption schemes of two flavors: 
\begin{itemize}
    \item First, we consider search security (\Cref{fig:product-uniform-search-anti-piracy-puncturable-functions}). In this security definition, $\bob$ and $\charlie$ receive ciphertexts of random messages and they win if they are able to predict the messages. 
    \item Next, we consider selective security (\Cref{fig:correlated-sde-cpa-style-anti-piracy}). In this security definition, $\bob$ and $\charlie$ receive encryptions of one of two messages adversarially chosen and they are supposed to predict which of the two messages was used in the encryption. Moreover, the adversarially chosen messages need to be declared before the security experiment begins and hence, the term selective security. Once we achieve this, we propose a generic lifting theorem to lift SDE security satisfying selective security to full adaptive security (\Cref{fig:correlated-sde-full-blown-cpa-style-anti-piracy}) where the challenge messages can be chosen later in the experiment. 
\end{itemize}

\subsubsection{Construction of UPO}
We move on to the construction of UPO.\\

\noindent \textsc{Starting Point: Decoupling Unclonability and Computation.} We consider the following template to design $\UPO$. To obfuscate a circuit $C$, we build two components. The first component is an unclonable quantum state that serves the purpose of authentication. The second component is going to aid in computation once the authentication passes. In more detail, given an input $x$, we first use the unclonable quantum state to authenticate $x$ and then execute the second component on the authenticated tag along with $x$ to obtain the output $C(x)$. 
\par The purpose of designing the obfuscation scheme this way is two-fold. Firstly, the fact that the first component is an unclonable quantum state means that an adversary cannot create multiple copies of this. And by design, without this state, it is not possible to execute the second component. Secondly, decoupling the unclonability and the computation part allows us to put less burden on the unclonable state, and in particular, only require the first component for authentication. Moreover, this approach helps us leverage existing tools in a modular way to construct $\UPO$. 
\par To implement the above approach, we use a copy-protection scheme for pseudorandom functions~\cite{CLLZ21}, denoted by $\cp$, and a post-quantum indistinguishability obfuscation scheme, denoted by $\iO$. In the $\UPO$ scheme, to obfuscate $C$, we do the following: 
\begin{enumerate}
    \item Copy-protect a pseudorandom function $f_k(\cdot)$ and,
    \item Obfuscate a circuit, with the PRF key $k$ hardcoded in it, that takes as input $(x,y)$ and outputs $C(x)$ if and only if $f_k(x)=y$. 
\end{enumerate}

\noindent \textsc{First Issue.} While syntactically the above template makes sense, when proving security we run into an issue. To invoke the security of $\cp$, we need to argue that the obfuscated circuit does not reveal any information about the PRF key $k$. This suggests that we need a much stronger object like virtual black box obfuscation instead of $\iO$ which is in general known to be impossible~\cite{BGIRSVY01}. Taking a closer look, we realize that this issue arose because we wanted to completely decouple the $\cp$ part and the $\iO$ part. \\

\noindent \textsc{Second Issue.} Another issue that arises when attempting to work out the proof. At a high level, in the security proof, we reach a hybrid where we need to hardwire the outputs of the PRF on the challenge inputs $x^{\bob}$ and $x^{\charlie}$ in the obfuscated circuit (i.e., in bullet 2 above). This creates an obstacle when we need to invoke the security of copy-protection: the outputs of the PRF are only available in the challenge phase (i.e., \textit{after} $\alice$ splits) whereas we need to know these outputs in order to generate the input to $\alice$. 
\\

\noindent \textsc{Addressing the above issues.} We first address the second issue. We introduce a new security notion of copy-protection for PRFs, referred to as copy-protection with \textit{preponed security}. Roughly speaking, in the preponed security experiment, $\alice$ receives the outputs of the PRF on the challenge inputs instead of being delayed until the challenge phase. By design, this stronger security notion solves the second issue. 
\par In order to resolve the aforementioned problem, we pull back and only partially decouple the two components. In particular, we tie both the $\cp$ and $\iO$ parts together by making non-black-box use of the underlying copy-protection scheme. Specifically, we rely upon the scheme by Colandangelo et al.~\cite{CLLZ21}. Moreover, we show that Colandangelo et al. ~\cite{CLLZ21} scheme satisfies preponed security by reducing their security to the security of their single-decryptor encryption construction; our proof follows along the same lines as theirs. Unfortunately, we do not know how to go further. While they did show that their single-decryptor encryption construction can be based on well studied cryptographic assumptions, the type of single-decryptor encryption scheme we need has a different flavor. In more detail, in their scheme, they consider \textit{independent} challenge distribution (i.e., both $\bob$ and $\charlie$ receive ciphertexts where the challenge bit is picked independently), whereas we consider \textit{identical} challenge distribution (i.e., the challenge bit for both $\bob$ and $\charlie$ is identical). We show how to modify their construction to satisfy security with respect to identical challenge distribution based on the simultaneous inner product conjecture. 


\ \\
\noindent \textsc{Summary.} To summarise, we design UPO for keyed circuit classes in $\mathrm{P/poly}$ as follows: 
\begin{itemize}
    \item We show that if the copy-protection scheme of~\cite{CLLZ21} satisfies preponed security, UPO for $\mathrm{P/poly}$ exists. This step makes heavy use of iO techniques. 
    \item We reduce the task of proving that the copy-protection scheme of~\cite{CLLZ21} satisfies preponed security to the task of proving that the single-decryptor encryption construction of~\cite{CLLZ21} is secure in the identical challenge setting. 
\end{itemize}

%% file: prelims.tex
\section{Preliminaries}
We refer the reader to~\cite{nielsen_chuang_2010} for a comprehensive reference on the basics of quantum information and quantum computation. We use $I$ to denote the identity operator. We use $\cS(\cH)$ to denote the set of unit vectors in the Hilbert space $\cH$. We use $\dmx(\cal{H})$ to denote the set of density matrices in the Hilbert space $\cal{H}$. Let $P,Q$ be distributions. We use $d_{TV}(P,Q)$ to denote the total variation distance between them. Let $\rho,\sigma \in \dmx(\cal{H})$ be density matrices. We write $\TD(\rho,\sigma)$ to denote the trace distance between them, i.e.,
\[
    \TD(\rho,\sigma) = \frac{1}{2} \| \rho - \sigma \|_1
\]
where $\norm{X}_1 = \Tr(\sqrt{X^\dagger X})$ denotes the trace norm.
We denote $\norm{X} := \sup_{\ket\psi}\{\braket{\psi|X|\psi}\}$ to be the operator norm where the supremum is taken over all unit vectors.
For a vector $\ket{x}$, we denote its Euclidean norm to be $\norm{\ket{x}}_2$.
We use the notation $M\ge 0$ to denote the fact that $M$ is positive semi-definite.

\subsection{Quantum Algorithms}
\label{sec:algorithms}

A quantum algorithm $A$ is a family of generalized quantum circuits $\{A_\lambda\}_{\lambda \in \N}$ over a discrete universal gate set (such as $\{ CNOT, H, T \}$). By generalized, we mean that such circuits can have a subset of input qubits that are designated to be initialized in the zero state and a subset of output qubits that are designated to be traced out at the end of the computation. Thus a generalized quantum circuit $A_\lambda$ corresponds to a \textit{quantum channel}, which is a completely positive trace-preserving (CPTP) map. When we write $A_\lambda(\rho)$ for some density matrix $\rho$, we mean the output of the generalized circuit $A_\lambda$ on input $\rho$. If we only take the quantum gates of $A_\lambda$ and ignore the subset of input/output qubits that are initialized to zeroes/traced out, then we get the \textit{unitary part} of $A_\lambda$, which corresponds to a unitary operator which we denote by $\hat{A}_\lambda$. The \textit{size} of a generalized quantum circuit is the number of gates in it, plus the number of input and output qubits.

We say that $A = \{A_\lambda\}_\lambda$ is a \textit{quantum polynomial-time (QPT) algorithm} if there exists a polynomial $p$ such that the size of each circuit $A_\lambda$ is at most $p(\lambda)$. We furthermore say that $A$ is \textit{uniform} if there exists a deterministic polynomial-time Turing machine $M$ that on input $1^\lambda$ outputs the description of $A_\lambda$. 

We also define the notion of a \textit{non-uniform} QPT algorithm $A$ that consists of a family $\{(A_\lambda,\rho_\lambda) \}_\lambda$ where $\{A_\lambda\}_\lambda$ is a polynomial-size family of circuits (not necessarily uniformly generated), and for each $\lambda$ there is additionally a subset of input qubits of $A_\lambda$ that are designated to be initialized with the density matrix $\rho_\lambda$ of polynomial length. This is intended to model nonuniform quantum adversaries who may receive quantum states as advice.
Nevertheless, the reductions we show in this work are all uniform.

The notation we use to describe the inputs/outputs of quantum algorithms will largely mimic what is used in the classical cryptography literature. For example, for a state generator algorithm $G$, we write $G_\lambda(k)$ to denote running the generalized quantum circuit $G_\lambda$ on input $\ketbra{k}{k}$, which outputs a state $\rho_k$.

Ultimately, all inputs to a quantum circuit are density matrices. However, we mix-and-match between classical, pure state, and density matrix notation; for example, we may write $A_\lambda(k,\ket{\theta},\rho)$ to denote running the circuit $A_\lambda$ on input $\ketbra{k}{k} \otimes \ketbra{\theta}{\theta} \otimes \rho$. In general, we will not explain all the input and output sizes of every quantum circuit in excruciating detail; we will implicitly assume that a quantum circuit in question has the appropriate number of input and output qubits as required by the context. 

%% file: upodefinition.tex
\newcommand{\gcktclass}{\mathfrak{G}}

\section{Unclonable Puncturable Obfuscation: Definition}\label{def:upo-definition}
\noindent We present the definition of an unclonable puncturable obfuscation scheme in this section. 

\paragraph{Keyed Circuit Class.} A class of classical circuits of the form $\cktclass=\{\cktclass_{\secparam}\}_{\secparam \in \mathbb{N}}$ is said to be a keyed circuit class if the following holds: $\cktclass_{\secparam}=\{C_k :k \in \keyspace_{\secparam} \}$, where $C_k$ is a (classical) circuit with input length $n(\secparam)$, output length $m(\secparam)$ and $\keyspace=\{\keyspace_{\secparam}\}_{\secparam \in \mathbb{N}}$ is the key space. We refer to $C_k$ as a keyed circuit. We note that any circuit class can be represented as a keyed circuit class using universal circuits. We will be interested in the setting when $C_k$ is a polynomial-sized circuit; henceforth, unless specified otherwise, all keyed circuit classes considered in this work will consist only of polynomial-sized circuits. We will also make a simplifying assumption that $C_{k}$ and $C_{k'}$ have the same size, where $k,k' \in \keyspace_{\secparam}$. 

\paragraph{Syntax.} An unclonable puncturable obfuscation (UPO) scheme $(\obf,\eval)$ for a keyed circuit class $\cktclass=\{\cktclass_{\secparam}\}_{\secparam \in \mathbb{N}}$, consists of the following QPT algorithms:
\begin{itemize}
    \item $\obf(1^{\secparam},C)$: on input a security parameter $\secparam$ and a keyed circuit $C \in \cktclass_{\secparam}$ with input length $n(\secparam)$, it outputs a quantum state $\rho_C$. 
    \item $\eval(\rho_C,x)$: on input a quantum state $\rho_C$ and an input $x \in \{0,1\}^{n(\secparam)}$, it outputs $(\rho'_C,y)$.
\end{itemize}

\paragraph{Correctness.}\label{par:upo-correctness}
An unclonable puncturable obfuscation scheme  $(\obf,\eval)$ for a keyed circuit class $\cktclass=\{\cktclass_{\secparam}\}_{\secparam \in \mathbb{N}}$ is $\delta$-correct, if for every $C \in \cktclass_{\secparam}$ with input length $n(\secparam)$, and for every $x\in \{0,1\}^{n(\secparam)}$,
$$ \Pr \left[ C(x)=y\ \mid\ \substack{\rho_C \leftarrow \obf(1^{\secparam},C)\\ \ \\(\rho'_C,y) \gets \eval(\rho_C,x)} \right] \geq \delta$$
\noindent If $\delta$ is negligibly close to 1 then we say that the scheme is correct (i.e., we omit mentioning $\delta$). 
\begin{remark}
If $(1-\delta)$ is a negligible function in $\secparam$, by invoking the almost as good as new lemma~\cite{aaronson2016complexity}, we can evaluate $\rho'_C$ on another input $x'$ to get $C(x')$ with probability negligibly close to 1. We can repeat this process polynomially many times and each time, due to the quantum union bound~\cite{gao2015quantum}, we get the guarantee that the output is correct with probability negligibly close to 1.  
\end{remark}

\subsection{Security} 
\label{sec:upo:security}

\paragraph{Puncturable Keyed Circuit Class.} Consider a keyed circuit class $\cktclass=\{\cktclass_{\secparam}\}_{\secparam \in \mathbb{N}}$, where $\cktclass_{\secparam}$ consists of circuits of the form $C_k(\cdot)$, where $k \in \keyspace_{\secparam}$, the input length of $C_k(\cdot)$ is $n(\secparam)$ and the output length is $m(\secparam)$. We say that $\cktclass_{\secparam}$ is said to be puncturable if there exists a deterministic polynomial-time puncturing algorithm $\puncture$ such that the following holds: on input $k \in \{0,1\}^{\secparam}$, 
 strings $x^{\bob} \in \{0,1\}^{n(\secparam)}, x^{\charlie} \in \{0,1\}^{n(\secparam)}$, it outputs a circuit $G_{k^*}$. Moreover, the following holds: for every $x \in \{0,1\}^{n(\secparam)}$,
$$G_{k^*}(x) = \left\{ \begin{array}{cc} 
C_k(x), & x \neq x^{\bob},x \neq x^{\charlie}, \\
\bot, 
& x \in \{x^{\bob},x^\charlie\}. 
\end{array} \right. $$
Without loss of generality, we can assume that the size of $G_{k^*}$ is the same as the size of $C_k$. Note that for every keyed circuit class, there exists a trivial $\puncture$ algorithm. The trivial $\puncture$ algorithm on any input $k,x_1,x_2,\mu_1,\mu_2$, constructs the circuit $C_k$ and then outputs the circuit $G$ that on input $x$, if $x=x_0$ or $x_1$ outputs $\bot$, else if $x\not\in\{x_1,x_2\}$ outputs $C_k(x)$\footnote{The output circuit $G_{k^*}$ is not of the same size as $C_k$, but this issue can be resolved by sufficient padding of the circuit class.}. \anote{Prabhanjan, is this fine?}

\begin{definition}[$\UPO$ Security]
\label{def:newcpsecurity_bot}
We say that a pair of QPT algorithms $(\obf,\eval)$ for a puncturable keyed circuit class $\cktclass$, associated with puncturing procedure $\puncture$, satisfies {\bf UPO security} with respect to a distribution $\distr_{{\cal X}}$ on $\{0,1\}^{n(\secparam)} \times \{0,1\}^{n(\secparam)}$ if for every QPT $(\alice,\bob,\charlie)$ in $\upoexpt$ (see \Cref{fig:upoexpt}), there exists a negligible function $\negl(\secparam)$ such that 
\[\prob\left[ 1 \leftarrow \upoexpt^{(\alice,\bob,\charlie),\distr_{\inpclass},\cktclass}\left(1^{\secparam},b\right)\ :\ b \xleftarrow{\$} \{0,1\} \right] \leq \frac{1}{2} + \negl(\secparam).\]

\end{definition}

\begin{figure}[!htb]
   \begin{center} 
   \begin{tabular}{|p{12cm}|}
    \hline 
\begin{center}
\underline{$\upoexpt^{\left(\alice,\bob,\charlie \right),\distr_{\inpclass},\cktclass}\left( 1^{\secparam},b \right)$}: 
\end{center}
\begin{itemize}
\item $\alice$ sends $k$, where $k \in \keyspace_{\secparam}$, to the challenger $\ch$.

\item $\ch$ samples $(x^\bob,x^\charlie) \leftarrow \distr_{\inpclass}(1^{\secparam})$ and generates $G_{k^*} \gets\puncture(k,x^{\bob},x^{\charlie})$. 
\item $\ch$ generates $\rho_b$ as follows:
\begin{itemize}
    \item $\rho_0 \leftarrow \obf(1^{\secparam},C_k(\cdot))$, 
    \item $\rho_1 \leftarrow \obf(1^{\secparam},G_{k^*}(\cdot))$
\end{itemize}
It sends $\rho_b$ to $\alice$. 
\item Apply $(\bob(x^\bob,\cdot) \otimes \charlie(x^\charlie,\cdot))(\sigma_{\bob,\charlie})$ to obtain $(b_{\bfB},b_{\bfC})$. 
\item Output $1$ if $b=b_{\bfB} = b_\bfC$. 
\end{itemize}
\ \\ 
\hline
\end{tabular}
    \caption{Security Experiment}
    \label{fig:upoexpt}
    \end{center}
\end{figure}

\subsubsection{Generalized Security}
\label{sec:upo:gensecurity}
For most applications, the security definition discussed in~\Cref{sec:upo:security} suffices, but for a couple of applications, we need a \ifllncs generalized definition as follows. \else generalized definition. The new definition generalizes the definition in~\Cref{sec:upo:security} in terms of puncturability as follows. \fi We allow the adversary to choose the outputs of the circuit generated by $\puncture$ on the punctured points. Previously, the circuit generated by the puncturing algorithm was such that on the punctured points, it output $\bot$. Instead, we allow the adversary to decide the values that need to be output on the points that are punctured. We emphasize that the adversary still would not know the punctured points itself until the challenge phase. 
Formally, the (generalized) puncturing algorithm $\genpuncture$ now takes as input $k \in \keyspace_{\secparam}$, polynomial-sized circuits $\mu^{\bob}:\{0,1\}^{n(\secparam)} \rightarrow \{0,1\}^{m(\secparam)}$, $\mu^{\charlie}:\{0,1\}^{n(\secparam)} \rightarrow \{0,1\}^{m(\secparam)}$, strings $x^{\bob} \in \{0,1\}^{n(\secparam)}, x^{\charlie} \in \{0,1\}^{n(\secparam)}$, if $x^\bob\neq x^\charlie$, it outputs a circuit $G_{k^*}$ such that for every $x \in \{0,1\}^{n(\secparam)}$,
$$G_{k^*}(x) = \left\{ \begin{array}{cc} 
C_k(x), & x \neq x^{\bob},x \neq x^{\charlie} \\
\mu_{\bob}(x^{\bob}), & x = x^{\bob} \\
\mu_{\charlie}(x^{\charlie}), & x = x^{\charlie},
\end{array} \right. $$
else it outputs a circuit $G_{k^*}$ such that for every $x \in \{0,1\}^{n(\secparam)}$,
$$G_{k^*}(x) = \left\{ \begin{array}{cc} 
C_k(x), & x \neq x^{\bob} \\
\mu_{\bob}(x^{\bob}), & x = x^{\bob}.
\end{array} \right. $$
As before, we assume that without loss of generality, the size of $G_{k^*}$ is the same as the size of $C_k$. 
\ifllncs \else \par \fi A keyed circuit class $\cktclass$ associated with a generalized puncturing algorithm $\genpuncture$ is referred to as a \textit{generalized puncturable keyed circuit class}.
Note that for every keyed circuit class $\cktclass=\{C_k\}_k$, there exists a trivial $\genpuncture$ algorithm, which on any input $k,x_1,x_2,\mu_1,\mu_2$, constructs the circuit $C_k$ and then outputs the circuit $G_{k^*}$\footnote{As before, the output circuit $G_{k^*}$ may not have the same size as $C_k$, but this can be resolved by sufficient padding of the complexity class.} that on input $x$, if $x=x_i$ for any $i\in \{0,1\}$, outputs $\mu_i(x_i)$, else if $x\not\in\{x_1,x_2\}$ outputs $C_k(x)$. 

\newcommand{\genupoexpt}{\mathsf{GenUPO.Expt}}
\begin{definition}[Generalized $\UPO$ security]
\label{def:newcpsecurity}
We say that a pair of QPT algorithms $(\obf,\eval)$ for a generalized keyed circuit class $\cktclass=\{\cktclass_{\secparam}\}_{\secparam \in \mathbb{N}}$ equipped with a puncturing algorithm $\genpuncture$, satisfies {\bf generalized UPO security} with respect to a distribution $\distr_{{\cal X}}$ on $\{0,1\}^{n(\secparam)} \times \{0,1\}^{n(\secparam)}$ if the following holds for every QPT $(\alice,\bob,\charlie)$ in $\genupoexpt$ defined in~\Cref{fig:genupo:expt}: 
\[ \prob\left[ 1 \leftarrow \genupoexpt^{(\alice,\bob,\charlie),\distr_{\inpclass},\cktclass}\left(1^{\secparam},b\right)\ :\ b \xleftarrow{\$} \{0,1\} \right]  \leq \frac{1}{2} + \negl(\secparam).\]
\end{definition}

\begin{figure}[!htb]
   \begin{center} 
   \begin{tabular}{|p{12cm}|}
    \hline 
\begin{center}
\underline{$\genupoexpt^{\left(\alice,\bob,\charlie \right),\distr_{\inpclass},\cktclass}\left( 1^{\secparam},b \right)$}: 
\end{center}
\begin{itemize}
\item $\alice$ sends $(k,\mu_{\bob},\mu_{\charlie})$, where $k \in \keyspace_{\secparam},\mu_{\bob}: \{0,1\}^{n(\secparam)} \rightarrow \{0,1\}^{m(\secparam)},\mu_{\charlie}: \{0,1\}^{n(\secparam)} \rightarrow \{0,1\}^{m(\secparam)}$, to the challenger $\ch$.

\item $\ch$ samples $(x^{\bob},x^{\charlie}) \leftarrow \distr_{\inpclass}(1^{\secparam})$ and generates $G_{k^*} \gets\puncture(k,x^{\bob},x^{\charlie},\mu_{\bob},\mu_{\charlie})$. 
\item $\ch$ generates $\rho_b$ as follows:
\begin{itemize}
    \item $\rho_0 \leftarrow \obf(1^{\secparam},C_k)$, 
    \item $\rho_1 \leftarrow \obf(1^{\secparam},G_{k^*})$
\end{itemize}
It sends $\rho_b$ to $\alice$. 
\item Apply $(\bob(x^\bob,\cdot) \otimes \charlie(x^\charlie,\cdot))(\sigma_{\bob,\charlie})$ to obtain $(b_{\bfB},b_{\bfC})$. 
\item Output $1$ if $b=b_{\bfB} = b_\bfC$. 
\end{itemize}
\ \\ 
\hline
\end{tabular}
    \caption{Generalized Security Experiment}
    \label{fig:genupo:expt}
    \end{center}
\end{figure}
\paragraph{Instantiations of $\distr_\inpclass$.}
In the applications, we will be considering the following two distributions:
\begin{enumerate}
    \item $\calU_{\{0,1\}^{2n}}$: the uniform distribution on $\{0,1\}^{2n}$. When the context is clear, we simply refer to this distribution as $\calU$.\anote{I have changed $distrprod$ distribution to $\calU$, because this is where I have been refering to for the definition of $distrprod$. If you think $distrprod$ might have been used in other context, and hence this might be a problem let me know.}
    
    \item $\distrid{\{0,1\}^{n}}$: identical distribution on $\{0,1\}^{n} \times \{0,1\}^n$ with uniform marginals. That is, the sampler for $\distrid{\{0,1\}^{n}}$ is defined as follows: sample $x$ from $\calU_{\{0,1\}^n}$ and output $(x,x)$. When the context is clear, we simply refer to this distribution as $\distrid$.
\end{enumerate}
  \label{subsec:upo-definition}

\ifllncs
    
\else
\input{composition-theorem}

\fi

%% file: composition-theorem.tex
\ifllncs

\section{Compostion theorem for UPO with other compilers}
\else
\subsection{Composition Theorem}
\fi
\label{sec:composition-theorem}
\noindent We state a useful theorem that states that we can compose a secure UPO scheme with any functionality-preserving compiler without compromising on security.\\

\noindent Let $\compile$ be a circuit compiler, i.e., $\compile$ is a probabilistic algorithm that takes as input a security parameter $\secparam$, classical circuit $C$ and outputs another classical circuit $\widetilde{C}$ such that $C$ and $\widetilde{C}$ have the same functionality. For instance, program obfuscation~\cite{BGIRSVY01} is an example of a circuit compiler. 
\par Let $\cktclass$ be a generalized puncturable keyed circuit class associated with keyspace $\keyspace$ defined as follows: $\cktclass=\{\cktclass_{\secparam}\}_{\secparam \in \mathbb{N}}$, where every circuit in $\cktclass_{\secparam}$ is of the form $C_k$, where $k \in \keyspace_{\secparam}$, with input length $n(\secparam)$ and the output length $m(\secparam)$. We denote $\genpuncture$ to be a generalized puncturing algorithm associated with $\cktclass$. 
\par Let $\UPO=(\UPO.\obf,\UPO.\eval)$ be an unclonable puncturable obfuscation scheme for a generalized puncturable keyed circuit class $\gcktclass$ (defined below) with respect to the input distribution $\distr_\inpclass$. 
\par \label{pg:compile-construct} We define $\gcktclass=\{\gcktclass_{\secparam}\}_{\secparam \in \mathbb{N}}$, where every circuit in $\gcktclass_{\secparam}$ is of the form $G_{k||r}(\cdot)$, with input length $n(\secparam)$, output length $m(\secparam)$, $k \in \keyspace_{\secparam}$, and $r \in \{0,1\}^{t(\secparam)}$. Here, $t(\secparam)$ denotes the number of bits of randomness consumed by $\compile(1^{\secparam},C_k;\cdot)$. Moreover, the circuit $G_{k || r}$ takes as input $x \in \{0,1\}^n$, applies $\compile(1^{\secparam},C_k;r)$ to obtain $\widetilde{C_k}$ and then it outputs $\widetilde{C_k}(x)$. The puncturing algorithm associated with $\gcktclass$ is $\genpuncture'$ which on input $k\|r$ and the set of inputs $x_1,x_2$ and circuits $\mu_1,\mu_2$, generates $D_{k^*} \leftarrow \genpuncture(k,x_1,x_2,\mu_1,\mu_2)$, and then outputs the circuit $G_{k^*,r}$, where $G_{k^*,r}$ is defined as follows: it takes as input $x \in \{0,1\}^n$, applies $\compile(1^{\secparam},D_{k^*};r)$ to obtain $\widetilde{D_{k^*}}$ and then it outputs $\widetilde{D_{k^*}}$. The keyspace associated with $\gcktclass$ is $\keyspace'=\{\keyspace'_{\secparam}\}_{\secparam \in \mathbb{N}}$, where $\keyspace'_{\secparam} = \keyspace_{\secparam} \times \{0,1\}^{t(\secparam)}$. \\

\noindent We define $\UPO'=(\UPO'.\obf,\UPO'.\eval)$ as follows:
    \begin{itemize}
        \item $\UPO'.\obf(1^\secparam,C)=\UPO.\obf(1^\secparam,\widetilde{C})$, where $\widetilde{C} \gets \compile(1^{\secparam},C)$.
        \item $\UPO'.\eval=\UPO.\eval$.
    \end{itemize}

\noindent \pnote{is the below only for the regular UPO security and not for the generalized puncturing security?} \anote{No it should be for general puncturing security. I think the text was not adopted after we adopted the term generalized upo.}

\begin{proposition}
\label{prop:UPO-compile}
    Assuming $\UPO$ satisfies $\distr_\inpclass$-generalized $\upo$ security for $\gcktclass$ and $\compile$ is a circuit compiler for $\cktclass$, $\UPO'$ satisfies $\distr_\inpclass$-generalized $\upo$ security for $\cktclass$.
\end{proposition}
\begin{proof}
Suppose there is an adversary $(\alice,\bob,\charlie)$ that violates the security of $\UPO'$ with probability $p$. We construct a QPT reduction $({\cal R}_{\alice},{\cal R}_{\bob},{\cal R}_{\charlie})$ that violates the security of $\UPO$, also with probability $p$. From the security of $\UPO'$ it then follows that $p$ is at most $\frac{1}{2} + \varepsilon$, for some negligible function $\varepsilon$, which proves the theorem. 
\par ${\cal R}_{\alice}(1^{\secparam})$ first runs $\alice(1^{\secparam})$ to obtain $k \in \keyspace_{\secparam}$. It then samples $r \xleftarrow{\$} \{0,1\}^{t(\secparam)}$. Then, ${\cal R}_{\alice}$ forwards $k || r$ to the external challenger of $\UPO$. Then, ${\cal R}_{\alice}$ receives $\rho^*$ which it then duly forwards to $\alice$. Similarly, even in the challenge phase, ${\cal R}_{\bob}$ (resp., ${\cal R}_{\charlie}$) forwards the challenge from the challenger to $\bob$ (resp., $\charlie$).
\par It can be seen that the probability that $(\alice,\bob,\charlie)$ breaks the security of $\UPO'$ is the same as the probability that $({\cal R}_{\alice},{\cal R}_{\bob},{\cal R}_{\charlie})$ breaks the security of $\UPO$.
\end{proof}

\begin{theorem}[Composition theorem]
\label{thm:UPO-compile}
    Let $\compile$ be a circuit compiler, i.e., $\compile$ is a probabilistic algorithm that takes as input a classical circuit $C$ and outputs another classical circuit $\tilde{C}$ such that $C$ and $\tilde{C}$ have the same functionality. Let $\UPO=(\UPO.\obf,\UPO.\eval)$ be an unclonable puncturable obfuscation scheme that satisfies $\distr_\inpclass$-generalized $\upo$ security for any class of generalized puncturable keyed circuit classs in $\sf{P/poly}$, then the same holds for the $\upo$ scheme $\UPO'=(\UPO'.\obf,\UPO'.\eval)$ defined as follows:
    \begin{itemize}
        \item $\UPO'.\obf(1^\secparam,C)=\UPO.\obf(1^\secparam,\compile(C))$ for every circuit $C$.
        \item $\UPO'.\eval=\UPO.\eval$.
    \end{itemize}
\end{theorem}
\begin{proof}
    Let $\cktclass$ be an arbitrary generalized puncturable keyed class in $\ppoly$. Let $\gcktclass$ be the generalized puncturable keyed class in $\ppoly$ derived from $\cktclass$ as defined on \Cpageref{pg:compile-construct}. Note that by the assumption in the theorem, $\UPO$ satisfies $\distr_\inpclass$-generalized $\upo$ security for $\gcktclass$. Therefore, by \Cref{prop:UPO-compile}, $\UPO'$ satisfies $\distr_\inpclass$-generalized $\upo$ security for $\cktclass$. Since $\cktclass$ was arbitrary, we conclude that $\UPO'$ satisfies $\distr_\inpclass$-generalized unclonable puncturable obfuscation security for any generalized puncturable keyed circuit class in $\ppoly$.
\end{proof}

\noindent Instantiating $\compile$ with an indistinguishability obfuscation $\io$ in \cref{thm:UPO-compile}, the following corollary is immediate. 


\begin{corollary}\label{cor:UPO-io}
    Consider a keyed circuit class $\cktclass$. Suppose $\iO$ be an indistinguishability obfuscation scheme for $\cktclass$. Suppose $\UPO$ is an unclonable puncturable obfuscation scheme for $\gcktclass$ (as defined above). Then $\UPO'$ is a secure unclonable puncturable obfuscation scheme for $\cktclass$ where $\UPO'$ is defined as follows:
    
    Assuming $\UPO$ is a $\upo$ scheme that satisfies $\distr_\inpclass$-generalized $\upo$ security for any $\distr_\inpclass$-generalized puncturable keyed circuit class in $\sf{P/poly}$, then the same holds for the $\upo$ scheme $\UPO'=(\UPO'.\obf,\allowbreak \UPO'.\eval)$ defined as follows:
   \begin{itemize}
     \item $\UPO'.\obf(1^\secparam,C)=\UPO.\obf(1^\secparam,\io(1^{\secparam},C))$, where $C \in \cktclass_{\secparam}$.
     
    \item $\UPO'.\eval=\UPO.\eval$.
    \end{itemize}
\end{corollary}

\noindent In the corollary above, we assume that the indistinguishability scheme does not have an explicit evaluation algorithm. In other words, the obfuscation algorithm on input a circuit $C$ outputs another circuit $\widetilde{C}$ 
 that is functionally equivalent to $C$. This is without loss of generality since we can combine any indistinguishability obfuscation scheme (that has an evaluation algorithm) with universal circuits to obtain an obfuscation scheme with the desired format.

%% file: conjectures.tex
\newcommand{\simip}{\mathsf{simultIP}}
\newcommand{\bfx}{\mathbf{x}}
\newcommand{\bfr}{\mathbf{r}}
\newcommand{\bit}{\mathsf{bit}}
\vspace{-1 em}
\section{Conjectures}
\label{sec:conjectures} 
\noindent The security of our construction relies upon some novel conjectures. 
Towards understanding our conjectures, consider the following problem: suppose say an adversary $\bob$ is given a state $\rho_{\bfx}$ that is generated as a function of a secret string $\bfx \in \Zq^{n}$, where $q,n \in \mathbb{N}$ and $q$ is prime. We are given the guarantee that just given $\rho_{\bfx}$, it should be infeasible to compute $\bfx$ for most values of $\bfx$. Now, the goal of $\bob$ is to distinguish $(\bfu,\langle \bfu,\bfx \rangle)$, where $\bfu \xleftarrow{\$} \Zq^n$ versus $(\bfu,\langle \bfu,\bfx \rangle + m)$, where $m \xleftarrow{\$} \Zq$. The Goldreich-Levin precisely shows that $\bob$ cannot succeed; if $\bob$ did succeed then we can come up with an extractor that recovers $\bfx$. Our conjectures state that the problem should be hard even for two (possibly entangled) parties simultaneously distinguishing the above samples.  Depending on whether the samples are independently generated between these parties or they are correlated, we have two different conjectures. 
\ifllncs \else \par \fi Before we formally state these conjectures and prove them, we first define the following problem. 

\paragraph{$(\distr_{\calX},\distr_{\ch},\distr_{\bit})$-Simultaneous Inner Product Problem ($(\distr_{\calX},\distr_{\ch},\distr_{\bit})$-$\simip$).} Let $\distr_{\calX}$ be a distribution on $\Zq^n \times \Zq^n$, $\distr_{\ch}$ be a distribtion on $\Zq^{n+1} \times \Zq^{n+1}$ and finally, let $\distr_{\bit}$ be a distribution on $\{0,1\} \times \{0,1\}$, for prime $q \in \mathbb{N}$. Let $\bob'$ and $\charlie'$ be QPT algorithms.  Let $\rho=\{\rho_{\bfx^{\bob},\bfx^{\charlie}}\}_{\bfx^{\bob},\bfx^{\charlie} \in \mathbb{Z}_q^n}$ be a set of bipartite states. Consider the following game. 
\begin{itemize}
    \item Sample $(\bfx^{\bob},\bfx^{\charlie}) \leftarrow \distr_{\calX}$.
    \item Sample $\left( \left(\bfu^{\bob},m^{\bob} \right),\left(\bfu^{\charlie},m^{\charlie} \right) \right) \leftarrow \distr_{\ch}$ 
    \item Set $z_0^{\bob} = \langle \bfu^{\bob}, \bfx^{\bob} \rangle,z_0^{\charlie} = \langle \bfu^{\charlie}, \bfx^{\charlie} \rangle,z_1^{\bob} = m^{\bob} + \langle \bfu^{\bob}, \bfx^{\bob} \rangle,z_1^{\charlie} = m^{\charlie} + \langle \bfu^{\charlie}, \bfx^{\charlie} \rangle$ \anote{Prabhanjan, should this $+$ be $\oplus$ instead? Same also below. In the construction we use $\oplus$.}
    \item Sample $(b^{\bob},b^{\charlie}) \leftarrow \distr_{\bit}$ 
    \item $(\widehat{b}^{\bob},\widehat{b}^{\charlie}) \leftarrow (\bob'(\bfu^{\bob},z_{b^{\bob}}^{\bob},\cdot) \otimes \charlie'(\bfu^{\charlie},z_{b^{\charlie}}^{\charlie},\cdot))(\rho_{\bfx^{\bob},\bfx^{\charlie}})$
\end{itemize}
We say that $(\bob',\charlie')$ succeeds if $\widehat{b}^{\bob} = b^{\bob}$ and $\widehat{b}^{\charlie} = b^{\charlie}$.

\paragraph{Specific Settings.} Consider the following setting: (a) $q=2$, (b) $\distr_{\bit}$ is a uniform distribution on $\{0,1\}^2$, (c) $\distr_{\ch}$ is a uniform distribution on $\Zq^{2n+2}$ and $\distr_{\calX}$ is a uniform distribution on $\{(\bfx,\bfx)\ :\ \bfx \in \Zq^n\}$. In this setting, recent works~\cite{KT22,AKL23} showed, via a simultaneous version of quantum Goldreich-Levin theorem, that any non-local solver for the $(\distr_{\ch},\distr_{\bit})$-simultaneous inner product problem can succeed with probability at most $\frac{1}{2} + \varepsilon(n)$, for some negligible function $\varepsilon(n)$. Although not explicitly stated, the generic framework of upgrading classical reductions to non-local reductions, introduced in~\cite{AKL23}, can be leveraged to extend the above result to large values of $q$.  
\par In the case when $\distr_{\bit}$ is not a uniform distribution, showing the hardness of non-locally solving the above problem seems much harder.  \\

\noindent Specifically, we are interested in the following setting:  $\distr_{\bit}$ is a distribution on $\{0,1\} \times \{0,1\}$, where $(b,b)$ is sampled with probability $\frac{1}{2}$, for $b \in \{0,1\}$. In this case, we simply refer to the above problem as $(\distr_{\calX},\distr_{\ch})$-$\simip$ problem. 

\newcommand{\ind}{\mathsf{ind}}
\paragraph{Conjectures.} We state the following conjectures. In the conjectures, we assume that the order of the field $Q>2^\secparam$\footnote{We need this condition, since it is crucial for our applications that  the two distinguishing cases (i.e., $z^\bob_0,z^\charlie_0$ and $z^\bob_1,z^\charlie_1$) are disjoint distributions (up to negligible overlap), but for $Q\in \poly(\secparam)$, the two distinguishing cases of the challenge distributions have an inverse polynomial overlap}. We are interested in the following distributions: 
\begin{itemize}
    \item We define $\distr_{\ch}^{\ind}$ as follows: it samples $\left( \left(\bfu^{\bob},m^{\bob} \right),\left(\bfu^{\charlie},m^{\charlie} \right) \right)$, where $\bfu^{\bob} \xleftarrow{\$} \Zq^{n} ,\bfu^{\charlie} \xleftarrow{\$} \Zq^{n}, m^{\bob} \xleftarrow{\$} \Zq, m^{\charlie} \xleftarrow{\$} \Zq$.  We define $\distr_{\ch}^{\id}$ as follows: it samples $\left( \left(\bfu,m \right),\left(\bfu,m \right) \right)$, where $\bfu \xleftarrow{\$} \Zq^{n}, m \xleftarrow{\$} \Zq$. 
    
    \item Similarly, we define $\distr_{\calX}^{\ind}$ as follows: it samples $\left( \bfx^{\bob}, \bfx^{\charlie} \right)$, where $\bfx^{\bob} \xleftarrow{\$} \Zq^{n} ,\bfx^{\charlie} \xleftarrow{\$} \Zq^{n}$. We define $\distr_{\calX}^{\id}$ as follows: it samples $\left( \bfx, \bfx \right)$, where $\bfx \xleftarrow{\$} \Zq^{n}$.
\end{itemize}

\begin{conjecture}[$\left( \distr_{\calX}^{\id}, \distr_{\ch}^{\id} \right)$-$\simip$ Conjecture]
\label{conj:goldreich-levin-identical}
 Consider a set of bipartite states $\rho=\{\rho_{\bfx}\}_{\bfx \in \mathbb{Z}_q^n}$ satisfying the following property: for any QPT adversaries $\bob,\charlie$, 
$$\prob\left[ \left(\bfx,\bfx \right) \leftarrow \left( \bob \otimes \charlie \right)\left( \rho_{\bfx} \right)\ :\ \left(\bfx,\bfx \right) \leftarrow \distr_{\calX}^{\id} \right] \leq \nu(n)$$
for some negligible function $\nu(\secparam)$.

  Any QPT non-local solver for the $\left( \distr_{\calX}^{\id}, \distr_{\ch}^{\id} \right)$-$\simip$ problem succeeds with probability at most $\frac{1}{2} + \varepsilon(n)$, where $\varepsilon$ is a negligible function. \anote{Should we not say the success probability depends on the order of the field $q$? I think we should say that we assume $q> 2^\secparam$.}
\end{conjecture}

\begin{conjecture}[$(\distr_{\calX}^{\ind}, \distr_{\ch}^{\ind})$-$\simip$ Conjecture]
 \label{conj:goldreich-levin-correlated}
 Consider a set of bipartite states $\rho=\{\rho_{\bfx^{\bob},\bfx^{\charlie}}\}_{\bfx^{\bob},\bfx^{\charlie} \in \mathbb{Z}_q^n}$ satisfying the following property: for any QPT adversaries $\bob,\charlie$, 
$$\prob\left[ \left(\bfx^{\bob},\bfx^{\charlie} \right) \leftarrow \left( \bob \otimes \charlie \right)\left( \rho_{\bfx^{\bob},\bfx^{\charlie}} \right)\ :\ \left(\bfx^{\bob},\bfx^{\charlie} \right) \leftarrow \distr_{\calX}^{\ind} \right] \leq \nu(n)$$
for some negligible function $\nu(\secparam)$.

  Any QPT non-local solver for the $\distr_{\ch}^{\id}$-$\simip$ problem succeeds with probability at most $\frac{1}{2} + \varepsilon(n)$, where $\varepsilon$ is a negligible function. 
\end{conjecture}

\paragraph{Discussion.}  The independent version of the conjecture was proven using a simultaneous version of the quantum Goldreich-Levin theorem. To recall, at a high level, the quantum Goldreich-Levin extractor creates a superposition over all the challenge messages and then runs the adversary on it. In the simultaneous version, we need to extract from two parties, say, Bob and Charlie, simultaneously. Moreover, Bob’s extractor and Charlie’s extractor cannot communicate with each other. In the independent case, when both Bob and Charlie receive independent challenges, it is easy to come up with extractors for both Bob and Charlie. Specifically, Bob’s extractor and Charlie’s extractor can indeed work independently; this is due to the fact that Bob’s extractor can create a superposition over Bob’s challenges, and independently, Charlie’s extractor can create a superposition over Charlie’s challenges. However, in the identical and the correlated case, the superposition over Bob’s and Charlie’s challenge messages is an entangled state, and it is unclear how both Bob’s and Charlie’s extractors can create such entangled states using only local operations. We will definitely add more discussion about this in the paper.

%% file: upoconstruction.tex
\newcommand{\prepexpt}{\mathsf{PreponedExpt}}
\newcommand{\sdeexpt}{\mathsf{SDEncExpt}}
\newcommand{\cpscheme}{\mathsf{CP}}
\newcommand{\sdescheme}{\mathsf{SDE}}





\section{Direct Construction}

\noindent In this section, we construct unclonable puncturable obfuscation for all efficiently computable generalized puncturable keyed circuit classes, with respect to $\distrprod$ and $\distrid$ challenge distribution (see \Cref{subsec:upo-definition}). Henceforth, we assume that any keyed circuit class we consider will consist of circuits that are efficiently computable. 
\ifllncs \else \par \fi We present the construction in three steps. 
\begin{enumerate}
    \item In the first step (\Cref{subsec:postproc-sde}), we construct a $\sde$ ($\SDE$) scheme based on the $\cllz$ scheme~\cite{CLLZ21} (see \Cref{fig:CLLZ-SDE-construction}) and show that it satisfies $\distrcipherind$-indistinguishability from random anti-piracy (and $\distrcipherind$-indistinguishability from random anti-piracy respectively) (see \Cref{subsec:sde}), based on the conjectures, \Cref{conj:goldreich-levin-correlated,conj:goldreich-levin-identical}. 
    \item In the second step (\Cref{subsec:preponed-cp-prf}), we define a variant of the security definition considered in~\cite{CLLZ21} with respect to two different challenge distributions and prove that the copy-protection construction for $\prf$s in~\cite{CLLZ21}      (see \Cref{fig:copy-protect-CLLZ}) satisfies this security notion, based on the indistinguishability from random anti-piracy guarantees of the $\SDE$ scheme considered in the first step. 
    \item In the third step (\Cref{subsec:UPO-from-preponed-CP}), we show how to transform the copy-protection scheme obtained from the first step into UPO for a keyed circuit class with respect to the $\distrprod{}$ and $\distrid{}$ challenge distribution.
\end{enumerate}

\subsection{A New Public-Key Single-Decryptor Encryption Scheme}\label{subsec:postproc-sde}
\noindent The first step is to construct a $\SDE$ scheme of the suitable form. While $\SDE$ schemes have been studied in prior works~\cite{GZ20,CLLZ21}, we require a weaker version of security called indistinguishability from random anti-piracy, see \Cref{subsec:sde}, which has not been considered in prior works. 

Our construction is based on the $\SDE$ scheme in~\cite[Section 6.3]{CLLZ21} which we recall in \Cref{fig:CLLZ-SDE-construction}.
From here on, we will refer to it as the $\cllz$ $\SDE$ scheme, given in \Cref{fig:CLLZ-SDE-construction}. \ifllncs
Next, we define a family of $\SDE$ schemes based on the $\cllz$ $\SDE$, called $\postproc$ schemes, as follows.
\else
Next, we define a family of $\SDE$ schemes based on the $\cllz$ $\SDE$, called $\postproc$ schemes, and then in \Cref{subsubsec:construction-postproc-sde}, we give a construction of $\postproc$ $\SDE$ scheme (\Cref{fig:CLLZ-SDE-postproc-construction}). Unfortunately, we are able to prove the required security guarantees of this construction only assuming conjectures that state the simultaneous inner-product conjectures, see \Cref{conj:goldreich-levin-correlated,conj:goldreich-levin-identical}, given in \Cref{sec:conjectures}. For our purposes, we will consider the message length to be at least polynomial in the security parameter. 
\fi

\begin{figure}[!htb]
   \begin{center} 
   \begin{tabular}{|p{12cm}|}
    \hline 
\noindent\textbf{Assumes:} post-quantum indistinguishability obfuscation $\io$.\\
\ \\
\noindent$\gen(1^\secparam)$: 
\begin{compactenum}
    \item Sample $\ell_0$ uniformly random subspaces $\{A_i\}_{i\in [\ell_0]}$ and for each $i\in [\ell_0]$, sample $s_i,{s'}_i$. 
    \item Compute $\{R^0_i,R^1_i\}_{i\in \ell_0}$, where for every $i\in [\ell_0]$, $R^0_i\gets\io(A_i + s_i)$ and $R^1_i\gets \io(A_i^\perp + {s'}_i)$ are the membership oracles.
    \item Output $\sk=\{\{{A_i}_{s_i,{s'}_i}\}_i\}$ and $\pk=\{R^0_i,R^1_i\}_{i\in \ell_0}$
\end{compactenum}
\ \\
\noindent$\qkeygen(\sk)$: 
\begin{compactenum}
    \item Interprete $\sk$ as $\{\{{A_i}_{s_i,{s'}_i}\}_i\}$.
    \item Output $\rho_sk=\{\{\ket{{A_i}_{s_i,{s'}_i}}\}_i\}$.
\end{compactenum}
\ \\
\noindent$\enc(\pk,m)$: 
\begin{compactenum}
    \item Interprete $\pk=\{R^0_i,R^1_i\}_{i\in \ell_0}$.
    \item Sample $r\xleftarrow{\$}\{0,1\}^n$.
    \item Generate $\tilde{Q}\gets \io(Q_{m,r})$ where $Q_{m,r}$ has $\{R^0_i,R^1_i\}_{i\in \ell_0}$ hardcoded inside, and on input $v_1,\ldots,v_{\ell_0}\in \{0,1\}^{n\ell_0}$, checks if $R^{r_i}_i(v_i)=1$ for every $i\in [\ell_0]$ and if the check succeeds, outputs $m$, otherwise output $\bot$.
    \item Output $\ct=(r,\tilde{Q})$
\end{compactenum}
\ \\
\noindent$\dec(\rho_\sk,\ct)$
\begin{compactenum}
    \item Interprete $\ct=(r,\tilde{Q})$.
    \item For every $i\in [\ell_0]$, if $r_{i}=1$ apply $H^{\otimes n}$ on $\ket{{A_i}_{s_i,{s'}_i}}$. Let the resulting state be $\ket{\psi_x}$.
    \item Run the circuit $\tilde{Q}$  in superposition on the state $\ket{\psi_x}$ and measure the output register and output the measurement result $m$.
\end{compactenum}

\ \\ 
\hline
\end{tabular}
    \caption{The $\cllz$ $\sde$ scheme, see~\cite[Construction 1]{CLLZ21}.}
    \label{fig:CLLZ-SDE-construction}
    \end{center}
\end{figure}

\subsubsection{Definition of a  CLLZ post-processing single decryptor encryption scheme}
\label{subsubsec:definition-postproc-sde}
\pnote{this is for the next version: I'm not a fan of defining primitives keeping in mind the CLLZ construction I think we should aim to have the definition as abstract as possible and CLLZ should only come up in the instantiation.} We call a $\SDE$ scheme $(\gen,\qkeygen,\enc,\dec)$ a CLLZ post-processing if there exists polynomial time classical algorithms $(\encproc,\decproc)$, such that 
 $\decproc$ is a deterministic algorithm. 
\begin{savenotes}
\begin{figure}[!htb]
   \begin{center} 
   \begin{tabular}{|p{16cm}|}
    \hline 
\noindent\textbf{Assumes:} $\cllz$ $\SDE$ scheme given in \Cref{fig:CLLZ-SDE-construction}.\\
\ \\
\noindent$\gen(1^\secparam)$:  Same as $\cllz.\gen(1^\secparam)$.\\
\ \\
\noindent$\qkeygen(\sk)$: Same as $\cllz.\qkeygen(\sk)$.\\
\ \\

\noindent$\enc(\pk,m)$: 
\begin{compactenum}
    \item Sample $r\xleftarrow{\$}\ZQ^\secparam$, where $Q$ is the smallest prime greater than or equal to $M\cdot 2^{\secparam}$ with $M$ being the size of the message space, i.e., $M=2^{|m|}$ and $|m|$ is the bit-size of the message $m$.\anote{Prabhanjan, this is the correct version right?}
    \item Generate $c\gets \encproc(m,r)$ and generate $c'\gets \cllz.\enc(\pk,r)$\footnote{We would like to note that the obfuscated circuit may be padded more than what is required in the $\cllz$ $\SDE$ scheme, for the security proofs of the $\postproc$ $\SDE$.}.
    \item Output $\ct=(c,c')$.
\end{compactenum}
\ \\
\noindent$\dec(\rho_\sk,\ct)$
\begin{compactenum}
    \item Interprete $\ct=(c,c')$.
            \item Generate $r\gets \cllz.\dec(\rho_\sk,c')$.
            \item Output $m \gets \decproc(c,r)$.
\end{compactenum}
\ \\ 
\hline
\end{tabular}
    \caption{Definition of a $\postproc$ $\SDE$ scheme.}
    \label{fig:def:postproc-cllz}
    \end{center}
\end{figure}
For correctness of, a $\postproc$ $\SDE$ scheme (see \Cref{fig:def:postproc-cllz}) we require that for every string $r,m$,
\begin{equation}\label{eq:postproc}
c'\gets \encproc(m,r), m'\gets \decproc(c',r)\implies  m=m'.
\end{equation}
It is easy to verify that assuming \Cref{eq:postproc}, $\delta$-correctness of the $\cllz$ $\SDE$ implies $\delta$-correctness of a $\postproc$ $\SDE$ for every $\delta\in [0,1]$.
\anote{This can be generalized further but for our purposes this definition is enough}
Note that if the  above condition is satisfied then it holds that for every $\delta\in [0,1]$, $\delta$-correctness of the $\cllz$ $\SDE$ implies $\delta$-correctness of the $\postproc$ $\SDE$ (see \Cref{fig:def:postproc-cllz}).
\subsubsection{Construction of a CLLZ post-processing single decryptor encryption scheme}
\label{subsubsec:construction-postproc-sde}
We next consider the following $\postproc$ scheme given in \Cref{fig:CLLZ-SDE-postproc-construction}. As mentioned before, we will assume that the message length is at least polynomial in the security parameter.
Note that the algorithms $(\encproc,\decproc)$ in \Cref{fig:CLLZ-SDE-postproc-construction} satisfies \Cref{eq:postproc}, and hence if the $\cllz$ $\SDE$ scheme (depicted in \Cref{fig:CLLZ-SDE-construction}) satisfies $\delta$-correctness so does the $\SDE$ scheme in  \Cref{fig:CLLZ-SDE-postproc-construction}. It is also easy to see that $\decproc$ is a determinisitc algorithm.
\ifllncs
Next, we prove security for the $\SDE$ scheme in \Cref{fig:CLLZ-SDE-postproc-construction} based on the simultaneous inner product conjectures.
\else
Next we prove that the $\SDE$ scheme in  \Cref{fig:CLLZ-SDE-postproc-construction} satisfies $\distrcipherind$-indistinguishability from random anti-piracy and $\distrcipherid$-indistinguishability from random anti-piracy by exploiting the corresponding simultaneous inner product conjectures (see \Cref{conj:goldreich-levin-correlated,conj:goldreich-levin-identical}).
\fi

\begin{figure}[!htb]
   \begin{center} 
   \begin{tabular}{|p{16cm}|}
    \hline 
\noindent$\encproc(m,r)$: 
\begin{compactenum}
    \item Sample $u\xleftarrow{\$}\ZQ^\secparam$, where $Q$ is the smallest prime number greater than $2^{|m|+\secparam}$, and $|m|$ is the bit-size of the binary string $m$.\anote{Prabhanjan, just wanted to check if this is the correct version.}
    \item Generate $\tilde{m}\gets \embed_Q(m)$, where $\embed_Q$ randomly embeds the binary string $m$ in $\ZQ$, i.e., $\tilde{m}_Q\equiv kM+m_Q$ where $k\xleftarrow{\$}\{0,1,\ldots,L-1\}$, $M\equiv 2^{|m|}$, $L\equiv Q \% M$, and $m_Q$ is the canonical embedding of $m$ in $\ZQ$. 
    \item Output $u,\tilde{m}_Q+\langle u,r\rangle$, where the addition and inner product uses the product over the field $\ZQ$.
\end{compactenum}
\ \\
\noindent$\decproc(c,r)$: 
\begin{compactenum}
    \item Interprete $c$ as $u,z$.
    \item Generate $\tilde{m}_Q\gets z+ \langle u,r\rangle$. 
    \item Output $m$ where $m$ is the binary representation of $\tilde{m}_Q\mod M$.
\end{compactenum}

\ \\ 
\hline
\end{tabular}
    \caption{Construction of a $\postproc$ $\SDE$ scheme.}
    \label{fig:CLLZ-SDE-postproc-construction}
    \end{center}
\end{figure}

\end{savenotes}

\begin{remark}\label{remark:uniform-field_element-bit_strings-correspondence}
    By the definition of the randomized embedding $\embed_Q$ defined in the algorithm $\encproc$ given in \Cref{fig:CLLZ-SDE-postproc-construction}, it is easy to see that the ensemble \[\{\embed_Q(m)\}_{m\xleftarrow{\$}\{0,1\}^M}=\{\tilde{m}_Q\}_{\tilde{m}_Q\xleftarrow{\$}\{0,1,\ldots,LM-1\}}\approx_s\{\tilde{m}_Q\}_{\tilde{m}_Q\xleftarrow{\$}\ZQ},\]
    because $Q-L<M$ by definition of $L$, and hence, $\frac{Q-L}{Q}<\frac{M}{Q}$ which is at most $\frac{M}{M\cdot 2^\secparam}=\frac{1}{2^\secparam}$, by our choice of $Q$.
\end{remark}

\begin{theorem}\label{thm:postproc-sde-correlated-challenge}
    Assuming \Cref{conj:goldreich-levin-correlated}, the existence of post-quantum sub-exponentially secure $\io$ and one-way functions, and the quantum hardness of Learning-with-errors problem (LWE), the $\postproc$ $\SDE$ as defined in \Cref{fig:def:postproc-cllz} given in \Cref{fig:CLLZ-SDE-postproc-construction} satisfies  $\distrcipherind$-indistinguishability from random anti-piracy (see \Cref{subsec:sde}).
\end{theorem}

\begin{theorem}\label{thm:postproc-sde-identical-challenge}
    Assuming \Cref{conj:goldreich-levin-identical}, the existence of post-quantum sub-exponentially secure $\io$ and one-way functions, and quantum hardness of Learning-with-errors problem (LWE), the $\postproc$ $\SDE$ (as defined in \Cref{fig:def:postproc-cllz}) given in \Cref{fig:CLLZ-SDE-postproc-construction} satisfies  $\distrcipherid$-indistinguishability from random anti-piracy (see \Cref{subsec:sde}).
\end{theorem}

\ifllncs
The proof of \Cref{thm:postproc-sde-correlated-challenge,thm:postproc-sde-identical-challenge} is given in \Cref{app:proof:sde(cons)}.

\else
\input{proofofsde_construct}
\fi


\subsection{Copy-Protection for PRFs with Preponed Security} \label{subsec:preponed-cp-prf}
\noindent We first introduce the definition of $\cpppiracy$ in~\Cref{subsubsec:preponedsde} and then we present the constructions of copy-protection in~\Cref{subsubsec:cons:preponed}.  

\subsubsection{Definition}
\label{subsubsec:preponedsde}
\noindent We introduce a new security notion for copy-protection called $\cpppiracy$. 
\par Consider a pseudorandom function family $\fclass=\{\fclass_{\secparam}\}_{\secparam \in \mathbb{N}}$, where $\fclass_{\secparam} = \{ f_k:\{0,1\}^{\ell(\secparam)} \rightarrow \{0,1\}^{\kappa(\secparam)}\ :\ k \in \{0,1\}^{\secparam} \}$. Moreover, $f_k$ can be implemented using a polynomial-sized circuit, denoted by $C_k$. 

\begin{definition}[Preponed Security]\label{def:cpcpiracy_PRFs}
    A copy-protection scheme $\cpscheme=(\copyprotect,\eval)$ for $\fclass$ (\Cref{sec:def:copyprotection}) satisfies $\distr_\inpclass$-$\cpppiracy$ 
     if for any $QPT$ $(\alice,\bob,\charlie)$, there exists a negligible function $\negl$ such that:
\[\Pr[\prepexpt^{\left(\alice,\bob,\charlie \right),\fclass,\distrprod}\left( 1^{\secparam}\right)=1]\leq \frac{1}{2}+\negl.\]
where $\prepexpt$ is defined in~\Cref{fig:preponedsecurity}.

We consider two instantiations of $\distr_\inpclass$:
\begin{enumerate}
    \item $\distrprod$ which is the product of uniformly random distribution on $\{0,1\}^\ell$, meaning $x_1,x_2\gets \distrprod(1^\secparam)$ where $x_1,x_2\xleftarrow{\$}\{0,1\}^\ell$ independently.
    \item $\distrid$ which is the perfectly correlatd distribution on $\{0,1\}^\ell$ with uniform marginals, meaning $x,x\gets \distrid(1^\secparam)$ where $x\xleftarrow{\$}\{0,1\}^\ell$.
\end{enumerate}


\end{definition}

 \begin{figure}[!htb]
   \begin{center} 
   \begin{tabular}{|p{12cm}|}
    \hline 
\begin{center}
\underline{$\prepexpt^{\left(\alice,\bob,\charlie \right),\cpscheme,\distr_\inpclass}\left( 1^{\secparam} \right)$}: 
\end{center}
\begin{enumerate}
\item $\ch$ samples $k\gets \keygen(1^\secparam)$, then  generates $\rho_{C_k} \leftarrow \copyprotect(1^{\secparam},C_k)$ and sends $\rho_{f_k}$ to $\alice$. 
\item $\ch$ samples $x^\bob,x^\charlie \gets \distr_\inpclass(1^\secparam)$, $b\xleftarrow{\$} \{0,1\}$. Let $y^\bob_1=f(x^\bob),y^\charlie_1=f(x^\charlie)$
, and $y^\bob_0=y_1,y^\charlie_0=y_2$ where $y_1,y_2 \xleftarrow{\$} \{0,1\}^{\kappa(\secparam)}$. $\ch$ gives $(y^\bob_{b},y^\charlie_{b})$ to Alice.
\item $\alice(\rho_{C_k})$ outputs a bipartite state $\sigma_{\bob,\charlie}$.
\item Apply $(\bob(x^\bob,\cdot) \otimes \charlie(x^\charlie,\cdot))(\sigma_{\bob,\charlie})$ to obtain $(b_{\bfB},b_{\bfC})$. 
\item Output $1$ if $b_{\bfB}=b_\bfC=b$.
\end{enumerate}
\ \\ 
\hline
\end{tabular}
    \caption{Preponed security experiment for copy-protection of PRFs with respect to the distribution $\distr_\inpclass$.}
    \label{fig:preponedsecurity}
    \end{center}
\end{figure}

\subsubsection{Construction} 
\label{subsubsec:cons:preponed}

The $\cllz$ copy-protection scheme is given in \Cref{fig:copy-protect-CLLZ}.

\begin{figure}[!htb]
   \begin{center} 
   \begin{tabular}{|p{12cm}|}
    \hline 
\noindent\textbf{Assumes:} Punctrable and extractable $\prf$ family $F_1=(\keygen,\eval)$ (represented as $F_1(k,x)=\prf.\eval(k,\cdot)$) and secondary $\prf$ family $F_2,F_3$ with some special properties as noted in~\cite{CLLZ21}\\
\ \\
\noindent$\copyprotect(K_1)$: 
\begin{compactenum}
    \item Sample secondary keys $K_2,K_3$, and $\{\{\ket{{A_i}_{s_i,{s'}_i}}\}_i\}$, and compute the coset state $\{\{\ket{{A_i}_{s_i,{s'}_i}}\}_i\}$.
    \item Compute $\tilde{P}\gets \io(P)$ where $P$ is as given in \Cref{fig:CLLZ-circuit-P}.
    \item Output $\rho=(\tilde{P},\{\{\ket{{A_i}_{s_i,{s'}_i}}\}_i\} )$.
\end{compactenum}
\ \\
\noindent$\eval(\rho,x)$:
\begin{compactenum}
    \item Interprete $\rho=(\tilde{P},\{\{\ket{{A_i}_{s_i,{s'}_i}}\}_i\} )$.
    \item Let $x=x_0\|x_1\|x_2$, where $x_0=\ell_0$. For every $i\in [\ell_0]$, if $x_{0,i}=1$ apply $H^{\otimes n}$ on $\ket{{A_i}_{s_i,{s'}_i}}$. Let the resulting state be $\ket{\psi_x}$.
    \item Run the circuit $\tilde{C}$  in superposition on the input registers $(X,V)$ with the initial state $(x,\ket{\psi_x})$ and then measure the output register to get an output $y$.
\end{compactenum}
\ \\ 
\hline
\end{tabular}
    \caption{$\cllz$ copy-protection for $\prf$s.}
    \label{fig:copy-protect-CLLZ}
    \end{center}
\end{figure}

 \begin{figure}[!htb]
   \begin{center} 
   \begin{tabular}{|p{12cm}|}
    \hline 
\begin{center}
\underline{$P$}: 
\end{center}
 Hardcoded keys $K_1,K_2,K_3,R^0_i,R^1_i$ for every $i\in [\ell_0]$
 On input $x=x_0\|x_1\|x_2$ and vectors $v=v_1,\ldots v_{\ell_0}$.
 \begin{enumerate}
     \item If $F_3(K_3,x_1)\oplus x_2=x_0\|Q$ and $x_1=F_2(K_2,x_0\|Q)$:
     
     \textbf{Hidden trigger mode:} Treat $Q$ as a classical circuit and output $Q(v)$.
     \item Otherwise, check if the following holds: for all $i\in \ell_0$, $R^{x_{0,i}}(v_i)=1$ (where $x_{0,i}$ is the $i^{th}$ coordinate of $x_0$).
     
     \textbf{Normal mode:} If so, output $F_1(K_1,x)$ where $F_1()=\prf.\eval()$ is the primary pseudorandom function family that is being copy-protected. Otherwise output $\bot$.\label{line:normal-mode-cllz-construction}
 \end{enumerate}
\ \\ 
\hline
\end{tabular}
    \caption{Circuit $P$ in $\cllz$ copy-protection of $\prf$.}
    \label{fig:CLLZ-circuit-P-construction}
    \end{center}
\end{figure}

\paragraph{Construction of Copy-Protection.} 

\begin{proposition}\label{prop:cpcpiracy_PRF_from_cpcpiracysde-noncllz}
    Assuming the existence of post-quantum $\io$, and one-way functions, and if there exists a $\postproc$ $\SDE$ scheme that satisfies $\distrcipherind$-indistinguishability from random anti-piracy, see \Cref{subsec:sde}, then the $\cllz$ copy-protection construction in~\cite[Section 7.3]{CLLZ21} (see \Cref{fig:copy-protect-CLLZ}) satisfies $\distrprod$-$\cpppiracy$ (\Cref{def:cpcpiracy_PRFs}). 
\end{proposition}

\begin{proposition}\label{prop:cpcpiracy_PRF_from_cpcpiracysde-id-noncllz}
    Assuming the existence of post-quantum $\io$, and one-way functions, and if there exists a $\postproc$ $\SDE$ scheme that satisfies $\distrcipherid$-indistinguishability from random anti-piracy, see \Cref{subsec:sde}, then the $\cllz$ copy-protection construction in~\cite[Section 7.3]{CLLZ21} (see \Cref{fig:copy-protect-CLLZ}) satisfies $\distrid$-$\cpppiracy$ (\Cref{def:cpcpiracy_PRFs}). 
\end{proposition}

\ifllncs
The proof of lemmas \Cref{prop:cpcpiracy_PRF_from_cpcpiracysde-noncllz,prop:cpcpiracy_PRF_from_cpcpiracysde-id-noncllz} can be found in~\Cref{sec:app:proof:cp}. 
\else 
\input{proofofcp_construct}

\fi

\subsection{UPO for Keyed Circuits from Copy-Protection with Preponed Security} \label{subsec:UPO-from-preponed-CP}

\begin{theorem}\label{thm:strong-CLLZ-cp-prf_puncturable-CP-f-uniform}
    Assuming \Cref{conj:goldreich-levin-correlated}, the existence of post-quantum sub-exponentially secure $\io$ and one-way functions, and the quantum hardness of Learning-with-errors problem (LWE),
    there is a construction of $\upo$ satisfying $\distrprod$-generalized $\UPO$ security (see Definition~\ref{def:newcpsecurity}), for any generalized keyed puncturable circuit class $\cktclass$ in $\ppoly$, see \Cref{subsec:upo-definition}. 
\end{theorem}
\begin{proof}
    The proof follows by combining \Cref{lemma:copy-protection-construction-completeness,thm:copy-protection-construction-puncturable_anti-piracy}.
\end{proof}

\begin{theorem}\label{thm:strong-CLLZ-cp-prf_puncturable-CP-f-uniform-id}
    Assuming \Cref{conj:goldreich-levin-identical},  the existence of post-quantum sub-exponentially secure $\io$ and one-way functions, and the quantum hardness of Learning-with-errors problem (LWE), 
    there is a construction of $\upo$ satisfying $\distrid$-generalized $\UPO$ security (see Definition~\ref{def:newcpsecurity}), for any generalized keyed puncturable circuit class $\cktclass$ in $ \ppoly$, see \Cref{subsec:upo-definition}. 
\end{theorem}
\begin{proof}
    The proof follows by combining \Cref{lemma:copy-protection-construction-completeness,thm:copy-protection-construction-puncturable_anti-piracy-id}.
\end{proof}

The construction is as follows. 
In the construction given in \Cref{fig:copy-protect-construction}, the $\prf$ family $(\keygen,\eval)$ satisfies the requirements as in~\cite{CLLZ21} and has input length $n(\secparam)$ and output length $m$; $\prg$ is a length-doubling injective pseudorandom generator with input length $m$.

\begin{figure}[!htb]
   \begin{center} 
   \begin{tabular}{|p{12cm}|}
    \hline 
\noindent\textbf{Assumes:} $\prf$ family $(\keygen,\eval)$ with same properties as needed in~\cite{CLLZ21}, $\prg$, $\cllz$ copy-protection scheme $(\copyprotect,\eval)$.\\
\ \\
\noindent$\obf(1^{\secparam},W)$:
\begin{compactenum}
    \item Sample a random key $k\gets \prf.\keygen(1^\secparam)$. 
    \item Compute $\io(P),\{\{\ket{{A_i}_{s_i,{s'}_i}}\}_i\} \gets \cllz.\copyprotect(k)$. 
    \item Compute $\tilde{C}\gets \io(C)$ where $C= \prg\cdot \io(P)$. 
    \item Compute $\io(D)$ where $D$ takes as input $x,v,y$, and runs $C$ on $x,v$ to get $y'$ and outputs $\bot$ if $y'\neq y$ or $y'=\bot$, else it runs the circuit $W$ on $x$ to output $W(x)$. 
    \item Output $\rho=(\{\{\ket{{A_i}_{s_i,{s'}_i}}\}_i\},\tilde{C},\io(D))$.
\end{compactenum}
\ \\
\noindent$\eval(\rho,x)$
\begin{compactenum}
    \item Interprete $\rho=(\{\{\ket{{A_i}_{s_i,{s'}_i}}\}_i\},\tilde{C},\io(D))$.
    \item Let $x=x_0\|x_1\|x_2$, where $x_0=\ell_0$. For every $i\in [\ell_0]$, if $x_{0,i}=1$ apply $H^{\otimes n}$ on $\ket{{A_i}_{s_i,{s'}_i}}$. Let the resulting state be $\ket{\psi_x}$.\label{line:Hadamard}
    \item Run the circuit $\tilde{C}$  in superposition on the input registers $(X,V)$ with the initial state $(x,\ket{\psi_x})$ and then measure the output register to get an output $y$. Let the resulting state quantum state on register $V$ be $\sigma$.
    \item Run $\io(D)$ on the registers $X,V,Y$  in superposition where registers $X,Y$ are initialized to classical values $x,y$ and then measure the output register to get an output $z$. Output $z$.
    
    \label{line:outside-C-eval}
\end{compactenum}
\ \\ 
\hline
\end{tabular}
    \caption{Construction of a UPO scheme.}
    \label{fig:copy-protect-construction}
    \end{center}
\end{figure}

\begin{lemma}\label{lemma:copy-protection-construction-completeness}
    The construction given in \Cref{fig:copy-protect-construction} satisfies $(1-\negl)$-$\UPO$ correctness for any generalized puncturable keyed circuit class in $\ppoly$ for some negligible function $\negl$.
\end{lemma}

\ifllncs 
The proof is given in \Cref{app:proof:completeness}.
\else 
\input{proofcorrectness}
\fi 
\begin{theorem}\label{thm:copy-protection-construction-puncturable_anti-piracy}
    Assuming \Cref{conj:goldreich-levin-correlated},  post-quantum sub-exponentially secure $\io$ and one-way functions, and the quantum hardness of Learning-with-errors problem (LWE),
    the construction given in \Cref{fig:copy-protect-construction} satisfies $\distrprod$-generalized $\upo$ security (see \Cref{subsec:upo-definition}) for any generalized puncturable keyed circuit class in $\ppoly$. 
\end{theorem}

\begin{proof}
    The proof follows by combining \Cref{lemma:puncturable_anti-piracy_from_cpcpiracy_PRFs,prop:cpcpiracy_PRF_from_cpcpiracysde-noncllz,thm:postproc-sde-correlated-challenge}, and the observation that the quantum hardness of LWE implies post-quantum one-way functions.
\end{proof}
\begin{theorem}\label{thm:copy-protection-construction-puncturable_anti-piracy-id}
   Assuming \Cref{conj:goldreich-levin-identical}, the existence of post-quantum sub-exponentially secure $\io$ and one-way functions, and the quantum hardness of Learning-with-errors problem (LWE),
    the construction given in \Cref{fig:copy-protect-construction} satisfies $\distrid$-generalized unclonable puncturable obfuscation security (see \Cref{subsec:upo-definition}) for any generalized puncturable keyed circuit class in $\ppoly$.
\end{theorem}

\begin{proof}
    The proof follows by combining \Cref{lemma:puncturable_anti-piracy_from_cpcpiracy_PRFs-id,prop:cpcpiracy_PRF_from_cpcpiracysde-id-noncllz,thm:postproc-sde-identical-challenge}, and the observation that the quantum hardness of LWE implies post-quantum one-way functions.
\end{proof}

\begin{lemma}\label{lemma:puncturable_anti-piracy_from_cpcpiracy_PRFs}
   Assuming the existence of post-quantum $\io$, one-way functions, and that $\cllz$ copy protection construction for $\prf$s given in \Cref{fig:copy-protect-CLLZ}, satisfies $\distrprod$-$\cpppiracy$ (defined in \Cref{def:cpcpiracy_PRFs}, the construction given in \Cref{fig:copy-protect-construction} for $\cktclassw$ satisfies $\distrprod$-generalized $\UPO$ security guarantee (see \Cref{subsec:upo-definition}), for any puncturable keyed circuit class $\cktclassw=\{\{W_s\}_{s\in \keyspace_\secparam}\}_\secparam$ in $\ppoly$.
\end{lemma}

\begin{lemma}\label{lemma:puncturable_anti-piracy_from_cpcpiracy_PRFs-id}
    Assuming the existence of post-quantum $\io$, one-way functions, and that $\cllz$ copy protection construction for $\prf$s given in \Cref{fig:copy-protect-CLLZ}, satisfies $\distrid$-$\cpppiracy$ (defined in \Cref{def:cpcpiracy_PRFs}),  the construction given in \Cref{fig:copy-protect-construction} for $\cktclassw$ satisfies $\distrid$-generalized $\UPO$ security guarantee (see \Cref{subsec:upo-definition}), for any puncturable keyed circuit class $\cktclassw=\{\{W_s\}_{s\in \keyspace_\secparam}\}_\secparam$ in $\ppoly$. 
\end{lemma}

\ifllncs 
The proof of the \Cref{lemma:puncturable_anti-piracy_from_cpcpiracy_PRFs,lemma:puncturable_anti-piracy_from_cpcpiracy_PRFs-id} can be found in~\Cref{app:proof:lemma4}. 
\else

\input{proofupocons}
\fi

%% file: proofofsde_construct.tex
\ifllncs 
\subsection{Proof of~\Cref{thm:postproc-sde-correlated-challenge,thm:postproc-sde-identical-challenge}}
\label{app:proof:sde(cons)}
\else 

\fi

\begin{proof}[Proof of \Cref{thm:postproc-sde-correlated-challenge}]\anote{Need to address the $+$ vs $\oplus$ issue in the proof after talking with prabhanjan. Best would be to write everything as $+$ while keeping $m$ in $\Zq$.}
  Let $(\alice,\bob,\charlie)$ be an adversary against the $\sde$ scheme $\PPROC$ given  in \Cref{fig:CLLZ-SDE-construction} in the $\distrcipherind$-indistinguishability from random anti-piracy experiment (see Game~\ref{fig:indistinguishability_from_random_-sde-anti-piracy}). We will do a sequence of hybrids; the changes would be marked in blue.

\noindent \underline{$\hybrid_0$}: Same as $\indrsdeexpt^{\left(\alice,\bob,\charlie \right),\distrcipherind}\left( 1^{\secparam} \right)$  (see Game~\ref{fig:indistinguishability_from_random_-sde-anti-piracy}) where $\distrcipherind$ is the challenge distribution defined in \Cref{subsec:sde} for the single-decryptor encryption scheme, $\PPROC$ in~\Cref{fig:CLLZ-SDE-postproc-construction}.
\begin{enumerate}
\item $\ch$ samples $(\sk,\pk) \gets \keygen(1^\secparam)$ and $\rho_k\gets \qkeygen(k)$ and sends $\rho_{k},\pk$ to $\alice$. 
\item $\adversary(\rho_k,\pk)$ outputs $\sigma_{\bob,\charlie}$.
\item $\ch$ samples 
 $b\xleftarrow{\$} \{0,1\}$.
\item $\ch$ computes $\ct^\bob_b$ as follows:
    \begin{enumerate}
    \item Sample $r^\bob\xleftarrow{\$}\ZQ^\secparam$, and compute ${c'}^\bob\gets \cllz.\enc(\pk,r^\bob)$.
    \item Sample $u^\bob\xleftarrow{\$}\ZQ^\secparam$ and compute $c^\bob_b=(u^\bob,\langle u^\bob, r^\bob\rangle)$ if $b=0$, else sample $m^\bob\xleftarrow{\$}\{0,1\}^M$ (where $M$ is the bit-size of the messages), generate $\tilde{m}^\bob_Q\gets \embed_Q(m^\bob)$ (see \Cref{fig:CLLZ-SDE-postproc-construction} for the definition of $\embed_Q$) and compute $c^\bob_b=(u^\bob,\tilde{m}^\bob+ \langle u^\bob, r^\bob\rangle)$ (where the operations are in the field $\ZQ$) if $b=1$.
    \item Set $\ct^\bob_b=(c^\bob_b,{c'}^\bob)$.
    \end{enumerate}
\item $\ch$ computes $\ct^\charlie_b$ as follows:
    \begin{enumerate}
    \item Sample $r^\charlie\xleftarrow{\$}\ZQ^\secparam$, and compute ${c'}^\charlie\gets \cllz.\enc(\pk,r^\charlie)$.
    \item Sample $u^\charlie\xleftarrow{\$}\ZQ^\secparam$ and compute $c^\charlie_b=(u^\charlie,\langle u^\charlie, r^\charlie\rangle)$ if $b=0$, else sample $m^\charlie\xleftarrow{\$}\{0,1\}^M$,  generate $\tilde{m}^\charlie_Q\gets \embed_Q(m^\charlie)$ and compute $c^\charlie_b=(u^\charlie,\tilde{m}^\charlie+ \langle u^\charlie, r^\charlie\rangle)$ if $b=1$.
    \item Set $\ct^\charlie_b=(c^\charlie_b,{c'}^\charlie)$.
    \end{enumerate}    
\item Apply $(\bob(\ct^\bob_b,\cdot) \otimes \charlie(\ct^\charlie_b,\cdot))(\sigma_{\bob,\charlie})$ to obtain $(b_{\bfB},b_{\bfC})$. 
\item Output $1$ if $b_{\bfB}=b_\bfC=b$.
\end{enumerate}

\noindent \underline{$\hybrid_1$}: 
\begin{enumerate}
\item $\ch$ samples $(\sk,\pk) \gets \keygen(1^\secparam)$ and $\rho_k\gets \qkeygen(k)$ and sends $\rho_{k},\pk$ to $\alice$. 
\item $\adversary(\rho_k,\pk)$ outputs $\sigma_{\bob,\charlie}$.
\item $\ch$ samples 
 $b\xleftarrow{\$} \{0,1\}$.
\item $\ch$ computes $\ct^\bob_b$ as follows:
    \begin{enumerate}
    \item Sample $r^\bob\xleftarrow{\$}\ZQ^\secparam$, and compute ${c'}^\bob\gets \cllz.\enc(\pk,r^\bob)$.
    \item Sample $u^\bob\xleftarrow{\$}\ZQ^\secparam$ and compute $c^\bob_b=(u^\bob,\langle u^\bob, r^\bob\rangle)$ if $b=0$, else \sout{sample $m^\bob\xleftarrow{\$}\{0,1\}^M$ (where $M$ is the bit-size of the messages), generate $\tilde{m}^\bob_Q\gets \embed_Q(m^\bob)$ (see Fig.~\ref{fig:CLLZ-SDE-postproc-construction} for the definition of $\embed_Q$) and compute $c^\bob_b=(u^\bob,\tilde{m}^\bob+ \langle u^\bob, r^\bob\rangle)$ (where the operations are in the field $\ZQ$) if $b=1$}  \cblue{sample $\tilde{m}^\bob \xleftarrow{\$} \ZQ$, and compute $c^\bob_b=(u^\bob,\tilde{m}^\bob+ \langle u^\bob, r^\bob\rangle)$ (where the operations are in the field $\ZQ$) if $b=1$}.
    \item Set $\ct^\bob_b=(c^\bob_b,{c'}^\bob)$.
    \end{enumerate}
\item $\ch$ computes $\ct^\charlie_b$ as follows:
    \begin{enumerate}
    \item Sample $r^\charlie\xleftarrow{\$}\ZQ^\secparam$, and compute ${c'}^\charlie\gets \cllz.\enc(\pk,r^\charlie)$.
    \item Sample $u^\charlie\xleftarrow{\$}\ZQ^\secparam$ and compute $c^\charlie_b=(u^\charlie,\langle u^\charlie, r^\charlie\rangle)$ if $b=0$, else \sout{sample $m^\charlie\xleftarrow{\$}\{0,1\}^M$,  generate $\tilde{m}^\charlie_Q\gets \embed_Q(m^\charlie)$ and compute $c^\charlie_b=(u^\charlie,\tilde{m}^\charlie+ \langle u^\charlie, r^\charlie\rangle)$ if $b=1$} \cblue{sample $\tilde{m}^\charlie \xleftarrow{\$} \ZQ$, and compute $c^\charlie_b=(u^\charlie,\tilde{m}^\charlie+ \langle u^\charlie, r^\charlie\rangle)$ if $b=1$}.
    \end{enumerate}    
\item Apply $(\bob(\ct^\bob_b,\cdot) \otimes \charlie(\ct^\charlie_b,\cdot))(\sigma_{\bob,\charlie})$ to obtain $(b_{\bfB},b_{\bfC})$. 
\item Output $1$ if $b_{\bfB}=b_\bfC=b$.
\end{enumerate}

The only change from $\hybrid_0$ to $\hybrid_1$ was the distribution on $\tilde{m}^\bob_Q$ and $\tilde{m}^\charlie_Q$, which are independently and identically distributed in both the hybrids. In particular, the IID distribution on $\tilde{m}^\bob_Q$ and $\tilde{m}^\charlie_Q$ changes from $\{\embed_Q(m)\}_{m\xleftarrow{\$}\{0,1\}^M}
$ to $\{\tilde{m}_Q\}_{\tilde{m}_Q\xleftarrow{\$}\ZQ}$ across thre hybrids. Since these distributions are statistically indistinguishable by \Cref{remark:uniform-field_element-bit_strings-correspondence}, we conclude that statistically indistinguishability holds for $\hybrid_0$ and $\hybrid_1$. 

Consider the following independent search experiment against a pair of (uniform) efficient adversaries $\bob',\charlie'$.
\begin{enumerate}
\item $\ch$ samples $r^\bob,r^\charlie\xleftarrow{\$}\ZQ^\secparam$.
\item $\ch$ computes $\sigma_{\bob,\charlie}$ as follows:
\begin{enumerate}
    \item Sample $(\sk,\pk) \gets \keygen(1^\secparam)$ and prepares $\rho_k\gets \qkeygen(k)$.
    \item Run $\alice(\rho_k,\pk)$ to get $\sigma_{\bob,\charlie}$.
\end{enumerate}
\item $\ch$ computes ${c'}^\bob\gets \cllz.\enc(\pk,r^\bob)$, and  computes ${c'}^\charlie\gets \cllz.\enc(\pk,r^\charlie)$.
\item $\ch$ constructs the bipartite auxiliary state $\tau^{r^\bob,r^\charlie}_{B,C}={c'}^\bob,\sigma_{\bob,\charlie},{c'}^\charlie$, i.e., the ${c'}^\bob,\sigma_\bob$ and ${c'}^\charlie,\sigma_\charlie$ are the two partitions.
\item $\ch$ sends the respective registers of $\tau^{r^\bob,r^\charlie}_{B,C}$ to $\bob'$ and $\charlie'$, and gets back the responses ${r'}^\bob$ and ${r'}^\charlie$ respectively.
\item Ouptput $1$ if ${r'}^\bob=r^\bob$, and  ${r'}^\charlie=r^\charlie$.
\end{enumerate}

Clearly, the winning probability of $(\bob',\charlie')$ in the above game is the same as the winning probability of $(\alice,\bob',\charlie')$ in the independent search anti-piracy (see \Cref{subsec:sde}) of the $\cllz$ $\sde$ scheme given in \Cref{fig:CLLZ-SDE-construction}. It was shown in~\cite[Theorem 6.15]{CLLZ21} that the $\cllz$ $\sde$ satisfies independent search anti-piracy assuming the security guarantess of post-quantum sub-exponentially secure $\io$ and one-way functions, and quantum hardness of Learning-with-errors problem (LWE). 
Hence, under the security guarantees of the above assumptions, there exists a negligible function $\epsilon'()$ such that the winning probability of $(\bob',\charlie')$ in the above game is $\epsilon'(\secparam)$. Since the order of the field $Q>2^\secparam$, assuming \Cref{conj:goldreich-levin-correlated}, there exists a negligible function $\epsilon()$ such that the winning probability of $(\bob,\charlie)$ in the following indistinguishability game is at most $\frac{1}{2}+\epsilon(\secparam)$.

\begin{enumerate}
\item $\ch$ samples $r^\bob,r^\charlie\xleftarrow{\$}\ZQ^\secparam$.
\item $\ch$ computes $\sigma_{\bob,\charlie}$ as follows:
\begin{enumerate}
    \item Sample $(\sk,\pk) \gets \keygen(1^\secparam)$ and prepares $\rho_k\gets \qkeygen(k)$.
    \item Run $\alice(\rho_k,\pk)$ to get $\sigma_{\bob,\charlie}$.
\end{enumerate}
\item $\ch$ computes ${c'}^\bob\gets \cllz.\enc(\pk,r^\bob)$, and  computes ${c'}^\charlie\gets \cllz.\enc(\pk,r^\charlie)$.
\item $\ch$ constructs the bipartite auxiliary state $\tau^{r^\bob,r^\charlie}_{B,C}={c'}^\bob,\sigma_{\bob,\charlie},{c'}^\charlie$, i.e., the ${c'}^\bob,\sigma_\bob$ and ${c'}^\charlie,\sigma_\charlie$ are the two partitions.
\item $\ch$ samples  $b\xleftarrow{\$} \{0,1\}$.
\item $\ch$ samples $u^\bob\xleftarrow{\$}\ZQ^\secparam$ and compute $c^\bob_b=(u^\bob,\langle u^\bob, r^\bob\rangle)$ if $b=0$, else samples $\tilde{m}^\bob_Q\xleftarrow{\$}\ZQ$, and computes $c^\bob_b=(u^\bob,m^\bob+\langle u^\bob,r^\bob\rangle)$ if $b=1$.
\item Similarly,  $\ch$ samples  $u^\charlie\xleftarrow{\$}\ZQ^\secparam$ and computes $c^\charlie_b=(u^\charlie,\langle u^\charlie, r^\charlie\rangle)$ if $b=0$, else samples $\tilde{m}^\charlie_Q\xleftarrow{\$}\ZQ$, and computes $c^\charlie_b=(u^\charlie,m^\charlie+\langle u^\charlie,r^\charlie\rangle))$ if $b=1$.
\item $\ch$ sends $c^\bob_b$ and $c^\charlie_b$ along with the respective registers of $\tau^{r^\bob,r^\charlie}_{B,C}$ to $\bob'$ and $\charlie'$ respectively, and gets back the responses $b^\bob$ and $b^\charlie$ respectively.
\item Output $1$ if $b_{\bfB}=b_\bfC=b$.
\end{enumerate}

However, note that the view of the adversaries $\bob$ and $\charlie$ in the indistinguishability game above is the same as the view in $\hybrid_3$. Therefore, the winning probability of $(\alice,\bob,\charlie)$ in $\hybrid_1$ is at most $\frac{1}{2}+\epsilon(\secparam)$. This completes the proof of the theorem.

\end{proof}


\begin{proof}[Proof of \Cref{thm:postproc-sde-identical-challenge}]
    The proof directly follows by combining \Cref{lemma:identical-postproc-anti-piracy-to-identical-search-anti-piracy,lemma:identical-search-anti-piracy-postproc}, which we state and prove next.
\end{proof}

\begin{lemma}\label{lemma:identical-postproc-anti-piracy-to-identical-search-anti-piracy}
    Assuming \Cref{conj:goldreich-levin-identical}, the $\postproc$ $\sde$ as defined in \Cref{fig:def:postproc-cllz} given in \Cref{fig:CLLZ-SDE-postproc-construction} satisfies $\distrcipherid$-indistinguishability from random anti-piracy, if $\cllz$ $\sde$ (see \Cref{fig:CLLZ-SDE-construction}) satisfies $\distrid$-search anti-piracy (see \Cref{subsec:sde}).
\end{lemma}
\begin{proof}
\anote{Need to change the proof below similar to how it is done for the correlated case.}

  Let $(\alice,\bob,\charlie)$ be an adversary against the $\sde$ scheme $\PPROC$ given  in \Cref{fig:CLLZ-SDE-construction} in the $\distrcipherid$-indistinguishability from random anti-piracy experiment. We will do a sequence of hybrids; the changes will be marked in blue.

\noindent \underline{$\hybrid_0$}: Same as $\indrsdeexpt^{\left(\alice,\bob,\charlie \right),\distrcipherid}\left( 1^{\secparam} \right)$  (see Game~\ref{fig:indistinguishability_from_random_-sde-anti-piracy}) where $\distrcipherid$ is the challenge distribution defined in \Cref{subsec:sde} for the single-decryptor encryption scheme, $\PPROC$ in~\Cref{fig:CLLZ-SDE-postproc-construction}.
\begin{enumerate}
\item $\ch$ samples $(\sk,\pk) \gets \keygen(1^\secparam)$ and $\rho_k\gets \qkeygen(k)$ and sends $\rho_{k},\pk$ to $\alice$. 
\item $\adversary(\rho_k,\pk)$ outputs $\sigma_{\bob,\charlie}$.
\item $\ch$ samples
 $b\xleftarrow{\$} \{0,1\}$.
\item $\ch$ computes $\ct_b$ as follows:
    \begin{enumerate}
    \item Sample $r\xleftarrow{\$}\ZQ^\secparam$, and compute ${c'}\gets \cllz.\enc(\pk,r)$.
    \item Sample $u\xleftarrow{\$}\ZQ^\secparam$ and compute $c_b=(u,\langle u, r\rangle)$ if $b=0$, else sample $m\xleftarrow{\$}\{0,1\}^M$ (where $M$ is the bit-size of the messages), generate $\tilde{m}_Q\gets \embed_Q(m)$ (see \Cref{fig:CLLZ-SDE-postproc-construction} for the definition of $\embed_Q$) and compute $c_b=(u,\tilde{m}+ \langle u, r\rangle)$ (where the operations are in the field $\ZQ$) if $b=1$.
    \item Set $\ct_b=(c_b,{c'})$.
    \end{enumerate}
    

\item Apply $(\bob(\ct_b,\cdot) \otimes \charlie(\ct_b,\cdot))(\sigma_{\bob,\charlie})$ to obtain $(b_{\bfB},b_{\bfC})$. 
\item Output $1$ if $b_{\bfB}=b_\bfC=b$.
\end{enumerate}

\noindent \underline{$\hybrid_1$}: 
\begin{enumerate}
\item $\ch$ samples $(\sk,\pk) \gets \keygen(1^\secparam)$ and $\rho_k\gets \qkeygen(k)$ and sends $\rho_{k},\pk$ to $\alice$. 
\item $\adversary(\rho_k,\pk)$ outputs $\sigma_{\bob,\charlie}$.
\item $\ch$ samples 
 $b\xleftarrow{\$} \{0,1\}$.
\item $\ch$ computes $\ct_b$ as follows:
    \begin{enumerate}
    \item Sample $r\xleftarrow{\$}\ZQ^\secparam$, and compute ${c'}\gets \cllz.\enc(\pk,r)$.
    \item Sample $u\xleftarrow{\$}\ZQ^\secparam$ and compute $c_b=(u,\langle u, r\rangle)$ if $b=0$, else \sout{sample $m\xleftarrow{\$}\{0,1\}^M$ (where $M$ is the bit-size of the messages), generate $\tilde{m}_Q\gets \embed_Q(m)$ (see Fig.~\ref{fig:CLLZ-SDE-postproc-construction} for the definition of $\embed_Q$) and compute $c_b=(u,\tilde{m}+ \langle u, r\rangle)$ (where the operations are in the field $\ZQ$) if $b=1$}  \cblue{sample $\tilde{m} \xleftarrow{\$} \ZQ$, and compute $c_b=(u,\tilde{m}+ \langle u, r\rangle)$ (where the operations are in the field $\ZQ$) if $b=1$}.
    \item Set $\ct_b=(c_b,{c'})$.
    \end{enumerate}
\item Apply $(\bob(\ct_b,\cdot) \otimes \charlie(\ct_b,\cdot))(\sigma_{\bob,\charlie})$ to obtain $(b_{\bfB},b_{\bfC})$. 
\item Output $1$ if $b_{\bfB}=b_\bfC=b$.
\end{enumerate}

The indistinguishability holds since the overall distribution of $\ct_b$ did not change across hybrids $\hybrid_0$ and $\hybrid_1$.

Consider the following search experiment against a pair of (uniform) efficient adversaries $\bob',\charlie'$.
\begin{enumerate}
\item $\ch$ samples $r\xleftarrow{\$}\ZQ^\secparam$.
\item $\ch$ computes $\sigma_{\bob,\charlie}$ as follows:
\begin{enumerate}
    \item Sample $(\sk,\pk) \gets \keygen(1^\secparam)$ and prepares $\rho_k\gets \qkeygen(k)$.
    \item Run $\alice(\rho_k,\pk)$ to get $\sigma_{\bob,\charlie}$.
\end{enumerate}
\item $\ch$ computes ${c'}\gets \cllz.\enc(\pk,r)$.
\item $\ch$ constructs the bipartite auxiliary state $\tau^{r}_{B,C}={c'}^\bob,\sigma_{\bob,\charlie},{c'}^\charlie$, i.e., the ${c'}^\bob,\sigma_\bob$ and ${c'}^\charlie,\sigma_\charlie$ are the two partitions, where ${c'}^\bob={c'}^\charlie={c'}$.
\item $\ch$ sends the respective registers of $\tau^{r}_{B,C}$ to $\bob'$ and $\charlie'$, and gets back the responses ${r'}^\bob$ and ${r'}^\charlie$ respectively.
\item Ouptput $1$ if ${r'}^\bob={r'}^\charlie=r$.
\end{enumerate}
 
Clearly, the winning probability of $(\bob',\charlie')$ in the above game is the same as the winning probability of $(\alice,\bob',\charlie')$ in the $\distrid$-search anti-piracy (see \Cref{subsec:sde}) of the $\cllz$ $\sde$ scheme given in \Cref{fig:CLLZ-SDE-construction}. Assuming the $\cllz$ $\sde$ satisfies  $\distrid$-search anti-piracy (see \Cref{subsec:sde}), there exists a negligible function $\epsilon'()$ such that the winning probability of $(\bob',\charlie')$ in the above game is $\epsilon'(\secparam)$. Since $Q>2^\secparam$, by \Cref{conj:goldreich-levin-correlated}, there exists a negligible function $\epsilon()$ such that the winning probability of $(\bob,\charlie)$ in the following indistinguishability game is at most $\frac{1}{2}+\epsilon(\secparam)$

\begin{enumerate}
\item $\ch$ samples $r\xleftarrow{\$}\ZQ^\secparam$.
\item $\ch$ computes $\sigma_{\bob,\charlie}$ as follows:
\begin{enumerate}
    \item Sample $(\sk,\pk) \gets \keygen(1^\secparam)$ and prepares $\rho_k\gets \qkeygen(k)$.
    \item Run $\alice(\rho_k,\pk)$ to get $\sigma_{\bob,\charlie}$.
\end{enumerate}
\item $\ch$ computes ${c'}\gets \cllz.\enc(\pk,r)$.
\item $\ch$ constructs the bipartite auxiliary state $\tau^{r}_{B,C}={c'}^\bob,\sigma_{\bob,\charlie},{c'}^\charlie$, i.e., the ${c'}^\bob,\sigma_\bob$ and ${c'}^\charlie,\sigma_\charlie$ are the two partitions, where ${c'}^\bob={c'}^\charlie={c'}$.
\item $\ch$ samples  $b\xleftarrow{\$} \{0,1\}$.
\item $\ch$ samples $u\xleftarrow{\$}\{0,1\}^q$ and compute $c_b=(u,\langle u, r\rangle)$ if $b=0$, else computes $c_b=(u,m)$ if $b=1$.
\item $\ch$ sends $c^\bob_b$ and $c^\charlie_b$ along with the respective registers of $\tau^{r^\bob,r^\charlie}_{B,C}$ to $\bob'$ and $\charlie'$ respectively, where $c^\bob_b=c^\charlie_b=c_b$ and gets back the responses $b^\bob$ and $b^\charlie$ respectively.
\item Output $1$ if $b_{\bfB}=b_\bfC=b$.
\end{enumerate}

\begin{enumerate}
\item $\ch$ samples $r^\bob,r^\charlie\xleftarrow{\$}\ZQ^\secparam$.
\item $\ch$ computes $\sigma_{\bob,\charlie}$ as follows:
\begin{enumerate}
    \item Sample $(\sk,\pk) \gets \keygen(1^\secparam)$ and prepares $\rho_k\gets \qkeygen(k)$.
    \item Run $\alice(\rho_k,\pk)$ to get $\sigma_{\bob,\charlie}$.
\end{enumerate}
\item $\ch$ computes ${c'}^\bob\gets \cllz.\enc(\pk,r^\bob)$, and  computes ${c'}^\charlie\gets \cllz.\enc(\pk,r^\charlie)$.
\item $\ch$ constructs the bipartite auxiliary state $\tau^{r^\bob,r^\charlie}_{B,C}={c'}^\bob,\sigma_{\bob,\charlie},{c'}^\charlie$, i.e., the ${c'}^\bob,\sigma_\bob$ and ${c'}^\charlie,\sigma_\charlie$ are the two partitions.
\item $\ch$ sends the respective registers of $\tau^{r^\bob,r^\charlie}_{B,C}$ to $\bob'$ and $\charlie'$, and gets back the responses ${r'}^\bob$ and ${r'}^\charlie$ respectively.
\item Ouptput $1$ if ${r'}^\bob=r^\bob$, and  ${r'}^\charlie=r^\charlie$.
\end{enumerate}

However, note that the view of the adversaries $\bob$ and $\charlie$ in the indistinguishability game above is the same as the view in $\hybrid_3$. Therefore, the winning probability of $(\alice,\bob,\charlie)$ in $\hybrid_1$ is at most $\frac{1}{2}+\epsilon(\secparam)$. This completes the proof of the lemma.

\end{proof}

\begin{lemma}\label{lemma:identical-search-anti-piracy-postproc}
     Assuimng post-quantum sub-exponentially secure $\io$ and quantum hardness of Learning-with-errors problem (LWE), the $\cllz$ $\sde$ (see \Cref{fig:CLLZ-SDE-construction}) satisfies $\distrid$-search anti-piracy (see \Cref{subsec:sde}).
\end{lemma}
\begin{proof}
    By~\cite[Theorem 6.15]{CLLZ21}, assuming the security of post-quantum sub-exponentially secure $\io$ and one-way functions, and quantum hardness of Learning-with-errors problem (LWE)
    , the $\cllz$ $\sde$  (see \Cref{fig:CLLZ-SDE-construction}) satisfies independent search anti-piracy. 
    Since the trivial success probabilities of the $\distrprod$-search anti-piracy and $\distrid$-search anti-piracy expeirments for $\sde$ are both negligible, by the lifting result in~\cite[Theorem ]{AKL23}\anote{Prabhanjan, can you point to the lifting theorem from the cloning games paper?}, we conclude that \Cref{lemma:identical-search-anti-piracy-postproc} holds.
\end{proof}

%% file: proofofcp_construct.tex
\ifllncs 
\subsection{Proof of~\Cref{prop:cpcpiracy_PRF_from_cpcpiracysde-noncllz,prop:cpcpiracy_PRF_from_cpcpiracysde-id-noncllz}}
\label{sec:app:proof:cp}
\else 

\fi

\begin{proof}[\textbf{Proof of \cref{prop:cpcpiracy_PRF_from_cpcpiracysde-noncllz}}]
    To prove the proposition, we adopt the proof of~\cite[Theorem 7.12, Appendix F]{CLLZ21}. 
\par We will start with a series of hybrids. 
The changes are marked in \cblue{blue}. \\

\noindent \underline{$\hybrid_0$}: Same as $\prepexpt^{\left(\alice,\bob,\charlie \right),\cpscheme,\distr_\inpclass}\left( 1^{\secparam} \right)$ 
 (see Game~\ref{fig:preponedsecurity}) where $\distr=\distrprod$ (see the definition in \Cref{def:cpcpiracy_PRFs}) for the $\cllz$ copy-protection scheme see~\Cref{fig:copy-protect-CLLZ}.
\begin{enumerate}
    \item $\ch$ samples $K_1\gets \prf.\gen(1^\secparam)$ and generates $\rho=(\{\ket{{A_i}_{s_i,{s'}_i}}\}_{i\in \ell_0},\io(P))\gets \cllz.\qkeygen(K_1)$, and sends $\rho$ to $\alice$. $P$ has $K_1,K_2,K_3$ hardcoded in it where $K_2,K_3$ are the secondary  keys.
    \item $\ch$ generates $x^\bob,x^\charlie\xleftarrow{\$}\{0,1\}^n$, where $x^\bob=x^\bob_0\|x^\bob_1\|x^\bob_2, x^\charlie=x^\charlie_0\|x^\charlie_1\|x^\charlie_2$ and computes $y_0^\bob\gets \prf.\eval(K_1,x^\bob)$ and $y^\charlie_0\gets \prf.\eval(K_1,x^\charlie)$.
    \item $\ch$ also samples $y^\bob_1,y^\charlie_1 \xleftarrow{\$} \{0,1\}^m$.
    \item $\ch$ samples $b\xleftarrow{\$} \{0,1\}$, and sends $\alice$, $(\rho,y^\bob_b,y^\charlie_b)$.
    \item $\alice$ on receiving $(\rho,y^\bob_b,y^\charlie_b)$ produces a bipartite state  $\sigma_{\bob,\charlie}$.
    \item  Apply $(\bob(x^\bob,\cdot) \otimes \charlie(x^\charlie,\cdot))(\sigma_{\bob,\charlie})$ to obtain $(b^{\bob},b^{\charlie})$. 
    \item Output $1$ if $b^\bob=b^\charlie=b$, else $0$. 
\end{enumerate}

\noindent \underline{$\hybrid_1$}: We modify the sampling procedure of the challenge inputs $x^{\bob}$ and $x^{\charlie}$. 
\begin{enumerate}
    \item $\ch$ samples $K_1\gets \prf.\gen(1^\secparam)$ and generates $\rho=(\{\ket{{A_i}_{s_i,{s'}_i}}\}_{i\in \ell_0},\io(P))\gets \cllz.\qkeygen(K_1)$, and sends $\rho$ to $\alice$. $P$ has $K_1,K_2,K_3$ hardcoded in it where $K_2,K_3$ are the secondary  keys.
    \item $\ch$ generates $x^\bob,x^\charlie\xleftarrow{\$}\{0,1\}^n$, where $x^\bob=x^\bob_0\|x^\bob_1\|x^\bob_2, x^\charlie=x^\charlie_0\|x^\charlie_1\|x^\charlie_2$
    and computes $y_0^\bob\gets \prf.\eval(K_1,x^\bob)$ and $y^\charlie_0\gets \prf.\eval(K_1,x^\charlie)$. 
    \item \cblue{$\ch$ also computes $x^\bob_\trig\gets \gentrig(x^\bob_0,y^\bob_0,K_2,K_3,\{{A_i}_{s_i,{s'}_i}\}_{i\in \ell_0})$,
    
    and $x^\charlie_\trig\gets \gentrig(x^\charlie_0,y^\charlie,K_2,K_3,\{{A_i}_{s_i,{s'}_i}\}_{i\in \ell_0})$.}
    \item $\ch$ also samples $y^\bob_1,y^\charlie_1 \xleftarrow{\$} \{0,1\}^m$.
    \item $\ch$ samples $b\xleftarrow{\$} \{0,1\}$ and sends $\alice$ $(\rho,y^\bob_b,y^\charlie_b)$.
    \item $\alice$ on receiving $(\rho,y^\bob_b,y^\charlie_b)$ produces a bipartite state  $\sigma_{\bob,\charlie}$.
    \item  Apply $(\bob(($ \sout{$x^\bob$} \cblue{$x^\bob_\trig$} $,\cdot) \otimes \charlie(($ \sout{$x^\charlie$} \cblue{$x^\charlie_\trig$} $,\cdot))(\sigma_{\bob,\charlie})$ to obtain $(b^{\bob},b^{\charlie})$. 
    \item Output $1$ if $b^\bob=b^\charlie=b$, else $0$. 
\end{enumerate}

\begin{claim}
Assuming the security of $\prf$, hybrids $\hybrid_1$ and $\hybrid_2$ are computationally indistinguishable. 
\end{claim}
\begin{proof}
$\hybrid_1$ is computationally indistinguishable from $\hybrid_0$ due to~\cite[Lemma 7.17]{CLLZ21}. The same arguments via~\cite[Lemma 7.17]{CLLZ21} were made in showing the indistinguishability between hybrids $\hybrid_0$ and $\hybrid_1$ in the proof of~\cite[Theorem 7.12]{CLLZ21}. \pnote{I'm a bit confused: how does the indistinguishability follow from Lemma 7.17 and also 7.12? is it supposed to be just one. I have added the claim below, please check} \anote{Yeah I mentioned the theorem 7.12 since this hybrid argument was made there as well in a similar fashion, but I agree that it is confusing}\anote{Is it better now?}

\end{proof}

\noindent \underline{$\hybrid_2$}: We modify the generation of the outputs $y_0^{\bob}$ and $y_{0}^{\charlie}$. 
\begin{enumerate}
    \item $\ch$ samples $K_1\gets \prf.\gen(1^\secparam)$ and generates $\rho=(\{\ket{{A_i}_{s_i,{s'}_i}}\}_{i\in \ell_0},\io(P))\gets \cllz.\qkeygen(K_1)$, and sends $\rho$ to $\alice$. $P$ has $K_1,K_2,K_3$ hardcoded in it where $K_2,K_3$ are the secondary  keys.
    \item $\ch$ generates $x^\bob,x^\charlie\xleftarrow{\$}\{0,1\}^n$, where $x^\bob=x^\bob_0\|x^\bob_1\|x^\bob_2, x^\charlie=x^\charlie_0\|x^\charlie_1\|x^\charlie_2$
    and \sout{computes $y_0^\bob\gets \prf.\eval(K_1,x^\bob)$ and $y^\charlie_0\gets \prf.\eval(K_1,x^\charlie)$} \cblue{samples $y^\bob_0,y^\charlie_0\xleftarrow{\$}\{0,1\}^m$}. 
    \item $\ch$ also computes $x^\bob_\trig\gets \gentrig(x^\bob_0,y^\bob_0,K_2,K_3,\{{A_i}_{s_i,{s'}_i}\}_{i\in \ell_0})$,

    and $x^\charlie_\trig\gets \gentrig(x^\charlie_0,y^\charlie,K_2,K_3,\{{A_i}_{s_i,{s'}_i}\}_{i\in \ell_0})$.
    \item $\ch$ also samples $y^\bob_1,y^\charlie_1 \xleftarrow{\$} \{0,1\}^m$.
    \item $\ch$ samples $b\xleftarrow{\$} \{0,1\}$ and sends $\alice$ $(\rho,y^\bob_b,y^\charlie_b)$.
    \item $\alice$ on receiving $(\rho,y^\bob_b,y^\charlie_b)$  produces a bipartite state  $\sigma_{\bob,\charlie}$.
    \item  Apply $(\bob(x^\bob_\trig,\cdot) \otimes \charlie(x^\charlie_\trig,\cdot))(\sigma_{\bob,\charlie})$ to obtain $(b^{\bob},b^{\charlie})$. 
    \item Output $1$ if $b^\bob=b^\charlie=b$, else $0$. 
\end{enumerate}

\noindent $\hybrid_2$ is statistically indistinguishable from $\hybrid_1$ due to the extractor properties of the primary $\prf$ family. For more details, refer to the proof of see~\cite[Theorem 7.12]{CLLZ21}. 

\begin{claim}
Assuming the extractor properties of $\prf$, hybrids $\hybrid_2$ and $\hybrid_3$ are statistically indistinguishable. 
\end{claim}
\begin{proof}
The proof is identical to the proof of indistinguishability of $\hybrid_1$ and $\hybrid_2$ in the proof of~\cite[Theorem 7.12]{CLLZ21}.
\end{proof}

\noindent \underline{$\hybrid_3$}: This hybrid is a reformulation of $\hybrid_2$ in terms of the $\cllz$ $\sde$ scheme, see \cref{fig:CLLZ-SDE-construction}.
\begin{enumerate}
    \item \cblue{$\ch$ samples $\{{A_i}_{s_i,{s'}_i}\}_{i\in \ell_0}$ and generates $\{\ket{{A_i}_{s_i,{s'}_i}}\}_{i\in \ell_0}$, and treats it as the quantum decryption key for the $\cllz$ 
single-decryptor encryption scheme (see \cref{fig:CLLZ-SDE-construction}), where the secret key is $\{{A_i}_{s_i,{s'}_i}\}_{i\in \ell_0}$. $\ch$ also generates $\pk=\{R^0_i,R^1_i\}_{i\in \ell_0}$, where for every $i\in [\ell_0]$, $R^0_i=\io(A_i + s_i)$ and $R^1_i=\io(A_i^\perp + {s'}_i)$.}  
    \item $\ch$ \sout{generates $x^\bob,x^\charlie\xleftarrow{\$}\{0,1\}^n$, where $x^\bob=x^\bob_0\|x^\bob_1\|x^\bob_2, x^\charlie=x^\charlie_0\|x^\charlie_1\|x^\charlie_2$
    and} samples $y^\bob_0,y^\charlie_0\xleftarrow{\$}\{0,1\}^m$. 
    \item \sout{$\ch$ also computes $x^\bob_\trig\gets \gentrig(x^\bob_0,y^\bob_0,K_2,K_3,\{{A_i}_{s_i,{s'}_i}\}_{i\in \ell_0})$, \\

    and $x^\charlie_\trig\gets \gentrig(x^\charlie_0,y^\charlie,K_2,K_3,\{{A_i}_{s_i,{s'}_i}\}_{i\in \ell_0})$.}
    \item $\ch$ also samples $y^\bob_1,y^\charlie_1 \xleftarrow{\$} \{0,1\}^m$.
    \item $\ch$ samples $b\xleftarrow{\$} \{0,1\}$\cblue{, and generates $x^\bob_0,Q^\bob\gets \cllz.\enc(\pk,y^\bob_b)$ and $x^\charlie_0,Q^\charlie\gets \cllz.\enc(\pk,y^\charlie_b)$.   }
    \item \cblue{$\ch$ samples keys $K_1,K_2,K_3$ and constructs the program $P$ which hardcodes $K_1,K_2,K_3$. It then prepares $\rho=(\{\ket{{A_i}_{s_i,{s'}_i}}\}_{i\in \ell_0},\io(P))$ and sends to $\alice$.}
     \item $\alice$ on receiving  $(\rho,y^\bob_b,y^\charlie_b)$  produces a bipartite state  $\sigma_{\bob,\charlie}$.
     \item \cblue{$\ch$ then generates $x^\bob_\trig,x^\charlie_\trig\in \{0,1\}^n$  as follows:
     \begin{enumerate}
         \item Let ${x^\bob_\trig}_1=F_2(K_2,x^\bob_0\|Q^\bob)$ and ${x^\bob_\trig}_2=F_3(K_3,{x^\bob_\trig}_1)$. Let $x^\bob_\trig=x^\bob_0\|{x^\bob_\trig}_1\|{x^\bob_\trig}_2.$
         \item Let ${x^\charlie_\trig}_1=F_2(K_2,x^\charlie_0\|Q^\charlie)$ and ${x^\charlie_\trig}_2=F_3(K_3,{x^\charlie_\trig}_1)$. Let $x^\charlie_\trig=x^\charlie_0\|{x^\charlie_\trig}_1\|{x^\charlie_\trig}_2.$
     \end{enumerate}
     
     }
    \item  Apply $(\bob(x^\bob_\trig,\cdot) \otimes \charlie(x^\charlie_\trig,\cdot))(\sigma_{\bob,\charlie})$ to obtain $(b^{\bob},b^{\charlie})$. 
    \item Output $1$ if $b^\bob=b^\charlie=b$, else $0$. 
\end{enumerate}

\begin{claim}
The output distributions of the hybrids $\hybrid_2$ and $\hybrid_3$ are identically distributed. 
\end{claim}
\begin{proof}
The proof is identical to the proof of indistinguishability of $\hybrid_2$ and $\hybrid_3$ in the proof of~\cite[Theorem 7.12]{CLLZ21}.
\end{proof}


\noindent Finally we give a reduction from $\hybrid_{3}$ to the indistinguishability from random anti-piracy experiment (\cref{fig:indistinguishability_from_random_-sde-anti-piracy}) for $\postproc$ single-decryptor encryption scheme, where $\cllz$ $\sde$ is the one given in \cref{fig:CLLZ-SDE-construction}, for more details see~\cite[Construction 1, Section 6.3, pg. 39]{CLLZ21}. Let $(\alice,\bob,\charlie)$ be an adversary in $\hybrid_{3}$ above. Consider the following non-local adversary $(\reduc_\alice,\reduc_\bob,\reduc_\charlie)$:
\begin{enumerate}
    \item $\reduc_\alice$ samples $y^\bob_0,y^\bob_1,y^\charlie_0,y^\charlie_1\xleftarrow{\$}\{0,1\}^m$.
    \item $\reduc_\alice$ gets the quantum decryptor $\{\ket{{A_i}_{s_i,{s'}_i}}\}_{i\in \ell_0}$ and a public key $\pk=(R^0_i,R^1_i)$ from $\ch$, the challenger in the correlated challenge $\SDE$ anti-piracy experiment (see \cref{fig:correlated-sde-cpa-style-anti-piracy}) for the $\cllz$ $\SDE$ scheme. 
    \item $\reduc_\alice$ samples $K_1,K_2,K_3$ and prepares the circuit $P$ using $R^0_i,R^1_i$ and the keys $K_1,K_2,K_3$. Let $\rho=\{\ket{{A_i}_{s_i,{s'}_i}}\}_{i\in \ell_0},\io(P))$.

    \item $\reduc_\alice$ samples a bit $d\xleftarrow{\$}\{0,1\}$
     and runs $\alice$ on  $(\rho,y^\bob_d,y^\charlie_d)$ and gets back the output $\sigma_{\bob,\charlie}$.
    \item $\reduc_\alice$ sends $(K_1,K_2,K_3,d,\sigma_\bob)$ to $\reduc_\bob$ and  $(K_1,K_2,K_3,d,\sigma_\charlie)$ to $\reduc_\charlie$.
    \item $\reduc_\bob$ on receiving $(c^\bob,(x^\bob_0,T^\bob))$  as the challenge cipher text from $\ch$ as the challenge ciphertext and $K_1,K_2,K_3,d,\sigma_\bob$ from $\reduc_\alice$, does the following:
    \begin{enumerate}
        \item $\reduc_\bob$ generates the circuit $Q^\bob$ which on any input $x_0$ generates $r\gets T^\bob(x_0)$ and if the output is $\bot$ outputs $\bot$, else computes $\decproc(c^\bob,r)$ and if the outcome is $0$, output $y^\bob_0$, else output $y^\bob_1$. $\reduc_\bob$ generates $\tilde{Q}^\bob\gets \io(Q^\bob)$.
        \item $\reduc_\bob$ constructs $x^\bob_\trig$ as follows. Let ${x^\bob_\trig}_1=F_2(K_2,x^\bob_0\|\widetilde{Q}^\bob)$ and ${x^\bob_\trig}_2=F_3(K_3,{x^\bob_\trig}_1)$. Let $x^\bob_\trig=x^\bob_0\|{x^\bob_\trig}_1\|{x^\bob_\trig}_2.$
        \item $\reduc_\bob$ runs $\bob$ on $(x^\bob_\trig,\sigma_\bob)$ to get an output $b^\bob$.
        \item $\reduc_\bob$ outputs $b^\bob \oplus d$.
    \end{enumerate}

    \item Similarly, $\reduc_\charlie$ on receiving $(c^\charlie,(x^\charlie_0,T^\charlie))$ as the challenge cipher text from $\ch$ and $K_1,K_2,K_3,d,\sigma_\charlie$ from $\reduc_\alice$, does the following:
    \begin{enumerate}
         \item $\reduc_\charlie$ generates the circuit $Q^\charlie$ which on any input $x_0$ generates $r\gets T^\charlie(x_0)$ and if the output is $\bot$ outputs $\bot$, else computes $\decproc(c^\bob,r)$ and if the outcome is $0$, output $y^\charlie_0$, else output $y^\charlie_1$. $\reduc_\charlie$ generates $\tilde{Q}^\charlie\gets \io(Q^\charlie)$.
        \item $\reduc_\charlie$ constructs $x^\charlie_\trig$ as follows. Let ${x^\charlie_\trig}_1=F_2(K_2,x^\charlie_0\|\widetilde{Q}^\charlie)$ and ${x^\charlie_\trig}_2=F_3(K_3,{x^\charlie_\trig}_1)$. Let $x^\charlie_\trig=x^\charlie_0\|{x^\charlie_\trig}_1\|{x^\charlie_\trig}_2.$
        \item $\reduc_\charlie$ runs $\charlie$ on $(x^\charlie_\trig,\sigma_\charlie)$ to get an output $b^\charlie$.
        \item $\reduc_\charlie$ outputs $b^\charlie \oplus d$.
    \end{enumerate}

\end{enumerate}
Note that the functionality of $Q^\bob$ and $Q^\charlie$ are the same as that of $W^\bob,W^\charlie$ in the ciphertexts $(x^\bob_0,W^\bob)$ and $(x^\charlie_0,W^\charlie)$ obtained by running $\cllz.\enc(\pk,\cdot)$ algorithm on $y^\bob_b$ and $y^\charlie_b$ with $x^\bob_0$ and $x^\charlie_0$ as the randomness respectively.
Note that in $\hybrid_3$, $\bob$ (and similarly, $\charlie$) needs to distinguish between the following two inputs: a random string $y^\bob$ along with either a triggered input $x^\bob$ encoding $y^\bob$ which is also the view of the inside adversary in the reduction above in the event $b=d$ in the simulated experiment; or a  triggered input $x^\bob$ encoding $\tilde{y}^\bob$ random string where $\tilde{y}^\bob\xleftarrow{\$}$ sampled independent of $y^\bob$,  which is the view of the inside adversary in the reduction above in the event $b\neq d$ in the simulated experiment. Therefore, by the $\io$ guarantees, the view of the inside $\alice,\bob,\charlie$ is the same as that in $\hybrid_3$.  
   
\end{proof}

\begin{proof}[\textbf{Proof of \Cref{prop:cpcpiracy_PRF_from_cpcpiracysde-id-noncllz}}]
The proof is the same as the proof for \Cref{prop:cpcpiracy_PRF_from_cpcpiracysde-noncllz} up to minor changes. 
\par We will start with a series of hybrids. 
The changes are marked in \cblue{blue}. \\

\noindent \underline{$\hybrid_0$}: Same as $\prepexpt^{\left(\alice,\bob,\charlie \right),\cpscheme,\distr_\inpclass}\left( 1^{\secparam} \right)$ 
 (see Game~\ref{fig:preponedsecurity}) where $\distr=\distrid$ (see the definition in \Cref{def:cpcpiracy_PRFs}) for the $\cllz$ copy-protection scheme see~\Cref{fig:copy-protect-CLLZ}.
\begin{enumerate}
    \item $\ch$ samples $K_1\gets \prf.\gen(1^\secparam)$ and generates $\rho=(\{\ket{{A_i}_{s_i,{s'}_i}}\}_{i\in \ell_0},\io(P))\gets \cllz.\qkeygen(K_1)$, and sends $\rho$ to $\alice$. $P$ has $K_1,K_2,K_3$ hardcoded in it where $K_2,K_3$ are the secondary  keys.
    \item $\ch$ generates $x\xleftarrow{\$}\{0,1\}^n$, where $x=x_0\|x_1\|x_2$ and computes $y_0\gets \prf.\eval(K_1,x)$.
    \item $\ch$ also samples $y_1 \xleftarrow{\$} \{0,1\}^m$.
    \item $\ch$ samples $b\xleftarrow{\$} \{0,1\}$, and sends $\alice$, $(\rho,y_b,y_b)$.
    \item $\alice$ on receiving $(\rho,y_b,y_b)$ produces a bipartite state  $\sigma_{\bob,\charlie}$.
    \item  Apply $(\bob(x,\cdot) \otimes \charlie(x,\cdot))(\sigma_{\bob,\charlie})$ to obtain $(b^{\bob},b^{\charlie})$. 
    \item Output $1$ if $b^\bob=b^\charlie=b$, else $0$. 
\end{enumerate}

\noindent \underline{$\hybrid_1$}: We modify the sampling procedure of the challenge input $x$. 
\begin{enumerate}
    \item $\ch$ samples $K_1\gets \prf.\gen(1^\secparam)$ and generates $\rho=(\{\ket{{A_i}_{s_i,{s'}_i}}\}_{i\in \ell_0},\io(P))\gets \cllz.\qkeygen(K_1)$, and sends $\rho$ to $\alice$. $P$ has $K_1,K_2,K_3$ hardcoded in it where $K_2,K_3$ are the secondary  keys.
    \item $\ch$ generates $x\xleftarrow{\$}\{0,1\}^n$, where $x=x_0\|x_1\|x_2$ and computes $y_0\gets \prf.\eval(K_1,x)$.
    \item $\ch$ also samples $y_1 \xleftarrow{\$} \{0,1\}^m$.
    \item \cblue{$\ch$ also computes $x_\trig\gets \gentrig(x_0,y_0,K_2,K_3,\{{A_i}_{s_i,{s'}_i}\}_{i\in \ell_0})$.}
     \item $\ch$ also samples $y_1 \xleftarrow{\$} \{0,1\}^m$.
    \item $\ch$ samples $b\xleftarrow{\$} \{0,1\}$, and sends $\alice$, $(\rho,y_b,y_b)$.
    \item $\alice$ on receiving $(\rho,y_b,y_b)$ produces a bipartite state  $\sigma_{\bob,\charlie}$.
    \item  Apply $(\bob($ \sout{$x$} \cblue{$x_\trig$} $,\cdot) \otimes \charlie($ \sout{$x$} \cblue{$x_\trig$} $,\cdot))(\sigma_{\bob,\charlie})$ to obtain $(b^{\bob},b^{\charlie})$. 
    \item Output $1$ if $b^\bob=b^\charlie=b$, else $0$. 
\end{enumerate}

$\hybrid_1$ is computationally indistinguishable from $\hybrid_0$ due to~\cite[Lemma 7.17]{CLLZ21}. The same arguments via~\cite[Lemma 7.17]{CLLZ21} were made in showing the indistinguishability between hybrids $\hybrid_0$ and $\hybrid_1$ in the proof of~\cite[Theorem 7.12]{CLLZ21}. 

\begin{claim}
Assuming the security of $\prf$, hybrids $\hybrid_0$ and $\hybrid_1$ are computationally indistinguishable. 
\end{claim}
\begin{proof}
The proof is identical to the proof of indistinguishability of $\hybrid_0$ and $\hybrid_1$ in the proof of~\cite[Theorem 7.12]{CLLZ21}.
\end{proof}

\noindent \underline{$\hybrid_2$}: We modify the generation of the outputs $y_0$. 
\begin{enumerate}
    \item $\ch$ samples $K_1\gets \prf.\gen(1^\secparam)$ and generates $\rho=(\{\ket{{A_i}_{s_i,{s'}_i}}\}_{i\in \ell_0},\io(P))\gets \cllz.\qkeygen(K_1)$, and sends $\rho$ to $\alice$. $P$ has $K_1,K_2,K_3$ hardcoded in it where $K_2,K_3$ are the secondary  keys.
    \item $\ch$ generates $x\xleftarrow{\$}\{0,1\}^n$, where $x=x_0\|x_1\|x_2,$
    and \sout{computes $y_0^\gets \prf.\eval(K_1,x)$} \cblue{samples $y_0\xleftarrow{\$}\{0,1\}^m$}. 
    \item $\ch$ also computes $x_\trig\gets \gentrig(x_0,y_0,K_2,K_3,\{{A_i}_{s_i,{s'}_i}\}_{i\in \ell_0})$.
    \item $\ch$ also samples $y_1 \xleftarrow{\$} \{0,1\}^m$.
    \item $\ch$ samples $b\xleftarrow{\$} \{0,1\}$ and sends $\alice$ $(\rho,y_b,y_b)$.
    \item $\alice$ on receiving $(\rho,y_b,y_b)$  produces a bipartite state  $\sigma_{\bob,\charlie}$.
    \item  Apply $(\bob(x_\trig,\cdot) \otimes \charlie(x_\trig,\cdot))(\sigma_{\bob,\charlie})$ to obtain $(b^{\bob},b^{\charlie})$. 
    \item Output $1$ if $b^\bob=b^\charlie=b$, else $0$. 
\end{enumerate}

\noindent $\hybrid_2$ is statistically indistinguishable from $\hybrid_1$ due to the extractor properties of the primary $\prf$ family. For more details, refer to the proof of see~\cite[Theorem 7.12]{CLLZ21}. 

\begin{claim}
Assuming the extractor properties of $\prf$, hybrids $\hybrid_1$ and $\hybrid_2$ are statistically indistinguishable. 
\end{claim}
\begin{proof}
The proof is identical to the proof of indistinguishability of $\hybrid_1$ and $\hybrid_2$ in the proof of~\cite[Theorem 7.12]{CLLZ21}.
\end{proof}

\noindent \underline{$\hybrid_3$}: This hybrid is a reformulation of $\hybrid_2$.
\begin{enumerate}
    \item \cblue{$\ch$ samples $\{{A_i}_{s_i,{s'}_i}\}_{i\in \ell_0}$ and generates $\{\ket{{A_i}_{s_i,{s'}_i}}\}_{i\in \ell_0}$, and treats it as the quantum decryption key for the $\cllz$ 
single-decryptor encryption scheme (see \cref{fig:CLLZ-SDE-construction}), where the secret key is $\{{A_i}_{s_i,{s'}_i}\}_{i\in \ell_0}$. $\ch$ also generates $\pk=\{R^0_i,R^1_i\}_{i\in \ell_0}$, where for every $i\in [\ell_0]$, $R^0_i=\io(A_i + s_i)$ and $R^1_i=\io(A_i^\perp + {s'}_i)$.}  
    \item $\ch$ \sout{generates $x\xleftarrow{\$}\{0,1\}^n$, where $x=x_0\|x_1\|x_2$
    and} samples $y_0\xleftarrow{\$}\{0,1\}^m$. 
    \item \sout{$\ch$ also computes $x_\trig\gets \gentrig(x_0,y_0,K_2,K_3,\{{A_i}_{s_i,{s'}_i}\}_{i\in \ell_0})$, 
    }
    \item $\ch$ also samples $y_1\xleftarrow{\$} \{0,1\}^m$.
    \item $\ch$ samples $b\xleftarrow{\$} \{0,1\}$\cblue{, and generates $x_0,Q\gets \cllz.\enc(\pk,y_b)$.}
    \item \cblue{$\ch$ samples keys $K_1,K_2,K_3$ and constructs the program $P$ which hardcodes $K_1,K_2,K_3$. It then prepares $\rho=(\{\ket{{A_i}_{s_i,{s'}_i}}\}_{i\in \ell_0},\io(P))$ and sends to $\alice$.}
     \item $\alice$ on receiving  $(\rho,y_b,y_b)$  produces a bipartite state  $\sigma_{\bob,\charlie}$.
     \item \cblue{$\ch$ then generates $x_\trig\in \{0,1\}^n$  as follows:
     Let ${x_\trig}_1=F_2(K_2,x_0\|Q^\bob)$ and ${x_\trig}_2=F_3(K_3,{x_\trig}_1)$. Let $x^\bob_\trig=x_0\|{x_\trig}_1\|{x_\trig}_2.$

     }
    \item  Apply $(\bob(x^\bob_\trig,\cdot) \otimes \charlie(x^\charlie_\trig,\cdot))(\sigma_{\bob,\charlie})$ to obtain $(b^{\bob},b^{\charlie})$. 
    \item Output $1$ if $b^\bob=b^\charlie=b$, else $0$. 
\end{enumerate}

\begin{claim}
The output distributions of the hybrids $\hybrid_2$ and $\hybrid_3$ are identically distributed. 
\end{claim}
\begin{proof}
The proof is identical to the proof of indistinguishability of $\hybrid_2$ and $\hybrid_3$ in the proof of~\cite[Theorem 7.12]{CLLZ21}.
\end{proof}


\noindent Finally we give a reduction from $\hybrid_{3}$ to the indistinguishability from random anti-piracy experiment (\cref{fig:indistinguishability_from_random_-sde-anti-piracy}) for $\postproc$ single-decryptor encryption scheme, where $\cllz$ $\sde$ is the one given in \cref{fig:CLLZ-SDE-construction}, for more details see~\cite[Construction 1, Section 6.3, pg. 39]{CLLZ21}. Let $(\alice,\bob,\charlie)$ be an adversary in $\hybrid_{3}$ above. Consider the following non-local adversary $(\reduc_\alice,\reduc_\bob,\reduc_\charlie)$:
\begin{enumerate}
    \item $\reduc_\alice$ samples $y_0,y_1\xleftarrow{\$}\{0,1\}^m$.
    \item $\reduc_\alice$ gets the quantum decryptor $\{\ket{{A_i}_{s_i,{s'}_i}}\}_{i\in \ell_0}$ and a public key $\pk=(R^0_i,R^1_i)$ from $\ch$, the challenger in the correlated challenge $\SDE$ anti-piracy experiment (see \cref{fig:correlated-sde-cpa-style-anti-piracy}) for the $\cllz$ $\SDE$ scheme. 
    \item $\reduc_\alice$ samples $K_1,K_2,K_3$ and prepares the circuit $P$ using $R^0_i,R^1_i$ and the keys $K_1,K_2,K_3$. Let $\rho=\{\ket{{A_i}_{s_i,{s'}_i}}\}_{i\in \ell_0},\io(P))$.

    \item $\reduc_\alice$ samples a bit $d\xleftarrow{\$}\{0,1\}$
     and runs $\alice$ on  $(\rho,y_d,y_d)$ and gets back the output $\sigma_{\bob,\charlie}$.
    \item $\reduc_\alice$ samples a random string $s\xleftarrow{\$}$ of appropriate length as required by $\bob$ and $\charlie$ to run the $\io$ compiler.
    \item $\reduc_\alice$ sends $(K_1,K_2,K_3,d,s,\sigma_\bob)$ to $\reduc_\bob$ and  $(K_1,K_2,K_3,d,s,\sigma_\charlie)$ to $\reduc_\charlie$.
    \item $\reduc_\bob$ on receiving $(c,(x_0,T))$  as the challenge cipher text from $\ch$ as the challenge ciphertext and $K_1,K_2,K_3,d,s,\sigma_\bob$ from $\reduc_\alice$, does the following:
    \begin{enumerate}
        \item  $\reduc_\bob$ generates the circuit $Q$ which on any input $x_0$ generates $r\gets T(x_0)$ and if the output is $\bot$ outputs $\bot$, else computes $\decproc(c,r)$ and if the outcome is $0$, output $y_0$, else output $y_1$. $\reduc_\bob$ generates $\tilde{Q}\gets \io(Q;s)$.
        \item $\reduc_\bob$ constructs $x_\trig$ as follows. Let ${x_\trig}_1=F_2(K_2,x_0\|\widetilde{Q})$ and ${x_\trig}_2=F_3(K_3,{x_\trig}_1)$. Let $x_\trig=x_0\|{x_\trig}_1\|{x_\trig}_2.$
        \item $\reduc_\bob$ runs $\bob$ on $(x_\trig,\sigma_\bob)$ to get an output $b^\bob$.
        \item $\reduc_\bob$ outputs $b^\bob \oplus d$.
    \end{enumerate}

    \item Similarly, $\reduc_\charlie$ on receiving $(c,(x_0,T))$ as the challenge cipher text from $\ch$ and $K_1,K_2,K_3,d,s,\sigma_\charlie$ from $\reduc_\alice$, does the following:
    \begin{enumerate}
         \item $\reduc_\charlie$ generates the circuit $Q$ which on any input $x_0$ generates $r\gets T(x_0)$ and if the output is $\bot$ outputs $\bot$, else computes $\decproc(c,r)$ and if the outcome is $0$, output $y_0$, else output $y_1$. $\reduc_\charlie$ generates $\tilde{Q}\gets \io(Q;s)$.
        \item $\reduc_\charlie$ constructs $x_\trig$ as follows. Let ${x_\trig}_1=F_2(K_2,x_0\|\widetilde{Q})$ and ${x_\trig}_2=F_3(K_3,{x_\trig}_1)$. Let $x_\trig=x_0\|{x_\trig}_1\|{x_\trig}_2.$
        \item $\reduc_\bob$ runs $\bob$ on $(x_\trig,\sigma_\bob)$ to get an output $b^\charlie$.
        \item $\reduc_\charlie$ outputs $b^\charlie \oplus d$.
    \end{enumerate}

\end{enumerate}

Note that the functionality of $Q$ is the same as that of $W$ in the cipher text $(x_0,W)$ obtained by running $\cllz.\enc(\pk,\cdot)$ algorithm on $y_b$ with $x_0$ as the randomness. Note that in $\hybrid_3$, $\bob$ (and similarly, $\charlie$) needs to distinguish between the following two inputs: a random string $y$ along with either a triggered input $x$ encoding $y$ which is also the view of the inside adversary in the reduction above in the event $b=d$ in the simulated experiment; or a  triggered input $x$ encoding $\tilde{y}$ random string where $\tilde{y}\xleftarrow{\$}$ sampled independent of $y$,  which is the view of the inside adversary in the reduction above in the event $b\neq d$ in the simulated experiment. Therefore, by the $\io$ guarantees, the view of the inside $\alice,\bob,\charlie$ is the same as that in $\hybrid_3$.  
   
\end{proof}

%% file: proofcorrectness.tex
\ifllncs 
\subsection{Proof of~\Cref{lemma:copy-protection-construction-completeness}}
\label{app:proof:completeness}
\else 

\fi 

\begin{proof}[Proof of \Cref{lemma:copy-protection-construction-completeness}]\anote{I have slightly changed the proof}
Let $W$ be the circuit that is obfuscated, and let the resulting obfuscated state be $\rho=(\{\{\ket{{A_i}_{s_i,{s'}_i}}\}_i\},\tilde{C},\io(D))$. We will show that for every input $x=(x_0,x_1,x_2)$, the $\eval$ algorithm on $(\rho,x)$ outputs $W(x)$ except with negligible probability. Let $\ket{\phi_x}$ be the state obtained after running the Hadamard operation on  $\{\{\ket{{A_i}_{s_i,{s'}_i}}\}_i\}$ (see \Cref{line:Hadamard} of the $\eval$ algorithm in \Cref{fig:copy-protect-construction}). 
It is easy to check that for every input $x$, by the correctness of $\cllz$ copy-protection, running $\tilde{C}$ that is generated as $\tilde{C}\gets \io(C)$ on $(x,\ket{\phi_x}$ in superposition, and then measuring the output register results in $y$ which is equal to $\prg(\prf.\eval(k,x))$, except with negligible probability. By the almost as good as new lemma~\cite{aaronson2016complexity}, this would mean that the resulting quantum state $\sigma$ which is negligibly close to $\ket{\psi_x}\bra{\psi_x}$ in trace distance.\anote{Prabhanjan are you happy with this proof?}  
Hence, running $C$  on $\sigma$ in \Cref{line:outside-C-eval} and inside $\io(D)$ in superposition and then checking if the output is equal to $y$ in superposition (see \Cref{line:outside-C-eval} of the $\eval()$ algorithm in \Cref{fig:copy-protect-construction}), must succeed and $\io(D)$ will output $W(x)$, except with negligible probability. Therefore, except with negligible probability, $\eval(\rho,x)$ outputs $W(x)$.\anote{Change $\io(P),\io(D)$ to $\tilde{P}$ and $\tilde{D}$ respectively.}\anote{If we have time, simplify the proof by measuring after step 2.}
\end{proof}\anote{Prabhanjan, is this explanation satisfactory?}

%% file: proofupocons.tex
\ifllncs 
\subsection{Proof of \Cref{lemma:puncturable_anti-piracy_from_cpcpiracy_PRFs,lemma:puncturable_anti-piracy_from_cpcpiracy_PRFs-id}}
\label{app:proof:lemma4}
\else 

\fi 

\begin{proof}[Proof of \Cref{lemma:copy-protection-construction-completeness}]\anote{I have slightly changed the proof}
Let $W$ be the circuit that is obfuscated, and let the resulting obfuscated state be $\rho=(\{\{\ket{{A_i}_{s_i,{s'}_i}}\}_i\},\tilde{C},\io(D))$. We will show that for every input $x=(x_0,x_1,x_2)$, the $\eval$ algorithm on $(\rho,x)$ outputs $W(x)$ except with negligible probability. Let $\ket{\phi_x}$ be the state obtained after running the Hadamard operation on  $\{\{\ket{{A_i}_{s_i,{s'}_i}}\}_i\}$ (see \Cref{line:Hadamard} of the $\eval$ algorithm in \Cref{fig:copy-protect-construction}). 
It is easy to check that for every input $x$, by the correctness of $\cllz$ copy-protection, running $\tilde{C}$ that is generated as $\tilde{C}\gets \io(C)$ on $(x,\ket{\phi_x}$ in superposition, and then measuring the output register results in $y$ which is equal to $\prg(\prf.\eval(k,x))$, except with negligible probability. By the almost as good as new lemma~\cite{aaronson2016complexity}, this would mean that the resulting quantum state $\sigma$ which is negligibly close to $\ket{\psi_x}\bra{\psi_x}$ in trace distance.\anote{Prabhanjan are you happy with this proof?}  
Hence, running $C$  on $\sigma$ in \Cref{line:outside-C-eval} and inside $\io(D)$ in superposition and then checking if the output is equal to $y$ in superposition (see \Cref{line:outside-C-eval} of the $\eval()$ algorithm in \Cref{fig:copy-protect-construction}), must succeed and $\io(D)$ will output $W(x)$, except with negligible probability. Therefore, except with negligible probability, $\eval(\rho,x)$ outputs $W(x)$.\anote{Change $\io(P),\io(D)$ to $\tilde{P}$ and $\tilde{D}$ respectively.}\anote{If we have time, simplify the proof by measuring after step 2.}
\end{proof}\anote{Prabhanjan, is this explanation satisfactory?}

\begin{proof}[\textbf{Proof of \cref{lemma:puncturable_anti-piracy_from_cpcpiracy_PRFs}}]
Let $(\alice,\bob,\charlie)$ be a QPT adversary in the security experiment given in~\cref{fig:genupo:expt} with $\distr_\inpclass=\distrprod$ as mentioned in the lemma. 
We mark the changes in blue.

\noindent $\hybrid_0$: \\
Same as the security experiment given in~\cref{fig:genupo:expt}. 
\begin{enumerate}
\item $\alice$ sends a key $s\in \keyspace_\secparam$ and functions $\mu_\bob$ and $\mu_\charlie$ to $\ch$.
\item $\ch$ samples $x^\bob,x^\charlie \xleftarrow{\$} \{0,1\}^n$.
\item $\ch$ samples $k\gets \keygen$, and generates $\io(P),\{\ket{{A_i}_{s_i,{s'}_i}}\}_i \leftarrow \cllz.\copyprotect(1^{\secparam},k)$.
\item $\ch$ constructs $\tilde{C}\gets \io(C)$ where $C=\prg\cdot \io(P)$. 
\item $\ch$ constructs the circuit $\io(D_0),\io(D_1)$ where $D_0$, $D_1$ are as depicted in \cref{fig:D_0-hybrid-0,fig:D_1-hybrid-0}. 
\item $\ch$ samples $b\xleftarrow{\$} \{0,1\}$ and sends $(\io({C}),\{\ket{{A_i}_{s_i,{s'}_i}}\}_i,\io(D_b))$ to $\alice$. 
 
\item $\alice(\tilde{C},\{\ket{{A_i}_{s_i,{s'}_i}}\}_i,\io(D_b))$ outputs a bipartite state $\sigma_{\bob,\charlie}$.
\item Apply $(\bob(x^\bob,\cdot) \otimes \charlie(x^\charlie,\cdot))(\sigma_{\bob,\charlie})$ to obtain $(b_{\bfB},b_{\bfC})$. 
\item Output $1$ if $b_{\bfB}=b_\bfC=b$.
\end{enumerate}

 \begin{figure}[!htb]
   \begin{center} 
   \begin{tabular}{|p{12cm}|}
    \hline 
\begin{center}
\underline{$P$}: 
\end{center}
 Hardcoded keys $K_1,K_2,K_3,R^0_i,R^1_i$ for every $i\in [\ell_0]$
 On input $x=x_0\|x_1\|x_2$ and vectors $v=v_1,\ldots v_{\ell_0}$.
 \begin{enumerate}
     \item If $F_3(K_3,x_1)\oplus x_2=x_0\|Q$ and $x_1=F_2(K_2,x_0\|Q)$:
     
     \textbf{Hidden trigger mode:} Treat $Q$ as a classical circuit and output $Q(v)$.
     \item Otherwise, check if the following holds: for all $i\in \ell_0$, $R^{x_{0,i}}(v_i)=1$ (where $x_{0,i}$ is the $i^{th}$ coordinate of $x_0$).
     
     \textbf{Normal mode:} If so, output $F_1(K_1,x)$ where $F_1()=\prf.\eval()$ is the primary pseudorandom function family that is being copy-protected. Otherwise output $\bot$.\label{line:normal mode}
 \end{enumerate}
\ \\ 
\hline
\end{tabular}
    \caption{Circuit $P$ in $\hybrid_0$.}
    \label{fig:CLLZ-circuit-P}
    \end{center}
\end{figure}

\begin{figure}[!htb]
   \begin{center} 
   \begin{tabular}{|p{12cm}|}
    \hline 
\begin{center}
\underline{$D_0$}: 
\end{center}
Hardcoded keys $W_s,C$.
On input: $x,v,y$.
\begin{enumerate}
\item Run $y'\gets C(x,v)$.
\item If $y'\neq y$ or $y'=\bot$ output $\bot$.
\item If $y=y'\neq \bot$, output $W_s(x)$.
\end{enumerate}
\ \\ 
\hline
\end{tabular}
    \caption{Circuit $D_0$ in $\hybrid_{0}$}
    \label{fig:D_0-hybrid-0}
    \end{center}
\end{figure}

\begin{figure}[!htb]
   \begin{center} 
   \begin{tabular}{|p{12cm}|}
    \hline 
\begin{center}
\underline{$D_1$}: 
\end{center}
Hardcoded keys {$W_{s,x^\bob,x^\charlie,\mu_\bob,\mu_\charlie},$}$C$.
On input: $x,v,y$.
\begin{enumerate}
\item Run $y'\gets C(x,v)$.
\item If $y'\neq y$ or $y'=\bot$ output $\bot$.
\item {If $y=y'\neq \bot$, output $W_{s,x^\bob,x^\charlie,\mu_\bob,\mu_\charlie}(x)$.}
\end{enumerate}
\ \\ 
\hline
\end{tabular}
    \caption{Circuit $D_1$ in $\hybrid_{0}$}
    \label{fig:D_1-hybrid-0}
    \end{center}
\end{figure}

\noindent \underline{$\hybrid_1$}: 
\begin{enumerate}
\item $\alice$ sends a key $s\in \keyspace_\secparam$ and functions $\mu_\bob$ and $\mu_\charlie$ to $\ch$.
\item $\ch$ samples $x^\bob,x^\charlie \xleftarrow{\$} \{0,1\}^n$.
\item $\ch$ samples $k\gets \keygen$, and generates $\io(P),\{\ket{{A_i}_{s_i,{s'}_i}}\}_i \leftarrow \cllz.\copyprotect(1^{\secparam},k)$.
\item $\ch$ constructs $\tilde{C}\gets \io(C)$ where $C=\prg\cdot \io(P)$. 
\item \cblue{$\ch$ samples $y^\bob_0,y^\charlie_1 \xleftarrow{\$} \{0,1\}^{2m}$. }
\item $\ch$ constructs the circuit $\io(D_0),\io(D_1)$ where $D_0$ and \cblue{$D_1$} are as depicted in \cblue{\cref{fig:D_0-hybrid-1}} and \cref{fig:D_1-hybrid-0}, respectively. 
\item $\ch$ samples $b\xleftarrow{\$} \{0,1\}$ and sends $(\io({C}),\{\ket{{A_i}_{s_i,{s'}_i}}\}_i,\io(D_b))$ to $\alice$. 
 
\item $\alice(\tilde{C},\{\ket{{A_i}_{s_i,{s'}_i}}\}_i,\io(D_b))$ outputs a bipartite state $\sigma_{\bob,\charlie}$.
\item Apply $(\bob(x^\bob,\cdot) \otimes \charlie(x^\charlie,\cdot))(\sigma_{\bob,\charlie})$ to obtain $(b_{\bfB},b_{\bfC})$. 
\item Output $1$ if $b_{\bfB}=b_\bfC=b$.
\end{enumerate}

\begin{figure}[!htb]
   \begin{center} 
   \begin{tabular}{|p{12cm}|}
    \hline 
\begin{center}
\underline{$D_0$}: 
\end{center}
Hardcoded keys $W_s$,\cblue{$\mu_\bob,\mu_\charlie$},$C$.
On input: $x,v,y$.
\begin{enumerate}
\item Run $y'\gets C(x,v)$.
\item If $y'\neq y$ or $y'=\bot$ output $\bot$.
\item \cblue{If $y=y'\neq \bot$ and $y\in \{y^\bob_0, y^\charlie_0\}$:, output $g(x)$.}\label{line:dummy-check}
\begin{enumerate}
    \item \cblue{If $y=y^\bob_0$ output $\mu_\bob(x^\bob)$.}
    \item \cblue{If $y=y^\charlie_0$ output $\mu_\charlie(x^\charlie)$.}
\end{enumerate}
\item \cblue{If $y=y'\neq \bot$ and $y\not\in \{y^\bob_0, y^\charlie_0\}$, output $W_s(x)$.}
\end{enumerate}
\ \\ 
\hline
\end{tabular}
    \caption{Circuit $D_0$ in $\hybrid_{1}$}
    \label{fig:D_0-hybrid-1}
    \end{center}
\end{figure}

\noindent \underline{$\hybrid_2$}: 
\begin{enumerate}
\item $\alice$ sends a key $s\in \keyspace_\secparam$ and functions $\mu_\bob$ and $\mu_\charlie$ to $\ch$.
\item $\ch$ samples $x^\bob,x^\charlie \xleftarrow{\$} \{0,1\}^n$.
\item $\ch$ samples $k\gets \keygen$, and generates $\io(P),\{\ket{{A_i}_{s_i,{s'}_i}}\}_i \leftarrow \cllz.\copyprotect(1^{\secparam},k)$.
\item $\ch$ constructs $\tilde{C}\gets \io(C)$ where $C=\prg\cdot \io(P)$. 
\item $\ch$ samples \sout{$y^\bob_0,y^\charlie_1 \xleftarrow{\$} \{0,1\}^{2m}$} \cblue{$y_1,y_2\xleftarrow{\$}\{0,1\}^m$, and computes $y^\bob_0\gets \prg(y_1),y^\charlie_0\gets \prg(y_2)$. }
\item $\ch$ constructs the circuit $\io(D_0),\io(D_1)$ where $D_0$ and $D_1$ are as depicted in \cref{fig:D_0-hybrid-1,fig:D_1-hybrid-0}, respectively. 
\item $\ch$ samples $b\xleftarrow{\$} \{0,1\}$ and sends $(\io({C}),\{\ket{{A_i}_{s_i,{s'}_i}}\}_i,\io(D_b))$ to $\alice$. 
 
\item $\alice(\tilde{C},\{\ket{{A_i}_{s_i,{s'}_i}}\}_i,\io(D_b))$ outputs a bipartite state $\sigma_{\bob,\charlie}$.
\item Apply $(\bob(x^\bob,\cdot) \otimes \charlie(x^\charlie,\cdot))(\sigma_{\bob,\charlie})$ to obtain $(b_{\bfB},b_{\bfC})$. 
\item Output $1$ if $b_{\bfB}=b_\bfC=b$.
\end{enumerate}

\noindent \underline{$\hybrid_3$}: 
\begin{enumerate}
\item $\alice$ sends a key $s\in \keyspace_\secparam$ and functions $\mu_\bob$ and $\mu_\charlie$ to $\ch$.
\item $\ch$ samples $x^\bob,x^\charlie \xleftarrow{\$} \{0,1\}^n$.
\item $\ch$ samples $k\gets \keygen$, and generates $\io(P),\{\ket{{A_i}_{s_i,{s'}_i}}\}_i \leftarrow \cllz.\copyprotect(1^{\secparam},k)$.
\item $\ch$ constructs $\tilde{C}\gets \io(C)$ where $C=\prg\cdot \io(P)$. 
\item $\ch$ samples {$y_1,y_2\xleftarrow{\$}\{0,1\}^m$, and computes $y^\bob_0\gets \prg(y_1),y^\charlie_0\gets \prg(y_2)$. }
\item $\ch$ constructs the circuit $\io(D_0),\io(D_1)$ where $D_0$ and \cblue{$D_1$} are as depicted in {\cref{fig:D_0-hybrid-1}} and \cblue{\cref{fig:D_1-hybrid-3}}, respectively. 
\item $\ch$ samples $b\xleftarrow{\$} \{0,1\}$ and sends $(\io({C}),\{\ket{{A_i}_{s_i,{s'}_i}}\}_i,\io(D_b))$ to $\alice$. 
 
\item $\alice(\tilde{C},\{\ket{{A_i}_{s_i,{s'}_i}}\}_i,\io(D_b))$ outputs a bipartite state $\sigma_{\bob,\charlie}$.
\item Apply $(\bob(x^\bob,\cdot) \otimes \charlie(x^\charlie,\cdot))(\sigma_{\bob,\charlie})$ to obtain $(b_{\bfB},b_{\bfC})$. 
\item Output $1$ if $b_{\bfB}=b_\bfC=b$.
\end{enumerate}

\begin{figure}[!htb]
   \begin{center} 
   \begin{tabular}{|p{12cm}|}
    \hline 
\begin{center}
\underline{$D_1$}: 
\end{center}
Hardcoded keys $f,g,${$C$}.
On input: $x,v,y$.
\begin{enumerate}
\item Run $y'\gets C(x,v)$
\item If $y'\neq y$ or $y'=\bot$ output $\bot$.
\item \sout{If $y=y'\neq \bot$, output $W_{s,x^\bob,x^\charlie,\mu_\bob,\mu_\charlie}(x)$.}
\item \cblue{If $y=y'\neq \bot$ and $x\in \{x^\bob, x^\charlie\}$:}
\begin{enumerate}
    \item \cblue{If $x=x^\bob$ output $\mu_\bob(x^\bob)$.}
    \item \cblue{If $x=x^\charlie$ output $\mu_\charlie(x^\charlie)$.}
\end{enumerate}
\item \cblue{If $y=y'\neq \bot$ and $x\not\in \{x^\bob, x^\charlie\}$}, output $W_s(x)$.
\end{enumerate}
\ \\ 
\hline
\end{tabular}
    \caption{Circuit $D_1$ in $\hybrid_{3}$}
    \label{fig:D_1-hybrid-3}
    \end{center}
\end{figure}

\noindent \underline{$\hybrid_4$}: 
\begin{enumerate}
\item $\alice$ sends a key $s\in \keyspace_\secparam$ and functions $\mu_\bob$ and $\mu_\charlie$ to $\ch$.
\item $\ch$ samples $x^\bob,x^\charlie \xleftarrow{\$} \{0,1\}^n$.
\item $\ch$ samples $k\gets \keygen$, and runs the $\cllz.\copyprotect(1^\secparam,k)$ algorithm as follows:\sout{generates $\io(P),\{\ket{{A_i}_{s_i,{s'}_i}}\}_i \leftarrow \cllz.\copyprotect(1^{\secparam},k)$.}\footnote{There is no change in this line compared to $\hybrid_3$, we only spell out the $\cllz.\copyprotect(1^\secparam,k)$ explicitly in order to use intermediate information in the next few steps.}
\begin{enumerate}
    \item \cblue{Samples $\ell_0$ coset states $\ket{{A_i}_{s_i,{s'}_i}}_i$ and construct $R^0_i=\io(A_i+s_i)$ and $R^1_i=\io(A_i+{s'}_i)$ for every $i\in [\ell_0]$.}
    \item \cblue{Samples keys $K_2,K_3$ from the respective secondary $\prf$s and use $R^0_i=\io(A_i+s_i)$ and $R^1_i=\io(A_i+{s'}_i)$ along with {$k$} to construct $P$, as given in \cref{fig:CLLZ-circuit-P}.}
\end{enumerate}
\item \cblue{$\ch$ computes $y^\bob_1=\prg(\prf.\eval(k,x^\bob)),y^\charlie_1=\prg(\prf.\eval(k,x^\charlie))$, and uses $y^\bob_1,y^\charlie_1$ along with $R^0_i,R^1_i,\io(P),\prg$} to construct \cblue{$C$ as depicted in \cref{fig:C-hybrid-4}}. 
\item $\ch$ samples {$y_1,y_2\xleftarrow{\$}\{0,1\}^m$, and computes $y^\bob_0\gets \prg(y_1),y^\charlie_0\gets \prg(y_2)$. }
\item $\ch$ constructs the circuit $\io(D_0),\io(D_1)$ where $D_0$ and $D_1$ are as depicted in \cref{fig:D_0-hybrid-1,fig:D_1-hybrid-3}, respectively. 
\item $\ch$ samples $b\xleftarrow{\$} \{0,1\}$ and sends $(\io({C}),\{\ket{{A_i}_{s_i,{s'}_i}}\}_i,\io(D_b))$ to $\alice$. 
 
\item $\alice(\tilde{C},\{\ket{{A_i}_{s_i,{s'}_i}}\}_i,\io(D_b))$ outputs a bipartite state $\sigma_{\bob,\charlie}$.
\item Apply $(\bob(x^\bob,\cdot) \otimes \charlie(x^\charlie,\cdot))(\sigma_{\bob,\charlie})$ to obtain $(b_{\bfB},b_{\bfC})$. 
\item Output $1$ if $b_{\bfB}=b_\bfC=b$.
\end{enumerate}

\begin{figure}[!htb]
   \begin{center} 
   \begin{tabular}{|p{12cm}|}
    \hline 
\begin{center}
\underline{${C}$}: 
\end{center}
Hardcoded keys {$\io(P),\prg$,}\cblue{$x^\bob,x^\charlie,y^\bob_1,y^\charlie_1,{k_{x^\bob,x^\charlie}}_1,R^0_i,R^1_i$ for all $i\in \ell_0$ (where $\ell_0$ is the number of coset states.)}.
On input: $x,v$.
\begin{enumerate}
\item \cblue{If $x\in (x^\bob,x^\charlie)$:}
\begin{enumerate}
    \item \cblue{Check if $R^{x_{0,i}}_i(v_i)=1$ for all $i\in \ell_0$, and reject otherwise.}
    \item \cblue{If $x= x^\bob$, output $y^\bob_1$.}
    \item \cblue{If $x= x^\charlie$, output $y^\charlie_1$.}
\end{enumerate}
\item If $x\not\in \{x^\bob, x^\charlie\}$, output \cblue{$\prg(\io(P)(x))$}.
\end{enumerate}
\ \\ 
\hline
\end{tabular}
    \caption{Circuit ${C}$ in $\hybrid_{4}$}
    \label{fig:C-hybrid-4}
    \end{center}
\end{figure}

\noindent \underline{$\hybrid_5$}: 
\begin{enumerate}
\item $\alice$ sends a key $s\in \keyspace_\secparam$ and functions $\mu_\bob$ and $\mu_\charlie$ to $\ch$.
\item $\ch$ samples $x^\bob,x^\charlie \xleftarrow{\$} \{0,1\}^n$.
\item $\ch$ samples $k\gets \keygen$, and does the following:
\begin{enumerate}
    \item  \cblue{Computes $k_{x^\bob,x^\charlie}\gets \prf.\puncture(k,\{x^\bob,x^\charlie\})$.}
    \item {Samples $\ell_0$ coset states $\ket{{A_i}_{s_i,{s'}_i}}_i$ and construct $R^0_i=\io(A_i+s_i)$ and $R^1_i=\io(A_i+{s'}_i)$ for every $i\in [\ell_0]$.}
    \item {Samples keys $K_2,K_3$ from the respective secondary $\prf$s and use $R^0_i=\io(A_i+s_i)$ and $R^1_i=\io(A_i+{s'}_i)$ along with \cblue{$k_{x^\bob,x^\charlie}$} to construct $P$, as given in $\cref{fig:CLLZ-circuit-P}$.}
\end{enumerate}
\item $\ch$ computes $y^\bob_1=\prg(\prf.\eval(k,x^\bob))$, $y^\charlie_1=\prg(\prf.\eval(k,x^\charlie))$ and uses $y^\bob_1,y^\charlie_1$ along with $R^0_i,R^1_i,\io(P),\prg$ to construct {$C$ as depicted in \cref{fig:C-hybrid-4}}.  
\item $\ch$ samples {$y_1,y_2\xleftarrow{\$}\{0,1\}^m$, and computes $y^\bob_0\gets \prg(y_1),y^\charlie_0\gets \prg(y_2)$. }
\item $\ch$ constructs the circuit $\io(D_0),\io(D_1)$ where $D_0$ and $D_1$ are as depicted in \cref{fig:D_0-hybrid-1,fig:D_1-hybrid-3}, respectively. 
\item $\ch$ samples $b\xleftarrow{\$} \{0,1\}$ and sends $(\io({C}),\{\ket{{A_i}_{s_i,{s'}_i}}\}_i,\io(D_b))$ to $\alice$. 
 
\item $\alice(\tilde{C},\{\ket{{A_i}_{s_i,{s'}_i}}\}_i,\io(D_b))$ outputs a bipartite state $\sigma_{\bob,\charlie}$.
\item Apply $(\bob(x^\bob,\cdot) \otimes \charlie(x^\charlie,\cdot))(\sigma_{\bob,\charlie})$ to obtain $(b_{\bfB},b_{\bfC})$. 
\item Output $1$ if $b_{\bfB}=b_\bfC=b$.
\end{enumerate}

\noindent \underline{$\hybrid_6$}:
\begin{enumerate}
\item $\alice$ sends a key $s\in \keyspace_\secparam$ and functions $\mu_\bob$ and $\mu_\charlie$ to $\ch$.
\item $\ch$ samples $x^\bob,x^\charlie \xleftarrow{\$} \{0,1\}^n$.
\item $\ch$ samples $k\gets \keygen$, and does the following:
\begin{enumerate}
    \item  Computes $k_{x^\bob,x^\charlie}\gets \prf.\puncture(k,\{x^\bob,x^\charlie\})$.
    \item {Samples $\ell_0$ coset states $\ket{{A_i}_{s_i,{s'}_i}}_i$ and construct $R^0_i=\io(A_i+s_i)$ and $R^1_i=\io(A_i+{s'}_i)$ for every $i\in [\ell_0]$.}
    \item {Samples keys $K_2,K_3$ from the respective secondary $\prf$s and use $R^0_i=\io(A_i+s_i)$ and $R^1_i=\io(A_i+{s'}_i)$ along with {$k_{x^\bob,x^\charlie}$} to construct $P$, as given in \cref{fig:CLLZ-circuit-P}.}
\end{enumerate}
\item \cblue{$\ch$ samples $u^\bob,u^\charlie \xleftarrow{\$}\{0,1\}^m$ and computes $y^\bob_1=\prg(u^\bob),y^\charlie_1=\prg(u^\charlie)$} \sout{$\ch$ computes $y^\bob_1=\prf.\eval(k,x^\bob)$, $y^\charlie_1=\prf.\eval(k,x^\charlie)$} and uses $y^\bob_1,y^\charlie_1$ along with $R^0_i,R^1_i,\io(P),\prg$ to construct $C$ as depicted in \cref{fig:C-hybrid-4}.
\item $\ch$ samples {$y_1,y_2\xleftarrow{\$}\{0,1\}^m$, and computes $y^\bob_0\gets \prg(y_1),y^\charlie_0\gets \prg(y_2)$. }
\item $\ch$ constructs the circuit $\io(D_0),\io(D_1)$ where $D_0$ and $D_1$ are as depicted in \cref{fig:D_0-hybrid-1,fig:D_1-hybrid-3}, respectively. 
\item $\ch$ samples $b\xleftarrow{\$} \{0,1\}$ and sends $(\io({C}),\{\ket{{A_i}_{s_i,{s'}_i}}\}_i,\io(D_b))$ to $\alice$. 
 
\item $\alice(\tilde{C},\{\ket{{A_i}_{s_i,{s'}_i}}\}_i,\io(D_b))$ outputs a bipartite state $\sigma_{\bob,\charlie}$.
\item Apply $(\bob(x^\bob,\cdot) \otimes \charlie(x^\charlie,\cdot))(\sigma_{\bob,\charlie})$ to obtain $(b_{\bfB},b_{\bfC})$. 
\item Output $1$ if $b_{\bfB}=b_\bfC=b$.
\end{enumerate}

\noindent \underline{$\hybrid_7$}:
\begin{enumerate}
\item $\alice$ sends a key $s\in \keyspace_\secparam$ and functions $\mu_\bob$ and $\mu_\charlie$ to $\ch$.
\item $\ch$ samples $x^\bob,x^\charlie \xleftarrow{\$} \{0,1\}^n$.
\item $\ch$ samples $k\gets \keygen$, and does the following:
\begin{enumerate}
    \item  {Computes $k_{x^\bob,x^\charlie}\gets \prf.\puncture(k,\{x^\bob,x^\charlie\})$.}
    \item {Samples $\ell_0$ coset states $\ket{{A_i}_{s_i,{s'}_i}}_i$ and construct $R^0_i=\io(A_i+s_i)$ and $R^1_i=\io(A_i+{s'}_i)$ for every $i\in [\ell_0]$.}
    \item {Samples keys $K_2,K_3$ from the respective secondary $\prf$s and use $R^0_i=\io(A_i+s_i)$ and $R^1_i=\io(A_i+{s'}_i)$ along with {$k_{x^\bob,x^\charlie}$} to construct $P$, as given in \cref{fig:CLLZ-circuit-P}.}
\end{enumerate}
\item \cblue{$\ch$ samples $y^\bob_1,y^\charlie_1 \xleftarrow{\$}\{0,1\}^{2m}$} \sout{$\ch$ samples $u^\bob,u^\charlie \xleftarrow{\$}\{0,1\}^m$ and computes $y^\bob_1=\prg(u^\bob),y^\charlie_1=\prg(u^\charlie)$} and uses $y^\bob_1,y^\charlie_1$ along with $R^0_i,R^1_i,\io(P),\prg$ to construct $C$ as depicted in \cref{fig:C-hybrid-4}.
\item $\ch$ samples {$y_1,y_2\xleftarrow{\$}\{0,1\}^m$, and computes $y^\bob_0\gets \prg(y_1),y^\charlie_0\gets \prg(y_2)$. }
\item $\ch$ constructs the circuit $\io(D_0),\io(D_1)$ where $D_0$ and $D_1$ are as depicted in \cref{fig:D_0-hybrid-1,fig:D_1-hybrid-3}, respectively. 
\item $\ch$ samples $b\xleftarrow{\$} \{0,1\}$ and sends $(\io({C}),\{\ket{{A_i}_{s_i,{s'}_i}}\}_i,\io(D_b))$ to $\alice$. 
 
\item $\alice(\tilde{C},\{\ket{{A_i}_{s_i,{s'}_i}}\}_i,\io(D_b))$ outputs a bipartite state $\sigma_{\bob,\charlie}$.
\item Apply $(\bob(x^\bob,\cdot) \otimes \charlie(x^\charlie,\cdot))(\sigma_{\bob,\charlie})$ to obtain $(b_{\bfB},b_{\bfC})$. 
\item Output $1$ if $b_{\bfB}=b_\bfC=b$.
\end{enumerate}

\noindent \underline{$\hybrid_8$}:
\begin{enumerate}
\item $\alice$ sends a key $s\in \keyspace_\secparam$ and functions $\mu_\bob$ and $\mu_\charlie$ to $\ch$.
\item $\ch$ samples $x^\bob,x^\charlie \xleftarrow{\$} \{0,1\}^n$.
\item $\ch$ samples $k\gets \keygen$, and does the following:
\begin{enumerate}
    \item  {Computes $k_{x^\bob,x^\charlie}\gets \prf.\puncture(k,\{x^\bob,x^\charlie\})$.}
    \item {Samples $\ell_0$ coset states $\ket{{A_i}_{s_i,{s'}_i}}_i$ and construct $R^0_i=\io(A_i+s_i)$ and $R^1_i=\io(A_i+{s'}_i)$ for every $i\in [\ell_0]$.}
    \item {Samples keys $K_2,K_3$ from the respective secondary $\prf$s and use $R^0_i=\io(A_i+s_i)$ and $R^1_i=\io(A_i+{s'}_i)$ along with {$k_{x^\bob,x^\charlie}$} to construct $P$.}
\end{enumerate}
\item $\ch$ samples $y^\bob_1,y^\charlie_1 \xleftarrow{\$}\{0,1\}^{2m}$ and uses $y^\bob_1,y^\charlie_1$ along with $R^0_i,R^1_i,\io(P),\prg$ to construct $C$ as depicted in \cref{fig:C-hybrid-4}.
\item $\ch$ samples {$y_1,y_2\xleftarrow{\$}\{0,1\}^m$, and computes $y^\bob_0\gets \prg(y_1),y^\charlie_0\gets \prg(y_2)$. }
\item $\ch$ constructs the circuit $\io(D_0),\io(D_1)$ where $D_0$ and \cblue{$D_1$} are as depicted in \cref{fig:D_0-hybrid-1} and \cblue{\cref{fig:D_1-hybrid-8}}, respectively. 
\item $\ch$ samples $b\xleftarrow{\$} \{0,1\}$ and sends $(\io({C}),\{\ket{{A_i}_{s_i,{s'}_i}}\}_i,\io(D_b))$ to $\alice$. 
 
\item $\alice(\tilde{C},\{\ket{{A_i}_{s_i,{s'}_i}}\}_i,\io(D_b))$ outputs a bipartite state $\sigma_{\bob,\charlie}$.
\item Apply $(\bob(x^\bob,\cdot) \otimes \charlie(x^\charlie,\cdot))(\sigma_{\bob,\charlie})$ to obtain $(b_{\bfB},b_{\bfC})$. 
\item Output $1$ if $b_{\bfB}=b_\bfC=b$.
\end{enumerate}

\begin{figure}[!htb]
   \begin{center} 
   \begin{tabular}{|p{12cm}|}
    \hline 
\begin{center}
\underline{$D_1$}: 
\end{center}
Hardcoded keys $f,g,${$C$}.
On input: $x,v,y$.
\begin{enumerate}
\item Run $y'\gets C(x,v)$
\item If $y'\neq y$ or $y'=\bot$ output $\bot$.
\item If $y=y'\neq \bot$ and \cblue{$y\in \{y^\bob_1, y^\charlie_1\}$}\sout{$x\in \{x^\bob, x^\charlie\}$}:
\begin{enumerate}
    \item If \cblue{$y=y^\bob_1$} \sout{$x=x^\bob$} output $\mu_\bob(x^\bob)$.
    \item If \cblue{$y=y^\charlie_1$} \sout{$x=x^\charlie$} output $\mu_\charlie(x^\charlie)$.
\end{enumerate}
\item If $y=y'\neq \bot$ and \cblue{$y\not\in \{y^\bob_1, y^\charlie_1\}$}\sout{$x\not\in \{x^\bob, x^\charlie\}$}, output $W_s(x)$.
\end{enumerate}
\ \\ 
\hline
\end{tabular}
    \caption{Circuit $D_1$ in $\hybrid_8$}
    \label{fig:D_1-hybrid-8}
    \end{center}
\end{figure}

\noindent \underline{$\hybrid_{9}$}:
\begin{enumerate}
\item $\alice$ sends a key $s\in \keyspace_\secparam$ and functions $\mu_\bob$ and $\mu_\charlie$ to $\ch$.
\item $\ch$ samples $x^\bob,x^\charlie \xleftarrow{\$} \{0,1\}^n$.
\item $\ch$ samples $k\gets \keygen$, and does the following:
\begin{enumerate}
    \item  {Computes $k_{x^\bob,x^\charlie}\gets \prf.\puncture(k,\{x^\bob,x^\charlie\})$.}
    \item {Samples $\ell_0$ coset states $\ket{{A_i}_{s_i,{s'}_i}}_i$ and construct $R^0_i=\io(A_i+s_i)$ and $R^1_i=\io(A_i+{s'}_i)$ for every $i\in [\ell_0]$.}
    \item {Samples keys $K_2,K_3$ from the respective secondary $\prf$s and use $R^0_i=\io(A_i+s_i)$ and $R^1_i=\io(A_i+{s'}_i)$ along with {$k_{x^\bob,x^\charlie}$} to construct $P$, as given in \cref{fig:CLLZ-circuit-P}.}
\end{enumerate}
\item \cblue{$\ch$ samples $u^\bob,u^\charlie \xleftarrow{\$}\{0,1\}^m$ and computes $y^\bob_1=\prg(u^\bob),y^\charlie_1=\prg(u^\charlie)$} \sout{$\ch$ samples $y^\bob_1,y^\charlie_1 \xleftarrow{\$}\{0,1\}^{2m}$} and uses $y^\bob_1,y^\charlie_1$ along with $R^0_i,R^1_i,\io(P),\prg$ to construct $C$ as depicted in \cref{fig:C-hybrid-4}.
\item $\ch$ samples {$y_1,y_2\xleftarrow{\$}\{0,1\}^m$, and computes $y^\bob_0\gets \prg(y_1),y^\charlie_0\gets \prg(y_2)$. }
\item $\ch$ constructs the circuit $\io(D_0),\io(D_1)$ where $D_0$ and $D_1$ are as depicted in \cref{fig:D_0-hybrid-1,fig:D_1-hybrid-8}, respectively. 
\item $\ch$ samples $b\xleftarrow{\$} \{0,1\}$ and sends $(\io({C}),\{\ket{{A_i}_{s_i,{s'}_i}}\}_i,\io(D_b))$ to $\alice$. 
 
\item $\alice(\tilde{C},\{\ket{{A_i}_{s_i,{s'}_i}}\}_i,\io(D_b))$ outputs a bipartite state $\sigma_{\bob,\charlie}$.
\item Apply $(\bob(x^\bob,\cdot) \otimes \charlie(x^\charlie,\cdot))(\sigma_{\bob,\charlie})$ to obtain $(b_{\bfB},b_{\bfC})$. 
\item Output $1$ if $b_{\bfB}=b_\bfC=b$.
\end{enumerate}

\noindent \underline{$\hybrid_{10}$}: 
\begin{enumerate}
\item $\alice$ sends a key $s\in \keyspace_\secparam$ and functions $\mu_\bob$ and $\mu_\charlie$ to $\ch$.
\item $\ch$ samples $x^\bob,x^\charlie \xleftarrow{\$} \{0,1\}^n$.
\item $\ch$ samples $k\gets \keygen$, and does the following:
\begin{enumerate}
    \item  {Computes $k_{x^\bob,x^\charlie}\gets \prf.\puncture(k,\{x^\bob,x^\charlie\})$.}
    \item {Samples $\ell_0$ coset states $\ket{{A_i}_{s_i,{s'}_i}}_i$ and construct $R^0_i=\io(A_i+s_i)$ and $R^1_i=\io(A_i+{s'}_i)$ for every $i\in [\ell_0]$.}
    \item {Samples keys $K_2,K_3$ from the respective secondary $\prf$s and use $R^0_i=\io(A_i+s_i)$ and $R^1_i=\io(A_i+{s'}_i)$ along with {$k_{x^\bob,x^\charlie}$} to construct $P$, as given in \cref{fig:CLLZ-circuit-P}.}
\end{enumerate}
\item \cblue{$\ch$ computes $y^\bob_1=\prg(\prf.\eval(k,x^\bob))$, $y^\charlie_1=\prg(\prf.\eval(k,x^\charlie))$} \sout{$\ch$ samples $u^\bob,u^\charlie \xleftarrow{\$}\{0,1\}^m$ and computes $y^\bob_1=\prg(u^\bob),y^\charlie_1=\prg(u^\charlie)$} and uses $y^\bob_1,y^\charlie_1$ along with $R^0_i,R^1_i,\io(P),\prg$ to construct $C$ as depicted in \cref{fig:C-hybrid-4}.
\item $\ch$ samples {$y_1,y_2\xleftarrow{\$}\{0,1\}^m$, and computes $y^\bob_0\gets \prg(y_1),y^\charlie_0\gets \prg(y_2)$. }
\item $\ch$ constructs the circuit $\io(D_0),\io(D_1)$ where $D_0$ and $D_1$ are as depicted in \cref{fig:D_0-hybrid-1,fig:D_1-hybrid-8}, respectively. 
\item $\ch$ samples $b\xleftarrow{\$} \{0,1\}$ and sends $(\io({C}),\{\ket{{A_i}_{s_i,{s'}_i}}\}_i,\io(D_b))$ to $\alice$. 
 
\item $\alice(\tilde{C},\{\ket{{A_i}_{s_i,{s'}_i}}\}_i,\io(D_b))$ outputs a bipartite state $\sigma_{\bob,\charlie}$.
\item Apply $(\bob(x^\bob,\cdot) \otimes \charlie(x^\charlie,\cdot))(\sigma_{\bob,\charlie})$ to obtain $(b_{\bfB},b_{\bfC})$. 
\item Output $1$ if $b_{\bfB}=b_\bfC=b$.
\end{enumerate}

\noindent \underline{$\hybrid_{11}$}: 
\begin{enumerate}
\item $\alice$ sends a key $s\in \keyspace_\secparam$ and functions $\mu_\bob$ and $\mu_\charlie$ to $\ch$.
\item $\ch$ samples $x^\bob,x^\charlie \xleftarrow{\$} \{0,1\}^n$.
\item $\ch$ samples $k\gets \keygen$, and does the following:
\begin{enumerate}
    \item  \sout{Computes $k_{x^\bob,x^\charlie}\gets \prf.\puncture(k,\{x^\bob,x^\charlie\})$.}
    \item {Samples $\ell_0$ coset states $\ket{{A_i}_{s_i,{s'}_i}}_i$ and construct $R^0_i=\io(A_i+s_i)$ and $R^1_i=\io(A_i+{s'}_i)$ for every $i\in [\ell_0]$.}
    \item {Samples keys $K_2,K_3$ from the respective secondary $\prf$s and use $R^0_i=\io(A_i+s_i)$ and $R^1_i=\io(A_i+{s'}_i)$ along with \sout{$k_{x^\bob,x^\charlie}$} \cblue{k} to construct $P$, as given in \cref{fig:CLLZ-circuit-P}.}
\end{enumerate}
\item $\ch$ computes $y^\bob_1=\prg(\prf.\eval(k,x^\bob))$, $y^\charlie_1=\prg(\prf.\eval(k,x^\charlie))$ and uses $y^\bob_1,y^\charlie_1$ along with $R^0_i,R^1_i,\io(P),\prg$ to construct $C$ as depicted in \cref{fig:C-hybrid-4}.
\item $\ch$ samples {$y_1,y_2\xleftarrow{\$}\{0,1\}^m$, and computes $y^\bob_0\gets \prg(y_1),y^\charlie_0\gets \prg(y_2)$. }
\item $\ch$ constructs the circuit $\io(D_0),\io(D_1)$ where $D_0$ and $D_1$ are as depicted in \cref{fig:D_0-hybrid-1,fig:D_1-hybrid-8}, respectively. 
\item $\ch$ samples $b\xleftarrow{\$} \{0,1\}$ and sends $(\io({C}),\{\ket{{A_i}_{s_i,{s'}_i}}\}_i,\io(D_b))$ to $\alice$. 
 
\item $\alice(\tilde{C},\{\ket{{A_i}_{s_i,{s'}_i}}\}_i,\io(D_b))$ outputs a bipartite state $\sigma_{\bob,\charlie}$.
\item Apply $(\bob(x^\bob,\cdot) \otimes \charlie(x^\charlie,\cdot))(\sigma_{\bob,\charlie})$ to obtain $(b_{\bfB},b_{\bfC})$. 
\item Output $1$ if $b_{\bfB}=b_\bfC=b$.
\end{enumerate}

\noindent \underline{$\hybrid_{12}$}:
\begin{enumerate}
\item $\alice$ sends a key $s\in \keyspace_\secparam$ and functions $\mu_\bob$ and $\mu_\charlie$ to $\ch$.
\item $\ch$ samples $x^\bob,x^\charlie \xleftarrow{\$} \{0,1\}^n$.
\item $\ch$ samples $k\gets \keygen$, \cblue{and computes $\io(P),\ket{{A_i}_{s_i,{s'}_i}}_i\gets \cllz.\copyprotect(k)$.}
\item {$\ch$ computes $y^\bob_1=\prg(\prf.\eval(k,x^\bob)),y^\charlie_1=\prg(\prf.\eval(k,x^\charlie))$.}
\item $\ch$ constructs \cblue{$C=\prg\cdot \io(P)$}. 
\item $\ch$ samples {$y_1,y_2\xleftarrow{\$}\{0,1\}^m$, and computes $y^\bob_0\gets \prg(y_1),y^\charlie_0\gets \prg(y_2)$. }
\item $\ch$ constructs the circuit $\io(D_0),\io(D_1)$ where $D_0$ and $D_1$ are as depicted in \cref{fig:D_0-hybrid-1,fig:D_1-hybrid-8}, respectively. 
\item $\ch$ samples $b\xleftarrow{\$} \{0,1\}$ and sends $(\io({C}),\{\ket{{A_i}_{s_i,{s'}_i}}\}_i,\io(D_b))$ to $\alice$. 
 
\item $\alice(\tilde{C},\{\ket{{A_i}_{s_i,{s'}_i}}\}_i,\io(D_b))$ outputs a bipartite state $\sigma_{\bob,\charlie}$.
\item Apply $(\bob(x^\bob,\cdot) \otimes \charlie(x^\charlie,\cdot))(\sigma_{\bob,\charlie})$ to obtain $(b_{\bfB},b_{\bfC})$. 
\item Output $1$ if $b_{\bfB}=b_\bfC=b$.
\end{enumerate}

\noindent Next, we give a reduction from $\hybrid_12$ to the $\cpppiracy$ of the $\cllz$ copy-protection (for the $\prf$s with the required property having the key-generation algorithm $\keygen$ as mentioned above) to finish the proof.
The reduction does the following.
\begin{enumerate}
    \item $R_\alice$ runs $\alice$ to get a circuit $f$ and $g$.
    \item $R_\alice$ on receiving the copy-protected $\prf$, $\io(P),\{\ket{{A_i}_{s_i,{s'}_i}}\}_i$ and $u^\bob,u^\charlie$, computes $y^\bob\gets \prg(u^\bob)$ and $y^\charlie=\prg(u^\charlie)$, and  
     creates the circuit $\tilde{C}\gets \io(C)$ where ${C}=\prg\cdot \io(P)$. $R_\alice$ also creates $\io(D)$ where $D$ on input $x,v,y$ runs ${C}$ on $x,v$ to get $y'$ and outputs $\bot$ if $y'\neq y$ or $y'=\bot$, else if $y'\in \{y^\bob_0,y^\charlie\}$ outputs $g(x)$, else it runs the circuit $W_s$ to output  $W_s(x)$. 
     $R_\alice$ runs $\alice$ on $\rho_k,\io(D)$ and gets an output $\sigma_{\bob,\charlie}$, it then sends the corresponding registers of $\sigma_{\bob,\charlie}$ 
      to both $R_\bob$ and  $R_\charlie$.
    \item $R_\bob$ and $R_\charlie$ receive $x^\bob$ and $x^\charlie$ from the challenger and run the adversaries $\bob(x^\bob,\cdot)$ and $\charlie(x^\charlie,\cdot)$ respectively on $\sigma_{\bob,\charlie}$, to get the outputs $b^\bob$ and $b^\charlie$ respectively,. 
     $R_\bob$ and $R_\charlie$ output $1-b^\bob$ and $1-b^\charlie$, respectively.
\end{enumerate}

Finally, we prove the indistinguishability of the hybrids to finish the proof.

\paragraph{Indistinguishability of hybrids}

\begin{claim}\label{claim:hyb_0-hyb_1}
Assuming the security of $\io$, hybrids $\hybrid_0$ and $\hybrid_1$ are computationally indistinguishable. 
\end{claim}
\begin{proof}[Proof of Claim~\ref{claim:hyb_0-hyb_1}]
For any function $f$, let $\imgclass_f$ denote the image of $f$.
Since $\imgclass_\prg$ is a negligible fraction of $\{0,1\}^{2m}$ and $y^\bob_0,y^\charlie_0$ were chosen uniformly at random, with overwhelming probability $y^\bob_0,y^\charlie_0\not \in \imgclass_\prg$ and hence not in $\imgclass_{{C}}$. Therefore with overwhelming probability over the choice of $y^\bob_0,y^\charlie_0$, any $(x,v,y)$ that satisfies this check also satisfies $y\not\in\{y^\bob_0,y^\charlie_0\}$. Hence with overwhelming probability, if $y'=y\neq \bot$, the penultimate check (\cref{line:dummy-check} in \cref{fig:D_0-hybrid-1}) will always fail, and therefore, $D_0$ will always output $W_s(x)$. Hence with overwhelming probability, $D_0$ has the same functionality in both the hybrids, and therefore by $\io$ guarantees, the indistinguishability of the hybrids holds.

\end{proof}

\begin{claim}\label{claim:hyb_1-hyb_2}
Assuming the pseudorandomness of $\prg$, hybrids $\hybrid_1$ and $\hybrid_2$ are computationally indistinguishable. 
\end{claim}
\begin{proof}[Proof of Claim~\ref{claim:hyb_0-hyb_1}]
The proof is immediate. 
\end{proof}

\begin{claim}\label{claim:hyb_2-hyb_3}
Assuming the security of $\io$, hybrids $\hybrid_{2}$ and $\hybrid_{3}$ are computationally indistinguishable. 
\end{claim}
\begin{proof}[Proof of Claim~\ref{claim:hyb_2-hyb_3}]
The modification did not change the functionality of $D_1$ in this hybrid compared to the previous hybrid by the definition of $W_{s,x^\bob,x^\charlie,\mu_\bob,\mu_\charlie}$ and the $\puncture$ algorithm associated with $\cktclassw$.
 Hence, the indistinguishability follows from the $\io$ guarantees.
\end{proof}

\begin{claim}\label{claim:hyb_3-hyb_4}
Assuming the security of $\io$, 
hybrids $\hybrid_3$ and $\hybrid_4$ are computationally indistinguishable. 
\end{claim}
\begin{proof}[Proof of Claim~\ref{claim:hyb_3-hyb_4}]
The indistinguishability follows by the $\io$ guarantees and the claim that with overwhelming probability, the functionalities of $\prg\cdot \io(P)$ and $C$ in this hybrid are the same. The proof of the claim is as follows.


In the proof of correctness~\cite[Lemma 7.13]{CLLZ21} of the $\cllz$ copy-protection scheme, it was shown that the probability over the keys for the secondary pseudorandom functions, that $x^\bob,x^\charlie$ are in the hidden triggers, is negligible. Hence, with overwhelming probability over the secondary pseudorandom function keys, $(x^\bob,v)$ and $(x^\charlie,v)$ will not satisfy the trigger condition for  $P$ and therefore, not run in the hidden-trigger mode\footnote{Note that this property depends only on the secondary keys $k_2$ and $k_3$. Since, over the hybrids, we only punctured the primary key and not the two secondary keys, the same correctness guarantee holds in this hybrid as in the unpunctured case of hybrid $0$.}. Hence with the same overwhelming probability, the functionality of $P$ will not change even if we skip the hidden trigger check for $\{x^\bob,x^\charlie\}$. Note that conditioned on the functionality does not change for $P$ by skipping the check for $\{x^\bob,x^\charlie\}$, the functionality of $C$ in $\hybrid_2$ and $\hybrid_3$ are the same. 
Hence, with overwhelming probability, the functionality of ${C}$ in $\hybrid_3$ is the same as that of $\prg\cdot \io(P)$. 
\end{proof}

\begin{claim}\label{claim:hyb_4-hyb_5}
Assuming the security of $\io$, hybrids $\hybrid_4$ and $\hybrid_5$ are computationally indistinguishable. 
\end{claim}
\begin{proof}
The indistinguishability holds because $P$ was hardcoded directly only in the circuit in the circuit $C$ in the previous hybrid, and in $C$, we never use the key $P$ to evaluate on $\{x^\bob,x^\charlie\}$, and hence the functionality did not change even after we punctured the $\prf$ key hardcoded inside $P$ in $\hybrid_5$, due to the puncturing correctness of the $\prf$. Hence the indistinguishability follows from the $\io$ guarantee since we did not change the functionality of ${C}$.
\end{proof}

\begin{claim}\label{claim:hyb_5-hyb_6}
Assuming the security of the pseudorandom function family $\prf$, hybrids $\hybrid_5$ and $\hybrid_6$ are computationally indistinguishable. 
\end{claim}
\begin{proof}
The proof is immediate.
\end{proof}

\begin{claim}\label{claim:hyb_6-hyb_7}
Assuming the pseudorandomness of $\prg$, hybrids $\hybrid_6$ and $\hybrid_7$ are computationally indistinguishable. 
\end{claim}
\begin{proof}
The proof is immediate.
\end{proof}

\begin{claim}\label{claim:hyb_7-hyb_8}
Assuming the security of $\io$, hybrids $\hybrid_7$ and $\hybrid_8$ are computationally indistinguishable. 
\end{claim}
\begin{proof}

We will show that the functionality of $D_1$ did not change across the hybrids $\hybrid_7$ and $\hybrid_8$ (see \cref{fig:D_1-hybrid-8,fig:D_1-hybrid-3}), and hence indistinguishability of the hybrids follows from the $\io$ guarantees.
Note that since ${C}$ in $\hybrid_8$ satisfies ${C}(x^\bob,v^\bob)=y^\bob$ and  ${C}(x^\charlie,v^\charlie)=y^\charlie$ $\forall v^\bob\in V^\bob$ and $v^\charlie\in V^\charlie$, where $V^\bob$ (respectively, $V^\charlie$) is the set of all $v$ such that $(x^\bob,v)$ (respectively, $(x^\charlie,v)$) passes the coset check in the normal mode (see \cref{line:normal mode}), respectively.
    Moreover, the image of ${C}$ restricted to $\inpclass_{{C}}\setminus \left((x^\bob,v^\bob)\cup (x^\charlie,v^\charlie)\right)$, i.e.,  \[\imgclass_{{{C}}_{\inpclass_{{C}}\setminus (x^\bob,v^\bob)\cup (x^\charlie,v^\charlie)}}\subset \imgclass_{\prg\left(\{0,1\}^m\right)},\] where $m$ is the output length of the $\prf$ family, $(x^\bob,v^\bob)$ (respectively, $(x^\charlie,v^\charlie)$) is the short hand notation for $\{(x^\bob,v)\mid w\in V^\bob\}$ (respectively, $\{(x^\charlie,v)\mid w\in V^\charlie\}$). 
    Since $\imgclass_{\imgclass_\prg}$ is a negligible fraction of  $\{0,1\}^{2m}$, $\imgclass_{{{C}}_{\inpclass_{{C}}\setminus (x^\bob,v^\bob)\cup (x^\charlie,v^\charlie)}}$ is also a negligible fraction of $\{0,1\}^{2m}$.
    Since $y^\bob_1,y^\charlie_1$ are sampled uniformly at random independent of the set $\imgclass_{{{C}}_{\inpclass_{{C}}\setminus (x^\bob,v^\bob)\cup (x^\charlie,v^\charlie)}}$, except with negligible probability, \[y^\bob_1,y^\charlie_1\not \in \imgclass_{{{C}}_{\inpclass_{{C}}\setminus (x^\bob,v^\bob)\cup (x^\charlie,v^\charlie)}}.\] Note that we did not change the description of ${C}$ after $\hybrid_3$, hence as noted in $\hybrid_3$, \[{C}(x^\bob,v)\in \{y^\bob_1 , \bot\}, \quad {C}(x^\charlie,v)\in \{y^\charlie_1,\bot\}.\]
    Therefore, combining the last two statements, except with negligible probability, the preimage(s) of $y^\bob_1$ are of the form $(x^\bob,v)$, and the only non-$\bot$ image of $x^\bob$ is $y^\bob_1$, and similarly for $y^\charlie_1$ and $x^\charlie$. Hence except with negligible probability, the check that $y'=y\neq \bot$ and $y\in \{y^\bob_1,y^\charlie_1\}$ is equivalent to $y'=y\neq \bot$ and $x\in \{x^\bob,x^\charlie\}$. Therefore with overwhelming probability, the functionality of $D_1$ in $\hybrid_7$ (see \cref{fig:D_1-hybrid-3}) and in $\hybrid_8$ (see \cref{fig:D_1-hybrid-8}) are the same.

\end{proof}

\begin{claim}\label{claim:hyb_8-hyb_9}
Assuming the pseudorandomness of $\prg$, hybrids $\hybrid_8$ and $\hybrid_9$ are computationally indistinguishable. 
\end{claim}
\begin{proof}
The proof is immediate.
\end{proof}

\begin{claim}\label{claim:hyb_9-hyb_10}
Assuming the puncturing security of the pseudorandom function family $\prf$, hybrids $\hybrid_9$ and $\hybrid_{10}$ are computationally indistinguishable. 
\end{claim}
\begin{proof}
The proof is immediate.
\end{proof}

\begin{claim}\label{claim:hyb_10-hyb_11}
Assuming the security of $\io$, hybrids $\hybrid_{10}$ and $\hybrid_{11}$ are computationally indistinguishable. 
\end{claim}
\begin{proof}
The proof is the same as that of Claim~\ref{claim:hyb_4-hyb_5}.
\end{proof}

\begin{claim}\label{claim:hyb_11-hyb_12}
Assuming the security of $\io$, hybrids $\hybrid_{11}$ and $\hybrid_{12}$ are computationally indistinguishable. 
\end{claim}
\begin{proof}
The proof is the same as that of~Claim~\ref{claim:hyb_3-hyb_4}.
\end{proof}



\end{proof}

\begin{proof}[\textbf{Proof of \Cref{lemma:puncturable_anti-piracy_from_cpcpiracy_PRFs-id}}]
The proof is the same as that of \Cref{lemma:puncturable_anti-piracy_from_cpcpiracy_PRFs} upto minor adaptations and hence we omit the proof.
\end{proof}

%% file: qsio.tex
\newcommand{\lrue}{\mathsf{lrUE}}
\newcommand{\lrueexpt}{\mathsf{lrUE.Expt}}
\newcommand{\qsio}{\mathsf{qsio}}
\newcommand{\test}{\mathsf{Test}}

\section{Construction from Quantum State iO} 
\label{sec:cons:qsio}
Recently, Coladangelo and Gunn~\cite{CG23} proposed the definition of quantum state iO and presented a candidate construction of qsiO. In this section, we show how to construct UPO from qsiO assuming unclonable encryption. As an intermediate tool, we consider a variant of private-key unclonable encryption introduced in~\cite{CG23}, which is referred to as key-testable (private-key) unclonable encryption.
\paragraph{Key-testable unclonable encryption} A key-testable private-key unclonable encryption~\cite{CG23} is a private unclonable encryption scheme with an additional QPT algorithm $\test$ that takes a key $\sk \in \{0,1\}^\secparam$ and a quantum cipher $\rho$ and outputs a bit $b$ such that for every two keys $sk,sk'\in \{0,1\}^\secparam$, such that $\sk\neq \sk'$, and a message $m\in \{0,1\}^n$, such that
\[\Pr[\rho\gets \enc(\sk,m), b\gets \test(\sk',\rho): b=\delta_\sk(\sk')]=1,\]
where $\delta_\sk$ is the function that is $1$ at $\sk$ and $0$ everywhere else, and $\enc$ is the encryption algorithm for the unclonable encryption scheme.

\paragraph{Unclonable encryption schemes with a specific kind of key-generation algorithm} In addition to the key-testable property, for the purpose of our construction of $\UPO$, we also require that the key-generation algorithm of the underlying unclonable encryption scheme simply samples a key uniformly at random from $\{0,1\}^\secparam$.
We need this restriction on the key-generation algorithm because, in our construction, the output distribution of the key-generation algorithm, i.e., the distribution of the secret key, determines the challenge distribution $\distr_\inpclass$, i.e., the distribution of the point to be punctured. We note that without this restriction on the key-generation algorithm, our construction still yields an $\upo$ scheme that satisfies generalized $\UPO$ security with respect to the challenge distribution $\distr_\inpclass=\{x,x\}_{x\gets\gen(1^\secparam)}$.

We next show that given a private-key unclonable encryption scheme, we can transform it into another scheme satisfying the above-mentioned restriction on the key-generation algorithm. 

\begin{theorem}\label{thm:generic-uenc-to-uenc-with-uniform-keys}
    Let $n:\mathbb{R} \rightarrow \mathbb{R}$ be arbitrary. Any private key unclonable encryption scheme $\UE=(\gen,\enc,\dec)$ on $n(\secparam)$ long messages where $\secparam$ is the security parameter can be transformed into another private-key unclonable encryption scheme $\UE'=(\gen',\enc',\dec')$ with the same message length such that the key-generation algorithm $\gen'$ is the uniformly random sampler from the keyspace of $\UE'$.

    In particular, any private key unclonable bit encryption scheme $\UE$ can be transformed into another private-key unclonable bit encryption scheme $\UE'=(\gen',\enc',\dec')$ such that $\gen'$ is of the above-mentioned form.
\end{theorem}
\begin{proof} 
Given $\UE=(\gen,\enc,\dec)$, we define $\UE'=(\gen',\enc',\dec')$ as follows.
\begin{itemize}
    \item $\gen'(1^\secparam)$: Sample $k'\xleftarrow{\$}\{0,1\}^\secparam$, and output $k'$.
    \item $\enc'(k',m)$: Generate $k\gets \gen(1^\secparam)$, and then generate $\rho\gets\enc(k,m)$. Output $\rho'=(\rho,k\oplus k')$.
    \item $\dec'(k',\rho')$: Interprete $\rho'=\rho,c$. Generate $k=k'\oplus c$, and then generate $m\gets \dec(k,\rho)$. Output $m$.
\end{itemize}
Clearly, the correctness of $\UE'$ is immediate from the correctness of $\UE$, and $\gen'$ satisfies the property mentioned in the theorem. For security, let $(\alice,\bob,\charlie)$ be an adversary against $\UE'$ in the unclonable indistinguishability game (see \Cref{subsec:uenc}). Consider the following reduction $(\reduc_\alice,\reduc_\bob,\reduc_\charlie)$ that uses  $(\alice,\bob,\charlie)$ to win the unclonable indistinguishability game against $\UE$.

\begin{enumerate}
    \item $\reduc_\alice$ runs $\alice$ on the security parameter and get backs a message pair $(m_0,m_1)$, which she sends to the challenger $\ch$.
    \item $\ch$ sends a cipher $\rho$.
    \item $\reduc_\alice$ samples $r\xleftarrow{\$}\{0,1\}^\secparam$, and feeds $\rho'=(\rho,r)$ to $\alice$ who then outputs a bipartite state $\sigma_(\bob,\charlie)$. $\reduc_\alice$ outputs $(r_\bob,\sigma_{\bob,\charlie},r_\charlie)$ where $r_\bob=r_\charlie=r$.
    \item $\reduc_\bob$ (respectively, $\reduc_\charlie$) on receiving $(r_\bob,\sigma_\bob)$ (respectively, $(r_\charlie,\sigma_\charlie)$) from $\reduc_\alice$ and a key $k$ from the challenger, runs $\bob$ on $(r_\bob\oplus k,\sigma_\bob)$ (respectively, $\charlie$ on $(r_\charlie\oplus k,\sigma_\charlie)$), and outputs the output of $\bob$ (respectively, $\charlie$).
\end{enumerate}
It is easy to see that the success probability of $(\reduc_\alice,\reduc_\bob,\reduc_\charlie)$ is the same as that of $(\alice,\bob,\charlie)$, which completes the proof of the theorem. 
\end{proof}

It was shown in~\cite{CG23} that assuming $\qsio$, the key testable property can be generically attached to any unclonable encryption scheme that satisfies the above-mentioned restriction on the key-generation algorithm.
\begin{theorem}[{Adapted from~\cite[Theorem 16]{CG23}}]\label{thm:key-testing-from-qsio}
    If injective one-way functions and $\qsio$ exist, then any private key unclonable bit\footnote{In~\cite{CG23}, the authors consider unclonable encryption scheme without fixed message length, i.e., the adversary chooses the message length during the security game. However, their proof also works for unclonable encryptions with fixed message length and, in particular, bit encryption.} encryption scheme (with the key-generation algorithm being the uniformly random sampler from the keyspace, $\{0,1\}^\secparam$) can be compiled into one with key testing (and the same key-generation algorithm).
\end{theorem}
For the rest of the section, for any key testable private-key unclonable encryption scheme, we will assume that the $\gen$ algorithm is just the uniformly random sampler from the keyspace and hence use a triplet of algorithm $(\enc,\dec,\test)$ to represent a key testable unclonable encryption scheme.
\paragraph{UPO from qsiO.} We consider the following tools: 
\begin{itemize}
    \item  A key-testable private-key unclonable bit encryption scheme $\UE=(\enc,\dec,\test)$.
    \item Quantum state iO scheme, denoted by $\qsio=(\obf,\eval)$. 
\end{itemize}
\begin{theorem}\label{thm:upo-from-key-testable-ue+qsio}
    Suppose there exists a key-testable private-key unclonable encryption scheme, $\UE=(\enc,\dec\test)$. 
     Then, any $\qsio$ scheme $(\obf,\eval)$ is also a $\UPO$ scheme satisfying $\distrid$-generalized $\UPO$ security guarantee (see \Cref{subsec:upo-definition}), for any puncturable keyed circuit class $\cktclassw=\{\{W_s\}_{s\in \keyspace_\secparam}\}_\secparam$ in $\ppoly$.
\end{theorem}

\begin{proof}
The correctness follows immediately from the correctness of the $\qsio$ scheme.    

Next, we prove security.
Let $(\alice,\bob,\charlie)$ be a QPT adversary in the generalized $\UPO$ security experiment given in~\cref{fig:genupo:expt} with $\distr_\inpclass=\distrid$.

\noindent $\hybrid_0:$
Same as the security experiment given in~\cref{fig:genupo:expt}.
\begin{enumerate}
    \item $\alice$ sends a key $s\in \keyspace_\secparam$ and function $\mu$\footnote{In the security experiment in~\cref{fig:genupo:expt}, $\alice$ sends two functions $\mu_\bob$, $\mu_\charlie$ but since in the context of the proof, $\distr_\inpclass=\distrid$, the second function $\mu_\charlie$ is redundant and do not play any part. Therefore, for the sake of the proof, we can assume, without loss of generality, that $\alice$ just sends a single function $\mu$ to the challenger.} to $\ch$.
    \item $\ch$ samples $x^* \xleftarrow{\$} \{0,1\}^{n(\secparam)}$, and a bit $b\xleftarrow{\$}\{0,1\}$.
    \item $\ch$ generates $\tilde{\rho}_0\gets \obf(1^\secparam,W_s)$, and $\tilde{\rho}_1\gets \obf(1^\secparam,W_{s,x^*,\mu})$, where $W_{s,x^*,\mu} \gets \genpuncture(s,x^*,x^*,\mu,\mu)$.
    \item $\ch$ sends $\tilde{\rho}_b$ to $\alice$.
    \item $\alice(\tilde{\rho}_b)$ outputs a bipartite state $\sigma_{\bob,\charlie}$.
    \item Apply $(\bob(x^*,\cdot) \otimes \charlie(x^*,\cdot))(\sigma_{\bob,\charlie})$ to obtain $(b_{\bfB},b_{\bfC})$. 
    \item Output $1$ if $b_{\bfB}=b_\bfC=b$.
\end{enumerate}

\noindent $\hybrid_1:$
\begin{enumerate}
    \item $\alice$ sends a key $s\in \keyspace_\secparam$ and function $\mu$ to $\ch$.
    \item $\ch$ samples $x^* \xleftarrow{\$} \{0,1\}^{n(\secparam)}$, and a bit $b\xleftarrow{\$}\{0,1\}$.
    \item $\ch$ generates \sout{$\tilde{\rho}_0\gets \obf(1^\secparam,W_s)$, and $\tilde{\rho}_1\gets \obf(1^\secparam,W_{s,x^*,\mu})$, where $W_{s,x^*,\mu} \gets \genpuncture(s,x^*,\mu)$.} \cblue{$\tilde{\rho}_b\gets \obf(1^\secparam,(C,\rho_b))$ where $\rho_b\gets \UE.\enc(x^*,b)$ and $C$ is the circuit that on input $x$, first checks if $\UE.\test(x,\rho_b)$ rejects, in which case, $C$ outputs $W_s(x)$. Else, $C$ runs $d\gets  \UE.\dec(x,\rho_b)$ and if $d=0$ outputs $W_s(x)$ else outputs $\mu(x)$.}
    \item $\ch$ sends $\tilde{\rho}_b$ to $\alice$.
    \item $\alice(\tilde{\rho}_b)$ outputs a bipartite state $\sigma_{\bob,\charlie}$.
    \item Apply $(\bob(x^*,\cdot) \otimes \charlie(x^*,\cdot))(\sigma_{\bob,\charlie})$ to obtain $(b_{\bfB},b_{\bfC})$. 
    \item Output $1$ if $b_{\bfB}=b_\bfC=b$.
\end{enumerate}
\par Hybrids $\hybrid_0$ and $\hybrid_1$ are computationally indistinguishable since the implementations $W_s$ and $C,\rho_0$, as well as $W_{s,x^*,\mu}$ and $C,\rho_1$ are functionally equivalent, i.e., $(1-\negl(\secparam))$ implementation of the same function for some negligible function $\negl$ (this, in turn, follows because of the key-testing property of the $\UE.\test$ algorithm).

Next, we give a reduction $(R_\alice,R_\bob,R_\charlie)$ from $\hybrid_1$ to unclonable indistinguishability experiment for $\UE$  as follows.

\begin{itemize}
    \item $R_\alice$ gives $0$ and $1$ as the challenge messages to the challenger.\anote{Maybe this step is redundant considering UE is a bit encryption or we can replace $0$ with $0^n$ and $1$ with $1^n$.}
    \item Challenger sends a cipher $\rho$. 
    \item $R_\alice$ computes $(C,\rho)$ using the $\UE.\test$ algorithm and $\rho$,  and then computes $\tilde{\rho}\gets \obf(1^\secparam, (C,\rho))$.
    \item $R_\alice$ feeds $\tilde{\rho}$ to $\alice$ and gets a bipartite state $\sigma_{B,C}$.
    \item $R_\bob$ (respectively, $R_\charlie$) on receiving $x$ from the challenger, runs $\bob$ (respectively, $\charlie$) on $\sigma_\bob$ (respectively, $\sigma_\charlie$) and $x$, and the output the bit outputted by $\bob$ (respectively, $\charlie$).
\end{itemize}

\noindent It follows that the advantage of the QPT adversary $(\alice,\bob,\charlie)$ in breaking UPO security is within a negligible additive factor of the advantage of the QPT adversary in breaking the unclonable indistinguishability of $\UE$. This completes the proof of generalized $\UPO$ security for $(\obf,\eval)$.
\end{proof}

Combining \Cref{thm:generic-uenc-to-uenc-with-uniform-keys}, \Cref{thm:key-testing-from-qsio}, and \Cref{thm:upo-from-key-testable-ue+qsio}, we conclude the following.
\begin{corollary}\label{cor:upo-from-qsio+uenc}
    Suppose there exists a post-quantum injective one-way function and a private key unclonable bit encryption scheme $\UE$. Then, any $\qsio$ scheme $(\obf,\eval)$ is also a $\UPO$ scheme satisfying $\distrid$-generalized $\UPO$ security guarantee (see \Cref{subsec:upo-definition}), for any puncturable keyed circuit class $\cktclassw=\{\{W_s\}_{s\in \keyspace_\secparam}\}_\secparam$ in $\ppoly$.
\end{corollary}

Combining \Cref{cor:UEnc-from-UPO} with \Cref{cor:upo-from-qsio+uenc} and using the fact that puncturable pseudorandom functions can be built from injective one-way functions, and that any $\qsio$ scheme is also an $\io$ scheme, we get the following feasibility result for unclonable encryptions in terms of increasing the message length.
\begin{corollary}[Extending bit unclonable encryption to multi-bit unclonable encryption]\label{cor:bit-unclonable_encryption_to_multi-bit_unclonable encryption}
    Assuming post-quantum injective one-way function and a private key unclonable bit encryption scheme $\UE$, and a $\qsio$ scheme, there exists a secure public-key unclonable encryption for multiple bits (see \Cref{subsec:uenc} for the definition).
\end{corollary}



%% file: upoapplications.tex
\section{Applications}
We discuss the applications of unclonable puncturable obfuscation: 
\begin{itemize}
    \item We identify an interesting class of circuits and show that copy-protection for this class of functionalities exist. We show this in~\Cref{sec:qcp:punc}. 
    \item We generalize the result from bullet 1 to obtain an approach to copy-protect certain family of cryptographic schemes. This is discussed in~\Cref{sec:cp:punc:cryptoscheme}. 
    \item We show how to copy-protect evasive functions in~\Cref{sec:qcp:evasive}. 
    \item We show how to construct public-key single-decryptor encryption from UPO in~\Cref{sec:pksde:app}. 
    
\end{itemize}

\subsection{Notations for the applications}\label{subsec:notation-applications}
\noindent All the search-based applications (i.e., the security of which can be written as a cloning game with trivial success probability negligible) are with respect to independent challenge distribution. By the generic transformation in~\cite{AKL23}, this implies the applications also achieve security with respect to arbitrarily correlated challenge distribution. \anote{Prabhanjan, is this fine?}

\noindent A function class $\fclass=\{\fclass_\secparam\}_{\secparam\in\NN}$ is said to have a keyed circuit implementation $\cktclass=\{\{C_{k}\}_{k\in \keyspace_\secparam}\}_{\secparam\in\NN}$ if for every $f\in \fclass$, there is a circuit $C_k$ in $\cktclass$ that implements $f$, i.e., the canonical map $S_\secparam$ mapping a circuit $C$ to its functionality when seen as a map $\cktclass_\secparam \mapsto \fclass_\secparam$, is surjective. In addition, if there exists  a distribution $\distr_\fclass$ on $\fclass$, and an efficiently samplable distribution $\distr_\keyspace$ on $\keyspace$ such that
\[\{S_\secparam(C_k)\}_{k\gets \distr_\keyspace(1^\secparam)}\approx \{f\}_{f\gets \distr_\fclass(1^\secparam)},\]
then $(\distr_\keyspace,\cktclass)$ is called a keyed circuit implementation of $(\distr_\fclass,\fclass)$.
Since any circuit class can be represented as a keyed circuit class using universal circuits, there is no loss of generality in our definition of keyed circuit implementation\anote{Is this fine, Prabhanjan?}.



\subsection{Copy-Protection for Puncturable Function Classes}
\label{sec:qcp:punc}
\noindent We identify a class of circuits associated with a security property defined below. We later show that this class of circuits can be copy-protected.

\begin{definition}[Puncturable Security]\label{def:puncturable-security-function-classes}
Let $\cktclass=\{\cktclass_{\secparam}\}_{\secparam \in \mathbb{N}}$ be a puncturable keyed circuit class (as defined in~\Cref{sec:upo:security}). Let $\puncture$ be the puncturing algorithm and $\keyspace$ be the key space associated with $\cktclass$. 
\par We say that $(\cktclass,\puncture)$ satisfies $\distr_{\keyspace}$-puncturable security, where $\distr_{\keyspace}$ is a distribution on $\keyspace$, where $n$ is the input length of the circuits in $\cktclass_{\secparam}$, if the following holds: for any quantum polynomial time adversary $\alice$,
$$\prob\left[y = C_k(x_1) \ :\ \substack{k \leftarrow \distr_{\keyspace}(1^{\secparam})\\ \ \\ (x_1,x_2) \xleftarrow{\$} \{0,1\}^{2n} \\ \ \\ G_{k^*} \gets \puncture(k,x_1,x_2) \\ \ \\ y \gets \alice(x_1,G_{k^*})} \right] \leq \frac{1}{2^m} + \negl(\secparam),$$
for some negligible function $\negl$. In the above expression, $C_k \in \cktclass_{\secparam}$ and $n$ is the input length and $m$ is the output length of $C_k$.
\end{definition}

\begin{remark}
    A possible objection to the definition could be the inclusion of $x_2$ in the definition. The sole purpose of including $x_2$ is to help in the proof.
\end{remark}

\noindent \pnote{condensed the remark}

\begin{remark}
We may abuse the notation and denote $\distr_{\keyspace}$ to be a distribution on $\cktclass$. Specifically, circuit $C$ is sampled from  $\distr_{\keyspace}(1^{\secparam})$ as follows: first sample $k \leftarrow \keyspace_{\secparam}$ and then set $C=C_k$. 
\end{remark}

\begin{theorem}
\label{thm:search-copy-protection-application}
Suppose  $\fclass={\fclass_\secparam}_{\secparam\in\NN}$ be a function class equipped with a distribution $\distr_\fclass$ such that there exists a keyed circuit implementation (see \Cref{subsec:notation-applications}) $(\distr_\keyspace,\cktclass)$ satisfying the following:
\begin{enumerate}
    \item $\cktclass$ is a puncturable keyed circuit class associated with the puncturing algorithm $\puncture$ and key space $\keyspace$
    \item $\cktclass$ satisfies $\distr_{\keyspace}$-puncturable security (\Cref{def:puncturable-security-function-classes}).
\end{enumerate}
Suppose $\UPO=(\obf,\eval)$ is a secure unclonable puncturable obfuscation scheme for $\cktclass$ associated with distribution $\distr_{\calX}$, where $\distr_{\calX}$ is defined to be a uniform distribution. 
\par Then there exists a copy-protection scheme $(\copyprotect,\eval)$ for $\fclass$ satisfying $(\distr_{\keyspace},\distr_{\calX})$-anti-piracy, with respect to $\cktclass$ as the keyed circuit implementation of $\fclass$, and $(\distr_\keyspace,\cktclass)$  as the keyed circuit implementation of $(\distr_\fclass,\fclass)$. 
\end{theorem}
\begin{proof}
\par We define the algorithms $\cp=(\copyprotect,\eval)$ as follows: 
\begin{itemize}
    \item $\copyprotect(1^{\secparam},C)$: on input $C \in \cktclass_{\secparam}$ with input length $n(\secparam)$, it outputs $\rho_C$, where $\rho_C \leftarrow \UPO.\obf(1^{\secparam},C)$. 
    \item $\eval(\rho_C,x)$: on input $\rho_C$, input $x \in \{0,1\}^n$, it outputs the result of $\UPO.\eval(\rho_C,x)$. 
\end{itemize}
The correctness of the copy-protection scheme follows from the correctness of $\UPO$.\\

\noindent Next, we prove $(\distr_{\keyspace},\distr_{\calX})$-anti-piracy with respect to the keyed circuit implementation $(\distr_\keyspace,\cktclass)$ (see \Cref{sec:def:copyprotection}).
Let $(\alice,\bob,\charlie)$ be a non-local adversary in the anti-piracy experiment $\cpexpt^{\left(\alice,\bob,\charlie \right),\distr_{\keyspace},\distr_{\calX}}\left( 1^{\secparam} \right)$ defined in \Cref{fig:product-uniform-search-anti-piracy-puncturable-functions}. Consider the following adversary $(\reduc_\alice,\reduc_\bob,\reduc_\charlie)$ in the $\UPO$ security experiment $\upoexpt^{\left(\reduc_\alice,\reduc_\bob,\reduc_\charlie \right),\distr_{\calX},\cktclass}\left( 1^{\secparam},\cdot \right)$ (\Cref{fig:upoexpt}), defined as follows:
\begin{itemize}
    \item $\reduc_\alice$ samples $k\gets \distr_{\keyspace}(1^\secparam)$, and sends $k$ to the challenger $\ch$ in the $\UPO$ security experiment.
    \item $\reduc_\alice$ runs $\alice$ on the received obfuscated state $\rho$ from $\ch$ to get a bipartite state $\sigma_{\bob,\charlie}$ on registers ${\bf B}$ and ${\bf C}$.
    \item $\reduc_\alice$ sends register ${\bf B}$ and key $k$ to $\bob$. Similarly, $\reduc_\alice$ sends register ${\bf C}$ and key $k$ to $\charlie$.  
    \item $\ch$ generates $(x^{\bob},x^{\charlie}) \leftarrow \distr_{\calX}$.
    \item $\reduc_\bob$ on receiving the challenge $x^\bob$, runs $\bob$ on $(k,\sigma_{\bob},x^{\bob})$ to obtain $y^\bob$. $\reduc_\bob$ outputs $0$ if and only if $y^\bob=C_{k_\bob}(x^\bob)$, otherwise outputs $1$.  
    \item $\reduc_\charlie$ receives the challenge $x^{\charlie}$ and does the same as $\reduc_\bob$ but on $(k,\sigma_\charlie,x^{\charlie})$.
\end{itemize}
Define the following quantities: 
\begin{itemize}
    \item $p^{\cp}$: probability that $(\bob,\charlie)$ simultaneously output $(C_k(x^{\bob}),C_k(x^{\charlie}))$ in $\cpexpt^{\left(\alice,\bob,\charlie \right),\distr_{\keyspace},\distr_{\calX}}\left( 1^{\secparam} \right)$. 
    \item For $b \in \{0,1\}$, $p_b^{\UPO}$: probability that $(\reduc_\bob,\reduc_\charlie)$ simultaneously output $b$ in  $\upoexpt^{\left(\reduc_\alice,\reduc_\bob,\reduc_\charlie \right),\distr_{\calX},\cktclass}\left( 1^{\secparam},b \right)$.
\end{itemize}

\noindent In order to prove the security of $\cp$, we have to upper bound $p^{\cp}$. We have the following:
\begin{itemize}
    \item From the description of $(\reduc_{\alice},\reduc_{\bob},\reduc_{\charlie})$, $p^{\cp} = p_0^{\UPO}$.
    \item From the security of $\UPO$, we have that $\frac{1}{2} p_{0}^{\UPO} + \frac{1}{2}p_1^{\UPO} \leq \frac{1}{2} + \nu_1(\secparam)$ for some negligible function $\nu_1(\secparam)$.
\end{itemize}
Combining the two, we have:
\begin{equation}
\label{eqn:first}
\frac{1}{2} p^{\cp} + \frac{1}{2}p_1^{\UPO} \leq \frac{1}{2} + \nu_1(\secparam)
\end{equation}

\begin{claim}
Assuming $\distr_{\keyspace}$-puncturable security of $\cktclass$, there exists a negligible function $\nu_2(\secparam)$ such that $p_1^{\UPO} \geq 1 -  \nu_2(\secparam)$. 
\end{claim}
\begin{proof}
Define the following quantities. Let $q_1^{\reduc_\bob}$ (respectively, $q_1^{\reduc_\charlie}$) be the probability that $\reduc_\bob$ (respectively, $\reduc_\charlie$) outputs $0$. Hence, $p_1^{\UPO} \geq 1 - q_1^{\reduc_\bob} - q_1^{\reduc_\charlie}$. We prove that $q_1^{\reduc_\bob} \leq \nu_3(\secparam)$, for some negligible function $\nu_3(\secparam)$ and symmetrically, it would follow that $q_1^{\reduc_{\charlie}} \leq \nu_4(\secparam)$. 
\par Suppose $q_1^{\reduc_\bob}$ is not negligible. We design an adversary $\alice_{{\sf punc}}$ participating in the security experiment of~\Cref{def:puncturable-security-function-classes}. Adversary $\alice_{{\sf punc}}$ proceeds as follows: 

\begin{itemize}
    \item $\alice_{{\sf punc}}$ on receiving $(x_1,G_{k^*})$, where $G_{k^*} \leftarrow \puncture(k,x_1,x_2)$, generates $\rho\gets \obf(1^\secparam,G_{k^*})$.
    \item It then runs $\sigma_{\bob\charlie}\gets \reduc_\alice(\rho)$, where $\sigma_{\bob\charlie}$ is defined on two registers ${\bf B}$ and ${\bf C}$.
    \item Finally, it outputs the result of $\reduc_\bob$ on the register ${\bf B}$ and $x_1$. 
\end{itemize}

\noindent By the above description, the event that $\alice_{{\sf punc}}$ wins exactly corresponds to the event that $\reduc_\bob$ outputs $0$. That is, the probability that  $\alice_{{\sf punc}}$ wins is $q_1^{\reduc_{\bob}}$. Since $q_1^{\reduc_{\bob}}$ is not negligible, it follows that $\alice_{{\sf punc}}$ breaks the puncturable security of $\cktclass$ with non-negligible probability, a contradiction. Thus, $q_1^{\reduc_{\bob}}$ is negligible and symmetrically, $q_1^{\reduc_{\charlie}}$ is negligible. 
\end{proof}

\noindent From the above claim, we have:
\begin{equation}
\label{eqn:second}
\frac{1}{2} p^{\cp} + \frac{1}{2}p_1^{\UPO} \geq \frac{1}{2} p^{\cp} + \frac{1}{2} - \frac{1}{2}\nu_2(\secparam)
\end{equation} 

\noindent Combining~\Cref{eqn:first} and~\Cref{eqn:second}, we have: 
$$p^{\cp} \leq 2\nu_1(\secparam) + \nu_2(\secparam),$$
which concludes the theorem. 
\end{proof}

\paragraph{Instantiations.} In the theorem below, we call a pseudorandom function (PRF) to be a 2-point puncturable PRF if it can be punctured at 2 points. Such a function family can be instantiated, for instance, from post-quantum one-way functions~\cite{BGI14,BW13}. We obtain the following corollary. 

\begin{corollary}
Let $\fclass$ be a class of 2-point puncturable PRF with an evaluation circuit $\eval$ and keyspace $\{\keyspace_\secparam\}_\secparam$, and let $\cktclass=\{\{\eval(k,\cdot)\}_{k\in \keyspace_\secparam}\}_\secparam$. Assuming the existence of unclonable puncturable obfuscation for $\cktclass$, there exists a copy-protection scheme for $\fclass$. 
\end{corollary}

Combined with \Cref{thm:strong-CLLZ-cp-prf_puncturable-CP-f-uniform}, we can rephrase \Cref{thm:search-copy-protection-application} in terms of concrete assumption as follows.
\begin{corollary}\label{cor:search-copy-protection-application-concrete-assumptions}
Suppose $\fclass$ be a function class satisfying all the properties as in \Cref{thm:search-copy-protection-application}, then
assuming \Cref{conj:goldreich-levin-correlated}, the existence of post-quantum sub-exponentially secure $\io$ and one-way functions, and the quantum hardness of LWE, there exists a copy-protection scheme for $\fclass$ satisfying anti-piracy with respect to the same circuit implementation and anti-piracy notion as mentioned in \Cref{thm:search-copy-protection-application}. 

In particular, under the above assumptions, there exists a copy-protection scheme for every class of 2-point puncturable PRF.
\end{corollary}

\ifllncs

\else 
\input{punccrypto}

\input{sde}
\input{evasive}
\fi

%% file: punccrypto.tex
\subsection{Copy-Protection for Puncturable Cryptographic Schemes} 
\label{sec:cp:punc:cryptoscheme}
We generalize the approach in the previous section to capture puncturable cryptographic schemes, rather than just puncturable functionalities.

\paragraph{Syntax.} A cryptographic primitive that is a tuple of probabilistic polynomial time algorithms $(\gen,\eval,\puncture,\verify)$ such that
\begin{itemize}
    \item $\gen(1^\secparam)$: takes a security parameter and generates a secret key $\sk$ and a public auxiliary information $\aux$. We will assume without loss of generality that $\sk \in \{0,1\}^{\secparam}$.
    \item $\eval(\sk,x)$: takes a secret key $\sk$ and an input $x$ and outputs a output string $y$. This is a deterministic algorithm. 
    \item $\puncture(\sk,x_1,x_2)$: takes a secret key $\sk$ and a set of inputs $(x_1,x_2)$ and outputs a circuit $G_{\sk,x_1,x_2}$. This is a deterministic algorithm. 
    
    \item $\verify(\sk,\aux,x,y)$: takes a secret key $\sk$, an auxiliary information $\aux$, an input $x$ and an output $y$ and either accepts or rejects.
\end{itemize}

\begin{definition}[{Puncturable cryptographic schemes}]\label{def:puncturable-crypto}
A cryptographic scheme $(\gen,\eval,\puncture,\allowbreak \verify)$ is a puncturable cryptographic scheme if it satisfies the following  properties: 
\begin{itemize}
    \item {\bf Correctness}: The correctness property states that for any input $x$, $\verify(\sk,\aux,x,\allowbreak \eval(x))$ accepts, where $(\sk,\aux) \leftarrow \gen(1^{\secparam})$. 
    \item {\bf Correctness of Punctured Circuit}: The correctness of punctured circuit states that for any set of inputs $\{x_1,x_2\}$, and $G_{\sk,x_1,x_2}\gets \puncture(\sk,x_1,x_2)$, where $(\sk,\aux) \leftarrow \gen(1^{\secparam})$, it holds that $G_{\sk,x_1,x_2}(x)=\eval(\sk,x)$ for all $x\not\in \{x_1,x_2\}$ and $G_{\sk,x_1,x_2}(x)$ outputs $\bot$ if $x\in \{x_1,x_2\}$.
    \item {\bf Security}: We say that a puncturable cryptographic scheme $(\gen,\eval,\puncture,\verify)$ satisfies puncturable security if the following holds: for any quantum polynomial time adversary $\alice$,
$$\prob\left[\verify(\sk,\aux,x_1,y)=1 \ :\ \substack{(\sk,\aux) \gets \gen(1^{\secparam})\\ \ \\ x_1,x_2 \xleftarrow{\$} \{0,1\}^{n}\\ \ \\ G_{\sk,x_1,x_2} \gets \puncture(\sk,x_1,x_2) \\ \ \\ y \gets \alice(x_1,\aux,G_{\sk,x_1,x_2})} \right] \leq  \negl(\secparam),$$
for some negligible function $\negl$.
\end{itemize}

\par 
\end{definition}

\begin{remark}\label{remark:double-puncture-from-one-puncture-arbitray-crypto}
    A possible objection to the definition could be the inclusion of $x_2$ in the definition. The sole purpose of including $x_2$ is to help in the proof. Assuming $\io$ and length-doubling $\prg$, this added restriction does not rule out function classes further since, given $\io$ and $\prg$, any function class that satisfies the above definition without the additional puncture point has a circuit representation that satisfies the puncturing security with this additional point of puncture $x_2$.
\end{remark}

\newcommand{\pcsexpt}{{\sf PCS.Expt}}
\begin{figure}[!htb]
\begin{center} 
\begin{tabular}{|p{12cm}|}
    \hline 
\begin{center}
\underline{$\pcsexpt^{\left(\alice,\bob,\charlie \right),}\left( 1^{\secparam} \right)$}: 
\end{center}
\begin{itemize}
\item $\ch$ samples $\sk,\aux \gets \gen(1^{\secparam})$, and generates $\rho_\sk\gets \UPO.\obf(1^\secparam,\eval(\sk,\cdot))$ and sends $(\rho_\sk,\aux)$ to $\alice$.
\item $\alice$ produces a bipartite state $\sigma_{\bob,\charlie}$.
\item $\ch$ samples $x^\bob,x^\charlie\xleftarrow{\$} \{0,1\}^n$.
\item Apply $(\bob(x^{\bob},\cdot) \otimes \charlie(x^{\charlie},\cdot))(\sigma_{\bob,\charlie})$ to obtain $(y^{\bob},y^{\charlie})$. 
\item Output $1$ if $\verify(\sk,\aux,x^\bob,y^\bob)=1$ and $\verify(\sk,\aux,x^\charlie,y^\charlie)=1$.
\end{itemize}
\ \\ 
\hline
\end{tabular}
    \caption{Anti-piracy experiment with uniform and independent challenge distribution:}
    \label{fig:product-uniform-arb-crypto-experiment-anti-piracy}
    \end{center}
\end{figure}

\begin{lemma}\label{lemma:abstract-crypto-experiment}
    Suppose $(\gen,\eval,\puncture,\verify)$ is a puncturable cryptographic scheme. Let $\UPO$ be a $\upo$ for the puncuturable keyed circuit class $\{\cktclass_\secparam=\{\eval(\sk,\cdot)\}_{\sk \in \{0,1\}^{\secparam}}$ parametrized by the secret keys, equipped with $\puncture$ as the puncturing algorithm. 
    Then for every QPT adversary $(\alice,\bob,\charlie)$, there exists a negligible function $\negl$ such that the following holds:
    \[ \prob\left[ 1 \leftarrow \pcsexpt^{\left(\alice,\bob,\charlie \right)}\left( 1^{\secparam} \right)\ \right] \leq \negl(\secparam),\]
    where $\pcsexpt^{\left(\alice,\bob,\charlie \right)}$ is defined in~\Cref{fig:product-uniform-arb-crypto-experiment-anti-piracy}.
\end{lemma}

\begin{proof}[Proof of lemma~\ref{lemma:abstract-crypto-experiment}]
Let $(\alice,\bob,\charlie)$ be a non-local adversary in the anti-piracy experiment $\pcsexpt^{\left(\alice,\bob,\charlie \right)}$ (\Cref{fig:product-uniform-signature-anti-piracy}). Consider the following adversary $(\reduc_\alice,\reduc_\bob,\reduc_\charlie)$ in the $\UPO$ security experiment $\upoexpt^{\left(\reduc_\alice,\reduc_\bob,\reduc_\charlie \right),\distr_{\inpclass},\cktclass}$ (\Cref{fig:upoexpt}), defined as follows:
\begin{itemize}
    \item $\reduc_\alice$ samples $(\sk,\aux)\gets \gen(1^\secparam)$, and sends $\sk$ to the challenger $\ch$ in the $\UPO$ security experiment.
    \item $\reduc_\alice$ receives $\rho$ from $\ch$ and runs $\alice$ on $(\rho,\aux)$ from $\ch$ to get a bipartite state $\sigma_{\bob,\charlie}$.
    \item $\reduc_\alice$ outputs $\sk_\bob,\sk_\charlie,\aux_\bob,\aux_\charlie,\sigma_{\bob,\charlie}$ where $\sk_\bob=\sk_\charlie=\sk$ and  $\aux_\bob=\aux_\charlie=\aux$.
    \item $\reduc_\bob$ receives the challenge $x^\bob$ from $\ch$ and $(\sk_\bob,\aux_\bob,\sigma_\bob)$ from $\reduc_\alice$ and runs $\bob$ on $\sigma_{\bob}$ to obtain $y^\bob$. $\reduc_\bob$ outputs $0$ if and only if $\verify(\sk,\aux,x^\bob,y^\bob)=1$, otherwise outputs $1$.  
    \item $\reduc_\charlie$ does the same but on $(\aux_\charlie,\sigma_\charlie)$ and the challenge $x^\charlie$.
\end{itemize}

\noindent Note that the view of $(\alice,\bob,\charlie)$ in $\expt^{\left(\reduc_\alice,\reduc_\bob,\reduc_\charlie \right),\cal{U}\times\cal{U},\cktclass}\left( 1^{\secparam},0 \right)$ is identical to the UPO experiment, and the event $1\gets \expt^{\left(\reduc_\alice,\reduc_\bob,\reduc_\charlie \right),\cal{U}\times\cal{U},\cktclass}\left( 1^{\secparam},0 \right)$ corresponds to $1\gets \expt^{\left(\alice,\bob,\charlie \right),(\gen,\eval,\puncture,\verify),\UPO}\left( 1^{\secparam} \right)$.
Let \[p_b\equiv \prob[b\gets \expt^{\left(\reduc_\alice,\reduc_\bob,\reduc_\charlie \right),\cal{U}\times\cal{U},\cktclass}\left( 1^{\secparam},b \right)], \forall b\in \{0,1\}.\]
Hence,
\begin{align}
p_0=&\prob[0\gets \expt^{\left(\reduc_\alice,\reduc_\bob,\reduc_\charlie \right),\cal{U}\times\cal{U},\cktclass}\left( 1^{\secparam},0 \right)]\\
&=\prob\left[ 1 \leftarrow \expt^{\left(\alice,\bob,\charlie \right),(\gen,\eval,\puncture,\verify),\UPO}\left( 1^{\secparam} \right)\ \right].
\end{align}

Therefore, it is enough to show that $p_0$ is negligible.

Note that by the $\UPO$-security (see \Cref{def:newcpsecurity_bot}) of the $\UPO$ scheme,
there exists a negligible function $\negl(\secparam)$ such that
\[\prob[b=0]p_0+\prob[b=1]p_1=\frac{p_0+p_1}{2}\leq \frac{1}{2}+\negl(\secparam).\]
Hence,
\begin{equation}\label{eq:invoke-upo-security-sign}
    p_0\leq 1+2\negl(\secparam) - p_1.
\end{equation}

\pnote{I wasn't able to follow the next few lines. Can you double check? }\anote{I double checked. There was a  mistake which I corrected and simplified and expanded a bit. Does this make sense now?} Let $q^{\reduc_\bob}_1$ (respectively, $q^{\reduc_\charlie}_1$) be the probability that $\reduc_\bob$ ($\reduc_\charlie$) outputs $0$, i.e., the inside adversary $\bob$ (respectively, $\charlie$) passed verification, in the experiment $\expt^{\left(\reduc_\alice,\reduc_\bob,\reduc_\charlie \right),\cal{U}\times\cal{U},\cktclass}\left( 1^{\secparam},1 \right)$.

Note that the event $0\gets \expt^{\left(\reduc_\alice,\reduc_\bob,\reduc_\charlie \right),\cal{U}\times\cal{U},\cktclass}\left( 1^{\secparam},1 \right)$ corresponds to either $\reduc_\bob$ outputs $0$ or $\reduc_\charlie$ outputs $0$ in $\expt^{\left(\reduc_\alice,\reduc_\bob,\reduc_\charlie \right),\cal{U}\times\cal{U},\cktclass}\left( 1^{\secparam},1 \right)$. Hence,
\[\prob\left[0\gets \expt^{\left(\reduc_\alice,\reduc_\bob,\reduc_\charlie \right),\cal{U}\times\cal{U},\cktclass}\left( 1^{\secparam},1 \right)\right]\leq q^{\reduc_\bob}_1+q^{\reduc_\charlie}_1.\]
Therefore,
\[p_1=1-\prob\left[0\gets \expt^{\left(\reduc_\alice,\reduc_\bob,\reduc_\charlie \right),\cal{U}\times\cal{U},\cktclass}\left( 1^{\secparam},1 \right)\right]\geq 1-q^{\reduc_\bob}_1-q^{\reduc_\charlie}_1.\]
Combining with \Cref{eq:invoke-upo-security-sign}, we conclude
\begin{equation}\label{eq:translating-upo-security-sign}
    p_0\leq 1 +2\negl(\secparam) -(1-q^{\reduc_\bob}_1-q^{\reduc_\charlie}_1)=q^{\reduc_\bob}_1+q^{\reduc_\charlie}_1+2\negl(\secparam).
\end{equation}

Hence, it is enough to show that $q^{\reduc_\charlie}_1$ and $q^{\reduc_\bob}_1$ are negligible.

Consider the adversary $A_{\alice,\bob}$ in the puncturing security experiment given in \Cref{def:puncturable-digital-signatures} for the puncturable signature scheme $(\gen,\eval,\puncture,\verify)$.

\begin{itemize}
    \item $A_{\alice,\bob,}$ on receiving $x_1,G_{\sk,x_1,x_2}$ generates $\rho\gets \obf(1^\secparam,\eval(G_{\sk,x_1,x_2},\cdot))$.
    \item Then, runs $\sigma_{\bob,\charlie}\gets \alice(\rho)$.
    \item Finally, outputs $\bob(\sigma_\bob)$. 
\end{itemize}
It is easy to see that the event of $A_{\alice,\bob}$ winning the puncturing security experiment exactly corresponds with the event of $\reduc_\bob$ outputting $1$ in $\expt^{\left(\reduc_\alice,\reduc_\bob,\reduc_\charlie \right),\cal{U}\times\cal{U},\cktclass}\left( 1^{\secparam},1 \right)$, where $x_1$ corresponds to $x^\bob$. Therefore, by the puncturing security of $(\gen,\eval,\puncture,\verify)$, there exists a negligible function $\epsilon_1(\secparam)$ such that, 
$$q^{\reduc_\bob}_1=\prob\left[\verify(\sk,\aux,x_1,\sig) = 1 \ :\ \substack{(\sk,\aux) \gets \gen(1^{\secparam})\\ \ \\ x_1,x_2 \xleftarrow{\$} \{0,1\}^{n}\\ \ \\ G_{\sk,x_1,x_2} \leftarrow \puncture(\sk,\{x_1,x_2\}) \\ \ \\ \sig \leftarrow A_{\alice,\bob}(x_1,\aux,G_{\sk,x_1,x_2})} \right]\leq \epsilon_1.$$
Similarly, by considering the adversary $A_{\alice,\charlie}$ which is $A_{\alice,\bob}$ with the $\bob$ replaced as $\charlie$, we conclude that there exists a negligible function $\epsilon_2(\secparam)$ such that 
$$q^{\reduc_\charlie}_1=\prob\left[\verify(\sk,\aux,x_1,\sig) = 1 \ :\ \substack{(\sk,\aux) \gets \gen(1^{\secparam})\\ \ \\ x_1,x_2 \xleftarrow{\$} \{0,1\}^{n}\\ \ \\ G_{\sk,x_1,x_2} \leftarrow \puncture(\sk,\{x_1,x_2\}) \\ \ \\ \sig \leftarrow A_{\alice,\charlie}(x_1,\aux,G_{\sk,x_1,x_2})} \right]\leq \epsilon_2.$$

Therefore, we conclude that both $q^{\reduc_\charlie}_1$ and $q^{\reduc_\bob}_1$ are negligible in $\secparam$, which in combination with \Cref{eq:translating-upo-security-sign} completes the proof of the anti-piracy.


\end{proof}

\subsubsection{Copy-Protection for Signatures}
\begin{definition}[{Puncturable digital signatures~\cite{BSW16}}]\label{def:puncturable-digital-signatures}
Suppose  $\digitalsig=(\gen,\sign,\verify)$ be a digital signature with message length $n=n(\secparam)$ and signature length $s=s(\secparam)$. Let $\puncture,\sign^*$ be efficient polynomial time algorithms such that $\puncture()$ takes as input a secret key and a message (or a polynomial number of messages) $(\sk,m)$ and outputs $\sk_m$, and $\sign^*$ is the signing algorithm for punctured keys such that $\sign^*(\sk_m,\cdot)$ has the same functionality as $\sign^*(\sk_m,\cdot)$ on all messages
 $m' \neq m$ and $\sign^*(\sk_m,m')$ outputs $\bot$. 
\par We say that a puncturable digital signature scheme $(\gen,\sign,\puncture,\verify,\sign^*)$ satisfies puncturable security if the following holds: for any quantum polynomial time adversary $\alice$,
$$\prob\left[\verify(\vk,x_1,\sig)=1 \ :\ \substack{(\sk,\vk) \gets \gen(1^{\secparam})\\ \ \\ m_1,m_2 \xleftarrow{\$} \{0,1\}^{n}\\ \ \\ \sk_{m_1,m_2} \gets \puncture(\sk,\{m_1,m_2\}) \\ \ \\ \sig \gets \alice(m_1,\vk,\sk_{m_1,m_2})} \right] \leq  \negl(\secparam),$$
for some negligible function $\negl()$.
\end{definition}

\begin{remark}\label{remark:double-puncture-from-one-puncture}
    A possible objection to the definition could be the inclusion of $m_2$ in the definition. The sole purpose of including $m_2$ is to help in the proof. Assuming $\io$ and length-doubling $\prg$, this added restriction does not rule out function classes further since, given $\io$ and $\prg$, it can be shown that any function class that satisfies the above definition without the additional puncture point has a circuit representation that satisfies the puncturing security with this additional point of puncture $m_2$.
\end{remark}
In~\cite{BSW16}, the authors constructed a puncturable digital signature scheme from one-way functions and sub-exponentially secure indistinguishability obfuscation. We observe that their construction when instantiated with a post-quantum one-way function, and post-quantum sub-exponentially secure $\io$, satisfies post-quantum security. \anote{IMP: Prabhanjan, is this okay?}
\begin{theorem}[{Adapted from~\cite[Theorem 3.1]{BSW16}}]\label{thm:puncturable-sig-from-owf+io}
    Assuming post-quantum one-way function, and post-quantum sub-exponentially secure $\io$, there exists a post-quantum puncturable digital signature, see \Cref{def:puncturable-digital-signatures}
\end{theorem}
\begin{definition}[{Adapted from~\cite{LLQ+22}}]
    A copy-protection scheme for a signature scheme with message length $n(\secparam$ and signature length $s(\secparam)$ consists of the following algorithms:
   \begin{itemize}
    \item $(\sk,\vk)\gets \gen(1^\lambda):$ on input a security parameter $1^\lambda$, returns a classical secret key $\sk$ and a classical verification key $\vk$.
    \item $\rho_{\sk}\gets \qkeygen(\sk):$ takes a classical secret key $\sk$ and outputs a quantum signing key $\rho_\sk$. 
    \item $\sig\gets \sign(\rho_\sk,m):$ takes a quantum signing key $\rho_\sk$ and a message $m$ for $m \in \{0,1\}^{n(\lambda)}$, and outputs a classical signature $\sig$. 
    \item $b\gets \verify(\vk,m,\sig)$ takes a classical verification key $\vk$, a message $m$ and a classical signature $\sig$, and outputs a bit $b$ indicating accept ($b=1$) or reject ($b=0$).
  \end{itemize}  
\paragraph*{Correctness} For every message $m\in\{0,1\}^{n(\secparam)}$, there exists a negligible function $\delta(\secparam)$, (also called the correctness precision) such that
\[\prob[\sk,\vk\gets\gen(\secparam);\rho_\sk\gets \qkeygen(\sk), \sig\gets \sign(\rho_\sk,m) : \verify(\vk,\sig)=1]\geq 1-\delta(\secparam).\]

\paragraph*{Security}

\begin{figure}[!htb]
\begin{center} 
\begin{tabular}{|p{12cm}|}
    \hline 
\begin{center}
\underline{$\expt^{\left(\alice,\bob,\charlie \right),\cpdigitalsig}\left( 1^{\secparam} \right)$}: 
\end{center}
\begin{itemize}
\item $\ch$ samples $\sk,\vk \gets \gen(1^{\secparam})$ and generates $\rho_\sk\gets \qkeygen(\sk)$ and sends $(\rho_\sk,\vk)$ to $\alice$.
\item $\alice$ produces a bipartite state $\sigma_{\bob,\charlie}$.
\item $\ch$ samples $m^\bob,m^\charlie\xleftarrow{\$} \{0,1\}^n$.
\item Apply $(\bob(m^{\bob},\cdot) \otimes \charlie(m^{\charlie},\cdot))(\sigma_{\bob,\charlie})$ to obtain $(\sig^{\bob},\sig^{\charlie})$. 
\item Output $1$ if $\verify(\vk,m^\bob,\sig^\bob)=1$ and $\verify(\vk,m^\charlie,\sig^\charlie)=1$.
\end{itemize}
\ \\ 
\hline
\end{tabular}
    \caption{Anti-piracy experiment with uniform and independent challenge distribution for copy-protection of signatures.}
    \label{fig:product-uniform-signature-anti-piracy}
    \end{center}
\end{figure}

We say that a copy-protection scheme for signatures $\cpdigitalsig=(\gen,\qkeygen,\sign,\verify)$ satisfies anti-piracy with respect to the product distribution $\mathcal{U}\otimes \mathcal{U}$  if for every efficient adversary $(\alice,\bob,\charlie)$ in Experiment~\ref{fig:product-uniform-signature-anti-piracy} there exists a negligible function $\negl()$ such that
\[ \prob\left[ 1 \leftarrow \expt^{\left(\alice,\bob,\charlie \right),\cpdigitalsig}\left( 1^{\secparam} \right)\ \right] \leq \negl(\secparam).\]
\end{definition}

\begin{theorem}\label{thm:signature-copy-protection-application}

Suppose  $\digitalsig=(\gen,\sign,\puncture,\verify,\sign^*)$ be a puncturable digital signature with messge length $n(\secparam)$ and signature length $s(\secparam)$. 

Given a $\upo$ scheme $(\obf,\eval)$ with $\UPO$-security (see \Cref{def:newcpsecurity_bot})  for $\fclass=\{\fclass_\secparam\}_\secparam$ where $\fclass_\secparam=\{\sign(k,\cdot)\}_{k\in Support(\gen(1^\secparam))}$, equipped with $\puncture$ as the puncturing algorithm, with respect to $\distr_\inpclass=\cal{U}\times\cal{U}$, there exists a copy-protection scheme for signature $\cpdigitalsig=(\gen,\qkeygen,\sign,\verify)$ 
where the algorithms $\cpdigitalsig.\gen,\cpdigitalsig.\verify$ are the same as that of the puncturable signature scheme and $\cpdigitalsig.\qkeygen(\sk)=\obf(\sign(\sk,\cdot))$ and the $\cpdigitalsig.\sign()$ algorithm is the same as the $\eval()$ algorithm of the $\UPO$ scheme.  
\end{theorem}
\begin{proof}[Proof of \Cref{thm:signature-copy-protection-application}]
The correctness of the copy-protection of signatures scheme directly follows from the $\UPO$-correctness guarantees, see \Cref{def:upo-definition}.
Next, we prove anti-piracy.
Let $(\alice,\bob,\charlie)$ be a non-local adversary in the anti-piracy experiment $\expt^{\left(\alice,\bob,\charlie \right),\cpdigitalsig}\left( 1^{\secparam} \right)$ given in \Cref{fig:product-uniform-signature-anti-piracy}.
By the puncturing security and correctness of  $\digitalsig=(\gen,\sign,\puncture,\verify,\sign^*)$, $(\gen,\sign,\puncture,\verify')$ is a puncturable cryptographic scheme where $\vk$ is the auxiliary information $\aux$, the message space is the input space, the signature is the output space, $\gen=\digitalsig.\gen$, $\eval=\digitalsig.\sign$, $\puncture=\digitalsig.\puncture$ and $\verify'(\sk,\vk,m,\sig)=\digitalsig.\verify(\vk,m,\sig)$. 

Therefore, by \Cref{lemma:abstract-crypto-experiment}, for any adversary $(\alice,\bob,\charlie)$ in the anti-piracy experiment

 $\expt^{\left(\alice,\bob,\charlie \right),(\gen,\sign,\puncture,\verify),\UPO}\left( 1^{\secparam} \right)$, there exists a negligible function $\negl()$ such that,
\[ \prob\left[ 1 \leftarrow \expt^{\left(\alice,\bob,\charlie \right),(\gen,\sign,\puncture,\verify),\UPO}\left( 1^{\secparam} \right)\ \right] \leq \negl(\secparam).\]
However, $\expt^{\left(\alice,\bob,\charlie \right),(\gen,\sign,\puncture,\verify),\UPO}\left( 1^{\secparam} \right)$ and $\expt^{\left(\alice,\bob,\charlie \right),\cpdigitalsig}\left( 1^{\secparam} \right)$  are the same experiments and therefore, we conclude that anti-piracy holds for the $\cpdigitalsig$ with respect to uniform and independent challenge distribution.

\end{proof}
\begin{remark}\label{remark:signature-copy-protection-with-identical-challenge-distribution}
  By the same arguments as in the proof of \cref{thm:signature-copy-protection-application}, it can be shown that any $\upo$ scheme $(\obf,\eval)$ with $\distrid$-$\UPO$ security (see \Cref{def:newcpsecurity_bot}) for any puncturable keyed circuit class in ${\sf P/poly}$ (see \Cref{subsec:upo-definition}), is also a copy-protection scheme $(\copyprotect,\eval)$ for $\fclass=\{\fclass_{\secparam}\}_{\secparam \in \mathbb{N}}$ with uniform and \textit{identical} challenge distribution, where $\copyprotect()=\obf()$. 
\end{remark}\anote{Prabhanjan, we can also use the cloning games paper of yours to go from independent to arbitrary correlated challenge distribution. I think it is worth mentioning here because that way we can claim written proof instead of this verbal argument. What do you think?}\anote{I have added this in the beginning of the application section}

Since copy-protection for signatures implies public-key quantum money schemes, we get the following corollary.
\begin{corollary}\label{cor:public-quantum-money-from-puncturablesig+upo}
Suppose  $\digitalsig=(\gen,\sign,\puncture,\verify,\sign^*)$ be a puncturable digital signature with messge length $n(\secparam)$ and signature length $s(\secparam)$. 

Given a $\upo$ scheme $(\obf,\eval)$ with $\UPO$-security (see \Cref{def:newcpsecurity_bot})  for $\fclass=\{\fclass_\secparam\}_\secparam$ where $\fclass_\secparam=\{\sign(k,\cdot)\}_{k\in Support(\gen(1^\secparam))}$, equipped with $\puncture$ as the puncturing algorithm, with respect to $\distr_\inpclass=\cal{U}\times\cal{U}$, there exists a public-key quantum money scheme.
\end{corollary}

Combined with \Cref{thm:strong-CLLZ-cp-prf_puncturable-CP-f-uniform} and \Cref{thm:puncturable-sig-from-owf+io}, we conclude the following feasibility results for copy-protection scheme for signature and public quantum money from concrete assumptions.

\begin{corollary}\label{cor:signature-copy-protection-from-concrete assumption}
Suppose  $\digitalsig=(\gen,\sign,\puncture,\verify,\sign^*)$ be a puncturable digital signature with messge length $n(\secparam)$ and signature length $s(\secparam)$. 

 Assuming \Cref{conj:goldreich-levin-correlated}, the existence of post-quantum sub-exponentially secure $\io$ and one-way functions, and the quantum hardness of Learning-with-errors problem (LWE), 
    there exists a copy-protection scheme for signature scheme. Hence under the same assumption, a public-key quantum money scheme exists.
\end{corollary}

%% file: sde.tex
\subsection{Public-key Single-Decryptor Encryption} 
\label{sec:pksde:app}


\paragraph*{Construction} Our construction is based on copy-protecting the decryption functionality of the Sahai-Waters public-key encryption scheme based on $\io$, $\prf$ (mapping $n(\secparam)$ bits to $n(\secparam)$ bits), and $\prg$ (mapping $\frac{n(\secparam)}{2}$ bits to $n(\secparam)$ bits). We assume  a $\upo$ scheme $\UPO=(\obf,\eval)$ satisfying $\distrprod$-generalized security (see \Cref{def:newcpsecurity}) for any generalized puncturable keyed circuit class in $\ppoly$. In the security proofs, we will be considering the circuit class $\cktclass=\{\{\prf.\eval(k,\cdot)\}_{k\in {\sf Supp}(\keygen(1^\secparam))}\}_{\secparam}$ equipped with the distribution $\prf.\gen(1^\secparam)$ on the $\prf$ keys, and a puncturing or a generalized puncturing algorithms, derived accordingly from the $\prf.\puncture$ algorithm. 
\begin{savenotes}
\begin{figure}[!htb]
   \begin{center} 
   \begin{tabular}{|p{12cm}|}
    \hline 
\noindent\textbf{Assumes:} $\prf$ family $(\gen,\eval,\puncture)$, length-doubling $\prg$, $\io$, $\UPO$ scheme $(\obf,\eval)$.\\
\ \\
\noindent$\gen(1^\secparam)$
\begin{compactenum}
    \item Sample a key $k\gets \prf.\gen(1^\secparam)$. 
    \item Generate the circuit $C$ that on input $r\gets \{0,1\}^{\frac{n(\secparam)}{2}}$ (the input space of $\prg$) and a message $m\in \{0,1\}^n$, outputs $(\prg(r),\prf.\eval(k,\prg(r))\oplus m)$. 
    \item Compute $\tilde{C}\gets \io(C)$.
    \item Output $(\sk,\pk)=(k,\tilde{C})$.
\end{compactenum}
\ \\
\noindent$\qkeygen(\sk)$
\begin{compactenum}
    \item Compute $\tilde{F}\gets \io(\prf.\eval(\sk,\cdot))$.
    \item Output $\rho_\sk\gets\UPO.\obf(1^\secparam,\tilde{F})$\footnote{We assume that it is possible to read off the security parameter from the secret key $\sk$. For example, the secret key could start with $1^\secparam$ followed by a special symbol, and then followed by the actual key.}.
\end{compactenum}
\ \\

\noindent$\enc(\pk,m)$
\begin{compactenum}
    \item Interprete $\pk=\tilde{C}$
    \item Sample $r\xleftarrow{\$}\{0,1\}^{\frac{n}{2}}$.
    \item Output $\ct=\tilde{C}(r,m)$.
\end{compactenum}
\ \\
\noindent$\dec(\rho_\sk,\ct)$
\begin{compactenum}
    \item Interprete $\ct=y,z$.
    \item Output $m=\UPO.\eval(\rho_\sk,y)\oplus z$.
\end{compactenum}
\ \\
\hline
\end{tabular}
    \caption{A construction of $\sde$ based on~\cite{SW14} public-key encryption.}
    \label{fig:SDE-construction}
    \end{center}
\end{figure}
\end{savenotes}

\begin{theorem}
Assuming an indistinguishability obfuscation scheme $\io$ for $\ppoly$, a puncturable pseudorandom function family $\prf=(\gen,\eval,\puncture)$ and a generalized $\upo$ $\UPO$ for any generalized puncturable keyed circuit class in $\ppoly$ with respect to $\distr_\inpclass=U\times U$, 
 there exists a $\sde$ scheme given in \Cref{fig:SDE-construction} that satisfies correctness, search anti-piracy with independent and uniform distribution and $\distrcor$-selective $\cpa$-style anti-piracy (see \Cref{subsec:sde}). 
\end{theorem}
\begin{proof}
    The proof follows by combining \Cref{lemma:correctness-SDE-construction,prop:random-challenge-anti-piracy-sde,prop:cpa-style-anti-piracy-sde}.
\end{proof}

\begin{lemma}\label{lemma:correctness-SDE-construction}
    The $\sde$ construction given in \Cref{fig:SDE-construction} satisfies correctness with the same correctness precision as the underlying $\UPO$ scheme. 
\end{lemma}
The proof is immediate, so we omit the proof.

\begin{proposition}\label{prop:random-challenge-anti-piracy-sde}
    The $\sde$ construction given in \Cref{fig:SDE-construction} satisfies search anti-piracy with independent and uniform distribution (see \Cref{subsec:sde}) if the underlying $\UPO$ scheme satisfies $\upo$ security for any puncturable keyed circuit class in $\ppoly$. 
\end{proposition}

We first identify a scheme $(\gen,\eval,\verify,\puncture)$ (defined in \Cref{fig:SW14-puncturable-crypto-construction}) based on the public-key encryption scheme given in~\cite{SW14}, and show that it is a puncturable cryptographic scheme, as defined in \Cref{def:puncturable-crypto}, see \Cref{lemma:Sahai-waters-puncturability}. This result would be required in the proof of \Cref{prop:random-challenge-anti-piracy-sde} given on \Cpageref{pf:prop:random-challenge-anti-piracy-sde}. 

\begin{figure}[!htb]
   \begin{center} 
   \begin{tabular}{|p{12cm}|}
    \hline 
\noindent\textbf{Assumes:} $\prf$ family $(\gen,\eval,\puncture)$, length-doubling $\prg$, $\io$, $\UPO$ scheme $(\obf,\eval)$\\
\ \\
\noindent$\gen(1^\secparam)$: 
Generate $(k,\tilde{C})\gets \SDE.\gen(1^\secparam)$ where $\SDE$ is the $\sde$ given in \Cref{fig:SDE-construction}, and output $(\sk,\aux)$ where $\sk=k$ and $\aux=\pk$.\\
\ \\
\noindent $\eval(\sk,x)$: Same as $\prf.\eval(\sk,x)$.\\ 
\ \\
\noindent$\verify(\sk,\aux,x,y)$: Check if $\prf.\eval(\sk,x)=y$ and if true outputs $1$ else $0$.\\
\ \\
\noindent$\puncture(\sk,x_1,x_2)$: Generate $\sk_{x_1,x_2}\leftarrow \prf.\puncture(\sk,x_1,x_2)$ and output $\prf.\eval(\sk_{x_1,x_2},\cdot)$. 

\ \\
\hline
\end{tabular}
    \caption{A construction of puncturable cryptographic scheme based on~\cite{SW14} public-key encryption.}
    \label{fig:SW14-puncturable-crypto-construction}
    \end{center}
\end{figure}

\begin{lemma}\label{lemma:Sahai-waters-puncturability}
    The scheme $(\gen,\eval,\puncture,\verify)$ given in \Cref{fig:SW14-puncturable-crypto-construction} is a puncturable cryptographic scheme, as defined in \Cref{def:puncturable-crypto}.
\end{lemma}
\begin{proof}
The correctness and correctness of punctured circuit for $(\gen,\eval,\puncture,\verify)$ is immediate.
Next, we prove the puncturable security.

Let $A$ be an adversary in the puncturing experiment given in~\Cref{def:puncturable-crypto} for the puncturable cryptographic scheme $(\gen,\eval,\puncture,\verify)$.
\noindent $\hybrid_0$: \\
Same as the puncturing security experiment given in~\Cref{def:puncturable-crypto}.
\begin{itemize}
\item $\ch$ samples $k\gets \prf.\gen(1^{\secparam})$.
\item $\ch$ generates the circuit $\tilde{C}\gets \io(C)$ where $C$ has $k$ hardcoded and on input $r\gets \{0,1\}^{\frac{n}{2}}$ (the input space of $\prg$) and a message $m\in \{0,1\}^n$, outputs $(\prg(r),\prf.\eval(k,\prg(r))\oplus m)$. 
\item $\ch$ samples $x_1,x_2\xleftarrow{\$}\{0,1\}^{n}$.
\item $\ch$ generates $k_{x_1,x_2}\gets \prf.\puncture(k,\{x_1,x_2\})$.
\item $\ch$ sends $(x_1,k_{x_1,x_2},\tilde{C})$ to $A$ and gets back $y$.
\item $\ch$ computes $y_1\gets \prf.\eval(k,x_1)$.
\item Output $1$ if $y=y_1$.
\end{itemize}    

\noindent $\hybrid_1$: \\

\begin{itemize}
\item $\ch$ samples $k\gets \prf.\gen(1^{\secparam})$.
\item $\ch$ generates the circuit $\tilde{C}\gets \io(C)$ where $C$ has \sout{$k$}\cblue{$k_{x^\bob,x^\charlie}$} hardcoded and on input $r\gets \{0,1\}^{\frac{n}{2}}$ (the input space of $\prg$) and a message $m\in \{0,1\}^n$, outputs $(\prg(r),\prf.\eval($\sout{$k$}\cblue{$k_{x^\bob,x^\charlie}$}$,\prg(r))\oplus m)$.  
\item $\ch$ samples $x_1,x_2\xleftarrow{\$}\{0,1\}^{n}$.
\item $\ch$ generates $k_{x_1,x_2}\gets \prf.\puncture(k,\{x_1,x_2\})$.
\item $\ch$ sends $(x_1,k_{x_1,x_2},\tilde{C})$ to $A$ and gets back $y$.
\item $\ch$ computes $y_1\gets \prf.\eval(k,x_1)$.
\item Output $1$ if $y=y_1$.
\end{itemize}

The proof of indistinguishability between $\hybrid_0$ and $\hybrid_1$ is as follows. Note that $x_1,x_2\xleftarrow{\$}\{0,1\}^n$ and ${\sf Supp}(\prg)\subset\{0,1\}^n$ has size $2^{\frac{n}{2}}$, and hence is a negligible fraction of $\{0,1\}^n$. Hence, with overwhelming probability $x_1,x_2\not\in {\sf Supp}(\prg)$. Therefore with overwhelming probability, $C$ as in $\hybrid_0$ never computes $\prf.\eval(k,\cdot)$ on $x_1$ or $x_2$ on any input query. Hence, replacing $k$ with $k_{x_1,x_2}$ inside $C$ does not change the functionality of $C$, by the puncturing correctness of $\prf$. Therefore, indistinguishability holds by the $\io$ guarantee.

\noindent $\hybrid_2$: \\

\begin{itemize}
\item $\ch$ samples $k\gets \prf.\gen(1^{\secparam})$.
\item $\ch$ generates the circuit $\tilde{C}\gets \io(C)$ where $C$ has \sout{$k$}\cblue{$k_{x^\bob,x^\charlie}$} hardcoded and on input $r\gets \{0,1\}^{\frac{n}{2}}$ (the input space of $\prg$) and a message $m\in \{0,1\}^n$, outputs $(\prg(r),\prf.\eval($\sout{$k$}\cblue{$k_{x^\bob,x^\charlie}$}$,\prg(r))\oplus m)$.  
\item $\ch$ samples $x_1,x_2\xleftarrow{\$}\{0,1\}^{n}$.
\item $\ch$ generates $k_{x_1,x_2}\gets \prf.\puncture(k,\{x_1,x_2\})$.
\item $\ch$ sends $(x_1,k_{x_1,x_2},\tilde{C})$ to $A$ and gets back $y$.
\item $\ch$ \sout{computes $y_1\gets \prf.\eval(k,x_1)$} \cblue{samples $y_1\xleftarrow{\$}\{0,1\}^n$}.
\item Output $1$ if $y=y_1$.
\end{itemize}    

The indistinguishability holds because the view of $A$ in $\hybrid_1$ depends only on $k_{x_1,x_2}$ and not on $k$. Hence, $A$ cannot distinguish between $y_1\gets \prf.\eval(k,x_1)$ with $y_1\xleftarrow{\$}\{0,1\}^n$. Therefore, checking if $y$, the response of $A$ is equal to $y_1$ when $y_1\gets \prf.\eval(k,x_1)$ should be indistinguishable from the same experiment but with $y_1\xleftarrow{\$}\{0,1\}^n$.

Finally, we argue that since $y_1$ is sampled independent of $y$, the probability that $y=y_1$, i.e., the output of $\hybrid_2$ is $1$, is exactly $\frac{1}{2^n}$, which is a  negligible function of $\secparam$ since $n(\secparam)\in \poly(\secparam)$.

\end{proof}

\begin{proof}[Proof of \Cref{prop:random-challenge-anti-piracy-sde}]
\label{pf:prop:random-challenge-anti-piracy-sde}

Let $(\alice,\bob,\charlie)$ be any adversary in $\srchsdeexpt^{\left(\alice,\bob,\charlie \right),\distr}\left( 1^{\secparam} \right)$ (see \Cref{fig:product-uniform-sde-random-challenge-anti-piracy}). We will do a sequence of hybrids starting from the original anti-piracy experiment $\srchsdeexpt^{\left(\alice,\bob,\charlie \right),\distr}\left( 1^{\secparam} \right)$ for the $\sde$ scheme given in \Cref{fig:SDE-construction}.
The changes are marked in \cblue{blue}.

\noindent $\hybrid_0$: \\
Same as $\srchsdeexpt^{\left(\alice,\bob,\charlie \right),\distr}\left( 1^{\secparam} \right)$ given in~\Cref{fig:product-uniform-sde-random-challenge-anti-piracy} for the $\sde$ scheme in \Cref{fig:SDE-construction}.
\begin{itemize}
\item $\ch$ samples $k\gets \prf.\gen(1^{\secparam})$.
\item $\ch$ samples $r^\bob,r^\charlie\xleftarrow{\$}\{0,1\}^{\frac{n}{2}}$ and generates $x^\bob\gets\prg(r^\bob)$ and $x^\charlie\gets\prg(r^\charlie)$.
\item $\ch$ generates the circuit $\tilde{C}\gets \io(C)$ where $C$ has $k$ hardcoded and on input $r\gets \{0,1\}^{\frac{n}{2}}$ (the input space of $\prg$) and a message $m\in \{0,1\}^n$, outputs $(\prg(r),\prf.\eval(k,\prg(r))\oplus m)$.  
\item $\ch$ generates $\rho_\sk\gets \UPO.\obf(1^\secparam,\tilde{F})$ where $\tilde{F}\gets \io(\prf.\eval(k,\cdot))$ and sends $(\rho_\sk,\tilde{C})$ to $\alice$.
\item $\alice$ produces a bipartite state $\sigma_{\bob,\charlie}$.
\item $\ch$ samples $m^\bob,m^\charlie\xleftarrow{\$} \{0,1\}^n$.
\item $\ch$ computes $\ct^\bob=(x^\bob,z^\bob)$ and $\ct^\charlie=(x^\charlie,z^\charlie)$ where $z^\bob=\prf.\eval(k,x^\bob)\oplus m^\bob$ and $z^\charlie=\prf.\eval(k,x^\charlie)\oplus m^\charlie$.
\item Apply $(\bob(\ct^{\bob},\cdot) \otimes \charlie(\ct^{\charlie},\cdot))(\sigma_{\bob,\charlie})$ to obtain $(y^{\bob},y^{\charlie})$. 
\item Output $1$ if $y^\bob=m^\bob$ and $y^\charlie=m^\charlie$.
\end{itemize}    

\noindent $\hybrid_1$: \\
\begin{itemize}
\item $\ch$ samples $k\gets \prf.\gen(1^{\secparam})$.
\item $\ch$ samples \sout{$r^\bob,r^\charlie\xleftarrow{\$}\{0,1\}^{\frac{n}{2}}$ and generates $x^\bob\gets\prg(r^\bob)$ and $x^\charlie\gets\prg(r^\charlie)$} \cblue{$x^\bob,x^\charlie\xleftarrow{\$}\{0,1\}^n$}.
\item $\ch$ generates the circuit $\tilde{C}\gets \io(C)$ where $C$ has $k$ hardcoded and on input $r\gets \{0,1\}^{\frac{n}{2}}$ (the input space of $\prg$) and a message $m\in \{0,1\}^n$, outputs $(\prg(r),\prf.\eval(k,\prg(r))\oplus m)$.  
\item $\ch$ generates $\rho_\sk\gets \UPO.\obf(1^\secparam,\tilde{F})$ where $\tilde{F}\gets \io(\prf.\eval(k,\cdot))$ and sends $(\rho_\sk,\tilde{C})$ to $\alice$.
\item $\alice$ produces a bipartite state $\sigma_{\bob,\charlie}$.
\item $\ch$ samples $m^\bob,m^\charlie\xleftarrow{\$} \{0,1\}^n$.
\item $\ch$ computes $\ct^\bob=(x^\bob,z^\bob)$ and $\ct^\charlie=(x^\charlie,z^\charlie)$ where $z^\bob=\prf.\eval(k,x^\bob)\oplus m^\bob$ and $z^\charlie=\prf.\eval(k,x^\charlie)\oplus m^\charlie$.
\item Apply $(\bob(\ct^{\bob},\cdot) \otimes \charlie(\ct^{\charlie},\cdot))(\sigma_{\bob,\charlie})$ to obtain $(y^{\bob},y^{\charlie})$. 
\item Output $1$ if $y^\bob=m^\bob$ and $y^\charlie=m^\charlie$.
\end{itemize}    
The indistinguishability between $\hybrid_0$ and $\hybrid_1$ follows from the pseudorandomness of $\prg$.

\noindent $\hybrid_2$: \\
\begin{itemize}
\item $\ch$ samples $k\gets \prf.\gen(1^{\secparam})$.
\item $\ch$ samples $x^\bob,x^\charlie\xleftarrow{\$}\{0,1\}^n$.
\item $\ch$ generates the circuit $\tilde{C}\gets \io(C)$ where $C$ has $k$ hardcoded and on input $r\gets \{0,1\}^{\frac{n}{2}}$ (the input space of $\prg$) and a message $m\in \{0,1\}^n$, outputs $(\prg(r),\prf.\eval(k,\prg(r))\oplus m)$.  
\item $\ch$ generates $\rho_\sk\gets \UPO.\obf(1^\secparam,\tilde{F})$ where $\tilde{F}\gets \io(\prf.\eval(k,\cdot))$ and sends $(\rho_\sk,\tilde{C})$ to $\alice$.
\item $\alice$ produces a bipartite state $\sigma_{\bob,\charlie}$.
\item $\ch$ samples \sout{$m^\bob,m^\charlie\xleftarrow{\$} \{0,1\}^n$} \cblue{$z^\bob,z^\charlie\xleftarrow{\$} \{0,1\}^n$ and computes $m^\bob=\prf.\eval(k,x^\bob)\oplus z^\bob$, $m^\charlie=\prf.\eval(k,x^\charlie)\oplus z^\charlie$}.
\item $\ch$ computes $\ct^\bob=(x^\bob,z^\bob)$ and $\ct^\charlie=(x^\charlie,z^\charlie)$ \sout{ where $z^\bob=\prf.\eval(k,x^\bob)\oplus m^\bob$ and $z^\charlie=\prf.\eval(k,x^\charlie)\oplus m^\charlie$}.
\item Apply $(\bob(\ct^{\bob},\cdot) \otimes \charlie(\ct^{\charlie},\cdot))(\sigma_{\bob,\charlie})$ to obtain $(y^{\bob},y^{\charlie})$. 
\item Output $1$ if $y^\bob=m^\bob$ and $y^\charlie=m^\charlie$.
\end{itemize}    
The overall distribution on $(m^\bob,z^\bob)$ and $(m^\charlie,z^\charlie)$ across the hybrids $\hybrid_1$ and $\hybrid_2$, and hence the indistinguishability holds.

\noindent $\hybrid_3$: \\
\begin{itemize}
\item $\ch$ samples $k\gets \prf.\gen(1^{\secparam})$.
\item $\ch$ samples $x^\bob,x^\charlie\xleftarrow{\$}\{0,1\}^n$.
\item $\ch$ generates the circuit $\tilde{C}\gets \io(C)$ where $C$ has $k$ hardcoded and on input $r\gets \{0,1\}^{\frac{n}{2}}$ (the input space of $\prg$) and a message $m\in \{0,1\}^n$, outputs $(\prg(r),\prf.\eval(k,\prg(r))\oplus m)$.  
\item $\ch$ generates \sout{$\rho_\sk\gets \UPO.\obf(1^\secparam,\tilde{F})$ where $\tilde{F}\gets \io(\prf.\eval(k,\cdot))$} \cblue{$\rho_\sk\gets \UPO'.\obf(1^\secparam,\prf.\eval(k,\cdot))$} and sends $(\rho_\sk,\tilde{C})$ to $\alice$.
\item $\alice$ produces a bipartite state $\sigma_{\bob,\charlie}$.
\item $\ch$ samples $z^\bob,z^\charlie\xleftarrow{\$} \{0,1\}^n$ and computes $m^\bob=\prf.\eval(k,x^\bob)\oplus z^\bob$, $m^\charlie=\prf.\eval(k,x^\charlie)\oplus z^\charlie$.
\item $\ch$ computes $\ct^\bob=(x^\bob,z^\bob)$ and $\ct^\charlie=(x^\charlie,z^\charlie)$. 
\item Apply $(\bob(\ct^{\bob},\cdot) \otimes \charlie(\ct^{\charlie},\cdot))(\sigma_{\bob,\charlie})$ to obtain $(y^{\bob},y^{\charlie})$. 
\item Output $1$ if $y^\bob=m^\bob$ and $y^\charlie=m^\charlie$.
\end{itemize}    

$\hybrid_3$ is just a rewriting of $\hybrid_2$ in terms of the new $\upo$ scheme defined as:
 \begin{itemize}
        \item $\UPO'.\obf(1^\secparam,C)=\UPO.\obf(1^\secparam,\tilde{C})$ where $\tilde{C}\gets \io(C)$, for every circuit $C$.
        \item $\UPO'.\eval=\UPO.\eval$.
    \end{itemize}

Note that by \Cref{cor:UPO-io}, since  $\UPO$ is a $\upo$ for any generalized keyed circuit class in $\ppoly$ with respect to $\distr_\inpclass=\cal{U}\times \cal{U}$, the product of uniform distribution, so is $\UPO'$.

Next, we give a reduction from $\hybrid_3$ to an anti-piracy game with uniform and independent challenge distribution (see \Cref{fig:product-uniform-arb-crypto-experiment-anti-piracy}) for $(\gen,\eval,\puncture,\verify)$ with respect to $\UPO'$ where $\gen$ on input $1^\secparam$ samples a key $k\gets \prf.\gen(1^\secparam)$ and then constructs the circuit $\tilde{C}\gets \io(C)$ where $C$ has $k$ hardcoded and on input $r\gets \{0,1\}^{\frac{n}{2}}$ (the input space of $\prg$) and a message $m\in \{0,1\}^n$, outputs $(\prg(r),\prf.\eval(k,\prg(r))\oplus m)$, and finally outputs $(\sk,\aux)=(k,\tilde{C})$. $\eval$ is the same as $\prf.\eval$; the $\verify()$ algorithm on input $k,\tilde{C},x,y$ checks if $\prf.\eval(k,x)=y$ and if true outputs $1$ else $0$. Finally, the $\puncture()$ algorithm on input a key $k$ and a set of input points $(x_1,x_2)$, generates $k_{x_1,x_2}\gets \prf.\puncture(k,x_1,x_2)$ and outputs $\prf.\eval(k_{x_1,x_2},\cdot)$. 


Let $(\alice,\bob,\charlie)$ be an adversary  in $\hybrid_2$ above. Consider the following adversary $(\reduc_\alice,\reduc_\bob,\reduc_\charlie)$ in $\expt^{\left(\reduc_\alice,\reduc_\bob,\reduc_\charlie \right),(\gen,\eval,\puncture,\verify)}\left( 1^{\secparam} \right)$ (see \Cref{fig:product-uniform-arb-crypto-experiment-anti-piracy}):
\begin{itemize}
    \item $\reduc_\alice$ on receiving $(\rho_\sk,\tilde{C})$ from the challenger $\ch$ in  $\expt^{\left(\reduc_\alice,\reduc_\bob,\reduc_\charlie \right),(\gen,\eval,\puncture,\verify)}\left( 1^{\secparam} \right)$ (see \Cref{fig:product-uniform-arb-crypto-experiment-anti-piracy}), runs $\alice$ on it to generate $\sigma_{\bob,\charlie}$ and sends the respective registers to $\reduc_\bob$ and $\reduc_\charlie$.
    \item $\reduc_\bob$ (respectively, $\reduc_\charlie$) on receiving $x^\bob$ (respectively $x^\charlie$), samples $z^\bob\xleftarrow{\$}\{0,1\}^n$ (respectively, $z^\charlie$) and runs $\bob$ (respectively, $\charlie$) on $((z^\bob,x^\bob),\sigma_\bob)$ (respectively, $((z^\charlie,x^\charlie),\sigma_\charlie)$) to get $m^\bob$ (respectively, $m^\charlie$). $\reduc_\bob$ (respectively, $\reduc_\charlie$) outputs $m^\bob\oplus z^\bob$ (respectively, $m^\charlie\oplus z^\charlie$).   
\end{itemize}

Clearly, the event $1\gets \expt^{\left(\reduc_\alice,\reduc_\bob,\reduc_\charlie \right),(\gen,\eval,\puncture,\verify),\UPO'}\left( 1^{\secparam} \right)$ (see \Cref{fig:product-uniform-arb-crypto-experiment-anti-piracy}) exactly corresponds to the event $(\alice,\bob,\charlie)$ winning the security experiment in $\hybrid_3$.

By \Cref{lemma:Sahai-waters-puncturability}, we know that $(\gen,\eval,\puncture,\verify)$ is a puncturable cryptographic scheme. Hence by \Cref{lemma:abstract-crypto-experiment},  for every adversary $(\alice,\bob,\charlie)$ in \Cref{fig:product-uniform-arb-crypto-experiment-anti-piracy} against $(\gen,\eval,\puncture,\verify)$, there exists a negligible function $\negl()$ such that
\[ \prob\left[ 1 \leftarrow \expt^{\left(\alice,\bob,\charlie \right),(\gen,\eval,\puncture,\verify),\UPO}\left( 1^{\secparam} \right)\ \right] \leq \negl(\secparam).\]

Hence by the reduction, we conclude that $(\alice,\bob,\charlie)$ has negligible winning probability in the security experiment in $\hybrid_3$, which completes the proof.

\end{proof}

\begin{proposition}\label{prop:cpa-style-anti-piracy-sde}
    The $\sde$ construction given in \Cref{fig:SDE-construction} satisfies $\distrcor$-selective $\cpa$-style anti-piracy (see \Cref{subsec:sde}). 
\end{proposition}

\begin{proof}

Let $\UPO'$ be a new $\upo$ scheme defined as:
 \begin{itemize}
        \item $\UPO'.\obf(1^\secparam,C)=\UPO.\obf(1^\secparam,\tilde{C})$ where $\tilde{C}\gets \io(C)$,  for every circuit $C$.
        \item $\UPO'.\eval=\UPO.\eval$.
    \end{itemize}
By \Cref{cor:UPO-io}, since  $\UPO$ is a $\upo$ for any generalized keyed circuit class in $\ppoly$ with respect to the independent challenge distribution $\distr_\inpclass=\cal{U}\times \cal{U}$, $\UPO'$ also satisfies the same security guarantees.

Let $(\alice,\bob,\charlie)$ be any adversary in $\selcpasdeexpt^{\left(\alice,\bob,\charlie \right),\distrcor}\left( 1^{\secparam} \right)$ (see \Cref{fig:correlated-sde-cpa-style-anti-piracy}) against the $\sde$ construction in \Cref{fig:SDE-construction}. We will do a sequence of hybrids starting from the original anti-piracy experiment $\selcpasdeexpt^{\left(\alice,\bob,\charlie \right),\distr}\left( 1^{\secparam} \right)$ for the $\sde$ scheme given in \Cref{fig:SDE-construction}, and finally give a reduction to the generalized $\upo$ security game of $\UPO'$ for  $\fclass=\{\fclass_\secparam\}$, where $\fclass_\secparam=\{\prf.\eval(k,\cdot)\}_{k\in {\sf Supp}(\prf.\gen(1^\secparam))}$ with respect to the puncture algorithm $\genpuncture$ defined as follows:
the $\genpuncture$ algorithm, which takes as input $(k,x_1,x_2,\mu_1,\mu_2)$ and does the following:
\begin{itemize}
    \item Generates $k_{x_1,x_2}\gets \prf.\puncture(k,x_1,x_2)$.
    \item Constructs the circuit $G_{k_{x_1,x_2},x_1,x_2,\mu_1,\mu_2}$ which on input $x$, outputs $\prf.\eval(k_{x_1,x_2},x)$ if $x\not\in\{x_1,x_2\}$, and outputs $\mu_1(x_1)$ if $x=x_1$ and  $\mu_2(x_2)$ if $x=x_2$.
    \item Output $E$.
\end{itemize}

The changes are marked in \cblue{blue}.

\noindent $\hybrid_0$: \\ \anote{Self: Need to change the prepone the step of adversary submitting challenge messages}
Same as $\selcpasdeexpt^{\left(\alice,\bob,\charlie \right),\distrcor}\left( 1^{\secparam} \right)$ given in~\Cref{fig:correlated-sde-cpa-style-anti-piracy} for the $\sde$ scheme in \Cref{fig:SDE-construction}.
 \begin{itemize}
\item $\alice$ sends two same-length message pairs $(m^\bob_0,m^\bob_1, m^\charlie_0,m^\charlie_1)$.
\item $\ch$ samples $k\gets \prf.\gen(1^{\secparam})$.
\item $\ch$ samples $r^\bob,r^\charlie\xleftarrow{\$}\{0,1\}^{\frac{n}{2}}$ and generates $x^\bob\gets\prg(r^\bob)$ and $x^\charlie\gets\prg(r^\charlie)$ as well as  generates $y^\bob\gets \prf.\eval(k,x^\bob)$ and $y^\charlie\gets \prf.\eval(k,x^\charlie)$.
\item $\ch$ generates the circuit $\tilde{C}\gets \io(C)$ where $C$ has $k$ hardcoded and on input $r\gets \{0,1\}^{\frac{n}{2}}$ (the input space of $\prg$) and a message $m\in \{0,1\}^n$, outputs $(\prg(r),\prf.\eval(k,\prg(r))\oplus m)$.  
\item $\ch$ generates $\rho_\sk\gets \UPO.\obf(1^\secparam,\tilde{F})$ where $\tilde{F}\gets \io(\prf.\eval(k,\cdot))$ and sends $(\rho_\sk,\tilde{C})$ to $\alice$.
\item $\alice$ produces a bipartite state $\sigma_{\bob,\charlie}$.
\item $\ch$ samples $b\xleftarrow{\$} \{0,1\}$.
\item $\ch$ computes $\ct^\bob=(x^\bob,z^\bob)$ and $\ct^\charlie=(x^\charlie,z^\charlie)$ where $z^\bob=y^\bob\oplus m^\bob_b$ and $z^\charlie=y^\charlie\oplus m^\charlie_b$.
\item Apply $(\bob(\ct^{\bob},\cdot) \otimes \charlie(\ct^{\charlie},\cdot))(\sigma_{\bob,\charlie})$ to obtain $(b^{\bob},b^{\charlie})$. 
\item Output $1$ if $b^\bob=b^\charlie=b$.
 
 \end{itemize}  

 \noindent $\hybrid_1$: \\
This is the same as $\hybrid_0$ up to re-ordering some of the steps performed by the $\ch$ without affecting view of the adversary.  
 \begin{itemize}
\item $\alice$ sends two same-length message pairs $(m^\bob_0,m^\bob_1, m^\charlie_0,m^\charlie_1)$.
\item $\ch$ samples $k\gets \prf.\gen(1^{\secparam})$.
\item \cblue{$\ch$ samples $b\xleftarrow{\$} \{0,1\}$.}
\item $\ch$ samples $r^\bob,r^\charlie\xleftarrow{\$}\{0,1\}^{\frac{n}{2}}$ and generates $x^\bob\gets\prg(r^\bob)$ and $x^\charlie\gets\prg(r^\charlie)$ as well as  generates $y^\bob\gets \prf.\eval(k,x^\bob)$ and $y^\charlie\gets \prf.\eval(k,x^\charlie)$.
\item $\ch$ generates the circuit $\tilde{C}\gets \io(C)$ where $C$ has $k$ hardcoded and on input $r\gets \{0,1\}^{\frac{n}{2}}$ (the input space of $\prg$) and a message $m\in \{0,1\}^n$, outputs $(\prg(r),\prf.\eval(k,\prg(r))\oplus m)$.  
\item $\ch$ generates $\rho_\sk\gets \UPO.\obf(1^\secparam,\tilde{F})$ where $\tilde{F}\gets \io(\prf.\eval(k,\cdot))$ and sends $(\rho_\sk,\tilde{C})$ to $\alice$.

\item $\alice$ produces a bipartite state $\sigma_{\bob,\charlie}$.
\item \sout{$\ch$ samples $b\xleftarrow{\$} \{0,1\}$.}
\item $\ch$ computes $\ct^\bob=(x^\bob,z^\bob)$ and $\ct^\charlie=(x^\charlie,z^\charlie)$ where $z^\bob=y^\bob\oplus m^\bob_b$ and $z^\charlie=y^\charlie\oplus m^\charlie_b$.
\item Apply $(\bob(\ct^{\bob},\cdot) \otimes \charlie(\ct^{\charlie},\cdot))(\sigma_{\bob,\charlie})$ to obtain $(b^{\bob},b^{\charlie})$. 
\item Output $1$ if $b^\bob=b^\charlie=b$.
 
 \end{itemize}

\noindent $\hybrid_2$: \\
\begin{itemize}
\item $\alice$ sends two same-length message pairs $(m^\bob_0,m^\bob_1, m^\charlie_0,m^\charlie_1)$.
\item $\ch$ samples $k\gets \prf.\gen(1^{\secparam})$.
\item $\ch$ samples $b\xleftarrow{\$} \{0,1\}$.
\item $\ch$ samples \sout{$r^\bob,r^\charlie\xleftarrow{\$}\{0,1\}^{\frac{n}{2}}$ and generates $x^\bob\gets\prg(r^\bob)$ and $x^\charlie\gets\prg(r^\charlie)$} \cblue{$x^\bob,x^\charlie\xleftarrow{\$}\{0,1\}^n$} as well as  generates $y^\bob\gets \prf.\eval(k,x^\bob)$ and $y^\charlie\gets \prf.\eval(k,x^\charlie)$.
\item $\ch$ generates the circuit $\tilde{C}\gets \io(C)$ where $C$ has $k$ hardcoded and on input $r\gets \{0,1\}^{\frac{n}{2}}$ (the input space of $\prg$) and a message $m\in \{0,1\}^n$, outputs $(\prg(r),\prf.\eval(k,\prg(r))\oplus m)$.  
\item $\ch$ generates $\rho_\sk\gets \UPO.\obf(1^\secparam,\tilde{F})$ where $\tilde{F}\gets \io(\prf.\eval(k,\cdot))$ and sends $(\rho_\sk,\tilde{C})$ to $\alice$.
\item $\alice$ produces a bipartite state $\sigma_{\bob,\charlie}$.
\item $\ch$ samples $b\xleftarrow{\$} \{0,1\}$.
\item $\ch$ computes $\ct^\bob=(x^\bob,z^\bob)$ and $\ct^\charlie=(x^\charlie,z^\charlie)$ where $z^\bob=y^\bob\oplus m^\bob_b$ and $z^\charlie=y^\charlie\oplus m^\charlie_b$.
\item Apply $(\bob(\ct^{\bob},\cdot) \otimes \charlie(\ct^{\charlie},\cdot))(\sigma_{\bob,\charlie})$ to obtain $(b^{\bob},b^{\charlie})$. 
\item Output $1$ if $b^\bob=b^\charlie=b$.
\end{itemize}    
The indistinguishability between $\hybrid_1$ and $\hybrid_2$ follows from the pseudorandomness of $\prg$.

\noindent $\hybrid_3$: \\
\begin{itemize}
\item $\alice$ sends two same-length message pairs $(m^\bob_0,m^\bob_1, m^\charlie_0,m^\charlie_1)$.
\item $\ch$ samples $k\gets \prf.\gen(1^{\secparam})$.
\item $\ch$ samples $b\xleftarrow{\$} \{0,1\}$. 
\item $\ch$ samples $x^\bob,x^\charlie\xleftarrow{\$}\{0,1\}^n$ as well as  generates $y^\bob\gets \prf.\eval(k,x^\bob)$ and $y^\charlie\gets \prf.\eval(k,x^\charlie)$.
\item \cblue{$\ch$ generates $k_{x^\bob,x^\charlie}\gets \prf.\puncture(k,\{x^\bob,x^\charlie\})$.}
\item $\ch$ generates the circuit $\tilde{C}\gets \io(C)$ \sout{where $C$ has $k$ hardcoded and on input $r\gets \{0,1\}^{\frac{n}{2}}$ (the input space of $\prg$) and a message $m\in \{0,1\}^n$, outputs $(\prg(r),\prf.\eval(k,\prg(r))\oplus m)$.} \cblue{where $C$ is constructed depending on the bit $b$ as follows. If $b=0$ (respectively, $b=1$), $C$ has $k$ (respectively, $k_{x^\bob,x^\charlie}$) hardcoded and on input $r\gets \{0,1\}^{\frac{n}{2}}$ (the input space of $\prg$) and a message $m\in \{0,1\}^n$, outputs $(\prg(r),\prf.\eval(k,\prg(r))\oplus m)$ (respectively, $(\prg(r),\prf.\eval(k_{x^\bob,x^\charlie},\prg(r))\oplus m)$).}  
\item $\ch$ generates $\rho_\sk\gets \UPO.\obf(1^\secparam,\tilde{F})$ where $\tilde{F}\gets \io(\prf.\eval(k,\cdot))$ and sends $(\rho_\sk,\tilde{C})$ to $\alice$.
\item $\alice$ produces a bipartite state $\sigma_{\bob,\charlie}$.

\item $\ch$ computes $\ct^\bob=(x^\bob,z^\bob)$ and $\ct^\charlie=(x^\charlie,z^\charlie)$ where $z^\bob=y^\bob\oplus m^\bob_b$ and $z^\charlie=y^\charlie\oplus m^\charlie_b$.
\item Apply $(\bob(\ct^{\bob},\cdot) \otimes \charlie(\ct^{\charlie},\cdot))(\sigma_{\bob,\charlie})$ to obtain $(b^{\bob},b^{\charlie})$. 
\item Output $1$ if $b^\bob=b^\charlie=b$.
\end{itemize}       

The proof of indistinguishability between $\hybrid_2$ and $\hybrid_3$ is as follows. Note that $x^\bob,x^\charlie\xleftarrow{\$}\{0,1\}^n$ and ${\sf Supp}(\prg)\subset\{0,1\}^n$ has size $2^{\frac{n}{2}}$, and hence is a negligible fraction of $\{0,1\}^n$. Hence, with overwhelming probability $x^\bob,x^\charlie\not\in {\sf Supp}(\prg)$. Therefore with overwhelming probability, $C$ as in $\hybrid_0$ never computes $\prf.\eval(k,\cdot)$ on $x^\bob$ or $x^\charlie$ on any input query. Hence, replacing $k$ with $k_{x_1,x_2}$ inside $C$ in the $b=1$ case of the security experiment does not change the functionality of $C$, by the puncturing correctness of $\prf$. Therefore, indistinguishability holds by the $\io$ guarantee.

\noindent $\hybrid_4$: \\
\begin{itemize}
\item $\alice$ sends two same-length message pairs $(m^\bob_0,m^\bob_1, m^\charlie_0,m^\charlie_1)$.
\item $\ch$ samples $k\gets \prf.\gen(1^{\secparam})$.
\item $\ch$ samples $b\xleftarrow{\$} \{0,1\}$. 
\item $\ch$ samples $x^\bob,x^\charlie\xleftarrow{\$}\{0,1\}^n$ as well as  generates $y^\bob\gets \prf.\eval(k,x^\bob)$ and $y^\charlie\gets \prf.\eval(k,x^\charlie)$.
\item \cblue{$\ch$ generates $k_{x^\bob,x^\charlie}\gets \prf.\puncture(k,\{x^\bob,x^\charlie\})$.}
\item $\ch$ generates the circuit $\tilde{C}\gets \io(C)$ where $C$ is constructed depending on the bit $b$ as follows. If $b=0$ (respectively, $b=1$), $C$ has $k$ (respectively, $k_{x^\bob,x^\charlie}$) hardcoded and on input $r\gets \{0,1\}^{\frac{n}{2}}$ (the input space of $\prg$) and a message $m\in \{0,1\}^n$, outputs $(\prg(r),\prf.\eval(k,\prg(r))\oplus m)$ (respectively, $(\prg(r),\prf.\eval(k_{x^\bob,x^\charlie},\prg(r))\oplus m)$).
\item \cblue{If $b=0$,} $\ch$ generates $\rho_\sk\gets \UPO.\obf(1^\secparam,\tilde{F})$ where $\tilde{F}\gets \io(\prf.\eval(k,\cdot))$, \cblue{else, if $b=1$, generates  $\rho_\sk\gets \UPO.\obf(1^\secparam,\tilde{W})$, where $\tilde{W}\gets \io(W)$ and $W$ is as depicted in \Cref{fig:W-hybrid-4}}  and sends $(\rho_\sk,\tilde{C})$ to $\alice$.
\item $\alice$ produces a bipartite state $\sigma_{\bob,\charlie}$.

\item $\ch$ computes $\ct^\bob=(x^\bob,z^\bob)$ and $\ct^\charlie=(x^\charlie,z^\charlie)$ where $z^\bob=y^\bob\oplus m^\bob_b$ and $z^\charlie=y^\charlie\oplus m^\charlie_b$.
\item Apply $(\bob(\ct^{\bob},\cdot) \otimes \charlie(\ct^{\charlie},\cdot))(\sigma_{\bob,\charlie})$ to obtain $(b^{\bob},b^{\charlie})$. 
\item Output $1$ if $b^\bob=b^\charlie=b$.
\end{itemize}

\begin{figure}[!htb]
   \begin{center} 
   \begin{tabular}{|p{12cm}|}
    \hline 
\begin{center}
\underline{$W$}: 
\end{center}
Hardcoded keys $k_{x^\bob,x^\charlie},y^\bob,y^\charlie$.
On input: $x$.
\begin{itemize}
\item If $x=x^\bob$, output $y^\bob$.
\item Else if, $x=x^\charlie$, output $y^\charlie$.
\item Else, run $\prf.\eval(k_{x^\bob,x^\charlie},x)$ and output the result.
\end{itemize}
\ \\ 
\hline
\end{tabular}
    \caption{Circuit $W$ in $\hybrid_{4}$}
    \label{fig:W-hybrid-4}
    \end{center}
\end{figure}

Clearly, $W$ and $\prf.\eval(k,\cdot)$ has the same functionality and therefore indistinguishability holds by $\io$ guarantees.

\noindent $\hybrid_5$: \\
\begin{itemize}
\item $\alice$ sends two same-length message pairs $(m^\bob_0,m^\bob_1, m^\charlie_0,m^\charlie_1)$.
\item $\ch$ samples $k\gets \prf.\gen(1^{\secparam})$.
\item $\ch$ samples $b\xleftarrow{\$} \{0,1\}$. 
\item $\ch$ samples $x^\bob,x^\charlie\xleftarrow{\$}\{0,1\}^n$ as well as generates $y^\bob\gets \prf.\eval(k,x^\bob)$, $y^\charlie\gets \prf.\eval(k,x^\charlie)$ \cblue{if $b=0$; and $y^\bob\xleftarrow{\$}\{0,1\}^n$, and $y^\charlie\xleftarrow{\$}\{0,1\}^n$ if $b=1$}.
\item $\ch$ generates $k_{x^\bob,x^\charlie}\gets \prf.\puncture(k,\{x^\bob,x^\charlie\})$.
\item $\ch$ generates the circuit $\tilde{C}\gets \io(C)$ where $C$ is constructed depending on the bit $b$ as follows. If $b=0$ (respectively, $b=1$), $C$ has $k$ (respectively, $k_{x^\bob,x^\charlie}$) hardcoded and on input $r\gets \{0,1\}^{\frac{n}{2}}$ (the input space of $\prg$) and a message $m\in \{0,1\}^n$, outputs $(\prg(r),\prf.\eval(k,\prg(r))\oplus m)$ (respectively, $(\prg(r),\prf.\eval(k_{x^\bob,x^\charlie},\prg(r))\oplus m)$).
\item If $b=0$, $\ch$ generates $\rho_\sk\gets \UPO.\obf(1^\secparam,\tilde{F})$ where $\tilde{F}\gets \io(\prf.\eval(k,\cdot))$, else, if $b=1$, generates  $\rho_\sk\gets \UPO.\obf(1^\secparam,\tilde{W})$, where $\tilde{W}\gets \io(W)$ and $W$ is as depicted in \Cref{fig:W-hybrid-4}  and sends $(\rho_\sk,\io(C))$ to $\alice$.
\item $\alice$ produces a bipartite state $\sigma_{\bob,\charlie}$.
\item $\ch$ computes $\ct^\bob=(x^\bob,z^\bob)$ and $\ct^\charlie=(x^\charlie,z^\charlie)$ where $z^\bob=y^\bob\oplus m^\bob_b$ and $z^\charlie=y^\charlie\oplus m^\charlie_b$.
\item Apply $(\bob(\ct^{\bob},\cdot) \otimes \charlie(\ct^{\charlie},\cdot))(\sigma_{\bob,\charlie})$ to obtain $(b^{\bob},b^{\charlie})$. 
\item Output $1$ if $b^\bob=b^\charlie=b$.
\end{itemize}
Since the views of the adversary $(\alice,\bob,\charlie)$ in $b=1$ case in hybrids $\hybrid_4$ and $\hybrid_5$ are only dependent on $k_{x^\bob,x^\charlie}$, the indistinguishability between $\hybrid_4$ and $\hybrid_5$ holds by the puncturing security of $\prf$.


\noindent $\hybrid_6$: \\
\begin{itemize}
\item $\alice$ sends two same-length message pairs $(m^\bob_0,m^\bob_1, m^\charlie_0,m^\charlie_1)$.
\item $\ch$ samples $k\gets \prf.\gen(1^{\secparam})$.
\item $\ch$ samples $b\xleftarrow{\$} \{0,1\}$. 
\item $\ch$ samples $x^\bob,x^\charlie\xleftarrow{\$}\{0,1\}^n$\sout{ as well as generates $y^\bob\gets \prf.\eval(k,x^\bob)$, $y^\charlie\gets \prf.\eval(k,x^\charlie)$ if $b=0$; and $y^\bob\xleftarrow{\$}\{0,1\}^n$, and $y^\charlie\xleftarrow{\$}\{0,1\}^n$ if $b=1$}.
\item $\ch$ generates $k_{x^\bob,x^\charlie}\gets \prf.\puncture(k,\{x^\bob,x^\charlie\})$.
\item $\ch$ generates the circuit $\tilde{C}\gets \io(C)$ where $C$ is constructed depending on the bit $b$ as follows. If $b=0$ (respectively, $b=1$), $C$ has $k$ (respectively, $k_{x^\bob,x^\charlie}$) hardcoded and on input $r\gets \{0,1\}^{\frac{n}{2}}$ (the input space of $\prg$) and a message $m\in \{0,1\}^n$, outputs $(\prg(r),\prf.\eval(k,\prg(r))\oplus m)$ (respectively, $(\prg(r),\prf.\eval(k_{x^\bob,x^\charlie},\prg(r))\oplus m)$).
\item If $b=0$, $\ch$ generates $\rho_\sk\gets \UPO.\obf(1^\secparam,\tilde{F})$ where $\tilde{F}\gets \io(\prf.\eval(k,\cdot))$, else, if $b=1$, generates  $\rho_\sk\gets \UPO.\obf(1^\secparam,\tilde{W})$, where $\tilde{W}\gets \io(W)$ and $W$ is as depicted in \Cref{fig:W-hybrid-4}  and sends $(\rho_\sk,\io(C))$ to $\alice$.
\item \cblue{If $b=1$, $\ch$ samples $u^\bob,u^\charlie\xleftarrow{\$}\{0,1\}^n$ and computes $y^\bob=u^\bob\oplus m^\bob_0\oplus m^\bob_1$ and $y^\charlie=u^\charlie\oplus m^\charlie_0\oplus m^\charlie_1$, else if $b=0$, $\ch$ generates $y^\bob\gets \prf.\eval(k,x^\bob)$, $y^\charlie\gets \prf.\eval(k,x^\charlie)$}.
\item $\alice$ produces a bipartite state $\sigma_{\bob,\charlie}$.
\item $\ch$ computes $\ct^\bob=(x^\bob,z^\bob)$ and $\ct^\charlie=(x^\charlie,z^\charlie)$ where \sout{$z^\bob=y^\bob\oplus m^\bob_b$ and $z^\charlie=y^\charlie\oplus m^\charlie_b$} \cblue{$z^\bob=y^\bob\oplus m^\bob_0$ and $z^\charlie=y^\charlie\oplus m^\charlie_0$ if $b=0$, and $z^\bob=y^\bob\oplus m^\bob_1$ and $z^\charlie=y^\charlie\oplus m^\charlie_1$ if $b=1$}.
\item Apply $(\bob(\ct^{\bob},\cdot) \otimes \charlie(\ct^{\charlie},\cdot))(\sigma_{\bob,\charlie})$ to obtain $(b^{\bob},b^{\charlie})$. 
\item Output $1$ if $b^\bob=b^\charlie=b$.
\end{itemize}
The indistinguishability between $\hybrid_5$ and $\hybrid_6$ since we did not change the distribution on $y^\bob$, $y^\charlie$ in both the cases $b=0$ and $b=1$, and hence we did not change the distribution on $z^\bob$, $z^\charlie$  in both the $b=0$ and the $b=1$ cases across the hybrids $\hybrid_5$ and $\hybrid_6$.

\noindent $\hybrid_7$: \\
\begin{itemize}
\item $\alice$ sends two same-length message pairs $(m^\bob_0,m^\bob_1, m^\charlie_0,m^\charlie_1)$.
\item $\ch$ samples $k\gets \prf.\gen(1^{\secparam})$.
\item $\ch$ samples $b\xleftarrow{\$} \{0,1\}$. 
\item $\ch$ samples $x^\bob,x^\charlie\xleftarrow{\$}\{0,1\}^n$.
\item $\ch$ generates $k_{x^\bob,x^\charlie}\gets \prf.\puncture(k,\{x^\bob,x^\charlie\})$.
\item $\ch$ generates the circuit $\tilde{C}\gets \io(C)$ where $C$ is constructed depending on the bit $b$ as follows. If $b=0$ (respectively, $b=1$), $C$ has $k$ (respectively, $k_{x^\bob,x^\charlie}$) hardcoded and on input $r\gets \{0,1\}^{\frac{n}{2}}$ (the input space of $\prg$) and a message $m\in \{0,1\}^n$, outputs $(\prg(r),\prf.\eval(k,\prg(r))\oplus m)$ (respectively, $(\prg(r),\prf.\eval(k_{x^\bob,x^\charlie},\prg(r))\oplus m)$).
\item If $b=0$, $\ch$ generates $\rho_\sk\gets \UPO.\obf(1^\secparam,\tilde{F})$ where $\tilde{F}\gets \io(\prf.\eval(k,\cdot))$, else, if $b=1$, generates  $\rho_\sk\gets \UPO.\obf(1^\secparam,\tilde{W})$, where $\tilde{W}\gets \io(W)$ and $W$ is \cblue{as depicted in \Cref{fig:W-hybrid-7}}  and sends $(\rho_\sk,\io(C))$ to $\alice$.
\item If $b=1$, $\ch$ samples $u^\bob,u^\charlie\xleftarrow{\$}\{0,1\}^n$ \sout{and computes $y^\bob=u^\bob\oplus m^\bob_0\oplus m^\bob_1$ and $y^\charlie=u^\charlie\oplus m^\charlie_0\oplus m^\charlie_1$}, else if $b=0$, $\ch$ generates $y^\bob\gets \prf.\eval(k,x^\bob)$, $y^\charlie\gets \prf.\eval(k,x^\charlie)$.
\item $\alice$ produces a bipartite state $\sigma_{\bob,\charlie}$.

\item $\ch$ computes $\ct^\bob=(x^\bob,z^\bob)$ and $\ct^\charlie=(x^\charlie,z^\charlie)$ where $z^\bob=y^\bob\oplus m^\bob_0$ and $z^\charlie=y^\charlie\oplus m^\charlie_0$ if $b=0$, and \sout{$z^\bob=y^\bob\oplus m^\bob_1$ and $z^\charlie=y^\charlie\oplus m^\charlie_1$} \cblue{$z^\bob=u^\bob\oplus m^\bob_0$ and $z^\charlie=u^\charlie\oplus m^\charlie_0$} if $b=1$.
\item Apply $(\bob(\ct^{\bob},\cdot) \otimes \charlie(\ct^{\charlie},\cdot))(\sigma_{\bob,\charlie})$ to obtain $(b^{\bob},b^{\charlie})$. 
\item Output $1$ if $b^\bob=b^\charlie=b$.
\end{itemize}

\begin{figure}[!htb]
   \begin{center} 
   \begin{tabular}{|p{12cm}|}
    \hline 
\begin{center}
\underline{$W$}: 
\end{center}
Hardcoded keys $k_{x^\bob,x^\charlie}$\sout{$y^\bob,y^\charlie$}\cblue{$,u^\bob\oplus m^\bob_0\oplus m^\bob_1,u^\charlie\oplus m^\charlie_0\oplus m^\charlie_1$}.
On input: $x$.
\begin{itemize}
\item If $x=x^\bob$, output \sout{$y^\bob$} \cblue{$u^\bob\oplus m^\bob_0\oplus m^\bob_1$}.
\item Else if, $x=x^\charlie$, output \sout{$y^\charlie$} \cblue{$u^\charlie\oplus m^\charlie_0\oplus m^\charlie_1$}.
\item Else, run $\prf.\eval(k_{x^\bob,x^\charlie},x)$ and output the result.
\end{itemize}
\ \\ 
\hline
\end{tabular}
    \caption{Circuit $W$ in $\hybrid_{7}$}
    \label{fig:W-hybrid-7}
    \end{center}
\end{figure}

The indistinguishability between $\hybrid_6$ and $\hybrid_7$ holds because, in $\hybrid_7$, we just rewrote $y^\bob$ and $y^\charlie$ wherever it appeared in the $b=1$  case of $\hybrid_6$ in terms of $u^\bob$ and $u^\charlie$, respectively.

\noindent $\hybrid_8$: \\
\begin{itemize}
\item $\alice$ sends two same-length message pairs $(m^\bob_0,m^\bob_1, m^\charlie_0,m^\charlie_1)$.
\item $\ch$ samples $k\gets \prf.\gen(1^{\secparam})$.
\item $\ch$ samples $b\xleftarrow{\$} \{0,1\}$. 
\item $\ch$ samples $x^\bob,x^\charlie\xleftarrow{\$}\{0,1\}^n$.
\item $\ch$ generates $k_{x^\bob,x^\charlie}\gets \prf.\puncture(k,\{x^\bob,x^\charlie\})$.
\item $\ch$ generates the circuit $\tilde{C}\gets \io(C)$ where $C$ is constructed depending on the bit $b$ as follows. If $b=0$ (respectively, $b=1$), $C$ has $k$ (respectively, $k_{x^\bob,x^\charlie}$) hardcoded and on input $r\gets \{0,1\}^{\frac{n}{2}}$ (the input space of $\prg$) and a message $m\in \{0,1\}^n$, outputs $(\prg(r),\prf.\eval(k,\prg(r))\oplus m)$ (respectively, $(\prg(r),\prf.\eval(k_{x^\bob,x^\charlie},\prg(r))\oplus m)$).
\item If $b=0$, $\ch$ generates $\rho_\sk\gets \UPO.\obf(1^\secparam,\tilde{F})$ where $\tilde{F}\gets \io(\prf.\eval(k,\cdot))$, else, if $b=1$, generates  $\rho_\sk\gets \UPO.\obf(1^\secparam,\tilde{W})$, where $\tilde{W}\gets \io(W)$ and $W$ is as depicted in \Cref{fig:W-hybrid-7}  and sends $(\rho_\sk,\io(C))$ to $\alice$.
\item If $b=1$, $\ch$ \sout{samples $u^\bob,u^\charlie\xleftarrow{\$}\{0,1\}^n$} \cblue{generates $u^\bob\gets \prf.\eval(k,x^\bob)$, $u^\charlie\gets \prf.\eval(k,x^\charlie)$}, else if $b=0$, $\ch$ generates $y^\bob\gets \prf.\eval(k,x^\bob)$, $y^\charlie\gets \prf.\eval(k,x^\charlie)$.
\item $\alice$ produces a bipartite state $\sigma_{\bob,\charlie}$.

\item $\ch$ computes $\ct^\bob=(x^\bob,z^\bob)$ and $\ct^\charlie=(x^\charlie,z^\charlie)$ where $z^\bob=y^\bob\oplus m^\bob_0$ and $z^\charlie=y^\charlie\oplus m^\charlie_0$ if $b=0$, and $z^\bob=u^\bob\oplus m^\bob_0$ and $z^\charlie=u^\charlie\oplus m^\charlie_0$ if $b=1$.
\item Apply $(\bob(\ct^{\bob},\cdot) \otimes \charlie(\ct^{\charlie},\cdot))(\sigma_{\bob,\charlie})$ to obtain $(b^{\bob},b^{\charlie})$. 
\item Output $1$ if $b^\bob=b^\charlie=b$.
\end{itemize}

Since the views of the adversary $(\alice,\bob,\charlie)$ in $b=1$ case in hybrids $\hybrid_7$ and $\hybrid_8$ are only dependent on $k_{x^\bob,x^\charlie}$, the indistinguishability between $\hybrid_7$ and $\hybrid_8$ holds by the puncturing security of $\prf$.

\noindent $\hybrid_9$: \\
\begin{itemize}
\item $\alice$ sends two same-length message pairs $(m^\bob_0,m^\bob_1, m^\charlie_0,m^\charlie_1)$.
\item $\ch$ samples $k\gets \prf.\gen(1^{\secparam})$.
\item $\ch$ samples $b\xleftarrow{\$} \{0,1\}$. 
\item $\ch$ samples $x^\bob,x^\charlie\xleftarrow{\$}\{0,1\}^n$.
\item $\ch$ generates $k_{x^\bob,x^\charlie}\gets \prf.\puncture(k,\{x^\bob,x^\charlie\})$.
\item $\ch$ generates the circuit $\tilde{C}\gets \io(C)$ where $C$ is constructed depending on the bit $b$ as follows. If $b=0$ (respectively, $b=1$), $C$ has $k$ (respectively, $k_{x^\bob,x^\charlie}$) hardcoded and on input $r\gets \{0,1\}^{\frac{n}{2}}$ (the input space of $\prg$) and a message $m\in \{0,1\}^n$, outputs $(\prg(r),\prf.\eval(k,\prg(r))\oplus m)$ (respectively, $(\prg(r),\prf.\eval(k_{x^\bob,x^\charlie},\prg(r))\oplus m)$).
\item If $b=0$, $\ch$ generates $\rho_\sk\gets \UPO.\obf(1^\secparam,\tilde{F})$ where $\tilde{F}\gets \io(\prf.\eval(k,\cdot))$, else, if $b=1$, \cblue{generates the circuits $\mu_{k,m^\bob_0\oplus m^\bob_1}$ and $\mu_{k,m^\charlie_0\oplus m^\charlie_1}$ which on any input $x$ output $\prf.\eval(k,x)\oplus m^\bob_0\oplus m^\bob_1$ and $\prf.\eval(k,x)\oplus m^\charlie_0\oplus m^\charlie_1$ respectively, and also} generates  $\rho_\sk\gets \UPO.\obf(1^\secparam,\tilde{W})$, where $\tilde{W}\gets \io(W)$ and $W$ is as \cblue{depicted in \Cref{fig:W-hybrid-9}}  and sends $(\rho_\sk,\tilde{C)}$ to $\alice$.
\item If $b=1$, $\ch$  {generates $u^\bob\gets \prf.\eval(k,x^\bob)$, $u^\charlie\gets \prf.\eval(k,x^\charlie)$}, else if $b=0$, $\ch$ generates $y^\bob\gets \prf.\eval(k,x^\bob)$, $y^\charlie\gets \prf.\eval(k,x^\charlie)$.
\item $\alice$ produces a bipartite state $\sigma_{\bob,\charlie}$.
\item $\ch$ computes $\ct^\bob=(x^\bob,z^\bob)$ and $\ct^\charlie=(x^\charlie,z^\charlie)$ where $z^\bob=y^\bob\oplus m^\bob_0$ and $z^\charlie=y^\charlie\oplus m^\charlie_0$ if $b=0$, and $z^\bob=u^\bob\oplus m^\bob_0$ and $z^\charlie=u^\charlie\oplus m^\charlie_0$ if $b=1$.
\item Apply $(\bob(\ct^{\bob},\cdot) \otimes \charlie(\ct^{\charlie},\cdot))(\sigma_{\bob,\charlie})$ to obtain $(b^{\bob},b^{\charlie})$. 
\item Output $1$ if $b^\bob=b^\charlie=b$.
\end{itemize}

\begin{figure}[!htb]
   \begin{center} 
   \begin{tabular}{|p{12cm}|}
    \hline 
\begin{center}
\underline{$W$}: 
\end{center}
Hardcoded keys $k_{x^\bob,x^\charlie}$\sout{$,u^\bob\oplus m^\bob_0\oplus m^\bob_1,u^\charlie\oplus m^\charlie_0\oplus m^\charlie_1$},\cblue{$\mu_{k,m^\bob_0\oplus m^\bob_1}$,$\mu_{k,m^\charlie_0\oplus m^\charlie_1}$}.
On input: $x$.
\begin{itemize}
\item If $x=x^\bob$, output \sout{$u^\bob\oplus m^\bob_0\oplus m^\bob_1$} \cblue{$\mu_{k,m^\bob_0\oplus m^\bob_1}(x)$}.
\item Else if, $x=x^\charlie$, output  \sout{$u^\charlie\oplus m^\charlie_0\oplus m^\charlie_1$} \cblue{$\mu_{k,m^\charlie_0\oplus m^\charlie_1}(x)$}.
\item Else, run $\prf.\eval(k_{x^\bob,x^\charlie},x)$ and output the result.
\end{itemize}
\ \\ 
\hline
\end{tabular}
    \caption{Circuit $W$ in $\hybrid_{9}$}
    \label{fig:W-hybrid-9}
    \end{center}
\end{figure}

The functionality of $W$ did not change due to the changes made across hybrids $\hybrid_8$ and $\hybrid_9$, and hence by $\io$ guarantees, the indistinguishability between $\hybrid_8$ and $\hybrid_9$ holds.

\noindent $\hybrid_{10}$: \\
\begin{itemize}
\item $\alice$ sends two same-length message pairs $(m^\bob_0,m^\bob_1, m^\charlie_0,m^\charlie_1)$.
\item $\ch$ samples $k\gets \prf.\gen(1^{\secparam})$.
\item $\ch$ samples $b\xleftarrow{\$} \{0,1\}$. 
\item $\ch$ samples $x^\bob,x^\charlie\xleftarrow{\$}\{0,1\}^n$.
\item $\ch$ generates $k_{x^\bob,x^\charlie}\gets \prf.\puncture(k,\{x^\bob,x^\charlie\})$.
\item $\ch$ generates the circuit $\tilde{C}\gets \io(C)$ \sout{where $C$ is constructed depending on the bit $b$ as follows. If $b=0$ (respectively, $b=1$), $C$ has $k$ (respectively, $k_{x^\bob,x^\charlie}$) hardcoded and on input $r\gets \{0,1\}^{\frac{n}{2}}$ (the input space of $\prg$) and a message $m\in \{0,1\}^n$, outputs $(\prg(r)$, $\prf.\eval(k,\prg(r))\oplus m)$ (respectively, $(\prg(r),\allowbreak \prf.\eval(k_{x^\bob,x^\charlie},\prg(r))\oplus m)$).} \cblue{where $C$ has $k$
hardcoded and on input $r\gets \{0,1\}^{\frac{n}{2}}$
 (the input space of $\prg$) and a message $m\in \{0,1\}^n$, outputs $(\prg(r),\allowbreak \prf.\eval(k,\prg(r))\oplus m)$.}
\item If $b=0$, $\ch$ generates $\rho_\sk\gets \UPO.\obf(1^\secparam,\tilde{F})$ where $\tilde{F}\gets \io(\prf.\eval(k,\cdot))$, else, if $b=1$, generates the circuits $\mu_{k,m^\bob_0\oplus m^\bob_1}$ and $\mu_{k,m^\charlie_0\oplus m^\charlie_1}$ which on any input $x$ output $\prf.\eval(k,x)\oplus m^\bob_0\oplus m^\bob_1$ and $\prf.\eval(k,x)\oplus m^\charlie_0\oplus m^\charlie_1$ respectively, and also generates  $\rho_\sk\gets \UPO.\obf(1^\secparam,\tilde{W})$, where $\tilde{W}\gets \io(W)$ and $W$ is as depicted in \Cref{fig:W-hybrid-9}  and sends $(\rho_\sk,\tilde{C})$ to $\alice$.
\item \sout{If $b=1$, $\ch$ generates $u^\bob\gets \prf.\eval(k,x^\bob), u^\charlie\gets \prf.\eval(k,x^\charlie)$, else if $b=0$, $\ch$ generates $y^\bob\gets \prf.\eval(k,x^\bob)$, $y^\charlie\gets \prf.\eval(k,x^\charlie)$.} \cblue{$\ch$ generates $u^\bob\gets \prf.\eval(k,x^\bob), u^\charlie\gets \prf.\eval(k,x^\charlie)$.}
\item $\alice$ produces a bipartite state $\sigma_{\bob,\charlie}$.
\item \sout{$\ch$ computes $\ct^\bob=(x^\bob,z^\bob)$ and $\ct^\charlie=(x^\charlie,z^\charlie)$ where $z^\bob=y^\bob\oplus m^\bob_0$ and $z^\charlie=y^\charlie\oplus m^\charlie_0$ if $b=0$, and $z^\bob=u^\bob\oplus m^\bob_0$ and $z^\charlie=u^\charlie\oplus m^\charlie_0$ if $b=1$.} \cblue{$\ch$ computes $\ct^\bob=(x^\bob,z^\bob)$ and $\ct^\charlie=(x^\charlie,z^\charlie)$ where $z^\bob=u^\bob\oplus m^\bob_0$ and $z^\charlie=u^\charlie\oplus m^\charlie_0$.}
\item Apply $(\bob(\ct^{\bob},\cdot) \otimes \charlie(\ct^{\charlie},\cdot))(\sigma_{\bob,\charlie})$ to obtain $(b^{\bob},b^{\charlie})$. 
\item Output $1$ if $b^\bob=b^\charlie=b$.
\end{itemize}

Note that $y^\bob$ and $y^\charlie$ are defined only in the $b=0$ case, and $u^\bob$ and $u^\charlie$ are defined only in the $b=1$ case in $\hybrid_9$. However, replacing $y^\bob$, $y^\charlie$ in the $b=0$ by $u^\bob$, $u^\charlie$ (as defined in $b=1$ case) does not change the global distribution of the experiment in $b=0$ case. 
 Therefore, replacing $y^\bob,y^\charlie$ in $b=0$ with $u^\bob,u^\charlie$ (as defined in the $b=1$ case) in $\hybrid_9$, does not change the security experiment and hence, $\hybrid_9$ and $\hybrid_{10}$ have the same success probability. 


Finally, we give a reduction from $\hybrid_{10}$ to the generalized $\upo$ security experiment (see \cref{fig:genupo:expt}) of $\UPO'$ for  $\cktclass=\{\cktclass_\secparam\}$, where $\cktclass_\secparam=\{\prf.\eval(k,\cdot)\}_{k\in {\sf Supp}(\prf.\gen(1^\secparam))}$ with respect to the puncture algorithm $\genpuncture$ defined at the begining of the proof.

Let $(\alice,\bob,\charlie)$ be an adversary in $\hybrid_{10}$ above. Consider the following non-local adversary $(\reduc_\alice,\reduc_\bob,\reduc_\charlie)$:
\begin{itemize}
    \item $\reduc_\alice$ gets a pair of messages $m^\bob_0,m^\bob_1,m^\charlie_0,m^\charlie_1\gets\alice(1^\secparam)$ and samples a key $k\gets \prf.\gen(1^\secparam)$ and constructs the  circuits $\mu_{k,m^\bob_0\oplus m^\bob_1}$ and $\mu_{k,m^\charlie_0\oplus m^\charlie_1}$ which on any input $x$ outputs $\prf.\eval(k,x)\oplus m^\bob_0\oplus m^\bob_1$ and $\prf.\eval(k,x)\oplus m^\charlie_0\oplus m^\charlie_1$ respectively, and sends $k,\mu_\bob,\mu_\charlie$ to $\ch$ where $\mu_\bob=\mu_{k,m^\bob_0\oplus m^\bob_1}$ and $\mu_\charlie=\mu_{k,m^\charlie_0\oplus m^\charlie_1}$.
    \item $\reduc_\alice$ also constructs the circuit $\tilde{C}\gets \io(C)$ where $C$ has $k$ hardcoded and on input $r\gets \{0,1\}^{\frac{n}{2}}$ (the input space of $\prg$) and a message $m\in \{0,1\}^n$, outputs $(\prg(r),\prf.\eval(k,\prg(r))\oplus m)$.
    \item On getting $\rho$ from $\ch$, $\reduc_\alice $ feeds $\rho,\tilde{C}$ to $\alice$ and gets back a state $\sigma_{\bob,\charlie}$. $\reduc_\alice$ then sends the respective registers of $\sigma_{\bob,\charlie}$ to $\reduc_\bob$ and $\reduc_\alice$, along with the key $k$.
    \item $\reduc_\bob$ (respectively, $\reduc_\charlie$) on receiving $(\sigma_\bob,k)$ (respectively, $(\sigma_\charlie,k)$) from $\reduc_\alice$ and $x^\bob$ (respectively, $x^\charlie$) from $\ch$ computes $y^\bob\gets\prf.\eval(k,x^\bob)$ (respectively, $y^\charlie\gets \prf.\eval(k,x^\charlie)$) and $\ct^\bob=(x^\bob,y^\bob\oplus m^\bob_0)$ (respectively, $\ct^\charlie=(x^\charlie,y^\charlie\oplus m^\charlie_0)$) and runs $\bob$ on $\ct^\bob$ (respectively, $\charlie$ on $\ct^\charlie$) to get a bit $b^\bob$ (respectively, $b^\charlie$), and outputs $b^\bob$ (respectively, $b^\charlie$).
\end{itemize}

\end{proof}

\begin{remark}\label{rem:selective-cpa-from-id-upo}
    If we change the $\UPO$ security guarantee of the underlying $\UPO$ scheme from  $\distrprod$-generalized $\UPO$ security to $\distrid$-generalized $\UPO$ security (see \Cref{subsec:upo-definition}), then using the same proof as in \cref{prop:cpa-style-anti-piracy-sde} upto minor corrections, we achieve $\distriden$-selective $\cpa$ anti-piracy instead of $\distrcor$-selective $\cpa$ anti-piracy as in \Cref{prop:cpa-style-anti-piracy-sde} for the $\SDE$ scheme given in \Cref{fig:SDE-construction}.
\end{remark}

\begin{theorem}[$\SDE$ lifting theorem]\label{thm:sde-lift}
    Assuming post-quantum indistinguishability obfuscation for classical circuits and length-doubling injective pseudorandom generators, there is a generic lift that takes a $\distrcor$-selective $\cpa$ secure $\SDE$ scheme and outputs a new $\SDE$ that is full-blown $\distrcor$-$\cpa$ secure (see \Cref{subsec:sde}).
\end{theorem}

\begin{proof}
    Let $(\gen,\qkeygen,\enc,\dec)$ be a selectively $\cpa$ secure $\SDE$, and let $\io$ be an indistinguishability obfuscation. Consider the $\SDE$ scheme $(\gen',\qkeygen',\enc',\dec')$ given in \Cref{fig:CPA-SDE-construction}.
    \begin{figure}[!htb]
   \begin{center} 
   \begin{tabular}{|p{12cm}|}
    \hline 
\noindent\textbf{Assumes:} $\SDE$ scheme $(\gen,\qkeygen,\enc,\dec)$, post-quantum indistinguishability obfuscation $\io$.\\
\ \\
\noindent$\gen'(1^\secparam)$: Same as $\gen()$.\\
\ \\
\noindent$\qkeygen'(\sk)$: Same as $\qkeygen()$.\\
\ \\

\noindent$\enc'(\pk,m)$: 
\begin{compactenum}
    \item Sample $r\xleftarrow{\$}\{0,1\}^n$.
    \item Generate $c=\enc(\pk,r)$.
    \item Output $\ct=(\tilde{C},c)$, where $\tilde{C}\gets \io(C)$ and $C$ is the circuit that on input $r$ outputs $m$ and outputs $\bot$ on all other inputs.
\end{compactenum}
\ \\
\noindent$\dec'(\rho_\sk,\ct)$
\begin{compactenum}
    \item Interprete $\ct=\tilde{C},c$.
    \item Run $r\gets \dec(\rho_\sk,c)$.
    \item Output $m=\tilde{C}(r)$.
\end{compactenum}

\ \\ 
\hline
\end{tabular}
    \caption{A construction of $\cpa$-secure $\sde$ from a selectively $\cpa$-secure $\sde$.}
    \label{fig:CPA-SDE-construction}
    \end{center}
\end{figure}

\noindent The correctness of $(\gen',\qkeygen',\enc',\dec')$ follows directly from the correctness of $(\gen,\qkeygen,\allowbreak \enc,\allowbreak \dec)$.

\paragraph{$\cpa$ anti-piracy of $(\gen',\qkeygen',\enc',\dec')$ from selective security of $(\gen,\qkeygen,\enc,\dec)$.}

Let $(\alice,\bob,\charlie)$ be an adversary against the full-blown $\cpa$ security experiment for the $\cpasdeexpt^{\left(\alice,\bob,\charlie \right)}\left( 1^{\secparam} \right)$ (see \Cref{fig:correlated-sde-full-blown-cpa-style-anti-piracy}). We will do a sequence of hybrids starting from the original anti-piracy experiment $\cpasdeexpt^{\left(\alice,\bob,\charlie \right)}\left( 1^{\secparam} \right)$ for the $\sde$ scheme given in \Cref{fig:CPA-SDE-construction}, and then conclude with a reduction from the final to.
The changes are marked in \cblue{blue}.

\noindent $\hybrid_0$: \\
Same as $\cpasdeexpt^{\left(\alice,\bob,\charlie \right)}\left( 1^{\secparam} \right)$ given in~\Cref{fig:correlated-sde-full-blown-cpa-style-anti-piracy} for the $\sde$ scheme in \Cref{fig:CPA-SDE-construction}.
 \begin{itemize} 
\item $\ch$ samples $\sk,\pk \gets \gen(1^{\secparam})$ and generates $\rho_\sk\gets \qkeygen(\sk)$ and sends $(\rho_\sk,\pk)$ to $\alice$.
\item $\alice$ sends two same-length message pairs $(m^\bob_0,m^\bob_1, m^\charlie_0,m^\charlie_1)$.
\item $\alice$ produces a bipartite state $\sigma_{\bob,\charlie}$.
\item $\ch$ samples $b\xleftarrow{\$} \{0,1\}$ and generates $c^\bob\gets \enc(\pk,m^\bob_b)$ and $c^\charlie\gets \enc(\pk,m^\charlie_b)$.
\item $\alice$ samples $r^\bob_0,r^\bob_1,r^\charlie_0,r^\charlie_1\xleftarrow{\$}\{0,1\}^n$ and computes $\ct^\bob=(\io(C^\bob),c^\bob)$ and $\ct^\charlie=(\io(C^\charlie),c^\charlie)$, where $C^\bob$ and $C^\charlie$ are the circuits that on input $r^\bob_b$ and $r^\charlie_b$ respectively, outputs $m^\bob_b$ and $m^\charlie_b$, respectively. $C^\bob$ and $C^\charlie$ on all inputs except $r^\bob_b$ and $r^\charlie_b$ respectively,  outputs $\bot$.
\item Apply $(\bob(\ct^{\bob},\cdot) \otimes \charlie(\ct^{\charlie},\cdot))(\sigma_{\bob,\charlie})$ to obtain $(b^{\bob},b^{\charlie})$. 
\item Output $1$ if $b^\bob=b^\charlie=b$.
 \end{itemize}  

\noindent $\hybrid_1$: \\
\begin{itemize} 
\item $\ch$ samples $\sk,\pk \gets \gen(1^{\secparam})$ and generates $\rho_\sk\gets \qkeygen(\sk)$ and sends $(\rho_\sk,\pk)$ to $\alice$.
\item $\alice$ sends two same-length message pairs $(m^\bob_0,m^\bob_1, m^\charlie_0,m^\charlie_1)$.
\item $\alice$ produces a bipartite state $\sigma_{\bob,\charlie}$.
\item $\ch$ samples $b\xleftarrow{\$} \{0,1\}$ and generates $c^\bob\gets \enc(\pk,m^\bob_b)$ and $c^\charlie\gets \enc(\pk,m^\charlie_b)$.
\item $\alice$ samples $r^\bob_0,r^\bob_1,r^\charlie_0,r^\charlie_1\xleftarrow{\$}\{0,1\}^n$, \cblue{$y^\bob,y^\charlie\xleftarrow{\$}\{0,1\}^n$,} and computes $\ct^\bob=(\io(C^\bob),c^\bob)$ and $\ct^\charlie=(\io(C^\charlie),c^\charlie)$, where $C^\bob$ and $C^\charlie$ are the circuits \sout{that on input $r^\bob$ and $r^\charlie$ respectively, outputs $m^\bob$ and $m^\charlie$, respectively. $C^\bob$ and $C^\charlie$ on all inputs except $r^\bob$ and $r^\charlie$ respectively,  outputs $\bot$.} \cblue{are as depicted in \Cref{fig:C-bob-hybrid-1-sde,fig:C-charlie-hybrid-1-sde}, respectively.}
\item Apply $(\bob(\ct^{\bob},\cdot) \otimes \charlie(\ct^{\charlie},\cdot))(\sigma_{\bob,\charlie})$ to obtain $(b^{\bob},b^{\charlie})$. 
\item Output $1$ if $b^\bob=b^\charlie=b$.
\end{itemize}   

The indistinguishability between hybrids $\hybrid_0$ and $\hybrid_1$ holds because of the following. Since the $\prg$ is a length-doubling, except with negligible probability, the functionality of circuits $C^\bob$ and $C^\charlie$ did not change across the hybrids $\hybrid_0$ and $\hybrid_1$. Therefore the computational indistinguishability between $\hybrid_0$ and $\hybrid_1$ follows from the security guarantees of $\io$.

\begin{figure}[!htb]
   \begin{center} 
   \begin{tabular}{|p{12cm}|}
    \hline 
\begin{center}
\underline{$C^\bob$}: 
\end{center}
Hardcoded keys $r^\bob_b,m^\bob_b$\cblue{$,m^\bob_{1-b},y^\bob$}.
On input: $r$.
\begin{itemize}
\item If $r=r^\bob_b$, output $m^\bob_b$.
\item \cblue{If $\prg(r)=y^\bob$, output $m^\bob_{1-b}$.}
\item Otherwise, output $\bot$.
\end{itemize}
\ \\ 
\hline
\end{tabular}
    \caption{Circuit $C^\bob$ in $\hybrid_{1}$}
    \label{fig:C-bob-hybrid-1-sde}
    \end{center}
\end{figure}

\begin{figure}[!htb]
   \begin{center} 
   \begin{tabular}{|p{12cm}|}
    \hline 
\begin{center}
\underline{$C^\charlie$}: 
\end{center}
Hardcoded keys $r^\charlie_b,m^\charlie_b$\cblue{$,m^\bob_{1-b},y^\bob$}.
On input: $r$.
\begin{itemize}
\item If $r=r^\charlie_b$, output $m^\charlie_b$.
\item \cblue{If $\prg(r)=y^\charlie$, output $m^\charlie_{1-b}$.}
\item Otherwise, output $\bot$.
\end{itemize}
\ \\ 
\hline
\end{tabular}
    \caption{Circuit $C^\charlie$ in $\hybrid_{1}$}
    \label{fig:C-charlie-hybrid-1-sde}
    \end{center}
\end{figure}

\noindent $\hybrid_2$: \\
\begin{itemize} 
\item $\ch$ samples $\sk,\pk \gets \gen(1^{\secparam})$ and generates $\rho_\sk\gets \qkeygen(\sk)$ and sends $(\rho_\sk,\pk)$ to $\alice$.
\item $\alice$ sends two same-length message pairs $(m^\bob_0,m^\bob_1, m^\charlie_0,m^\charlie_1)$.
\item $\alice$ produces a bipartite state $\sigma_{\bob,\charlie}$.
\item $\ch$ samples $b\xleftarrow{\$} \{0,1\}$ and generates $c^\bob\gets \enc(\pk,m^\bob_b)$ and $c^\charlie\gets \enc(\pk,m^\charlie_b)$.
\item $\alice$ samples $r^\bob,r^\charlie\xleftarrow{\$}\{0,1\}^n$, \sout{$y^\bob,y^\charlie\xleftarrow{\$}\{0,1\}^n$,} \cblue{$y^\bob\gets \prg(r^\bob_{1-b})$, $y^\charlie\gets \prg(r^\charlie{1-b})$} and computes $\ct^\bob=(\io(C^\bob),c^\bob)$ and $\ct^\charlie=(\io(C^\charlie),c^\charlie)$, where $C^\bob$ and $C^\charlie$ are the circuits are as depicted in \Cref{fig:C-bob-hybrid-1-sde,fig:C-charlie-hybrid-1-sde}, respectively.
\item Apply $(\bob(\ct^{\bob},\cdot) \otimes \charlie(\ct^{\charlie},\cdot))(\sigma_{\bob,\charlie})$ to obtain $(b^{\bob},b^{\charlie})$. 
\item Output $1$ if $b^\bob=b^\charlie=b$.
\end{itemize}   

The indistinguishability between $\hybrid_1$ and $\hybrid_2$ holds due to pseudorandomness of $\prg$.

\noindent $\hybrid_3$: \\
\begin{itemize} 
\item $\ch$ samples $\sk,\pk \gets \gen(1^{\secparam})$ and generates $\rho_\sk\gets \qkeygen(\sk)$ and sends $(\rho_\sk,\pk)$ to $\alice$.
\item $\alice$ sends two same-length message pairs $(m^\bob_0,m^\bob_1, m^\charlie_0,m^\charlie_1)$.
\item $\alice$ produces a bipartite state $\sigma_{\bob,\charlie}$.
\item $\ch$ samples $b\xleftarrow{\$} \{0,1\}$ and generates $c^\bob\gets \enc(\pk,m^\bob_b)$ and $c^\charlie\gets \enc(\pk,m^\charlie_b)$.
\item $\alice$ samples $r^\bob,r^\charlie\xleftarrow{\$}\{0,1\}^n$, \sout{$y^\bob\gets \prg(r^\bob_{1-b})$, $y^\charlie\gets \prg(r^\charlie{1-b})$} and computes $\ct^\bob=(\io(C^\bob),c^\bob)$ and $\ct^\charlie=(\io(C^\charlie),c^\charlie)$, where $C^\bob$ and $C^\charlie$ are the circuits are as depicted in \cblue{\Cref{fig:C-bob-hybrid-3-sde,fig:C-charlie-hybrid-3-sde}, respectively}.
\item Apply $(\bob(\ct^{\bob},\cdot) \otimes \charlie(\ct^{\charlie},\cdot))(\sigma_{\bob,\charlie})$ to obtain $(b^{\bob},b^{\charlie})$. 
\item Output $1$ if $b^\bob=b^\charlie=b$.
\end{itemize}   
The indistinguishability between $\hybrid_2$ and $\hybrid_3$ holds immediately by the $\io$ guarantees since we did not change the functionality of $C^\bob$ and $C^\charlie$ across the hybrids $\hybrid_2$ and $\hybrid_3$.

\begin{figure}[!htb]
   \begin{center} 
   \begin{tabular}{|p{12cm}|}
    \hline 
\begin{center}
\underline{$C^\bob$}: 
\end{center}
Hardcoded keys $r^\bob_b,m^\bob_0,m^\bob_1$,\sout{$y^\bob$},\cblue{$r^\bob_{1-b}$}.
On input: $r$.
\begin{itemize}
\item If $r=r^\bob_b$, output $m^\bob_b$.
\item \sout{If $\prg(r)=y^\bob$, output $m^\bob_{1-b}$.} \cblue{If $r=r^\bob_{1-b}$, output $m^\bob_{1-b}$.}
\item Otherwise, output $\bot$.
\end{itemize}
\ \\ 
\hline
\end{tabular}
    \caption{Circuit $C^\bob$ in $\hybrid_{3}$}
    \label{fig:C-bob-hybrid-3-sde}
    \end{center}
\end{figure}

\begin{figure}[!htb]
   \begin{center} 
   \begin{tabular}{|p{12cm}|}
    \hline 
\begin{center}
\underline{$C^\charlie$}: 
\end{center}
Hardcoded keys $r^\charlie_b,m^\bob,m^\charlie$,\sout{$y^\charlie$}\cblue{$r^\charlie_{1-b}$}.
On input: $r$.
\begin{itemize}
\item If $r=r^\charlie_b$, output $m^\charlie_b$.
\item \sout{If $\prg(r)=y^\charlie$, output $m^\charlie_{1-b}$.} \cblue{If $r=r^\charlie_{1-b}$, output $m^\charlie_{1-b}$.}
\item Otherwise, output $\bot$.
\end{itemize}
\ \\ 
\hline
\end{tabular}
    \caption{Circuit $C^\charlie$ in $\hybrid_{3}$}
    \label{fig:C-charlie-hybrid-3-sde}
    \end{center}
\end{figure}

\noindent Finally we give a reduction from $\hybrid_{3}$ to the selective-$\cpa$ anti-piracy game for $(\gen,\qkeygen,\allowbreak \enc,\allowbreak \dec)$ given in \Cref{fig:correlated-sde-cpa-style-anti-piracy}. Let $(\alice,\bob,\charlie)$ be an adversary in $\hybrid_{3}$ above. Consider the following non-local adversary $(\reduc_\alice,\reduc_\bob,\reduc_\charlie)$:

\begin{itemize}
    \item $\reduc_\alice$ samples $r^\bob_0,r^\bob_1,r^\charlie_0,r^\charlie_1\xleftarrow{\$}\{0,1\}^n$, and sends $(r^\bob_0,r^\bob_1)$  and $(r^\charlie_0,r^\charlie_1)$ as the challenge messages to $\ch$, the challenger for the selective-$\cpa$ anti-piracy game for $(\gen,\qkeygen,\enc,\dec)$ given in \Cref{fig:correlated-sde-cpa-style-anti-piracy}.
    \item $\reduc_\alice$ on receiving the decryptor and the public key $(\rho,\pk)$ from $\ch$ runs $\alice$ on $(\rho,\pk)$ to gets back the output, two pairs of messages $(m^\bob_0,m^\bob_1)$ and $(m^\charlie_0,m^\charlie_1)$ and a bipartite state $\sigma_{\bob,\charlie}$. 
    \item $\reduc_\alice$ constructs the circuit $\io(C^\bob)$ and $\io(C^\charlie)$ where  $C^\bob$ and $C^\charlie$ are the circuits are as depicted in \Cref{fig:C-bob-hybrid-1-sde,fig:C-charlie-hybrid-1-sde}, respectively.
    \item $\reduc_\alice$ sends $\io(C^\bob),\sigma_\bob$ to $\reduc_\bob$ and $\io(C^\charlie),\sigma_\charlie$ to $\reduc_\charlie$.
    \item $\reduc_\bob$ on receiving $c^\bob$ from $\ch$ and $(\io(C^\bob),\sigma_\bob)$ from $\reduc_\alice$, runs $b^\bob\gets\bob(\sigma_\bob,(\io(C^\bob),c^\bob))$ and outputs $b^\bob$.
    \item $\reduc_\charlie$ on receiving $c^\charlie$ from $\ch$ and $(\io(C^\charlie),\sigma_\charlie)$ from $\reduc_\alice$, runs $b^\charlie\gets\charlie(\sigma_\charlie,(\io(C^\charlie),c^\charlie))$ and outputs $b^\charlie$.
\end{itemize}
\begin{figure}[!htb]
   \begin{center} 
   \begin{tabular}{|p{12cm}|}
    \hline 
\begin{center}
\underline{$C^\bob$}: 
\end{center}
Hardcoded keys $r^\bob_b,m^\bob_0,m^\bob_1$,\sout{$y^\bob$},\cblue{$r^\bob_{1-b}$}.
On input: $r$.
\begin{itemize}
\item If $r=r^\bob_b$, output $m^\bob_b$.
\item If $r=r^\bob_{1-b}$, output $m^\bob_{1-b}$.
\item Otherwise, output $\bot$.
\end{itemize}
\ \\ 
\hline
\end{tabular}
    \caption{Circuit $C^\bob$}
    \label{fig:C-bob-reduc-sde}
    \end{center}
\end{figure}

\begin{figure}[!htb]
   \begin{center} 
   \begin{tabular}{|p{12cm}|}
    \hline 
\begin{center}
\underline{$C^\charlie$}: 
\end{center}
Hardcoded keys $r^\charlie_b,m^\bob,m^\charlie$,\sout{$y^\charlie$}\cblue{$r^\charlie_{1-b}$}.
On input: $r$.
\begin{itemize}
\item If $r=r^\charlie_b$, output $m^\charlie_b$.
\item If $r=r^\charlie_{1-b}$, output $m^\charlie_{1-b}$.
\item Otherwise, output $\bot$.
\end{itemize}
\ \\ 
\hline
\end{tabular}
    \caption{Circuit $C^\charlie$}
    \label{fig:C-charlie-reduc-sde}
    \end{center}
\end{figure}

\end{proof}

\begin{remark}\label{rem:sde-lift-distriden}
The proof of \cref{thm:sde-lift} can be adapted to prove the same construction lifts a $\SDE$ with $\distriden$-selective $\cpa$ anti-piracy to $\distriden$-$\cpa$ anti-piracy.
\end{remark}

\Cref{rem:selective-cpa-from-id-upo,rem:sde-lift-distriden} together gives us the following corollary.
\begin{corollary}\label{cor:UEnc-from-UPO}
Assuming an indistinguishability obfuscation scheme $\io$ for $\ppoly$, a puncturable pseudorandom function family $\prf=(\gen,\eval,\puncture)$, and a $\distrid$-generalized $\UPO$ scheme for any generalized puncturable keyed circuit class in $\ppoly$ 
 (see \Cref{subsec:upo-definition} for the formal definition of $\distrid$), there exists a secure public-key unclonable encryption for multiple bits (see \Cref{subsec:uenc} for the definition).
\end{corollary}

\begin{proof}
By \cite{GZ20}, a $\SDE$ scheme for multiple-bit messages satisfying $\distriden$-selective $\cpa$ anti-piracy, implies private-key unclonable encryption for multiple bits. Then the result of~\cite{AK21} shows that there exists a transformation from one-time unclonable encryption to public-key unclonable encryption assuming post-quantum secure public-key encryption, which in turn can be instantiated using $\iO$ and puncturable pseudorandom functions~\cite{SW14}.
\end{proof}

Combining \Cref{cor:UEnc-from-UPO} with \Cref{thm:strong-CLLZ-cp-prf_puncturable-CP-f-uniform-id}, we get the following feasibility result for unclonable encryptions from concrete assumptions.
\begin{corollary}\label{cor:UEnc-from-concrete assumption}
    Assuming \Cref{conj:goldreich-levin-identical},  the existence of post-quantum sub-exponentially secure $\io$ and one-way functions, and the quantum hardness of Learning-with-errors problem (LWE), there exists a secure public-key unclonable encryption for multiple bits (see \Cref{subsec:uenc} for the definition).
\end{corollary}

Similarly, combining \Cref{prop:cpa-style-anti-piracy-sde,thm:sde-lift,lemma:correctness-SDE-construction}, with \Cref{thm:strong-CLLZ-cp-prf_puncturable-CP-f-uniform}, we get the following feasibility result for $\sde$ from concrete assumptions.
\begin{corollary}\label{cor:sde-distrprod-from-concrete assumption}
    Assuming \Cref{conj:goldreich-levin-correlated},  the existence of post-quantum sub-exponentially secure $\io$ and one-way functions, and the quantum hardness of Learning-with-errors problem (LWE), there exists a $\distrcor$-$\cpa$ secure $\sde$ encryption scheme (see \Cref{subsec:sde} for the definition).
\end{corollary}

\subsection{Unclonable Encryption}
\label{sec:simple-construct-uenc-upo}
We next present a direct construction of unclonable secret key encryption for bits from $\UPO$. 
Let $\vec{0}_\secparam$ be the circuit denoting the all-zero function on input length $\secparam$. Similarly, for $x\in \{0,1\}^\secparam$, let $\vec{1_x}$ be the circuit implementing a point function with point $x$ and input length $\secparam$.
\begin{figure}[!htb]
   \begin{center} 
   \begin{tabular}{|p{12cm}|}
    \hline 
\noindent\textbf{Assumes:} $\UPO$, a $\upo$ for the $\distrid$-generalized puncturable keyed circuit class $\{\{\vec{0}_\secparam\}\}_\secparam$, with the trivial $\genpuncture$ algorithm and keyspace $\{\{0_\secparam\}\}_\secparam$ (since there is only one key or one circuit for a fixed input length).\\
\ \\
\noindent$\keygen(1^\secparam)$: Sample $k\xleftarrow{\$}\{0,1\}^\secparam$, and output $k$.
\noindent$\enc(k,b)$: 
\begin{compactenum}
    \item If $b=0$, construct the all-zero circuit $C=\vec{0}$ and if $b=1$, construct the circuit $C\gets \genpuncture({0},k,k,\vec{1},\vec{1})$. 
    \item Output $\rho\gets \UPO.\obf(C)$.
\end{compactenum}
\ \\
\noindent$\dec(k,\rho)$: Output $b'\gets \UPO.\eval(\rho,k)$.
\ \\ 
\hline
\end{tabular}
    \caption{A direct construction of unclonable encryption from $\UPO$s.}
    \label{fig:direct-construction-UEnc-UPO}
    \end{center}
\end{figure}

\begin{theorem}
    The unclonable encryption scheme in \Cref{fig:direct-construction-UEnc-UPO} satisfies correctness and unclonable indistinguishability security.
\end{theorem}
\begin{proof}
    The proof of correctness follows directly from the correctness of the underlying $\upo$, $\UPO$.
    For unclonable indistinguishable security, let $(\alice,\bob,\charlie)$ be an adversary in the unclonable indistinguishability security game. Next, we give the following reduction $(\reduc_\alice,\reduc_\bob,\reduc_\charlie)$ to $\distrid$-generalized $\upo$ security of $\UPO$, to complete the proof of security.
    \begin{enumerate}
        \item $\reduc_\alice$ sends the key $0$ and circuits $\vec{1},\vec{1}$ to challenger.
        \item $\reduc_\alice$ on receiving $\rho$ from $\ch$, runs $\alice$ on $\rho$, and gets as output a bipartite state $\sigma_{\bob,\charlie}$.
        \item $\reduc_\bob$ and $\reduc_\charlie$ are the same as $\bob$ and $\charlie$ respetively.
    \end{enumerate}
    Clearly, for every $b\in \{0,1\}$, the view of $(\alice,\bob,\charlie)$ when the challenge message is $b$ is the same as the view of $(\reduc_\alice,\reduc_\bob,\reduc_\charlie)$ when the challenge bit is $b$. Therefore, $(\reduc_\alice,\reduc_\bob,\reduc_\charlie)$ and $(\alice,\bob,\charlie)$ have the same advantage of winning in their respective indistinguishability game.
\end{proof}

%% file: evasive.tex
\newcommand{\evaspuncture}{\mathsf{Evasive}\text{-}\genpuncture}
\anote{Notation about circuit at the end, to be done }
\subsection{Copy-Protection for Evasive Functions} 
\label{sec:qcp:evasive}


We start by recalling the definition of evasive function classes.
\begin{definition}\label{definition:evasive}
     A class of keyed boolean-valued functions with input-length $n=n(\secparam)$ $\fclass=\{\fclass_\secparam\}_{\secparam \in \NN}$ is evasive with respect to an efficiently samplable distribution $\distr_\fclass$ on $\fclass$, if for every fixed input point $x$, there exists a negligible function $\negl()$ such that 
    \[\prob[f\gets \distr_\fclass(1^\secparam): f(x)=1]=\negl(\secparam).\]
\end{definition}
\paragraph{Challenges in constructing copy-protection for evasive functions:}
The copy-protection of evasive functions though similar in many ways have key syntactic differences with the $\UPO$ security experiment $\genupoexpt$ (see \Cref{fig:genupo:expt}). In particular, the objective of the adversary $(\alice,\bob,\charlie)$ in $\genupoexpt$ is to guess whether the obfuscated circuit given to $\alice$ is punctured or not at a point $x$ revealed later to $\bob$ and $\charlie$, which is the opposite of the syntax in the copy-protection experiment, where $\alice$ always gets the same copy-protected circuit for the function and $\bob$ and $\charlie$ get a challenge input $x$ and they need to guess the boolean output of the function on $x$. In order to construct copy-protection of evasive functions, we need to bridge this gap, for which we consider the following subclass of evasive functions.

\begin{definition}[{$\presamp$ evasive functions}]
\label{def:preimage:sampleable:evasive}
   
    An evasive function class $\fclass=\{\fclass_\secparam\}_{\secparam \in \NN}$ equipped with a distribution $\distr_\fclass$ on $\keyspace$ is a $\presamp$ evasive function class if 
    \begin{enumerate}
        \item There exists a keyed circuit implementation $(\distr,\cktclassf)$ of $(\distr_\fclass,\fclass)$  where $\cktclassf=\{C^\fclass_{k}\}_{k\in \keyspace}$.
        \item There exists an auxiliary generalized puncturable keyed circuit class $\cktclass=\{{C}_{k'}\}_{k'\in \keyspace'}$ with $\evaspuncture$ as the generalized puncturing algorithm, (see \Cref{sec:upo:gensecurity}), and equipped with an efficiently samplable distribution $\distr'$ on its keyspace $\keyspace'$, such that 
        \begin{equation}\label{eq:presamp-f-g}
        \{{C^\fclass}_{k},x\}_{k\gets \distr(1^\secparam), x \xleftarrow{\$} {C^\fclass_k}^{-1}(1)}\approx_c \{{C}_{k',y,\vec{1}},y\}_{{C}_{k',y,\vec{1}} \gets\evaspuncture(k',y,y,\vec{1},\vec{1}), k'\gets \distr'(1^\secparam), y \xleftarrow{\$}\{0,1\}^n},
        \end{equation}
        where $\vec{1}$ is the constant-$1$ function, and ${C}_{k',y,\vec{1}}$ is the same as the circuit ${C}_{k',y,y,\vec{1},\vec{1}}$.
    \end{enumerate}

    In short, we call $(\distr_\fclass,\fclass)$ $\presamp$ evasive if $\fclass$ equipped with a distribution $\distr_\fclass$ on $\keyspace$ is a $\presamp$ evasive function class.
    \label{def:presamp-evasive}
\end{definition}
\paragraph{Explanation and usefulness of \Cref{def:presamp-evasive}:} The $\presamp$ condition for an evasive function class $\fclass$ broadly means that there is an auxiliary circuit class $\cktclass$ such that sampling a uniformly random function $f$ from $\fclass$ represented as a circuit implementing it, along with a uniformly random preimage of $1$ under $f$ is indistinguishable from $(C_x,x)$ where $x$ is sampled uniformly at random and $C_x$ is generated by first sampling a uniformly random circuit from $\cktclass$ and then puncturing it at $x$. We will see in \Cref{thm:evasive-functions-from-upo} that the $\presamp$ condition  allows us to rewrite the copy-protection experiment of an evasive function family, as an unclonable experiment concerning the auxiliary circuit class but with a flipped syntax, which makes this new unclonable experiment compatible with the syntax of $\genupoexpt$, thus making it possible to construct copy-protection for $\presamp$ evasive functions.
\paragraph{Instantiations:} 
We show that a large class of single-bit output evasive function classes that includes point functions are $\presamp$ evasive. In particular, assuming post-quantum $\io$, we show that $\fclass^r$, the boolean-output function class consisting of functions with exactly $r$ preimages of $1$ are $\presamp$ evasive, where the auxiliary circuit class consists of the obfuscation of circuits that implement the function class $\fclass^{r-1}$ defined analogously. Formally, we show the following.
\begin{theorem}\label{thm:evasive-class-instantiations}
    For every $t\in [2^n]$, let $\fclass^t=\{\fclass^t_\secparam\}$ defined as $\fclass^t_\secparam=\{f:\{0,1\}^n \mapsto \{0,1\} \mid |f^{-1}(1)|=t\}$, i.e, the set of all functions $f$ on $n$-bit input and $1$-bit output with exactly $t$ preimages of $1$. Suppose for $r = \poly(\secparam)$, the following holds:
    \begin{enumerate}
        \item $\fclass^r$ is evasive with respect to $\distrunr$, the uniform distribution.
        \item For every $t \in \{r-1,r\}$\footnote{This requirement might look odd. The reason we need it is that we want to use the $\presamp$ condition (see \Cref{eq:presamp-f-g}) on $\cktclass^{r}$ with $\cktclass^{t-1}$ as the auxiliary circuit class.}, there exists a keyed circuit implementation $(\distr^t,\cktclass^t)$ for $(\mathcal{U}_{\fclass^t},\fclass^t)$. 
    \end{enumerate}
    Then, assuming post-quantum indistinguishability obfuscation, 
     $(\distrunr,\fclass^r)$ is $\presamp$ evasive.
\end{theorem}

\begin{proof}
Let $r\in o(2^n)$ as given in the theorem. Fix the circuit descriptions $\cktclass^r$ and $\cktclass^{r-1}$ for $\fclass^{r}$ and $\fclass^{r-1}$ respectively, as mentioned in the theorem.

Note that for every circuit $k\in \keyspace^{r-1}_\secparam$ and set of inputs $\{x_1,x_2\}$ and circuits $\{\mu_1,\mu_2\}$, there is an efficient procedure to construct the circuit $C_{k,x_1,x_2,\mu_1,\mu_2}$ which on any input $x'$ first checks if $x'=x_i$ for some $i\in [2]$ in which case it outputs $\mu_i(x_i)$, otherwise it outputs $C^{r-1}_k(x)$. We call this procedure $\genpuncture$. For $x_1=x_2=y$ and $\mu_1=\mu_2=\mu$, we will use $C_{k,y,\mu}$ as a shorthand notation for $C_{k,x_1,x_2,\mu_1,\mu_2}$.

We assume that for every $\secparam \in \NN$, and for every $k\in \keyspace^r_\secparam$, and for every $k'\in \keyspace^{r-1}_\secparam$, and $x_1,x_2 \in \{0,1\}^n$,  circuit $C^{r}_k\in \cktclass^r$, $C^{r-1}_{k'}\in \cktclass^{r-1}$, and a punctured circuit $C_{k',x_1,x_2,\mu_1,\mu_2}\gets \genpuncture(k',x_1,x_2,\allowbreak \mu_1,\allowbreak \mu_2)$ have the same size. 
These conditions can be achieved by padding sufficiently many zeroes to smaller circuits.

    Let $\io$ be a post-quantum indistinguishability obfuscation.

 Next, we make the following claim
\begin{claim}\label{eq:main}
    \[\{\io(C^{r}_k),x\}_{k\gets \distr^r(1^\secparam), x \xleftarrow{\$} {C^r_k}^{-1}(1)}\approx_c \{\io(C_{k',y,\vec{1}}),y\}_{k'\gets\distr^{r-1}(1^\secparam), y \xleftarrow{\$}\{0,1\}^n}.\]
\end{claim}
We first prove the theorem assuming \Cref{eq:main} as follows. Let $a(\secparam)$ be the amount of randomness $\io$ uses to obfuscate the circuits in $\cktclass^r$ and the punctured circuits obtained by puncturing circuits in $\cktclass^{r-1}$ using the $\genpuncture$ algorithm. 

Fix a security parameter $\secparam$ arbitrarily.  
    
Let $\widetilde{\cktclass}^r=\{\{\io(C^r_k;t)\}_{k\in \keyspace^r_\secparam,t\in \{0,1\}^{a(\secparam)}}\}_\secparam$ be a keyed circuit class with keyspace $\keyspace^r\times \{0,1\}^a$. Note that by the correctness of $\io$,  for every $k\in \keyspace^r_\secparam$, the circuit $\io(C^r_k;t)$ has the same functionality as $C^r_k$ for every $t\in \{0,1\}^{a(\secparam)}$, i.e, $S_\secparam(\io(C^r_k;t))=S_\secparam(C^r_k)$ where $S_\secparam$ is the canonical circuit-to-functionality map. Therefore, since $\cktclass^r$ is a keyed implementation $\fclass^r$, so is $\widetilde{\cktclass}^r$ (see \Cref{subsec:notation-applications} for the definition of keyed implementation).
Moreover, since $S_\secparam(\io(C^r_k;t))=S_\secparam(C^r_k)$,
it holds that 
\[\{S_\secparam(C^r_k)\}_{k\gets \distr^r(1^\secparam)}=\{S_\secparam(\io(C^r_k;t))\}_{k\gets \distr^r(1^\secparam), t\xleftarrow{\$}\{0,1\}^{a(\secparam)}}.\]
Therefore, since $(\distr^r,\cktclass^r)$ is a keyed implementation of $(\distrunr,\fclass^r)$, so is $(\distr,\widetilde{\cktclass}^r)$ where $\distr$ is defined as $(k,t)\gets \distr(1^\secparam)\equiv k\gets \distr^r(1^\secparam), t\xleftarrow{\$}\{0,1\}^{a(\secparam)}$ (see \Cref{subsec:notation-applications} for the definition of keyed implementation).

Similarly, $(\distr',\widetilde{\cktclass}^{r-1})$ is  a generalized circuit implementation of $(\distrunrminus,\fclass^{r-1})$ where $\distr'$ is defined as $(k,t)\gets \distr'(1^\secparam)\equiv k\gets \distr^{r-1}(1^\secparam), t\xleftarrow{\$}\{0,1\}^{a(\secparam)}$ and $\widetilde{\cktclass}^{r-1}=\{\{\io(C^{r-1}_k;t)\}_{k\in \keyspace^{r-1}_\secparam,t\in \{0,1\}^{a(\secparam)}}\}_\secparam$.

Let $\evaspuncture$ be an efficient algorithm that on input $k'\in \keyspace^{r-1}_\secparam$, $t'\in \{0,1\}^a$, a set of points $y_1,y_2$ 
and circuits $\mu_1,\mu_2$, generates $C_{k',y_1,y_2,\mu_1,\mu_2}$ and outputs the circuit $\io(C_{k',y_1,y_2,\mu_1,\mu_2};t')$.

Note that by definition of $\distr$,
\[\{\io(C^r_k;t),x\}_{(k,t)\gets\distr(1^\secparam) x \xleftarrow{\$} \{C^r_k\}^{-1}(1)}=\{\io(C^{r}_k),x\}_{k\gets \distr^r(1^\secparam), x \xleftarrow{\$} \{C^r_k\}^{-1}(1)},\]
which is the LHS of \Cref{eq:main},
and,
\begin{align*}
&\{\tilde{C}_{k',t',y',\vec{1}},y\}_{\tilde{C}_{k',t',y',\vec{1}} \gets\evaspuncture((k',t'),y,y,\vec{1},\vec{1}), (k',t')\gets\distr'(1^\secparam), y \xleftarrow{\$}\{0,1\}^n}\\
&=\{\io(C_{k',y,\vec{1}};t'),y\}_{C_{k',y,\vec{1}} \gets\genpuncture(k',y,y,\vec{1},\vec{1}), (k',t')\gets\distr'(1^\secparam), y \xleftarrow{\$}\{0,1\}^n} &\text{By definition of $\evaspuncture$}\\
&=\{\io(C_{k',y,\vec{1}}),y\}_{k'\gets\distr^{r-1}(1^\secparam), y \xleftarrow{\$}\{0,1\}^n}, &\text{By definition of $\distr'$}
\end{align*}
which is the RHS of \Cref{eq:main}.
Hence by \Cref{eq:main}, we conclude that,
\begin{align*}
&\{\io(C^r_k;t),x\}_{k,t\gets\distr(1^\secparam)), x \xleftarrow{\$} \{C^r_k\}^{-1}(1)}\\
&\approx_c \{\tilde{C}_{k',t',y',\vec{1}},y\}_{\tilde{C}_{k',t',y',\vec{1}} \gets\evaspuncture(k',y,y,\vec{1},\vec{1}), k'\gets\distr'(1^\secparam), t'\xleftarrow{\$}\{0,1\}^a, y \xleftarrow{\$}\{0,1\}^n},
\end{align*}
which is exactly the $\presamp$ condition for $\distrunr,\fclass^r$ with the keyed circuit implementation, $(\distr,\widetilde{\cktclass}^r)$,  
 the auxiliary generalized puncturable keyed circuit class $\widetilde{\cktclass}^{r-1}$ equipped with $\evaspuncture$, and $\distr'$ as the corresponding distribution on the keyspace of $\widetilde{\cktclass}^{r-1}$. 
 
 Next, we give a proof of \Cref{eq:main} to complete the proof.
\paragraph*{Proof of \Cref{eq:main}}

    Fix $\secparam$ arbitrarily.
    Since $\fclass^r$ is evasive, so is $\fclass^{r-1}$. Hence, $k'\gets\distr^{r-1}(1^\secparam)$, $y\xleftarrow{\$}\{0,1\}^n\approx_s y\xleftarrow{\$}{\left\{C^{r-1}_{k'}\right\}}^{-1}(0)$ and hence, 
    \[\{\io(C_{k',y,\vec{1}}),y\}_{k'\gets\distr^{r-1}(1^\secparam), y \xleftarrow{\$}\{0,1\}^n}\approx_s \{\io(C_{k,y,\vec{1}}),y\}_{k\xleftarrow{\$} \keyspace^{r-1}_\secparam, y \xleftarrow{\$}{\left\{C^{r-1}_{k'}\right\}}^{-1}(0)}.\] 
   Hence it is enough to show that
    \[\{\io(C^r_k),x\}_{k\gets\distr^r(1^\secparam), x \xleftarrow{\$} {C^{r}_{k}}^{-1}(1)}\approx_c \{\io(C_{k',y,\vec{1}}),y\}_{k'\gets\distr^{r-1}(1^\secparam), y \xleftarrow{\$}{\left\{C^{r-1}_{k'}\right\}}^{-1}(0)}.\]
Recall the circuit-to-functionality map $S_\secparam$. 
Let $\inductdistr$ and $\inductdistrdash$ be the distribution that $\distr^r$ and $\distr^{r-1}$ respectively induces on $\fclass^r_\secparam$ and $\fclass^{r-1}_\secparam$ under $S_\secparam$. Since $(\distr^r,\cktclass^r)$ and $(\distr^{r-1},\cktclass^{r-1})$ are keyed implementation of $(\distrunr,\fclass^r)$ and $(\distrunrminus,\fclass^{r-1})$ respectively, it holds that,
\begin{equation}\label{eq:indisntinguishability-keyed-imp}
    \inductdistr\approx_s \distrunr,\text{ and similarly, } \inductdistrdash\approx_s \distrunrminus
\end{equation}

Since $\widetilde{\cktclass}^{r}$ and $\widetilde{\cktclass}^{r-1}$ are keyed implementations of $\fclass^r$ and $\fclass^{r-1}$ respectively, for every $f\in \fclass^r$ and $\gclass^{r-1}$ $\distr^r$ and $\distr^{r-1}$ induce distributions $\distrsf$ and $\distrsg$, on the class of circuits $S_\secparam^{-1}(f)$ and $S_\secparam^{-1}(g)$, respectively. For every $f\in \fclass^r, g\in \fclass^{r-1}$, let $k_f$ and ${k'}_g$ be the lexicographically first key in $\keyspace^r$ and $\keyspace^{r-1}$ such that $C^r_{k_f}\in S_\secparam^{-1}(f)$ and  $C^{r-1}_{{k'}_g}\in S_\secparam^{-1}(g)$.

Note that by the security of $\io$, for every $f\in \fclass^r$, and $C^r_k\in S_\secparam^{-1}(f)$
\[\{\io(C^r_{k};t)\}_{t\xleftarrow{\$}\{0,1\}^a}  \approx_c  \{\io(C^r_{k_f};t)\}_{t\xleftarrow{\$}\{0,1\}^a}.\]
Therefore it holds that,  for every $f\in \fclass^r$,

\begin{equation}\label{eq:io-sec}
\{\io(C^r_{k})\}_{k\gets \distrsf}=\{\io(C^r_{k};t)\}_{k\gets \distrsf, t\xleftarrow{\$}\{0,1\}^a}  \approx_c \{\io(C^r_{k_f};t)\}_{t\xleftarrow{\$}\{0,1\}^a}=\{\io(C^r_{k_f})\}.
\end{equation}

Next note that, 
\[\{\io(C^r_k),x\}_{k\gets\distr^r(1^\secparam), x \xleftarrow{\$} \{C^r_k\}^{-1}(1)}=\{\io(C^r_k),x\}_{k\gets\distrsf(1^\secparam), f\gets \inductdistr x \xleftarrow{\$} {C^r_{k}}^{-1}(1)}.\]
Therefore,
\begin{align*}
&\{\io(C^r_k),x\}_{k\gets\distr^r(1^\secparam), x \xleftarrow{\$} \{C^r_k\}^{-1}(1)}\\
&=\{\io(C^r_k),x\}_{k\gets\distrsf(1^\secparam), f\gets \inductdistr, x \xleftarrow{\$} \{\io(C^r_{k})\}^{-1}(1)}\\
&\approx_s \{\io(C^r_k),x\}_{k\gets\distrsf(1^\secparam), f\gets \distrunr, x \xleftarrow{\$} {C^r_{k}}^{-1}(1)}&\text{By \Cref{eq:indisntinguishability-keyed-imp}}\\
&\approx_c\{\io(C^r_{k_f};t),x\}_{t\xleftarrow{\$}\{0,1\}^a, f\gets \distrunr, x \xleftarrow{\$} \{\io(C^r_{k_f})\}^{-1}(1)}&\text{By \Cref{eq:io-sec}}\\
&=\{\io(C^r_{k_f};t),x\}_{t\xleftarrow{\$}\{0,1\}^a, f\gets \distrunr, x \xleftarrow{\$} f^{-1}(1)}.
\end{align*}
Similarly, it can be shown that
\[\{\io(C_{k',y,\vec{1}}),y\}_{k'\gets\distr^{r-1}(1^\secparam), y \xleftarrow{\$}{\left\{C^{r-1}_{k'}\right\}}^{-1}(0)} \approx_c \{\io(C_{{k'}_g,y,\vec{1}};t),x\}_{t\xleftarrow{\$}\{0,1\}^a, g\gets \distrunrminus, y \xleftarrow{\$} g^{-1}(0)}.\]
 Therefore to conclude \Cref{eq:main}, it is enough to prove that
\[\{\io(C^r_{k_f};t),x\}_{t\xleftarrow{\$}\{0,1\}^a, f\gets \distrunr, x \xleftarrow{\$} f^{-1}(1)} \approx_c \{\io(C_{{k'}_g,y,\vec{1}};t),x\}_{t\xleftarrow{\$}\{0,1\}^a, g\gets \distrunrminus, y \xleftarrow{\$} g^{-1}(0)}.\]
 
    This is the same as proving the following claim:
\begin{claim}\label{eq:interplay-r_r-1}    
    \[\{\io(C^r_{k_f}),x\}_{(f,x)\xleftarrow{\$} \mathrm{F}^{0,r}_\secparam}\approx_c \{\io(C_{{k'}_g,y,\vec{1}}),y\}_{(g,y)\xleftarrow{\$} \mathrm{F}^{1,r-1}_\secparam},\]
    where $\mathrm{F}^{v,b}_\secparam=\{(f,z)\mid f\in \fclass^v_\secparam, f(z)=b\}$, for every $v\in \NN$, $b\in \{0,1\}$, $s\in \keyspace^t_\secparam$.
\end{claim}
\paragraph*{Proof of \Cref{eq:interplay-r_r-1}}    
    Note that for every fixed pair $(f^*,x^*)\in \mathrm{F}^{r,b}_\secparam$, there exists a unique $(\tilde{g},\tilde{y})\in  \mathrm{F}^{r-1,0}_\secparam$, and vice versa, such that $C_{{k'}_{\tilde{g}},\tilde{y},\vec{1}}$ has the same functionality as $C^r_{k_f}$ and $\tilde{y}=x^*$. In other words, there  is a bijection $\mathcal{B}: \mathrm{F}^{r,1}_\secparam\mapsto\mathrm{F}^{r-1,0}_\secparam$ mapping $(f^*,x^*)$ to $(\tilde{g},\tilde{y})$ such that  $C_{{k'}_{\tilde{g}},\tilde{y},\vec{1}}$ has the same functionality as $C^r_{f^*}$ and $\tilde{y}=x^*$. In particular, $\tilde{y}=x^*$ and $\tilde{g}$ is the unique function that satisfies $\tilde{g}(x^*)=1$ and $\tilde{g}(x)=f^*(x)$ for every $x\neq x^*$. 
    
    By $\io$ guarantees, this implies that  for every fixed pair $(f^*,x^*)\in \mathrm{F}^{r,b}_\secparam$, the image under the bijection $\mathcal{B}$, $(\tilde{g},\tilde{y})\in  \mathrm{F}^{r-1,0}_\secparam$, satisfies
    \[\io(C^r_{k_{f^*}}),x^*\approx_c \io(C_{{k'}_{\tilde{g}},\tilde{y},\vec{1}}),y.\]
    Therefore,
    \[\{\io(C^r_{k_f}),x\}_{(k,x)\xleftarrow{\$} \mathrm{F}^{r,1}_\secparam}\approx_c \{\io(C_{{k'}_g,y,\vec{1}}),y\}_{(k',y)=\mathcal{B}(h,z), (h,z)\xleftarrow{\$} \mathrm{F}^{0,k}_\secparam}=\{\io(C_{{k'}_g,y,\vec{1}}),y\}_{(k',y)\xleftarrow{\$} \mathrm{F}^{r-1,0}_\secparam},\]
    where the last equality holds because $\mathcal{B}$ is a bijection.
\end{proof}
\begin{corollary}
    In particular,  assuming post-quantum indistinguishability obfuscation, point functions form a $\presamp$ evasive function class with respect to the uniform distribution, i.e., $(\mathcal{U}_{\fclass^1},\fclass^1)$ is $\presamp$ evasive.
\end{corollary}
\begin{theorem}\label{thm:evasive-functions-from-upo}
Let $\fclass=\{\fclass_\secparam\}_{\secparam\in \NN}$ equipped with a distribution $\distr_\fclass$ be a $\presamp$ evasive function class (see \Cref{def:presamp-evasive}) with input-length $n=n(\secparam)$, and $(\distr,\cktclassf)$ as the  corresponding keyed circuit implementation for the $\presamp$ condition (see \Cref{def:presamp-evasive}).

Assuming a $\distrid$-generalized unclonable puncturable obfuscation $\UPO$ for any generalized puncturable keyed circuit class in $\ppoly$ (see \Cref{subsec:upo-definition}), there 
 is a copy-protection scheme for $\fclass$ that satisfies $(\distr_\fclass,\distriden)$-anti-piracy (see \Cref{sec:def:copyprotection})  with respect to $\cktclassf$ as the keyed circuit implementation of $\fclass$, and $(\distr,\cktclassf)$  as the keyed circuit implementation of $(\distr_\fclass,\fclass)$, 
  where $\copyprotect()$ is the same as $\UPO.\obf()$, and the distribution $\distriden$ on pairs of inputs is as follows:
\begin{itemize}
    \item With probability $\frac{1}{2}$, output $(x^\bob_0,x^\charlie_0)=(x,x)$, where $x\xleftarrow{\$} \{0,1\}^n$.
    \item With probability $\frac{1}{2}$, output $(x^\bob_1,x^\charlie_1)=(x,x)$, where $x\xleftarrow{\$} {C^\fclass_k}^{-1}(1)$, and $C^\fclass_k\in \cktclassf$ is the circuit that is copy-protected.
\end{itemize}
\end{theorem}
\begin{proof}[Proof of \Cref{thm:evasive-functions-from-upo}]
The correctness of the copy-protection scheme follows directly from the correctness of the $\UPO$.



We fix the keyed circuit representation of $(\distr_\fclass,\fclass)$ to be $(\distr,\cktclassf$). Let the keyspace of $\cktclassf$ be $\keyspacef$, i.e., $\cktclassf=\{\{{C^\fclass}_k\}_{k\in \keyspacef_\secparam}\}_{\secparam \in \NN}$. 

Let $\cktclass=\{\{{C}_k\}_{k\in \keyspace_\secparam}\}_{\secparam \in \NN}$ be the auxiliary generalized puncturable keyed circuit class and $\distr'$ be the corresponding distribution on $\keyspace$ with respect to which the $\presamp$ condition (see \Cref{def:presamp-evasive}) holds for $(\distr_\fclass,\fclass)$ equipped with the keyed circuit description $(\distr,\cktclassf)$. Let $\evaspuncture$ be the generalized puncturing algorithm associated with $\cktclass$.

We give a reduction from the copy-protection security experiment to the generalized unclonable puncturable obfuscation security experiment of $\UPO$ for the generalized puncturable keyed circuit class $\cktclass$ (see \Cref{fig:genupo:expt}).
Let $(\alice,\bob,\charlie)$ be an adversary in the copy-protection security experiment.
We mark the changes in blue.\\

\noindent \underline{$\hybrid_{0}$}: \\
This is the same as the original copy-protection security experiment for the scheme $(\obf,\eval)$.
\begin{itemize}
    \item $\ch$ samples a bit $b\xleftarrow{\$}\{0,1\}$. 
    \item $\ch$ samples $k\gets \distr(1^\secparam)$ $\rho_{k}\gets \UPO.\obf(1^\secparam,{C^\fclass}_{k})$ and sends it to $\alice$.
    \item $\alice$ produces a bipartite state $\sigma_{\bob,\charlie}$.
    \item $\ch$ samples $x_0\xleftarrow{\$} \{0,1\}^n$ and $x_1\xleftarrow{\$} {C^\fclass_k}^{-1}(1)$.
    \item Apply $(\bob(x_{b},\cdot) \otimes \charlie(x_{b},\cdot))(\sigma_{\bob,\charlie})$ to obtain $(b_{\bob},b_{\charlie})$.
    \item Output $1$ if ${C^\fclass_k}(x_b)=b_\bob=b_\charlie$.
\end{itemize}

\noindent \underline{$\hybrid_{1}$}: 
\begin{itemize}
    \item $\ch$ samples a bit $b\xleftarrow{\$}\{0,1\}$. 
    \item $\ch$ samples $k\gets \distr(1^\secparam)$ $\rho_{k}\gets \UPO.\obf(1^\secparam,{C^\fclass}_{k})$ and sends it to $\alice$.
    \item $\alice$ produces a bipartite state $\sigma_{\bob,\charlie}$.
    \item $\ch$ samples $x_0\xleftarrow{\$} \{0,1\}^n$ and $x_1\xleftarrow{\$} \{C^\fclass_{k}\}^{-1}(1)$.
    \item Apply $(\bob(x_{b},\cdot) \otimes \charlie(x_{b},\cdot))(\sigma_{\bob,\charlie})$ to obtain $(b_{\bob},b_{\charlie})$.
    \item Output $1$ if \sout{${C^\fclass_k}(x_b)=b_\bob=b_\charlie$} \cblue{$b=b_\bob=b_\charlie$}.
\end{itemize}
Since $\fclass$ is evasive with respect to $\distr$, with overwhelming probability ${C^\fclass_k}(x_0)=0$. Hence, in the $b=0$ case outputting $1$ if ${C^\fclass_k}(x_0)=b_\bob=b_\charlie$ is indistinguishable from $0=b_\bob=b_\charlie$. Clearly, since $x_1\in {C^\fclass_k}^{-1}(1)$, in the $b=1$ case, ${C^\fclass_k}(x_1)=b_\bob=b_\charlie$ is the same as $1=b_\bob=b_\charlie$. 
Hence, the indistinguishability between $\hybrid_0$ and $\hybrid_1$ holds. \\

\noindent \underline{$\hybrid_{2}$}: 
\begin{itemize}
    \item $\ch$ samples a bit $b\xleftarrow{\$}\{0,1\}$. 
    \item $\ch$ samples \sout{$k\gets \distr(1^\secparam)$} \cblue{$k'\gets \distr'(1^\secparam), y\xleftarrow{\$}\{0,1\}^n$} and generates \sout{$\rho_{k}\gets\UPO.\obf(1^\secparam,{C}_{k})$} \cblue{$\rho_{k',y}\gets \UPO.\obf(1^\secparam,{C}_{k',y})$, where ${C}_{k',y}\gets \evaspuncture(k',y,y,\vec{1},\vec{1})$,} and sends it to $\alice$.
    \item $\alice$ produces a bipartite state $\sigma_{\bob,\charlie}$.
    \item $\ch$ samples $x_0\xleftarrow{\$} \{0,1\}^n$ and \sout{$x_1\xleftarrow{\$} {C^\fclass_k}^{-1}(1)$} \cblue{set $x_1=y$}.
    \item Apply $(\bob(x_{b},\cdot) \otimes \charlie(x_{b},\cdot))(\sigma_{\bob,\charlie})$ to obtain $(b_{\bob},b_{\charlie})$.
    \item Output $1$ if $b=b_\bob=b_\charlie$.
\end{itemize}
The indistinguishability between $\hybrid_1$ and $\hybrid_2$ holds by the $\presamp$ relation (in particular, \Cref{eq:presamp-f-g} for the $b=1$ and $b=0$ cases)  between $\fclass,\distr$ and $\gclass,\distr'$.\\

\noindent \underline{$\hybrid_{3}$}: 
\begin{itemize}
    \item $\ch$ samples a bit $b\xleftarrow{\$}\{0,1\}$. 
    \item $\ch$ samples $k'\gets \distr'(1^\secparam),y\xleftarrow{\$}\{0,1\}^n$ and \sout{generates $\rho_{k',y}\gets \UPO.\obf(1^\secparam,{C}_{k',y})$, where ${C}_{k',y}\gets$ $ \evaspuncture(k',y,y,\vec{1},\vec{1})$,} \cblue{if $b=0$ generates $\rho_{k'}\gets \obf(1^\secparam,{C}_{k'})$ else if $b=1$ generates $\rho_{k',y}\gets \UPO.\obf(1^\secparam,{C}_{k',y})$, where ${C}_{k',y}\gets \evaspuncture(k',y,y,\vec{1},\vec{1})$,} and sends it to $\alice$.
    \item $\alice$ produces a bipartite state $\sigma_{\bob,\charlie}$.
    \item $\ch$ samples $x_0\xleftarrow{\$} \{0,1\}^n$ and set $x_1=y$.
    \item Apply $(\bob(x_{b},\cdot) \otimes \charlie(x_{b},\cdot))(\sigma_{\bob,\charlie})$ to obtain $(b_{\bob},b_{\charlie})$.
    \item Output $1$ if $b=b_\bob=b_\charlie$.
\end{itemize}

The indistinguishability between $\hybrid_2$ and $\hybrid_3$ holds as follows. In the $b=0$ case of $\hybrid_2$, the view of $(\alice,\bob,\charlie)$ only depends on $\UPO.\obf(1^\secparam,{C}_{k',y}),x_0$, but in the $b=0$ case of $\hybrid_3$, the view depends on $\UPO.\obf(1^\secparam,{C}_{k'}),x_0$ where $x_0\xleftarrow{\$}\{0,1\}^n$ is sampled independent of $k'$ and $y$. Hence it is enough to show that 
\begin{equation}\label{eq:upo-implication}
\{\UPO.\obf(1^\secparam,{C}_{k',y})\}_{k'\gets\distr'(1^\secparam),y\xleftarrow{\$}\{0,1\}^n }\approx_c\{\UPO.\obf(1^\secparam,{C}_{k'})\}_{k'\gets\distr'(1^\secparam)},
\end{equation} which is a necessary condition for the generalized $\upo$ security of $\UPO$ (otherwise $\alice$ can itself distinguish between $b=0$ and $b=1$ case in the generalized $\upo$ security experiment given in \Cref{def:newcpsecurity} for the keyed circuitclass $\cktclass$). Therefore, \Cref{eq:upo-implication} holds by the generalized UPO security of $\UPO$ for the circuit class $\cktclass$.

\ \\
\noindent \underline{$\hybrid_{4}$}: 
\begin{itemize}
    \item $\ch$ samples a bit $b\xleftarrow{\$}\{0,1\}$. 
    \item $\ch$ samples $k'\gets \distr'(1^\secparam),y\xleftarrow{\$}\{0,1\}^n$ and if $b=0$ generates $\rho_{k'}\gets \obf(1^\secparam,{C}_{k'})$ else if $b=1$ generates $\rho_{k',y}\gets \UPO.\obf(1^\secparam,{C}_{k',y})$, where ${C}_{k',y}\gets \evaspuncture(k',y,y,\vec{1},\vec{1})$, and sends it to $\alice$.
    \item $\alice$ produces a bipartite state $\sigma_{\bob,\charlie}$.
    \item \sout{$\ch$ samples $x_0\xleftarrow{\$} \{0,1\}^n$ and set $x_1=y$.}
    \item Apply $(\bob($\sout{$x_b$}\cblue{$y$}$,\cdot) \otimes \charlie($\sout{$x_b$}\cblue{$y$}$,\cdot))(\sigma_{\bob,\charlie})$ to obtain $(b_{\bob},b_{\charlie})$.
    \item Output $1$ if $b=b_\bob=b_\charlie$.
\end{itemize}
The only change from $\hybrid_3$ to $\hybrid_4$ is replacing $x_0$ with $y$ in the $b=0$ case and $x_1$ with $y$ in the $b=1$ case.
The indistinguishability between $\hybrid_3$ and $\hybrid_4$ holds as follows. Note that replacing $x_1$ with $y$ in $\hybrid_3$ does not change anything since $x_1$ was set to $y$ in $\hybrid_3$. Next, in the $b=1$ case, the view of $(\alice,\bob,\charlie)$ only depends on $\UPO.\obf(1^\secparam,{C}_{k'}),x_0$ where $x_0\xleftarrow{\$}\{0,1\}^n$ is sampled independent of $k'$. Since $y\xleftarrow{\$}\{0,1\}^n$ is also sampled independent of $k'$,
\[\{{C}_{k'},x_0\}_{k'\gets\distr'(1^\secparam), x_0\xleftarrow{\$}\{0,1\}^n}= \{{C}_{k'},y\}_{k'\gets\distr'(1^\secparam),y\xleftarrow{\$}\{0,1\}^n}.\]
Hence, 
\begin{align*}
&\{\UPO.\obf(1^\secparam,{C}_{k'}),x_0\}_{k'\gets\distr'(1^\secparam),x_0\xleftarrow{\$}\{0,1\}^n}\\
&= \{\UPO.\obf(1^\secparam,{C}_{k'}),y\}_{k'\gets\distr'(1^\secparam),y\xleftarrow{\$}\{0,1\}^n}.
\end{align*}
Therefore, replacing $\UPO.\obf(1^\secparam,{C}_{k'}),x_0$ with $\UPO.\obf(1^\secparam,{C}_{k'}),y$ is indistinguishable and hence, $\hybrid_3$ and $\hybrid_4$ are indistinguishable with respect to the adversary.

We next give a reduction $(\reduc_\alice,\reduc_\bob,\reduc_\charlie)$ from $\hybrid_4$ to the $\distrid$-generalized UPO security experiment of  
$\UPO$ (\Cref{def:newcpsecurity}) for the generalized puncturable keyed circuitclass $\cktclass=\{\{{C}_{k'}\}_{k'\in \keyspace_\secparam}\}_\secparam$ equipped with $\evaspuncture$ as the generalized puncturing algorithm (see \Cref{subsec:sde}).

\begin{itemize}
    \item $\reduc_\alice$ samples $k'\gets \distr'(1^\secparam)$, and sends $k'$ along with $\mu_\bob=\mu_\charlie=\vec{1}$, the constant $1$ function.
    \item On receiving $\rho$ from $\ch$, the challenger for the generalized $\upo$ experiment, $\reduc_\alice$ runs $\alice(\rho)$ to get a bipartite state $\sigma_{\bob,\charlie}$, and sends $\sigma_{\bob},\sigma_{\charlie}$ to $\reduc_\bob$ and $\reduc_\charlie$ respectively.
    \item $\reduc_\bob$ (respectively, $\reduc_\charlie$) runs $\bob(x_\bob,\sigma_\bob)$ (respectively, $\charlie(x_\charlie,\sigma_\charlie)$) on receiving $x_\bob$ and $\sigma_\bob$ (respectively $x_\charlie$ and $\sigma_\charlie$) from $\ch$ and $\reduc_\alice$, respectively, and output the outcome.
\end{itemize}

Clearly, the view of $(\alice,\bob,\charlie)$ in the experiment (\Cref{fig:genupo:expt}) 
$\genupoexpt^{\left(\reduc_\alice,\reduc_\bob,\reduc_\charlie \right),\distrid,\cktclass}\left( 1^{\secparam},0 \right)$

(respectively, $\genupoexpt^{\left(\reduc_\alice,\reduc_\bob,\reduc_\charlie \right),\distrid,\cktclass}\left( 1^{\secparam},1 \right)$) is exactly the same as that in the $b=0$ (respectively, $b=1$) case in $\hybrid_4$, where $\distrid$ is as defined in \Cref{subsec:upo-definition}. This completes the reduction from the copy-protection security experiment to the generalized $\upo$ security experiment (\Cref{fig:genupo:expt}).
\end{proof}

\begin{corollary}\label{cor:copy-protect-fixed-length-functions-from-UPO+io}
Suppose $r$ is such that the following holds:
    \begin{enumerate}
        \item $\fclass^r$ is evasive with respect to $\distrunr$, the uniform distribution.
        \item There exists a keyed circuit implementation $(\distr^r,\cktclass^r)$ for $(\distrunr,\fclass^r)$, and similarly keyed circuit implementation $(\distr^{r-1},\cktclass^{r-1})$ for $(\distrunrminus,\fclass^{r-1})$. 
    \end{enumerate}
Then, assuming post-quantum indistinguishability obfuscation, a $\distrid$-generalized unclonable puncturable obfuscation $\UPO$ for any generalized puncturable keyed circuit class in $\ppoly$ (see \Cref{subsec:upo-definition}), there 
 is a copy-protection scheme for $\fclass^r$ that satisfies $(\distrunr,\distriden)$-anti-piracy (see \Cref{sec:def:copyprotection}) with respect to some keyed circuit implementation $(\distr,\cktclass)$ of $(\distrunr,\fclass)$, 
  where $\copyprotect()$ is the same as $\UPO.\obf()$, and the distribution $\distriden$ on pairs of inputs is as follows:
\begin{itemize}
    \item With probability $\frac{1}{2}$, output $(x^\bob_0,x^\charlie_0)=(x,x)$, where $x\xleftarrow{\$} \{0,1\}^n$.
    \item With probability $\frac{1}{2}$, output $(x^\bob_1,x^\charlie_1)=(x,x)$, where $x\xleftarrow{\$} {C_k}^{-1}(1)$, and $C_k\in \cktclass$ is the circuit that is copy-protected.
\end{itemize}
    

In particular, there exists a copy-protection for point functions that satisfies $(\cal{U},\distriden)$-anti-piracy, under the assumptions made above.
\end{corollary}

Combined with \Cref{thm:strong-CLLZ-cp-prf_puncturable-CP-f-uniform-id}, \Cref{cor:copy-protect-fixed-length-functions-from-UPO+io} gives us the following feasibility result for a generalization of point functions, namely, single bit output evasive function classes that consist of functions with a fixed number of preimages of $1$ (see the formal definition in \Cref{thm:evasive-class-instantiations}).

\begin{corollary}\label{cor:copy-protect-fixed-length-functions-from-concrete-assumptions}
    Suppose $r$ is such that the following holds:
    \begin{enumerate}
        \item $\fclass^r$ is evasive with respect to $\distrunr$, the uniform distribution.
        \item There exists a keyed circuit implementation $(\distr^r,\cktclass^r)$ for $(\distrunr,\fclass^r)$, and similarly keyed circuit implementation $(\distr^{r-1},\cktclass^{r-1})$ for $(\distrunrminus,\fclass^{r-1})$. 
    \end{enumerate}

Then, assuming \Cref{conj:goldreich-levin-identical},  the existence of post-quantum sub-exponentially secure $\io$ and one-way functions, and the quantum hardness of Learning-with-errors problem (LWE), there 
 is a copy-protection scheme for $\fclass^r$ that satisfies $(\distrunr,\distriden)$-anti-piracy (see \Cref{sec:def:copyprotection}) with respect to some keyed circuit implementation $(\distr,\cktclass)$ of $(\distrunr,\fclass)$, 
  where $\copyprotect()$ is the same as $\UPO.\obf()$, and the distribution $\distriden$ on pairs of inputs is as defined in \Cref{cor:copy-protect-fixed-length-functions-from-UPO+io}.  
In particular, there exists a copy-protection for point functions that satisfies $(\cal{U},\distriden)$-anti-piracy, under the assumptions made above.
\end{corollary}

%% file: ucdefs.tex
\newcommand{\pirateexp}[4]{{\sf PirExp}^{#1,#2}_{#3,#4}}
\section{Unclonable Cryptography: Definitions}
\subsection{Quantum Copy-Protection}
\label{sec:def:copyprotection}
   Consider a function class $\fclass$ with keyed circuit implementation $\cktclass=\{\cktclass_\secparam\}_{\secparam \in \mathbb{N}}$, where $\fclass_{\secparam}$ (respectively, $\cktclass_\secparam$) consists of functions (respectively, circuits) with input length $n(\secparam)$ and output length $m(\secparam)$. A copy-protection scheme is a pair of QPT algorithms $(\copyprotect,\eval)$ defined as follows: 
    \begin{itemize}
    \item $\copyprotect(1^{\secparam},C)$: on input a security parameter $\secparam$ and a circuit $C \in \cktclass_{\secparam}$, it outputs a quantum state $\rho_C$. 
    \item $\eval(\rho_k,x)$: on input a quantum state $\rho_C$ and an input $x \in \inpclass_{\secparam}$, it outputs $(\rho'_C,y)$.
    \end{itemize}
    
    \paragraph{Correctness.} A copy-protection scheme  $(\copyprotect,\eval)$ for a function class $\fclass$ with keyed circuit implementation $\cktclass=\{\cktclass_{\secparam}\}_{\secparam \in \mathbb{N}}$ is $\delta$-correct, if for every $C \in \cktclass_{\secparam}$, for every $x\in \{0,1\}^{n(\secparam)}$, there exists a negligible function $\delta(\secparam)$ such that: 
$$ \Pr \left[ C(x)=y\mid \substack{\rho_C \leftarrow \copyprotect(1^{\secparam},C)\\ \ \\ (\rho'_C,y) 
 \gets \eval(\rho_C,x)} \right] \geq 1-\delta(\secparam)$$
   
\begin{savenotes}
\begin{figure}[!htb]
\begin{center} 
\begin{tabular}{|p{12cm}|}
    \hline 
\begin{center}
\underline{$\cpexpt^{\left(\alice,\bob,\charlie \right),\distr_\keyspace,\distr_{{\cal X}}}\left( 1^{\secparam} \right)$}: 
\end{center}
\begin{itemize}
\item $\ch$ samples $k \gets \distr_\keyspace(1^{\secparam})$ and generates $\rho_k \gets \copyprotect(1^{\secparam},C_k)$ and sends $\rho_k$ to $\alice$.
\item $\alice$ produces a bipartite state $\sigma_{\bob,\charlie}$.
\item $\ch$ samples $(x^\bob,x^\charlie) \leftarrow \distr_{{\cal X}}$\footnote{$\distr_{{\cal X}}$ may potentially depend on the circuit $C_k$.}.
\item Apply $(\bob(x^{\bob},\cdot) \otimes \charlie(x^{\charlie},\cdot))(\sigma_{\bob,\charlie})$ to obtain $(y^{\bob},y^{\charlie})$. 
\item Output $1$ if $y^\bob=C(x^{\bob})$ and $y^\charlie=C_k(x^{\charlie})$, else $0$. 
\end{itemize}
\ \\ 
\hline
\end{tabular}
    \caption{$(\distr_\keyspace,\distr_{{\cal X}})$-anti-piracy experiment of copy-protection.}
    \label{fig:product-uniform-search-anti-piracy-puncturable-functions}
    \end{center}
\end{figure}
\end{savenotes}
\paragraph{$(\distr_\keyspace,\distr_{{\cal X}})$-anti-piracy.} Consider the experiment in~\Cref{fig:product-uniform-search-anti-piracy-puncturable-functions}. We define $p_{{\sf triv}}=\max\{p_{{\sf B}},p_{{\sf C}}\}$, where $p_{{\sf B}}$ is the maximum probability that the experiment outputs 1 when $\alice$ gives $\rho_C$ to $\bob$ and $\charlie$ outputs its best guess and $p_{{\sf \charlie}}$ is defined symmetrically. We refer to~\cite{AKL23} for a formal definition of trivial success probability. 

\par Suppose  $\distr_{{\cal X}}$ is a distribution on $\{0,1\}^{n(\secparam)} \times \{0,1\}^{n(\secparam)}$, and $\distr_\fclass$ is a distribution on $\fclass$. 

We say that a copy-protection scheme  $(\copyprotect,\eval)$ for $\fclass$ satisfies $(\distr_{\fclass},\distr_{{\cal X}})$-anti-piracy if there exists a keyed circuit implementation (see \Cref{subsec:notation-applications}) of the form $(\distr_\keyspace,\cktclass)$\footnote{It is crucial that $\cktclass$ is the same circuit class as the keyed implementation of $\fclass$ that we fixed before for correctness.} for $(\distr_\fclass,\fclass)$ such that
for every tuple of QPT adversaries $(\alice,\bob,\charlie)$ there exists a negligible function $\negl(\secparam)$ such that: 
\[ \prob\left[ 1 \leftarrow \cpexpt^{\left(\alice,\bob,\charlie \right),\distr_\keyspace,\distr_{{\cal X}}}\left( 1^{\secparam} \right)\ \right] \leq p_{{\sf triv}} + \negl(\secparam)\]
If $\distr_{{\cal X}}$ is a uniform distribution on $\{0,1\}^{n(\secparam)} \times \{0,1\}^{n(\secparam)}$ then we simply refer to this definition as $\distr_\keyspace$-anti-piracy.

\subsection{Public-Key Single-Decryptor Encryption}
\label{sec:pksde:def}
\noindent We adopt the following definition of public-key single-decryptor encryption from~\cite{CLLZ21}. \anote{Should we add GZ20, SW22 as well?}
 \par A public-key single-decryptor encryption scheme with message length $n(\secparam)$ and ciphertext length $c(\secparam)$ consists of the QPT algorithms $\SDE=(\gen,\qkeygen,\enc,\dec)$ defined below:
    \begin{itemize}
    \item $(\sk,\pk)\gets \gen(1^\lambda):$ on input a security parameter $1^\lambda$, returns a classical secret key $\sk$ and a classical public key $\pk$.
    \item $\rho_{\sk}\gets \qkeygen(\sk):$ takes a classical secret key $\sk$ and outputs a quantum decryptor key $\rho_\sk$. 
    \item $\ct\gets \enc(\pk,m)$ takes a classical public key $\pk$, a message $m\in \{0,1\}^n$ and outputs a classical ciphertext $\ct$.
    \item $m\gets \dec(\rho_\sk,\ct):$ takes a quantum decryptor key $\rho_\sk$ and a ciphertext $\ct$, and outputs a message $m\in \{0,1\}^n$. 
\end{itemize}

\paragraph*{Correctness} For every message $m\in\{0,1\}^{n(\secparam)}$, there exists a negligible function $\delta(\secparam)$ such that:
\[\prob\left[ \dec(\rho_\sk,\ct)=m\ \left| \ \substack{(\sk,\pk) \gets\gen(\secparam)\\ \ \\ \rho_\sk\gets \qkeygen(\sk)\\ \ \\ \ct\gets \enc(\pk,m)} \right. \right]\geq 1-\delta(\secparam).\]

\paragraph{Search Anti-Piracy}

\begin{figure}[!htb]
\begin{center} 
\begin{tabular}{|p{12cm}|}
    \hline 
\begin{center}
\underline{$\srchsdeexpt^{\left(\alice,\bob,\charlie \right),\distr}\left( 1^{\secparam} \right)$}: 
\end{center}
\begin{itemize}
\item $\ch$ samples $(\sk,\pk)  \gets \gen(1^{\secparam})$. It then generates $\rho_\sk\gets \qkeygen(\sk)$ and sends $(\rho_\sk,\pk)$ to $\alice$.
\item $\alice$ produces a bipartite state $\sigma_{\bob,\charlie}$.
\item $\ch$ samples $(m^\bob,m^\charlie)\gets \distr(1^\secparam)$ and generates $\ct^\bob\gets \enc(\pk,m^\bob)$ and $\ct^\charlie\gets \enc(\pk,m^\charlie)$.
\item Apply $(\bob(\ct^{\bob},\cdot) \otimes \charlie(\ct^{\charlie},\cdot))(\sigma_{\bob,\charlie})$ to obtain $(y^{\bob},y^{\charlie})$. 
\item Output $1$ if $y^\bob=m^\bob$ and $y^\charlie=m^\charlie$.
\end{itemize}
\ \\ 
\hline
\end{tabular}
    \caption{Search anti-piracy.}
    \label{fig:product-uniform-sde-random-challenge-anti-piracy}
    \end{center}
\end{figure}

We say that a single-decryptor encryption scheme $\SDE$ satisfies $\distr$-search anti-piracy if for every QPT  adversary $(\alice,\bob,\charlie)$ in~\Cref{fig:product-uniform-sde-random-challenge-anti-piracy} if there exists a negligible function $\negl$ such that:
\[ \prob\left[ 1 \leftarrow \srchsdeexpt^{\left(\alice,\bob,\charlie \right)}\left( 1^{\secparam} \right)\ \right] \leq \negl(\secparam).\]

The two instantiations of $\distr$ are $\distrprod$ and $\distrid$, as defined in \cref{subsec:upo-definition}.



\paragraph{Indistinguishability from random Anti-Piracy}

\begin{figure}[!htb]
\begin{center} 
\begin{tabular}{|p{12cm}|}
    \hline 
\begin{center}
\underline{$\indrsdeexpt^{\left(\alice,\bob,\charlie \right),\distr_\ct}\left( 1^{\secparam} \right)$}: 
\end{center}
\begin{itemize}
\item $\ch$ samples $(\sk,\pk)  \gets \gen(1^{\secparam})$. It then generates $\rho_\sk\gets \qkeygen(\sk)$ and sends $(\rho_\sk,\pk)$ to $\alice$.
\item $\alice$ produces a bipartite state $\sigma_{\bob,\charlie}$.
\item $\ch$ samples $b\xleftarrow{\$}\{0,1\}$, and generates $(\ct^\bob_b,\ct^\charlie_b)\gets \distr_\ct(1^\secparam,b,\pk)$. 
\item Apply $(\bob(\ct^{\bob}_b,\cdot) \otimes \charlie(\ct^{\charlie}_b,\cdot))(\sigma_{\bob,\charlie})$ to obtain $(b^{\bob},b^{\charlie})$. 
\item Output $1$ if $b^\bob=b^\charlie=b$.
\end{itemize}
\ \\ 
\hline
\end{tabular}
    \caption{Indistinguishability from random anti-piracy.}
    \label{fig:indistinguishability_from_random_-sde-anti-piracy}
    \end{center}
\end{figure}

We say that a single-decryptor encryption scheme $\SDE$ satisfies $\distr_\ct$-indistinguishability from random anti-piracy if for every QPT  adversary $(\alice,\bob,\charlie)$ in~\Cref{fig:product-uniform-sde-random-challenge-anti-piracy} if there exists a negligible function $\negl$ such that:
\[ \prob\left[ 1 \leftarrow \indrsdeexpt^{\left(\alice,\bob,\charlie \right),\distr_\ct}\left( 1^{\secparam} \right)\ \right] \leq \negl(\secparam).\]

The two instantiations of $\distr_\ct$ are as follows:
\begin{enumerate}
    \item $\distrcipherind(1^\secparam,b,\pk)$: 
    \begin{enumerate}
        \item  Sample $m^\bob,m^\charlie\xleftarrow{\$}\{0,1\}^q$, where $q(\secparam)$ is the message length.
        \item Generate $\ct^\bob_b\gets \enc(\pk,m^\bob_b)$ and $\ct^\charlie_b\gets \enc(\pk,m^\charlie_b)$, where $m^\bob_0=m^\charlie_0=0$, $m^\bob_1=m^\bob$ and $m^\charlie_1=m^\charlie$.
        \item Output $\ct^\bob_b,\ct^\charlie_b$.
    \end{enumerate}
    \item $\distrcipherid(1^\secparam,b,\pk)$:
    \begin{enumerate}
        \item Sample $m\xleftarrow{\$}\{0,1\}^q$, where $q(\secparam)$ is the message length.
        \item Generate $\ct_b\gets \enc(\pk,m_b)$ where $m_0=0$, and $m_1=m$.
        \item Set $\ct^\bob_b=\ct^\charlie_b=\ct_b$.
        \item Output $\ct^\bob_b,\ct^\charlie_b$.
    \end{enumerate}
\end{enumerate}

\paragraph{Selective CPA Anti-piracy}

\begin{figure}[!htb]
\begin{center} 
\begin{tabular}{|p{12cm}|}
    \hline 
\begin{center}
\underline{$\selcpasdeexpt^{\left(\alice,\bob,\charlie \right),\distr}\left( 1^{\secparam} \right)$}: 
\end{center}
\begin{enumerate}
\item $\adversary(\rho_k)$ outputs $(m_0^\bob,m_1^\bob,m_0^\charlie,m_1^\charlie)$, such that $|m_0^{\bob}|=|m_{1}^{\bob}|$ and $|m_0^{\charlie}|=|m_1^{\charlie}|$.
\item $\ch$ samples $(\rho_k,\pk) \gets \keygen(1^\secparam)$ and sends $(\rho_{k},\pk)$ to $\alice$. 
\item $\adversary(\rho_k)$ outputs a bipartite state $\sigma_{\bob,\charlie}$.
\item $\ch$ samples 
 $b\xleftarrow{\$} \{0,1\}$.
\item Let $\ct^\bob,\ct^\charlie\gets \distr(1^\secparam,b,\pk)$. 

\item Apply $(\bob(\ct^\bob,\cdot) \otimes \charlie(\ct^\charlie,\cdot))(\sigma_{\bob,\charlie})$ to obtain $(b_{\bfB},b_{\bfC})$. 
\item Output $1$ if $b_{\bfB}=b_\bfC=b$.
\end{enumerate}
\ \\ 
\hline
\end{tabular}
    \caption{Selective $\distr$-CPA anti-piracy.}
    \label{fig:correlated-sde-cpa-style-anti-piracy}
    \end{center}
\end{figure}

We say that a single-decryptor encryption scheme $\SDE$ satisfies $\distr$-selective $\cpa$ anti-piracy, for a distribution $\distr$ on $\{0,1\}^n \times \{0,1\}^n$, if for every QPT adversary $(\alice,\bob,\charlie)$ in~\Cref{fig:correlated-sde-cpa-style-anti-piracy}, there exists a negligible function $\negl$ such that:
\[ \prob\left[ 1 \leftarrow \selcpasdeexpt^{\left(\alice,\bob,\charlie \right),\distr}\left( 1^{\secparam} \right)\ \right] \leq \frac{1}{2}+\negl(\secparam).\]

The two instantiations of $\distr$ are:
\begin{enumerate}
    \item $\distrcor(1^\secparam,b,\pk)$: outputs $(\ct^\bob,\ct^\charlie)$ where $\ct^\bob\gets \enc(\pk,m^\bob_b)$ and $\ct^\charlie\gets \enc(\pk,m^\charlie_b)$.
    \item $\distriden(1^\secparam,b,\pk)$ outputs $(\ct,\ct)$ where $\ct\gets \enc(\pk,m^\bob_b)$\footnote{Ideally, in the identical challenge setting, there should be just two challenge messages $m_0,m_1$ and not $m^\bob_0,m^\bob_1,m^\charlie_0,m^\charlie_1$, but we chose to have this redundancy in order to unify the syntax for the identical and correlated challenge settings.}.
\end{enumerate}



\noindent This notion of selective $\distriden$-$\cpa$ security is equivalent to the selective $\cpa$-security in~\cite{GZ20}.
\paragraph{CPA anti-piracy}

\begin{figure}[!htb]
\begin{center} 
\begin{tabular}{|p{12cm}|}
    \hline 
\begin{center}
\underline{$\cpasdeexpt^{\left(\alice,\bob,\charlie \right),\distr}\left( 1^{\secparam} \right)$}: 
\end{center}
\begin{itemize}
\item $\ch$ samples $(\sk,\pk) \gets \gen(1^{\secparam})$ and generates $\rho_\sk\gets \qkeygen(\sk)$ and sends $(\rho_\sk,\pk)$ to $\alice$.
\item $\alice$ sends two pairs of same-length messages $((m^\bob_0,m^\bob_1),(m^\charlie_0,m^\charlie_1))$.
\item $\alice$ produces a bipartite state $\sigma_{\bob,\charlie}$.
\item $\ch$ samples 
 $b\xleftarrow{\$} \{0,1\}$.
\item Let $\ct^\bob,\ct^\charlie\gets \distr(1^\secparam,b,\pk)$. 
\item Apply $(\bob(\ct^{\bob},\cdot) \otimes \charlie(\ct^{\charlie},\cdot))(\sigma_{\bob,\charlie})$ to obtain $(b^{\bob},b^{\charlie})$. 
\item Output $1$ if $b^\bob=b_0$ and $b^\charlie=b_1$.
\end{itemize}
\ \\ 
\hline
\end{tabular}
    \caption{$\distr$-CPA anti-piracy}
    \label{fig:correlated-sde-full-blown-cpa-style-anti-piracy}
    \end{center}
\end{figure}

We say that a single-decryptor encryption scheme $\SDE$ satisfies $\cpa$ $\distr$-anti-piracy if for every QPT adversary $(\alice,\bob,\charlie)$ in Experiment~\ref{fig:correlated-sde-full-blown-cpa-style-anti-piracy}, there exists a negligible function $\negl$ such that
\[ \prob\left[ 1 \leftarrow \cpasdeexpt^{\left(\alice,\bob,\charlie \right),\distr}\left( 1^{\secparam} \right)\ \right] \leq \frac{1}{2}+\negl(\secparam).\]

The two instantiations of $\distr$ are $\distrcor$ and $\distriden$, defined in the selective $\cpa$ anti-piracy definition in the previous paragraph.




\noindent The definition of $\distrcor$-$\cpa$ anti-piracy is the same as the correlated version of the $1$-$2$ variant of ${\sf UD}-\cpa$ anti-piracy defined in~\cite{SW22} and the definition $\distrid$-$\cpa$ anti-piracy  is the same as the secret-key $\cpa$ secure defined in~\cite{GZ20}.
\label{subsec:sde}
\subsection{Unclonable Encryption}  \label{subsec:uenc}
An unclonable encryption scheme is a triple of QPT algorithms $\UE=(\gen,\enc,\dec)$ given below:
\begin{itemize}
    \item $\gen(1^\lambda):\sk$ on input a security parameter $1^\lambda$, returns a classical key $\sk$.
    \item $\enc(\sk,m):\rho_{ct}$ takes the key $\sk$, a message $m \in \{0,1\}^{n(\secparam)}$ and outputs a quantum ciphertext $\rho_{ct}$.
    \item $\dec(\sk,\rho_{ct}):\rho_m$ takes a secret key $\sk$, a quantum ciphertext $\rho_{ct}$ and outputs a message $m'$.
\end{itemize}

\paragraph{Correctness.} The following must hold for the encryption scheme. For every $m \in \{0,1\}^{n(\secparam)}$, the following holds: 
$$\prob\left[ m \leftarrow \dec(\sk,\rho_{ct})\ \left| \ \substack{\sk \leftarrow \gen(1^{\secparam})\\ 
\ \\ \rho_{ct} \leftarrow \enc(\sk,m)} \right. \right] \geq 1 - \negl(\secparam) $$

\paragraph{CPA security.} We say that an unclonable encryption scheme $\UE$ satisfies CPA security if for every QPT adversary $(\alice,\bob,\charlie)$, there exists a negligible function $\negl$ such that
\[ \prob\left[ 1 \leftarrow \ueexpt^{\left(\alice,\bob,\charlie \right)}\left( 1^{\secparam} \right)\ \right] \leq \frac{1}{2}+\negl(\secparam).\]

\begin{figure}[!htb]
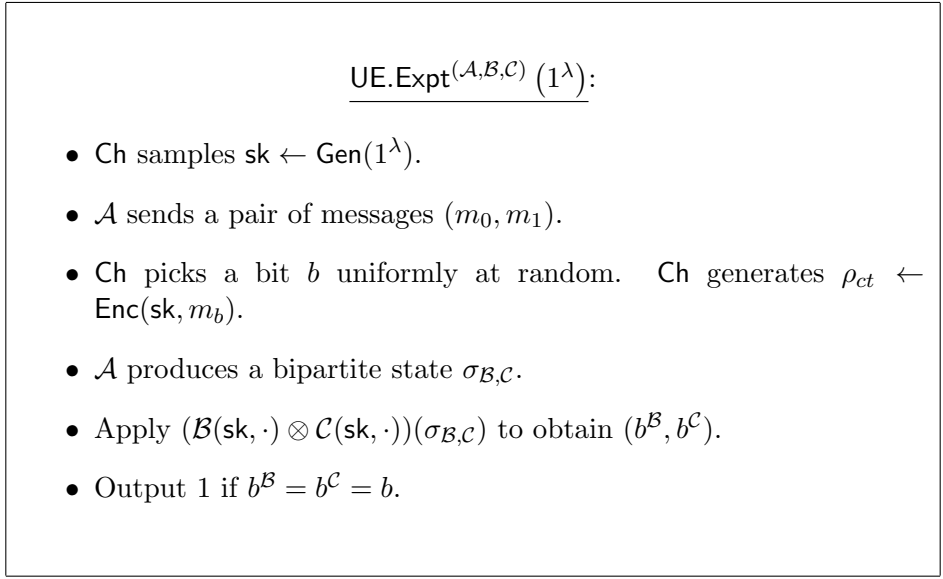

\begin{center} 
\begin{tabular}{|p{12cm}|}
    \hline 
\begin{center}
\underline{$\ueexpt^{\left(\alice,\bob,\charlie \right)}\left( 1^{\secparam} \right)$}: 
\end{center}
\begin{itemize}
\item $\ch$ samples $\sk \gets \gen(1^{\secparam})$.
\item $\alice$ sends a pair of messages $(m_0,m_1)$. 
\item $\ch$ picks a bit $b$ uniformly at random. $\ch$ generates $\rho_{ct} \leftarrow \enc(\sk,m_b)$. 
\item $\alice$ produces a bipartite state $\sigma_{\bob,\charlie}$.
\item Apply $(\bob(\sk,\cdot) \otimes \charlie(\sk,\cdot))(\sigma_{\bob,\charlie})$ to obtain $(b^{\bob},b^{\charlie})$. 
\item Output $1$ if $b^\bob=b^\charlie=b$.
\end{itemize}
\ \\ 
\hline
\end{tabular}
    \caption{CPA security}
    \label{fig:uesecurity}
    \end{center}
\end{figure}

%% file: relatedwork.tex
\section{Related Work}
\label{sec:relatedwork}
Unclonable cryptography is an emerging area in quantum cryptography. The origins of this area date back to 1980s when Weisner~\cite{Wiesner83} first conceived the idea of quantum money which leverages the no-cloning principle to design money states that cannot be counterfeited. Designing quantum money has been an active and an important research direction~\cite{Aar09,AC12,Zha17,Shm22,LMZ23,Zha23}. Since the inception of quantum money, there have been numerous unclonable primitives proposed and studied. We briefly discuss the most relevant ones to our work below. 

\paragraph{Copy-Protection.} Aaronson~\cite{Aar09} conceived the notion of quantum copy-protection. Roughly speaking, in a quantum copy-protection scheme, a quantum state is associated with functionality such that given this state, we can still evaluate the functionality while on the other hand, it should be hard to replicate this state and send this to many parties. Understanding the feasibility of copy-protection for unlearnable functionalities has been an intriguing direction. Copy-protecting arbitrary unlearnable functions is known to be impossible in the plain model~\cite{AL20} assuming cryptographic assumptions. Even in the quantum random oracle model, the existence of a restricted class of copy-protection schemes have been ruled out~\cite{AK22}. This was complemented by~\cite{ALLZZ20} who showed that any class of unlearnable functions can be copy-protected in the presence of a classical oracle. The breakthrough work of~\cite{CLLZ21} showed for the first time that copy-protection for interesting classes of unlearnable functions exists in the plain model. This was followed by the work of~\cite{LLQ+22} who identified some watermarkable functions that can be copy-protected. Notably, both~\cite{CLLZ21} and~\cite{LLQ+22} only focus on copy-protecting specific functionalities whereas we identify a broader class of functionalities that can be copy-protected. Finally, a recent work~\cite{CHV23} shows how to copy-protect point functions in the plain model. The same work also shows how to de-quantize communication in copy-protection schemes. 

\paragraph{Unclonable and Single-Decryptor Encryption.} Associating encryption schemes with unclonability properties were studied in the works of~\cite{BL20,BI20,GZ20}. In an encryption scheme, either we can protect the decryption key or the ciphertext from being cloned, resulting in two different notions. 
\par In an unclonable encryption scheme, introduced by~\cite{BL20}, given one copy of a ciphertext, it should be infeasible to produce many copies of the ciphertext. There are two ways to formalize the security of an unclonable encryption scheme. Roughly speaking, search security is defined as follows: if the adversary can produce two copies from one copy then it should be infeasible for two non-communicating adversaries $\bob$ and $\charlie$, who receive a copy each, to simultaneously recover the entire message. Specifically, the security notion does not prevent both $\bob$ and $\charlie$ from learning a few bits of the message. On the other hand, indistinguishability security is a stronger notion that disallows $\bob$ and $\charlie$ to simultaneously determine which of $m_0$ or $m_1$, for two adversarially chosen messages $(m_0,m_1)$, were encrypted.~\cite{BL20} showed that unclonable encryption with search security for long messages exists. Achieving indistinguishability security in the plain model has been left as an important open problem. A couple of recent works~\cite{AKLLZ22,AKL23} shows how to achieve indistinguishability security in the quantum random oracle model. Both~\cite{AKLLZ22,AKL23} achieve unclonable encryption in the one-time secret-key setting and this can be upgraded to a public-key scheme using the compiler of~\cite{AK21}. 
\par In a single-decryptor encryption scheme, introduced by~\cite{GZ20}, the decryption key is associated with a quantum state such that given this quantum state, we can still perform decryption but on the other hand, it should be infeasible for an adversary who receives one copy of the state to produce two states, each given to $\bob$ and $\charlie$, such that $\bob$ and $\charlie$ independently have the ability to decrypt. As before, we can consider both search and indistinguishability security; for the rest of the discussion, we focus on indistinguishability security.~\cite{CLLZ21} first constructed single-decryptor encryption in the public-key setting assuming indistinguishability obfuscation (iO) and learning with errors. Recent works~\cite{AKL23} and~\cite{KN23} present information-theoretic constructions and constructions based on learning with errors in the one-time setting. The challenge distribution in the security of single-decryptor encryption is an important parameter to consider. In the security experiment, $\bob$ and $\charlie$ each respectively receive ciphertexts $\ct_{\bob}$ and $\ct_{\charlie}$, where $(\ct_{\bob},\ct_{\charlie})$ is drawn from a distribution referred to as challenge distribution. Most of the existing results focus on the setting when the challenge distribution is a product distribution, referred to as independent challenge distribution. Typically, achieving independent challenge distribution is easier than achieving identical distribution, which corresponds to the case when both $\bob$ and $\charlie$ receive as input the same ciphertext. Indeed, there is a reason for this: single-decryptor encryption with security against identical challenge distribution implies unclonable encryption. In this work, we show how to achieve public-key single-decryptor encryption under identical challenge distribution. 



%% file: additionalprelims.tex
\section{Additional Preliminaries}

\subsection{Indistinguishability Obfuscation (IO)}
\label{sec:pre:iO}
An obfuscation scheme associated with a class of circuit ${\cal C} = \{{\cal C}_\secparam\}_{\secparam \in \mathbb{N}}$ consists of two probabilistic polynomial-time algorithms $\iO = (\obf, \eval)$ defined below. 
\begin{itemize}
\item {\bf Obfuscate}, $C' \leftarrow \obf(1^\secparam, C)$: takes as input security parameter $\secparam$, a circuit $C \in {\cal C}_\secparam$ and outputs an obfuscation of $C$, $C'$. 
\item {\bf Evaluation}, $y \leftarrow \eval(C', x)$: a deterministic algorithm that takes as input an obfuscated circuit $C'$, an input $x \in \bitspace^\secparam$ and outputs $y$. 
\end{itemize}

\begin{definition}[\cite{BGIRSVY01}]
 An obfuscation scheme $\iO = (\obf, \eval)$ is a {\bf post-quantum secure indistinguishability obfuscator} for a class of circuits $\mathcal{C} =
\{\mathcal{C}_{\lambda}\}_{\lambda\in \mathbb{N}}$, with every $C\in \mathcal{C}_{\lambda}$ has size $poly(\lambda)$, if it satisfies the following properties:
\begin{itemize}
    \item  {\bf Perfect correctness:} 
     For every $C : \{0, 1\}^{\lambda} \to \{0, 1\} \in \mathcal{C}_{\lambda}$, $x \in \{0, 1\}^{\lambda}$ it holds that:
     $$
     \Pr\big[\eval\big(\obf(1^{\lambda}, C), x\big) = C(x)\big] = 1 \ .
     $$

     \item {\bf Polynomial Slowdown:}  For every $C : \{0, 1\}^{\lambda} \to \{0, 1\} \in \mathcal{C}_{\lambda}$, we have the running time of $\obf$ on input $(1^{\lambda}, C)$ to be $poly(|C|, \lambda)$. Similarly, we have the running time of $\eval$ on input $(C', x)$ is ${\sf poly}(|C'|, \lambda)$
          \item {\bf Security:}  For 
     every QPT adversary ${\cal A}$, there exists a negligible function $\mu(\cdot)$, such that  for  every sufficiently large $\lambda \in \mathbb{N}$, for every $C_0, C_1 \in \mathcal{C}_{\lambda}$ with $C_0(x) = C_1(x)$ for every $x \in \{0, 1\}^{\lambda}$ and $|C_0| = |C_1|$, we have:
     $$
     \left|\Pr\left[{\cal A}\big(\obf(1^{\lambda},C_0), C_0, C_1 \big)=1\right] - \Pr\left[{\cal A}\big(\obf(1^{\lambda},C_1), C_0, C_1 \big)=1\right]\right| \leq \mu(\lambda) \ .
     $$
\end{itemize}

\end{definition}